\documentclass[letterpaper,11pt]{article}
\usepackage{fullpage}
\usepackage[ansinew]{inputenc}
\usepackage{verbatim}
\usepackage{graphicx}
\usepackage{tabularx}
\usepackage{array}
\usepackage{delarray}
\usepackage{dcolumn}
\usepackage{float}
\usepackage{amsmath}
\usepackage{psfrag}
\usepackage{booktabs}
\usepackage{multirow,hhline}
\usepackage{amstext}
\usepackage{pstricks}

\usepackage{pst-plot}
\usepackage{amssymb}
\usepackage{fancyhdr}
\usepackage{hyperref}
\usepackage{setspace}
\usepackage{amsthm}
\usepackage{dsfont}
\RequirePackage[round,authoryear,longnamesfirst]{natbib}
\usepackage{makeidx}
\usepackage{tikz}
\usetikzlibrary{arrows,calc}
\usepackage{mdframed}
\usepackage{lipsum} 
\usepackage{chapterbib}
\usepackage{url}
\usepackage{adjustbox}
\hypersetup{citecolor=true,colorlinks=true,allcolors=blue}
\makeindex
\allowdisplaybreaks

\numberwithin{equation}{section}
\newtheorem{lem}{Lemma}[section]
\newtheorem{thm}{Theorem}[section]

\newtheorem{ass}{Assumption}

\newtheorem{ex}{Example}
\renewcommand{\citep}[1]{\citeauthor{#1}, \citeyear{#1}}

\newcommand{\diag}{\text{diag}}
\newcommand{\Supp}{\text{Supp}}
\newcommand{\indep}{\perp\!\!\!\perp}

\newcommand{\convP}{\stackrel{p}{\longrightarrow}}
\newcommand{\convD}{\rightsquigarrow}

\newcommand{\N}{\mathcal{N}}
\newcommand{\eps}{\varepsilon}
\renewcommand{\epsilon}{\varepsilon}

\DeclareMathOperator*{\argmin}{arg\,min}
\makeatletter
\newcommand*{\rom}[1]{\expandafter\@slowromancap\romannumeral #1@}
\makeatother
\usepackage{xr}
\title{Quantile Treatment Effects and Bootstrap Inference under Covariate-Adaptive Randomization\thanks{We are grateful to Federico Bugni, Qu Feng, Sukjin Han, Yu-Chin Hsu, Shakeeb Khan, Frank Kleibergen, Michael Leung, Jia Li, Wenjie Wang, and
		seminar participants at NTU, Financial Econometrics and New Finance Conference at Zhejiang University, SH3 Conference on Econometrics at SMU, 2019 Shanghai Workshop of Econometrics, and Asian Meeting of the
		Econometric Society for their valuable comments. We also thank the editor and two anonymous referees for their valuable comments which greatly improve our paper. Zhang acknowledges
		the financial support from Singapore Ministry of Education Tier 2 grant under grant no. MOE2018-T2-2-169 and the Lee Kong Chian fellowship. Any and all errors are our own.  }
	\\ \vspace{2mm}
}
\author{Yichong Zhang\thanks{Singapore Management University.\ E-mail~address: yczhang@smu.edu.sg.} \and  Xin Zheng\thanks{The corresponding author. Singapore Management University. E-mail~address: xin.zheng.2015@phdecons.smu.edu.sg.}
	\date{}
}

\begin{document}
	\renewcommand\texteuro{FIXME} 
	\maketitle
\begin{abstract}
	In this paper, we study the estimation and inference of the quantile treatment effect under covariate-adaptive randomization. We propose two estimation methods: (1) the simple quantile regression and (2) the inverse propensity score weighted quantile regression. For the two estimators, we derive their asymptotic distributions uniformly over a compact set of quantile indexes, and show that, when the treatment assignment rule does not achieve strong balance, the inverse propensity score weighted estimator has a smaller asymptotic variance than the simple quantile regression estimator. For the inference of method (1), we show that the Wald test using a weighted bootstrap standard error under-rejects. But for method (2), its asymptotic size equals the nominal level. We also show that, for both methods, the asymptotic size of the Wald test using a covariate-adaptive bootstrap standard error equals the nominal level. We illustrate the finite sample performance of the new estimation and inference methods using both simulated and real datasets. \\
	
	\bigskip
	
\end{abstract}

\section{Introduction}
The randomized control trial (RCT), as pointed out by \cite{AP08}, is one of the five most common methods (along with instrumental variable regressions, matching estimations, differences-in-differences, and regression discontinuity designs) for causal inference. Researchers can use the RCT to estimate not only average treatment effects (ATEs) but also quantile treatment effects (QTEs), which capture the heterogeneity of the sign and magnitude of treatment effects, varying depending on their place in the overall distribution of outcomes. For example, \cite{MS11} estimate the QTE of teacher performance pay program on student learning via the difference of empirical quantiles of test scores between treatment and control groups. \cite{DGPR13} and \cite{BDGK15} estimate the QTEs of audits on endline pollution and a group-lending microcredit program on informal borrowing, respectively, via linear quantile regressions (QRs). \cite{CDDP15} estimate the QTE of microcredit on various household outcomes via a minimum distance method. \cite{BNM18} estimate the QTE of being informed on energy use via the inverse propensity score weighted (IPW) QR. Except \cite{CDDP15}, the other four papers all use the bootstrap to construct confidence intervals for their QTE estimates. However, RCTs have also been routinely implemented with covariate-adaptive randomization. Individuals are first stratified based on some baseline covariates, and then, within each stratum, the treatment status is assigned (independent of covariates) to achieve some balance between the sizes of treatment and control groups; as examples, see \citet[Chapter 9]{IR05} for a textbook treatment of the topic, and \cite{DGK07} and \cite{B09} for two excellent surveys on implementing RCTs in development economics. To achieve such balance, treatment status for different individuals usually exhibits a (negative) cross-sectional \textit{dependence}. The standard inference procedures that rely on cross-sectional \textit{independence} are usually conservative and lacking power. How do we consistently estimate QTEs under covariate-adaptive randomization? What are the asymptotic distributions for the QTE estimators, and how do we conduct proper bootstrap inference? These questions are as yet unaddressed.  

We propose two ways to estimate QTEs: (1) the simple quantile regression (SQR) and (2) the IPW QR. We establish the weak limits for both estimators uniformly over a compact set of quantile indexes and show that the IPW estimator has a smaller asymptotic variance than the SQR estimator when the treatment assignment rule does not achieve strong balance.\footnote{We will define ``strong balance" in Section \ref{sec:setup}.} If strong balance is achieved, then the two estimators are asymptotically first-order equivalent. For inference, we show that the Wald test combined with weighted bootstrap based critical values can lead to under-rejection for method (1), but its asymptotic size equals the nominal level for method (2). We also study the covariate-adaptive bootstrap which respects the cross-sectional dependence when generating the bootstrap sample. The estimator based on the covariate-adaptive bootstrap sample can mimic that of the original sample in terms of the standard error. Thus, using proper covariate-adaptive bootstrap based critical values, the asymptotic size of the Wald test equals the nominal level for both estimators.  

As originally proposed by \cite{D74}, the QTE, for a fixed quantile index, corresponds to the horizontal difference between the marginal distributions of the potential outcomes for treatment and control groups. \cite{F07} studies the identification and estimation of QTE under unconfoundedness. Our estimators (1) and (2) directly follow those in \cite{D74} and \cite{F07}, respectively.

\cite{SYZ10} first point out that, under covariate-adaptive randomization, the usual two-sample t-test for the ATE is conservative. They then propose a covariate-adaptive bootstrap which can produce the correct standard error. \cite{SY13} extend the results to generalized linear models. However, both groups of researchers parametrize the (transformed) conditional mean equation by a specific linear model and focus on a specific randomization scheme (covariate-adaptive biased coin method). \cite{MQLH18} derive the theoretical properties of ATE estimators based on general covariate-adaptive randomization under the linear model framework. \cite{BCS17} substantially generalize the framework to a fully nonparametric setting with a general class of randomization schemes. However, they mainly focus on the ATE and show that the standard two-sample t-test and the t-test based on the linear regression with strata fixed effects are conservative. They then obtain analytical estimators for the correct standard errors and study the validity of permutation tests. \cite{HHK11} study the IPW estimator for the ATE under adaptive randomization. However, they assume the treatment status is assigned completely independently across individuals. More recently, \cite{BCS18} study the estimation of ATE with multiple treatments and propose a fully saturated estimator. \cite{T18} study the estimation of ATE under an adaptive randomization procedure. 

Our paper complements the above papers in four aspects. First, we consider the estimation and inference of the QTE, which is a function of quantile index $\tau$. We rely on the empirical processes theories developed by \cite{VW96} and \cite{CCK14} to obtain uniformly weak convergence of our estimators over a compact set of $\tau$. Based on the uniform convergence, we can construct not only point-wise but also uniform confidence bands. Second, we study the asymptotic properties of the IPW estimator under covariate-adaptive randomization. When the treatment assignment rule does not achieve strong balance, the IPW estimator is more efficient than the SQR estimator. Third, we investigate the weighted bootstrap approximation to the asymptotic distributions of the SQR and IPW estimators. We show that the weighted bootstrap ignores the (negative) cross-sectional dependence due to the covariate-adaptive randomization and over-estimates the asymptotic variance for the SQR estimator. However, the asymptotic variance for the IPW estimator does not rely on the randomization scheme implemented. Thus, the asymptotic size of the Wald test using the IPW estimator paired with the weighted bootstrap based critical values equals the nominal level. Fourth, we investigate the covariate-adaptive bootstrap approximation to the asymptotic distributions of the SQR and IPW estimators. We establish that, using either estimator paired with its corresponding covariate-adaptive bootstrap based critical values, the asymptotic size of the Wald test  equals the nominal level. \cite{SYZ10} first propose the covariate-adaptive bootstrap and establish its validity for the ATE in a linear regression model under the null hypothesis that the treatment effect is not only zero but also homogeneous.\footnote{We say the average treatment effect is homogeneous if the conditional average treatment effect given covariates is the same as the unconditional one.} We modify the covariate-adaptive bootstrap and establish its validity for the QTE in the nonparametric setting proposed by \cite{BCS17}. In addition, our results do not rely on the homogeneity of the treatment effect. Compared with the analytical inference, the two bootstrap inferences for QTEs we study in this paper avoid estimating the infinite-dimensional nuisance parameters such as the densities of the potential outcomes, and thus, the choices of tuning parameters. In addition, unlike the permutation tests studied in \cite{BCS17}, the validity of bootstrap inferences does not require either strong balance condition or studentization. In particular, such studentization is cumbersome in the QTE context.  

As the asymptotic variance for the IPW estimator does not depend on the treatment assignment rule implemented in RCTs, this estimator (and equivalently, the fully saturated estimator for the ATE) is suitable for settings where the knowledge of the exact treatment assignment rule is not available. Such scenario occurs when researchers are using an experiment that was run in the past and the randomization procedure may not have been fully described. It also occurs in subsample analysis, where sub-groups are defined using variables that may have not been used to form the strata and the treatment assignment rule for each sub-group becomes unknown. We illustrate this fact in the subsample analysis of the empirical application in Section \ref{sec:app}.

The rest of the paper is organized as follows. In Section \ref{sec:setup}, we describe the model setup and notation. In Sections \ref{sec:SQR} and \ref{sec:ipw}, we discuss the asymptotic properties of estimators (1) and (2), respectively. In Sections \ref{sec:SBI} and \ref{sec:CABI}, we investigate the  weighted and covariate-adaptive bootstrap approximations to the asymptotic distributions of estimators (1) and (2), respectively. In Section \ref{sec:sim}, we examine the finite-sample performance of the estimation and inference methods. In Section \ref{sec:guide}, we  provide recommendations for practitioners. In Section \ref{sec:app}, we apply the new methods to estimate and infer the average and quantile treatment effects of iron efficiency on educational attainment. In Section \ref{sec:concl}, we conclude. We provide proofs for all results in an appendix.  We study the strata fixed effects quantile regression estimator and provide additional simulation results in the second online supplement.

\section{Setup and Notation}
\label{sec:setup}

First, denote the potential outcomes for treated and control groups as $Y(1)$ and $Y(0)$, respectively. The treatment status is denoted as $A$, where $A=1$ means treated and $A=0$ means untreated. The researcher can only observe $\{Y_i,Z_i,A_i\}_{i=1}^n$ where $Y_i = Y_i(1)A_i + Y_i(0)(1-A_i)$, and $Z_i$ is a collection of baseline covariates. Strata are constructed from $Z$ using a function $S: \Supp(Z)\mapsto \mathcal{S}$, where $\mathcal{S}$ is a finite set. For $1\leq i \leq n$, let $S_i = S(Z_i)$ and $p(s) = \mathbb{P}(S_i = s)$. Throughout the paper, we maintain the assumption that $p(s)$ is fixed w.r.t. $n$ and is positive for every $s \in \mathcal{S}$.\footnote{We can also allow for the DGP to depend on $n$ so that $p_n(s) = \mathbb{P}_n(S_i = s)$ and $p(s) = \lim p_n(s)$. All the results in this paper still hold as long as $n(s)\rightarrow \infty$ a.s. Interested readers can refer to the previous version of this paper on arXiv for more detail.} We make the following assumption for the data generating process (DGP) and the treatment assignment rule.
\begin{ass}
	\label{ass:assignment1}	
	\begin{enumerate}
		\item $\{Y_i(1),Y_i(0),S_i\}_{i=1}^n$ is i.i.d. 
		\item $\{Y_i(1),Y_i(0)\}^{n}_{i=1} \indep \{A_i\}^{n}_{i=1}|\{S_i\}^{n}_{i=1}$.
		\item $\left\{\left\{ \frac{D_n(s)}{\sqrt{n}}\right\}_{s \in \mathcal{S}}\biggl|\{S_i\}^{n}_{i=1}  \right\} \convD N(0,\Sigma_D)$ a.s., where 
		\begin{align*}
		D_n(s) = \sum_{i =1}^n (A_i-\pi)1\{S_i = s\}\quad \text{and} \quad	\Sigma_D = \diag\{p(s)\gamma(s):s \in \mathcal{S}\}
		\end{align*}
		with $0 \leq \gamma(s) \leq \pi(1-\pi)$.
		\item $\frac{D_n(s)}{n(s)} = o_p(1)$ for $s \in \mathcal{S}$, where $n(s) =\sum_{i =1}^n1\{S_i=s\}$. 
	\end{enumerate}
\end{ass}

Several remarks are in order. First, Assumptions \ref{ass:assignment1}.1--\ref{ass:assignment1}.3 are exactly the same as \citet[Assumption 2.2]{BCS17}. We refer interested readers to \cite{BCS17} for more discussion of these assumptions. Second, note that in Assumption \ref{ass:assignment1}.3 the parameter $\pi$ is the target proportion of treatment for each stratum and $D_n(s)$ measures the imbalance. \cite{BCS18} study the more general case that $\pi$ can take distinct values for different strata. Third, we follow the terminology in \cite{BCS17}, which follows that of \cite{E71} and \cite{HH12}, saying a treatment assignment rule achieves strong balance if $\gamma(s)=0$. Fourth, we do not require that the treatment status is assigned independently. Instead, we only require Assumption \ref{ass:assignment1}.3 or Assumption \ref{ass:assignment1}.4, which condition is satisfied by several treatment assignment rules such as simple random sampling (SRS), biased-coin design (BCD), adaptive biased-coin design (WEI), and stratified block randomization (SBR). \citet[Section 3]{BCS17} provide an excellent summary of these four examples. For completeness, we briefly repeat their descriptions below. Note that both BCD and SBR assignment rules achieve strong balance. Last, as $p(s)>0$, Assumption \ref{ass:assignment1}.3 implies Assumption \ref{ass:assignment1}.4.

\begin{ex}[SRS]
	\label{ex:srs}
	Let $\{A_i\}_{i=1}^n$ be drawn independently across $i$ and of $\{S_i\}_{i=1}^n$ as Bernoulli random variables with success rate $\pi$, i.e., for $k=1,\cdots,n$, 
	\begin{align*}
	\mathbb{P}\left(A_k = 1\big|\{S_i\}_{i=1}^n, \{A_{j}\}_{j=1}^{k-1}\right) = \mathbb{P}(A_k = 1) = \pi.
	\end{align*}
	Then, Assumption \ref{ass:assignment1}.3 holds with $\gamma(s) = \pi(1-\pi)$.
\end{ex}

\begin{ex}[WEI]
	\label{ex:wei}
	The design is first proposed by \cite{W78}. Let $n_{k-1}(S_k) = \sum_{i=1}^{k-1}1\{S_i = S_k\}$, $D_{k-1}(s) = \sum_{i=1}^{k-1}\left(A_i - \frac{1}{2} \right) 1\{S_i = s\}$, and  
	\begin{align*}
	\mathbb{P}\left(A_k = 1\big| \{S_i\}_{i=1}^k,\{A_i\}_{i=1}^{k-1}\right) = \phi\biggl(\frac{D_{k-1}(S_k)}{n_{k-1}(S_k)}\biggr),
	\end{align*}
	where $\phi(\cdot):[-1,1] \mapsto [0,1]$ is a pre-specified non-increasing function satisfying $\phi(-x) = 1- \phi(x)$. Here, $\frac{D_0(S_1)}{0}$ is understood to be zero. Then, \cite{BCS17} show that Assumption \ref{ass:assignment1}.3 holds with $\pi = \frac{1}{2}$ and $\gamma(s) = \frac{1}{4}(1 - 4\phi'(0))^{-1}$.
\end{ex} 

\begin{ex}[BCD]
	\label{ex:bcd}
	The treatment status is determined sequentially for $1 \leq k \leq n$ as
	\begin{align*}
	\mathbb{P}\left(A_k = 1| \{S_i\}_{i=1}^k,\{A_i\}_{i=1}^{k-1}\right) = \begin{cases}
	\frac{1}{2} & \text{if }D_{k-1}(S_k) = 0 \\
	\lambda & \text{if }D_{k-1}(S_k) < 0 \\
	1-\lambda & \text{if }D_{k-1}(S_k) > 0,	
	\end{cases}
	\end{align*}
	where $D_{k-1}(s)$ is defined as above and $\frac{1}{2}< \lambda \leq 1$. Then, \cite{BCS17} show that Assumption \ref{ass:assignment1}.3 holds with $\pi = \frac{1}{2}$ and $\gamma(s) = 0$. 	
\end{ex}

\begin{ex}[SBR]
	\label{ex:sbr}
	For each stratum, $\lfloor \pi n(s) \rfloor$ units are assigned to treatment and the rest is assigned to control. \cite{BCS17} then show that Assumption \ref{ass:assignment1}.3 holds with $\gamma(s) = 0$. 
\end{ex}

Our parameter of interest is the $\tau$-th QTE defined as 
\begin{align*}
q(\tau) = q_1(\tau) - q_0(\tau),
\end{align*}
where $\tau \in (0,1)$ is a quantile index and $q_j(\tau)$ is the $\tau$-th quantile of random variable $Y(j)$ for $j = 0,1.$ For inference, although we mainly focus on the Wald test for the null hypothesis that $q(\tau)$ equals some particular value, our method can also be used to test hypotheses involving multiple or even a continuum of quantile indexes. The following regularity conditions are common in the literature of quantile estimations. 
\begin{ass}
	\label{ass:tau}
	For $j=0,1$, denote $f_j(\cdot)$ and $f_j(\cdot|s)$ as the PDFs of $Y_i(j)$ and $Y_i(j)|S_i=s$, respectively.
	\begin{enumerate}
		\item $f_j(q_j(\tau))$ and $f_j(q_j(\tau)|s)$ are bounded and bounded away from zero uniformly over $\tau \in \Upsilon$ and $s \in \mathcal{S}$, where $\Upsilon$ is a compact subset of $(0,1)$.
		\item $f_j(\cdot)$ and $f_j(\cdot|s)$ are Lipschitz over $\{q_j(\tau):\tau \in \Upsilon\}.$
	\end{enumerate}
\end{ass}

\section{Estimation}
\label{sec:est}

\subsection{Simple Quantile Regression}
\label{sec:SQR}
In this section, we propose to estimate $q(\tau)$ by a QR of $Y_i$ on $A_i$. Denote $\beta(\tau) = (\beta_0(\tau),\beta_1(\tau))^\prime$, $\beta_0(\tau) = q_0(\tau)$, and $\beta_1(\tau) = q(\tau)$. We estimate $\beta(\tau)$ by $\hat{\beta}(\tau)$, where 
\begin{align*}
\hat{\beta}(\tau) = \argmin_{b = (b_0,b_1)^\prime \in \Re^2} \sum_{i=1}^{n}\rho_\tau\left(Y_i - \dot{A}_i'b\right),
\end{align*}
$\dot{A}_i = (1,A_i)^\prime$, and $\rho_\tau(u) = u(\tau - 1\{u\leq 0\})$ is the standard check function. We refer to $\hat{\beta}_1(\tau)$, the second element of $\hat{\beta}(\tau)$, as our SQR estimator for the $\tau$-th QTE. As $A_i$ is a dummy variable, $\hat{\beta}_1(\tau)$ is numerically the same as the difference between the $\tau$-th empirical quantiles of $Y$ in the treatment and control groups. 

\begin{thm}
	\label{thm:qr}
	If Assumptions \ref{ass:assignment1}.1--\ref{ass:assignment1}.3 and \ref{ass:tau} hold, then, uniformly over $\tau \in \Upsilon$, 
	\begin{align*}
	\sqrt{n}\left(\hat{\beta}_1(\tau) - q(\tau)\right) \convD \mathcal{B}_{sqr}(\tau),~\text{as}~n\rightarrow \infty,
	\end{align*}
	where $\mathcal{B}_{sqr}(\cdot)$ is a Gaussian process with covariance kernel $\Sigma_{sqr}(\cdot,\cdot)$. The expression for $\Sigma_{sqr}(\cdot,\cdot)$ can be found in the Appendix. 
\end{thm}
The asymptotic variance for $\sqrt{n}\left(\hat{\beta}_1(\tau) - \beta_1(\tau)\right)$ is $\zeta_Y^2(\pi,\tau)  + \zeta_A^2(\pi,\tau) + \zeta_S^2(\tau)$, where
\begin{align*}
\zeta_Y^2(\pi,\tau) = \frac{\tau (1-\tau) - \mathbb{E}m_1^2(S,\tau)}{\pi f_1^2(q_1(\tau))} + \frac{\tau (1-\tau) - \mathbb{E}m_0^2(S,\tau)}{(1-\pi) f_0^2(q_0(\tau))},
\end{align*}
\begin{align*}
\zeta_A^2(\pi,\tau) = \mathbb{E}\gamma(S) \left(\frac{m_1(S,\tau)}{\pi f_1(q_1(\tau))} + \frac{m_0(S,\tau)}{(1-\pi) f_0(q_0(\tau))} \right)^2,
\end{align*}
\begin{align*}
\zeta_S^2(\tau) = \mathbb{E}\left(\frac{m_1(S,\tau)}{f_1(q_1(\tau))}-\frac{m_0(S,\tau)}{ f_0(q_0(\tau))} \right)^2,
\end{align*}
and $m_j(s,\tau) = \mathbb{E}(\tau - 1\{Y(j)\leq q_j(\tau)\}|S=s)$. Note that, if the treatment assignment rule achieves strong balance or the stratification is irrelevant\footnote{It means $\mathbb{P}(Y(j) \leq q_j(\tau)|S=s) = \tau$ for $s \in \mathcal{S}, j=0,1$.} then $\zeta_A^2(\pi,\tau) = 0$.

\subsection{Inverse Propensity Score weighted Quantile Regression}
\label{sec:ipw}
Denote $\hat{\pi}(s) = n_1(s)/n(s)$, $n_1(s) = \sum_{i=1}^n A_i 1\{S_i = s\}$, and $n(s) = \sum_{i=1}^n 1\{S_i = s\}$. Note $\hat{\pi}(S_i)$ is an estimator for the propensity score, i.e., $\pi$. In addition,  Assumption \ref{ass:assignment1}.2 implies that the unconfoundedness condition holds. Thus, following the lead of \cite{F07}, we can estimate $q_j(\tau)$ by the IPW QR. Let 
\begin{align*}
\hat{q}_1(\tau) = \argmin_q \frac{1}{n}\sum_{i=1}^n\frac{A_i}{\hat{\pi}(S_i)}\rho_\tau(Y_i - q) \quad \text{and} \quad \hat{q}_0(\tau) = \argmin_q\frac{1}{n}\sum_{i=1}^n \frac{1-A_i}{1-\hat{\pi}(S_i)}\rho_\tau(Y_i - q). 
\end{align*}
We then estimate $q(\tau)$ by $\hat{q}(\tau) = \hat{q}_1(\tau) - \hat{q}_0(\tau)$. 
\begin{thm}
	\label{thm:ipw}
	If Assumptions \ref{ass:assignment1}.1, \ref{ass:assignment1}.2, \ref{ass:assignment1}.4 and \ref{ass:tau} hold, then, uniformly over $\tau \in \Upsilon$, 
	\begin{equation*}
	\sqrt{n}\left(\hat{q}(\tau)-q(\tau)\right)\convD \mathcal{B}_{ipw}(\tau),~\text{as}~n\rightarrow \infty,
	\end{equation*}
	where $\mathcal{B}_{ipw}(\cdot)$ is a scalar Gaussian process with covariance kernel $\Sigma_{ipw}(\cdot,\cdot)$. The expression for $\Sigma_{ipw}(\cdot,\cdot)$ can be found in the Appendix. 	
\end{thm}
Two remarks are in order. First, the asymptotic variance for $\hat{q}(\tau)$ is 
\begin{align*}
\zeta_Y^2(\pi,\tau) + \zeta_S^2(\tau).
\end{align*}
When strong balance is not achieved and the stratification is relevant, we have $\zeta_A^2(\pi,\tau)>0$. Thus, $\hat{q}(\tau)$ is more efficient than $\hat{\beta}_{1}(\tau)$ in the sense that
\begin{align*}
\Sigma_{ipw}(\tau,\tau) < \Sigma_{sqr}(\tau,\tau).
\end{align*}
When strong balance is achieved ($\gamma(s) = 0$), we have $\zeta_A^2(\pi,\tau)=0$. Thus, the two estimators are asymptotically first-order equivalent. Based on the same argument, one can potentially prove that, when strong balance is not achieved and the stratification is relevant, the IPW estimator for ATE has strictly smaller asymptotic variance than the simple two-sample difference and strata fixed effects estimators studied by \cite{BCS17}, and is asymptotically equivalent to the fully saturated linear regression estimator proposed by \cite{BCS18}. Second, since the amount of ``balance" of the treatment assignment rule does not play a role in the limiting distribution of the IPW estimator, Assumption \ref{ass:assignment1}.3 is replaced by Assumption \ref{ass:assignment1}.4.

\section{Weighted Bootstrap }
\label{sec:SBI}
In this section, we approximate the  asymptotic distributions of the SQR and IPW estimators via the weighted bootstrap. Let $\{\xi_i\}_{i=1}^n$ be a sequence of bootstrap weights which will be specified later. Further denote 
$n_1^w(s) = \sum_{i =1}^n\xi_iA_i1\{S_i=s\}$, $n^w(s) = \sum_{i =1}^n \xi_i 1\{S_i=s\}$, and $\hat{\pi}^w(s) = n_1^w(s)/n^w(s)$. The weighted bootstrap counterparts for the two estimators we study in this paper can then be written respectively as 
\begin{align*}
\hat{\beta}^w(\tau) = \argmin_b \sum_{i=1}^n \xi_i \rho_\tau\left(Y_i - \dot{A}_i'b\right)
\end{align*}
and 
\begin{align*}
\hat{q}^w(\tau) = \hat{q}_1^w(\tau) - \hat{q}_0^w(\tau),
\end{align*}
where 
\begin{align*}
\hat{q}_1^w(\tau) = \argmin_q \sum_{i =1}^n \frac{\xi_iA_i}{\hat{\pi}^w(S_i)}\rho_\tau\left(Y_i - q\right) \quad \text{and} \quad \hat{q}_0^w(\tau) = \argmin_q \sum_{i =1}^n \frac{\xi_i(1-A_i)}{1-\hat{\pi}^w(S_i)}\rho_\tau\left(Y_i - q\right).
\end{align*}
The second element $\hat{\beta}_1^w(\tau)$ of $\hat{\beta}^w(\tau)$ and $\hat{q}^w(\tau)$ are the SQR and IPW bootstrap estimators for the $\tau$-th QTE, respectively. Next, we specify the bootstrap weights. 
\begin{ass}
	\label{ass:weight}
	Suppose $\{\xi_i\}_{i=1}^n$ is a sequence of nonnegative i.i.d. random variables with unit expectation and variance and a sub-exponential upper tail. 	
\end{ass}
The nonnegativity is required to maintain the convexity of the quantile regression objective function. The other conditions in Assumption \ref{ass:weight} are common for the weighted bootstrap approximation. In practice, we generate $\{\xi_i\}_{i=1}^n$ by the standard exponential distribution. The corresponding weighted bootstrap is also known as the Bayesian bootstrap. 

\begin{thm}
	\label{thm:b}
	If Assumptions \ref{ass:assignment1}.1--\ref{ass:assignment1}.3, \ref{ass:tau}, and \ref{ass:weight} hold, then uniformly over $\tau \in \Upsilon$ and conditionally on data,
	\begin{align*}
	\sqrt{n}\left(\hat{\beta}_1^w(\tau) - \hat{\beta}_1(\tau)\right) \convD \tilde{\mathcal{B}}_{sqr}(\tau),~\text{as}~n\rightarrow \infty,
	\end{align*}
	where $\tilde{\mathcal{B}}_{sqr}(\tau)$ is a Gaussian process. In addition,  $\tilde{\mathcal{B}}_{sqr}(\tau)$ shares the same covariance kernel with $\mathcal{B}_{sqr}(\tau)$ defined in Theorems \ref{thm:qr} with $\gamma(s)$ there replaced by $\pi(1-\pi)$.
	
	If Assumptions \ref{ass:assignment1}.1, \ref{ass:assignment1}.2, \ref{ass:assignment1}.4, \ref{ass:tau}, \ref{ass:weight} hold,  then uniformly over $\tau \in \Upsilon$ and conditionally on data, 	
	\begin{align*}
	\sqrt{n}\left(\hat{q}^w(\tau) - \hat{q}(\tau)\right) \convD \mathcal{B}_{ipw}(\tau),~\text{as}~n\rightarrow \infty,
	\end{align*}
	where $\mathcal{B}_{ipw}(\tau)$ is the same Gaussian process defined in Theorem \ref{thm:ipw}. 
\end{thm}

Four remarks are in order. First, the weighted bootstrap sample does not preserve the negative cross-sectional dependence in the original sample. Asymptotic variances of the weighted bootstrap estimators equal those of their original sample counterparts as if SRS is applied. In fact, the asymptotic variance for $\hat{\beta}^w_{1}(\tau)$ is
\begin{align*}
\zeta_Y^2(\pi,\tau) +  \tilde{\zeta}_A^2(\pi,\tau) + \zeta_S^2(\tau),
\end{align*}
where
\begin{align*}
\tilde{\zeta}_A^2(\pi,\tau) = \mathbb{E}\pi(1-\pi) \left(\frac{m_1(S,\tau)}{\pi f_1(q_1(\tau))} + \frac{m_0(S,\tau)}{(1-\pi) f_0(q_0(\tau))} \right)^2.
\end{align*}
This asymptotic variance is intuitive as the weight $\xi_i$ is independent with each other, which implies that, conditionally on data, the bootstrap sample observations are independent. As $\gamma(s) \leq \pi (1-\pi)$, we have 
$$\zeta_A^2(\pi,\tau) \leq \tilde{\zeta}_A^2(\pi,\tau).$$
If the inequality is strict, then the weighted bootstrap overestimates the asymptotic variance of the SQR estimator, and thus, the Wald test constructed using the SQR estimator and its weighted bootstrap standard error is conservative.   

Second, the asymptotic distribution of the weighted bootstrap IPW estimator coincides with that of the original estimator. The asymptotic size of the Wald test constructed using the IPW estimator and its weighted bootstrap standard error then equals the nominal level. Theorem \ref{thm:ipw} shows that the asymptotic variance for $\hat{q}(\tau)$ is invariant in the treatment assignment rule applied. Thus, even though the weighted bootstrap sample ignores the cross-sectional dependence and behaves as if the treatment status is generated randomly, the asymptotic variance for $\hat{q}^w(\tau)$ is still 
\begin{align*}
\zeta_Y^2(\pi,\tau) + \zeta_S^2(\tau). 
\end{align*} 

Third, the validity of weighted bootstrap for the IPW estimator only requires Assumption \ref{ass:assignment1}.4 instead of \ref{ass:assignment1}.3, for the same reason mentioned after Theorem \ref{thm:ipw}.

Fourth, it is possible to consider the conventional nonparametric bootstrap which generates the bootstrap sample from the empirical distribution of the data. If the observations are i.i.d., \citet[Section 3.6]{VW96} show that the conventional bootstrap is first-order equivalent to a weighted bootstrap with Poisson(1) weights. However, in the current setting, $\{A_i\}_{i \geq 1}$ is dependent. It is technically challenging to rigorously show that the above equivalence still holds. We leave it as an interesting topic for future research.  

\section{Covariate-Adaptive Bootstrap}
\label{sec:CABI}
In this section, we consider the covariate-adaptive bootstrap procedure as follows:
\begin{enumerate}
	\item Draw $\{S_i^*\}_{i=1}^n$ from the empirical distribution of $\{S_i\}_{i=1}^n$ with replacement. 
	\item Generate $\{A_i^*\}_{i=1}^n$ based on $\{S_i^*\}_{i=1}^n$ and the treatment assignment rule. 
	\item For $A_i^* = a$ and $S_i^* = s$, draw $Y_i^*$ from the empirical distribution of $Y_i$ given $A_i=a$ and $S_i = s$ with replacement. 
\end{enumerate}
First, Step 1 is the conventional nonparametric bootstrap. The bootstrap sample $\{S_i^*\}_{i=1}^n$ is obtained by drawing from the empirical distribution of $\{S_i\}_{i=1}^n$ with replacement $n$ times. Second, Step 2 follows the treatment assignment rule, and thus preserves the cross-sectional dependence structure in the bootstrap sample, even after conditioning on data. The weighted bootstrap sample, by contrast, is cross-sectionally independent given data. Third, Step 3 applies the conventional bootstrap procedure to the outcome $Y_i$ in the cell $(S_i,A_i) = (s,a) \in \mathcal{S} \times \{0,1\}$. Given that the original data contain $n_a(s)$ observations in this cell, in this step, the bootstrap sample $\{Y_i^*\}_{i: A_i^* =a, S_i^*=s}$ is obtained by drawing from the empirical distribution of these $n_a(s)$ outcomes with replacement $n^*_a(s)$ times, where $n^*_a(s) = \sum_{i=1}^n 1\{A_i^*=a, S_i^*=s\}.$ Unlike the conventional bootstrap, here both $n_a(s)$ and $n_a^*(s)$ are random and are not necessarily the same. Last, to implement the covariate-adaptive bootstrap, researchers need to know the treatment assignment rule for the original sample. Unlike in observational studies, such information is usually available for RCTs. If one only knows that the treatment assignment rule achieves strong balance, then Theorem \ref{thm:cab} below still holds, provided that the bootstrap sample is generated from any treatment assignment rule that achieves strong balance. Even worse, if no information on the treatment assignment rule is available, then one cannot implement the covariate-adaptive bootstrap inference. In this case, the weighted bootstrap for the IPW estimator can still provide a non-conservative Wald test, as shown in Theorem \ref{thm:b}.

Using the bootstrap sample $\{Y_i^*,A_i^*,S_i^*\}_{i=1}^n$, we can estimate QTE by the two methods considered in the paper. Let 
$n_1^*(s) = \sum_{i =1}^nA_i^*1\{S_i^*=s\}$, $n^*(s) = \sum_{i =1}^n 1\{S_i^*=s\}$, $\hat{\pi}^*(s) = \frac{n_1^*(s)}{n^*(s)}$, and $\dot{A}_i^* = (1,A_i^*)'$. Then, the two bootstrap estimators can be written respectively as 
\begin{align*}
\hat{\beta}^*(\tau) = \argmin_b \sum_{i=1}^n\rho_\tau\left(Y_i^* - \dot{A}_i^*b\right)
\end{align*}
and 
\begin{align*}
\hat{q}^*(\tau) = \hat{q}_1^*(\tau) - \hat{q}_0^*(\tau),
\end{align*}
where 
\begin{align*}
\hat{q}_1^* = \argmin_q \sum_{i =1}^n \frac{A_i^*}{\hat{\pi}^*(S_i^*)}\rho_\tau(Y_i^* - q) \quad \text{and} \quad \hat{q}_0^* = \argmin_q \sum_{i =1}^n \frac{1-A_i^*}{1-\hat{\pi}^*(S_i^*)}\rho_\tau(Y_i^* - q).
\end{align*}
The second element $\hat{\beta}_1^*(\tau)$ of $\hat{\beta}^*(\tau)$ and $\hat{q}^*(\tau)$ are the SQR and IPW bootstrap estimators for the $\tau$-th QTE, respectively. Parallel to Assumption \ref{ass:assignment1}, we make the following assumption for the bootstrap sample. 
\begin{ass}
	\label{ass:bassignment}
	Let $D_n^*(s) = \sum_{i=1}^n (A_i^* - \pi)1\{S_i^* = s\}$. 
	\begin{enumerate}
		\item $\left\{\left\{ \frac{D^*_n(s)}{\sqrt{n}}\right\}_{s \in \mathcal{S}}\biggl|\{S_i^*\}_{i=1}^n  \right\} \convD N(0,\Sigma_D)$ a.s., where $\Sigma_D = \diag\{p(s)\gamma(s):s \in \mathcal{S}\}$.
		\item $\sup_{s \in \mathcal{S} }\frac{|D^*_n(s)|}{\sqrt{n^*(s)}} = O_p(1)$, $\sup_{s \in \mathcal{S} }\frac{|D_n(s)|}{\sqrt{n(s)}} = O_p(1)$.
	\end{enumerate}
\end{ass}
Assumption \ref{ass:bassignment}.1 is a high-level assumption. Obviously, it holds for SRS. For WEI, this condition holds by the same argument in \citet[Lemma B.12]{BCS17} with the fact that $\frac{n^*(s)}{n(s)} \convP 1.$ For BCD, as shown in \citet[Lemma B.11]{BCS17}, 
\begin{align*}
D_n^*(s) |\{S_i^*\}_{i=1}^n = O_p(1). 
\end{align*}
Therefore, $D_n^*(s)/\sqrt{n^*(s)} \convP 0$ and Assumption \ref{ass:bassignment}.1 holds with $\gamma(s) = 0$. For SBR, it is clear that $|D_n^*(s)| \leq 1.$ Thus, Assumption \ref{ass:bassignment}.1 holds with $\gamma(s) = 0$ as well. In addition, as $p(s)>0$, based on the standard bootstrap results, we have $n^*(s)/n \convP p(s)$ and $n(s)/n \convP p(s)$. Therefore, Assumption \ref{ass:bassignment}.1 is sufficient for Assumption \ref{ass:bassignment}.2. Last, note that Assumption \ref{ass:bassignment}.2 implies Assumption \ref{ass:assignment1}.4. 
\begin{thm}
	\label{thm:cab}
	Suppose Assumptions \ref{ass:assignment1}.1, \ref{ass:assignment1}.2, \ref{ass:tau}, and \ref{ass:bassignment}.2 hold. Then, uniformly over $\tau \in \Upsilon$ and conditionally on data,  
	\begin{align*}
	\sqrt{n}(\hat{q}^*(\tau) - \hat{q}(\tau)) \convD \mathcal{B}_{ipw}(\tau), ~\text{as}~n\rightarrow \infty.
	\end{align*}
	If, in addition, Assumptions \ref{ass:assignment1}.3 and \ref{ass:bassignment}.1 hold, then 
	\begin{align*}
	\sqrt{n}\left(\hat{\beta}_1^*(\tau) - \hat{q}(\tau)\right) \convD \mathcal{B}_{sqr}(\tau),~\text{as}~n\rightarrow \infty.
	\end{align*}
	Here, $\mathcal{B}_{sqr}(\tau)$ and $\mathcal{B}_{ipw}(\tau)$ are two Gaussian processes defined in Theorem \ref{thm:qr} and \ref{thm:ipw}, respectively. 
\end{thm}

Several remarks are in order. First, unlike the usual bootstrap estimator, the covariate-adaptive bootstrap SQR estimator is not centered around its corresponding counterpart from the original sample, but rather $\hat{q}(\tau)$. The reason is that the treatment status $A_i^*$ is not generated by bootstrap. In the linear expansion for the bootstrap estimator $\hat{\beta}_1^*(\tau)$, the part of the influence function that accounts for the variation generated by $A_i^*$ need not be centered. We also know from the proof of Theorem \ref{thm:ipw} that $\hat{q}(\tau)$ do not have an influence function that represents the variation generated by $A_i$. Thus, $\hat{q}(\tau)$ can be used to center $\hat{\beta}_1^*(\tau)$. 

Second, the choice of $\hat{q}(\tau)$ as the center is somehow ad-hoc. In fact, any estimator $\tilde{q}(\tau)$ that is first-order equivalent to $\hat{q}(\tau)$ in the sense that  
\begin{align*}
\sup_{ \tau \in \Upsilon}|\tilde{q}(\tau) - \hat{q}(\tau)| = o_p(1/\sqrt{n})
\end{align*} 
can serve as the center for the bootstrap estimators $\hat{q}^*(\tau)$ and $\hat{\beta}_1^*(\tau) $.

Third, when the treatment assignment rule achieves strong balance, $\hat{\beta}_1(\tau)$ and $\hat{q}(\tau)$ are first-order equivalent. In this case, $\hat{\beta}_1(\tau)$ can serve as the center for $\hat{\beta}_1^*(\tau)$ and various bootstrap inference methods are valid. On the other hand, when the treatment assignment rule does not achieve strong balance, $\hat{\beta}_1(\tau)$ and $\hat{q}(\tau)$ are not first-order equivalent. In this case, the asymptotic size of the percentile bootstrap for the SQR estimator using the quantiles of $\hat{\beta}_1^*(\tau)$ does not equal the nominal level. In the next section, we propose a way to compute the bootstrap standard error which does not depend on the choice of the center. Based on the bootstrap standard error, researchers can construct t-statistics and use standard normal critical values for inference.

Fourth, for ATE, we can use the same bootstrap sample to compute the standard errors for the simple and strata fixed effects estimators proposed in \cite{BCS17} as well as the IPW estimator. We expect that all the results in this paper hold for the ATE as well.

\section{Simulation}
\label{sec:sim}
We can summarize four bootstrap scenarios from the analysis in Sections \ref{sec:SBI} and \ref{sec:CABI}: (1) the SQR estimator with the weighted bootstrap, (2) the IPW estimator with either the weighted or covariate-adaptive bootstrap, (3) the SQR estimator with the covariate-adaptive bootstrap when the assignment rule achieves strong balance, and (4) the SQR estimator with the covariate-adaptive bootstrap when the assignment rule does not achieve strong balance. The results of Sections \ref{sec:SBI} and \ref{sec:CABI} imply that the bootstrap in scenario (1) produces conservative Wald-tests when the treatment assignment rule is not SRS. For scenarios (2) and (3), various bootstrap based inference methods are valid. However,  for scenario (4), researchers should be careful about the centering issue. In particular, the percentile bootstrap inference using the quantiles of $\hat{\beta}_1^*$ is invalid.  In the following, we propose one single bootstrap inference method that works for scenarios (2)--(4). In addition, the proposed method does not require the knowledge of the centering. 

We take the IPW estimator as an example. We can repeat the bootstrap estimation\footnote{For the IPW estimator, we can use either the weighted or covariate-adaptive bootstrap. For the SQR estimator, we can only use the covariate-adaptive bootstrap.} $B$ times and obtain $B$ bootstrap IPW estimates, denoted as $\{\hat{q}_b^*(\tau)\}_{b=1}^B$. Further denote $\hat{Q}(\alpha)$ as the $\alpha$-th empirical quantile of $\{\hat{q}_b^*(\tau)\}_{b=1}^B$. We can test the null hypothesis that $q(\tau) = q^0(\tau)$ via $1\left\{\left|\frac{\hat{q}(\tau) - q^0(\tau)}{\hat{\sigma}_n^*}\right| > z_{1-\alpha/2}\right\}$, where $\hat{q}(\tau)$, $z_{1-\alpha/2}$, and $\hat{\sigma}_n^*$ are the IPW estimator, the $(1-\alpha/2)$-th quantile of the standard normal distribution, and 
\begin{align*}
\hat{\sigma}_n^* = \frac{\hat{Q}(0.975) - \hat{Q}(0.025)}{z_{0.975} - z_{0.025}},
\end{align*}  
respectively. In scenarios (2)--(4), the asymptotic size of such test equals the nominal level $\alpha$. In scenarios (2) and (3), we recommend the t-statistic and confidence interval using this particular bootstrap standard error (i.e., $\hat{\sigma}_n^*$) over other bootstrap inference methods (e.g., bootstrap confidence interval, percentile bootstrap confidence interval, etc.) because based on unreported simulations, they have better finite sample performance.

\subsection{Data Generating Processes}
We consider two DGPs with parameters $\gamma =4$, $\sigma =2$, and $\mu$ which will be specified later. 
\begin{enumerate}
	\item Let $Z$ be standardized $\text{Beta}(2,2)$ distributed, $S_i = \sum_{j=1}^41\{Z_i \leq g_j\}$, and $(g_1,\cdots,g_4) = (-0.25\sqrt{20},0,0.25\sqrt{20},0.5\sqrt{20})$. The outcome equation is 
	$$Y_i = A_i \mu + \gamma Z_i + \eta_i,$$ 
	where $\eta_i = \sigma A_i \eps_{i,1} + (1-A_i)\eps_{i,2}$ and $(\eps_{i,1},\eps_{i,2})$ are jointly standard normal. 
	\item Let $Z$ be uniformly distributed on $[-2,2]$, $S_i = \sum_{j=1}^41\{Z_i \leq g_j\}$, and $(g_1,\cdots,g_4) = (-1,0,1,2)$. The outcome equation is 
	\begin{align*}
	Y_i = A_i \mu+A_i\nu_{i,1}   + (1-A_i)\nu_{i,0} + \eta_i,
	\end{align*}
	where $\nu_{i,0} = \gamma Z_i^2 1\{|Z_i|\geq 1\} + \frac{\gamma}{4}(2 - Z_i^2)1\{|Z_i|<1\}$, $\nu_{i,1} = -\nu_{i,0}$, $\eta_i = \sigma(1+Z_i^2)A_i\eps_{i,1} + (1+Z_i^2)(1-A_i)\eps_{i,2}$, and $(\eps_{i,1},\eps_{i,2})$ are mutually independent $T(3)/3$ distributed. 
\end{enumerate}

When $\pi = \frac{1}{2}$, for each DGP, we consider four randomization schemes:
\begin{enumerate}
	\item SRS: Treatment assignment is generated as in Example \ref{ex:srs}.
	\item WEI: Treatment assignment is generated as in Example \ref{ex:wei} with $\phi(x) = (1-x)/2$.
	\item BCD: Treatment assignment is generated as in Example \ref{ex:bcd} with $\lambda = 0.75$. 
	\item SBR: Treatment assignment is generated as in Example \ref{ex:sbr}. 
\end{enumerate}
When $\pi \neq 0.5$, BCD is not defined while WEI is not defined in the original paper (\citep{W78}). Recently, \cite{H16} generalizes the adaptive biased-coin design (i.e., WEI) to multiple treatment values and unequal target treatment ratios. Here, for $\pi \neq 0.5$, we only consider SRS and SBR as in \cite{BCS17}. We conduct the simulations with sample sizes $n=200$ and $400$. The numbers of simulation replications and bootstrap samples are 1000. Under the null, $\mu=0$ and we compute the true parameters of interest using simulations with $10^6$ sample size and $10^4$ replications. Under the alternative, we perturb the true values by $\mu = 1$ and $\mu = 0.75$ for $n = 200$ and $400$, respectively. We report the results for the median QTE. Section \ref{sec:addsim} contains additional simulation results for ATE and QTEs with $\tau = 0.25$ and $0.75$. All the observations made in this section still apply. 

\subsection{QTE, $\pi = 0.5$}
We consider the Wald test with six t-statistics and $95\%$ nominal rate. We construct the t-statistics using one of our two point estimates and some estimate of the standard error. We will reject the null hypothesis when the absolute value of the t-statistic is greater than 1.96. The details about the point estimates and standard errors are as follows: 
\begin{enumerate}
	\item ``s/naive": the point estimator is computed by the SQR and its standard error $\hat{\sigma}_{naive}(\tau)$ is computed as 
	\begin{align}
	\label{eq:stddev}
	\hat{\sigma}^2_{naive}(\tau) = & \frac{\tau(1-\tau) - \frac{1}{n}\sum_{i =1}^n\hat{m}^2_1(S_i,\tau) }{\pi\hat{f}_1^2(\hat{q}_1(\tau))}+\frac{\tau(1-\tau) - \frac{1}{n}\sum_{i =1}^n\hat{m}^2_0(S_i,\tau) }{(1-\pi)\hat{f}_0^2(\hat{q}_0(\tau))} \notag \\
	& + \frac{1}{n}\sum_{i =1}^n \pi(1-\pi)\left(\frac{\hat{m}_1(S_i,\tau)}{\pi \hat{f}_1(\hat{q}_1(\tau))} + \frac{\hat{m}_0(S_i,\tau)}{(1-\pi) \hat{f}_0(\hat{q}_0(\tau))}\right)^2 \notag \\
	& + \frac{1}{n}\sum_{i =1}^n \left(\frac{\hat{m}_1(S_i,\tau)}{ \hat{f}_1(\hat{q}_1(\tau))} - \frac{\hat{m}_0(S_i,\tau)}{ \hat{f}_0(\hat{q}_0(\tau))}\right)^2,
	\end{align}
	where $\hat{q}_j(\tau)$ is the $\tau$-th empirical quantile of $Y_i|A_i = j$, $$\hat{m}_{i,1}(s,\tau) = \frac{\sum_{i =1}^nA_i1\{S_i = s\}(\tau - 1\{Y_i \leq \hat{q}_1(\tau)\})}{n_1(s)},$$ 
	$$\hat{m}_{i,0}(s,\tau) = \frac{\sum_{i =1}^n(1-A_i)1\{S_i = s\}(\tau - 1\{Y_i \leq \hat{q}_0(\tau)\})}{n(s)-n_1(s)}.$$ 
	For $j=0,1$, $\hat{f}_j(\cdot)$ is computed by the kernel density estimation using the observations $Y_i$ provided that $A_i = j$, bandwidth $h_j = 1.06\hat{\sigma}_jn_j^{-1/5}$, Gaussian kernel function, standard deviation $\hat{\sigma}_j$ of the observations $Y_i$ provided that $A_i = j$, and $n_j = \sum_{i =1}^n 1\{A_i = j\}$. 
	\item ``s/adj": exactly the same as the ``s/naive" method with one difference: replacing $\pi(1-\pi)$ in \eqref{eq:stddev} by $\gamma(S_i)$.
	\item ``s/W": the point estimator is computed by the SQR and its standard error $\hat{\sigma}_W(\tau)$ is computed by the weighted bootstrap procedure. The bootstrap weights $\{\xi_i\}_{i=1}^n$ are generated from the standard exponential distribution. Denote $\{\hat{\beta}^w_{1,b}\}_{b=1}^B$ as the collection of $B$ weighted bootstrap SQR estimates. Then, 
	\begin{align*}
	\hat{\sigma}_{W}(\tau) = \frac{\hat{Q}(0.975) - \hat{Q}(0.025)}{z_{0.975} - z_{0.025}},
	\end{align*}
	where $\hat{Q}(\alpha)$ is the $\alpha$-th empirical quantile of $\{\hat{\beta}^w_{1,b}(\tau)\}_{b=1}^B$. 
	\item ``ipw/W": the same as above with one difference: the estimation method for both the original and bootstrap samples is the IPW QR.  
	\item  ``s/CA": the point estimator is computed by the SQR and its standard error $\hat{\sigma}_{CA}(\tau)$ is computed by the covariate-adaptive bootstrap procedure. Denote $\{\hat{\beta}^*_{1,b}\}_{b=1}^B$ as the collection of $B$ estimates obtained by the SQR applied to the samples generated by the covariate-adaptive bootstrap procedure. Then, 
	\begin{align*}
	\hat{\sigma}_{CA}(\tau) = \frac{\hat{Q}(0.975) - \hat{Q}(0.025)}{z_{0.975} - z_{0.025}},
	\end{align*}
	where $\hat{Q}(\alpha)$ is the $\alpha$-th empirical quantile of $\{\hat{\beta}^*_{1,b}(\tau)\}_{b=1}^B$. 
	\item ``ipw/CA": the same as above with one difference:  the estimation method for both the original and bootstrap samples is the IPW QR.  
\end{enumerate}

Tables \ref{tab:200_2} and \ref{tab:400_2} present the rejection probabilities (multiplied by 100) for the six t-tests under both the null hypothesis and the alternative hypothesis, with sample sizes $n=200$ and $400$, respectively. In these two tables, columns $M$ and $A$ represent DGPs and treatment assignment rules, respectively. From the rejection probabilities under the null, we can make five observations. First, the naive t-test (``s/naive") is conservative for WEI, BCD, and SBR, which is consistent with the findings for ATE estimators by \cite{SYZ10} and \cite{BCS17}. Second, although the asymptotic size of the adjusted t-test (``s/adj") is expected to equal the nominal level, it does not perform well for DGP2. The main reason is that, in order to analytically compute the standard error, we must compute nuisance parameters such as the unconditional densities of $Y(0)$ and $Y(1)$, which requires tuning parameters. We further compute the standard errors following \eqref{eq:stddev} with $\pi(1-\pi)$ and the tuning parameter $h_j$ replaced by $\gamma(S_i)$ and $1.06C_f\hat{\sigma}_jn_j^{-1/5}$, respectively, for some constant $C_f \in [0.5,1.5]$. Figure \ref{fig:BCD_50_200} plots the rejection probabilities of the ``s/adj" t-tests against $C_f$ for the BCD assignment rule with $n=200$, $\tau = 0.5$, and $\pi = 0.5$. We see that (1) the rejection probability is sensitive to the choice of bandwidth, (2) there is no universal optimal bandwidth across two DGPs, and (3) the covariate-adaptive bootstrap t-tests (``s/CA") represented by the dotted dash lines are quite stable across different DGPs and close to the nominal rate of rejection. Third, the weighted bootstrap t-test for the SQR estimator (``s/W") is conservative, especially for the BCD and SBR assignment rules which achieve strong balance. Fourth, the rejection probabilities of the weighted bootstrap t-test for the IPW estimator (``ipw/W") are close to the nominal rate even for sample size $n=200$, which is consistent with Theorem \ref{thm:b}. Last, the rejection rates for the two covariate-adaptive bootstrap t-tests (``s/CA" and ``ipw/CA") are close to the nominal rate, which is consistent with Theorem \ref{thm:cab}.

\begin{table}[H]
	\centering
	\caption{$n = 200,  \tau = 0.5, \pi=0.5$}
	\begin{adjustbox}{max width=\textwidth}
		\begin{tabular}{c|l|cccccc|c|cccccc}
			\multicolumn{8}{c|}{$H_0$}                                        &       & \multicolumn{6}{c}{$H_1$} \\ \hline 
			\multicolumn{1}{c}{M} & \multicolumn{1}{c}{A}   & \multicolumn{1}{l}{s/naive} & \multicolumn{1}{c}{s/adj} & \multicolumn{1}{c}{s/W} & \multicolumn{1}{c}{ipw/W} & \multicolumn{1}{c}{s/CA } & \multicolumn{1}{c|}{ipw/CA} & & \multicolumn{1}{l}{s/naive} & \multicolumn{1}{c}{s/adj} & \multicolumn{1}{c}{s/W} & \multicolumn{1}{c}{ipw/W} & \multicolumn{1}{c}{s/CA } & \multicolumn{1}{c}{ipw/CA} \\ \hline 
			1     & SRS   & 4.5   & 4.5   & 4.7   & 4.4   & 4.4   & 3.9   &       & 18.3  & 18.3  & 19.3  & 44.1  & 20.0  & 42.9 \\
			& WEI   & 1.2   & 4.0   & 1.4   & 4.3   & 3.7   & 3.5   &       & 11.6  & 29.5  & 13.8  & 44.7  & 29.8  & 43.6 \\
			& BCD   & 0.2   & 5.7   & 0.3   & 4.1   & 4.4   & 3.9   &       & 7.2   & 47.2  & 9.5   & 45.3  & 43.4  & 44.8 \\
			& SBR   & 0.1   & 5.7   & 0.1   & 4.6   & 4.5   & 4.4   &       & 8.5   & 48.5  & 9.9   & 46.0  & 45.7  & 44.8 \\ \hline 
			2     & SRS   & 0.4   & 0.4   & 4.7   & 5.2   & 5.2   & 5.3   &       & 79.7  & 79.7  & 90.4  & 91.6  & 90.2  & 91.3 \\ 
			& WEI   & 0.6   & 0.6   & 4.5   & 5.8   & 5.2   & 5.7   &       & 80.2  & 80.7  & 90.7  & 90.9  & 91.3  & 90.6 \\
			& BCD   & 1.0   & 1.0   & 4.5   & 5.1   & 5.0   & 5.3   &       & 79.6  & 80.4  & 90.2  & 91.1  & 90.8  & 90.6 \\
			& SBR   & 0.8   & 1.1   & 4.8   & 5.3   & 4.6   & 4.7   &       & 77.1  & 77.4  & 89.7  & 90.1  & 89.9  & 89.9 \\
		\end{tabular}%
	\end{adjustbox}
	\label{tab:200_2}%
\end{table}%

\begin{table}[H]
	\centering
	\caption{$n = 400,  \tau = 0.5, \pi=0.5$}
	\begin{adjustbox}{max width=\textwidth}
		\begin{tabular}{c|l|cccccc|c|cccccc}
			\multicolumn{8}{c|}{$H_0$}                                        &       & \multicolumn{6}{c}{$H_1$} \\ \hline 
			\multicolumn{1}{c}{M} & \multicolumn{1}{c}{A}   & \multicolumn{1}{l}{s/naive} & \multicolumn{1}{c}{s/adj} & \multicolumn{1}{c}{s/W} & \multicolumn{1}{c}{ipw/W} & \multicolumn{1}{c}{s/CA } & \multicolumn{1}{c|}{ipw/CA} & & \multicolumn{1}{l}{s/naive} & \multicolumn{1}{c}{s/adj} & \multicolumn{1}{c}{s/W} & \multicolumn{1}{c}{ipw/W} & \multicolumn{1}{c}{s/CA } & \multicolumn{1}{c}{ipw/CA} \\ \hline 		
			
			1     & SRS   & 4.2   & 4.2   & 5.4   & 4.0   & 4.6   & 4.1   &       & 21.8  & 21.8  & 23.2  & 50.2  & 23.5  & 50.2 \\
			& WEI   & 1.0   & 4.9   & 0.8   & 4.7   & 4.6   & 4.2   &       & 14.7  & 35.6  & 16.0  & 50.3  & 35.0  & 50.7 \\
			& BCD   & 0.3   & 4.5   & 0.2   & 4.3   & 3.5   & 4.0   &       & 8.9   & 52.6  & 11.7  & 50.2  & 49.3  & 49.6 \\
			& SBR   & 0.2   & 4.6   & 0.0   & 3.7   & 3.6   & 3.7   &       & 8.9   & 55.0  & 10.9  & 51.8  & 52.4  & 51.9 \\ \hline 
			2     & SRS   & 1.2   & 1.2   & 4.3   & 4.8   & 4.6   & 5.0   &       & 89.7  & 89.7  & 95.6  & 95.6  & 95.7  & 95.7 \\
			& WEI   & 1.4   & 1.6   & 5.7   & 6.0   & 5.5   & 5.7   &       & 89.2  & 89.2  & 95.4  & 94.8  & 95.1  & 94.8 \\
			& BCD   & 1.3   & 1.3   & 5.5   & 6.1   & 5.1   & 5.2   &       & 88.7  & 88.9  & 95.2  & 95.4  & 95.7  & 95.6 \\
			& SBR   & 0.6   & 0.6   & 4.0   & 3.9   & 3.8   & 3.8   &       & 90.0  & 90.2  & 95.4  & 95.4  & 95.8  & 95.7 \\
		\end{tabular}%
	\end{adjustbox}
	\label{tab:400_2}%
\end{table}

\begin{figure}[H]
	\centering
	\includegraphics[height=5cm,width = 10cm]{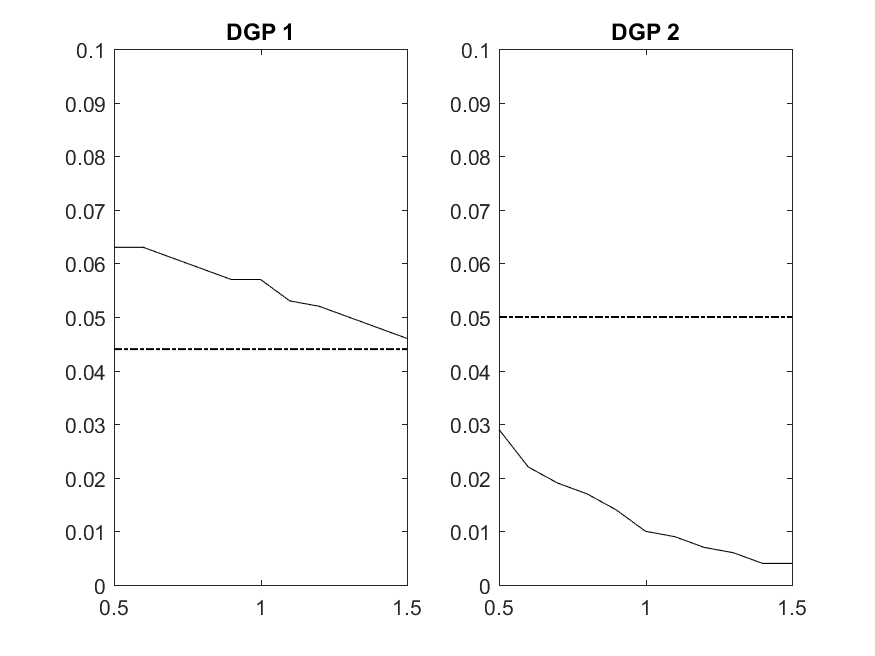}
	\begin{minipage}{0.75\linewidth}
		\footnotesize
		\flushleft Note: Rejection probabilities for BCD assignment rule with $n = 200$, $\pi = 0.5$, and $\tau = 0.5$. The X-axis is $C_f$. The solid lines are the rejection probabilities for ``s/adj". The densities of $Y_j$ is computed using the tuning parameters $h_j = 1.06C_f\hat{\sigma}_jn_j^{-1/5}$, for $j=0,1$. The dotted dash lines are the rejection probability for ``s/CA".   
	\end{minipage}
	\caption{Rejection Probabilities Across Different Bandwidth Values}
	\label{fig:BCD_50_200}
\end{figure}

Turning to the rejection rates under the alternative in Tables \ref{tab:200_2} and \ref{tab:400_2}, we can make two additional observations. First, for BCD and SBR, the rejection probabilities (power) for ``ipw/W", ``s/CA", and ``ipw/CA" are close. This is because both BCD and SBR achieve strong balance. In this case, the two estimators we propose are asymptotically first-order equivalent. Second, for DGP1 with SRS and WEI assignment rules, ``ipw/CA" is more powerful than ``s/CA". This confirms our theoretical finding that the IPW estimator is \textit{strictly} more efficient than the SQR estimator when the treatment assignment rule does \textit{not} achieve strong balance. For DGP2 the three t-tests, i.e., ``ipw/W", ``s/CA", and ``ipw/CA", have similar power.

\subsection{QTE, $\pi = 0.7$}
Tables \ref{tab:200_70_2} and \ref{tab:400_70_2} show the similar results with $\pi = 0.7$. The same comments for Tables \ref{tab:200_2} and \ref{tab:400_2} still apply.  
\begin{table}[H]
	\centering
	\caption{$n = 200,  \tau = 0.5, \pi=0.7$}
	\begin{adjustbox}{max width=\textwidth}
		\begin{tabular}{c|l|cccccc|c|cccccc}
			\multicolumn{8}{c|}{$H_0$}                                        &       & \multicolumn{6}{c}{$H_1$} \\ \hline 
			\multicolumn{1}{c}{M} & \multicolumn{1}{c}{A}   & \multicolumn{1}{l}{s/naive} & \multicolumn{1}{c}{s/adj} & \multicolumn{1}{c}{s/W} & \multicolumn{1}{c}{ipw/W} & \multicolumn{1}{c}{s/CA } & \multicolumn{1}{c|}{ipw/CA} & & \multicolumn{1}{l}{s/naive} & \multicolumn{1}{c}{s/adj} & \multicolumn{1}{c}{s/W} & \multicolumn{1}{c}{ipw/W} & \multicolumn{1}{c}{s/CA } & \multicolumn{1}{c}{ipw/CA} \\ \hline 	
			1     & SRS   & 4.8   & 4.8   & 5.2   & 4.7   & 3.4   & 4.4   &       & 17.0  & 17.0  & 17.2  & 42.5  & 16.7  & 40.6 \\
			& SBR   & 0.1   & 0.7   & 0.2   & 4.0   & 4.4   & 3.7   &       & 4.3   & 21.2  & 6.0   & 45.5  & 45.7  & 43.4 \\ \hline
			2     & SRS   & 1.6   & 1.6   & 5.2   & 5.4   & 5.1   & 5.3   &       & 77.1  & 77.1  & 89.1  & 90.3  & 89.5  & 89.4 \\
			& SBR   & 0.4   & 0.5   & 3.9   & 4.8   & 4.5   & 4.8   &       & 76.0  & 76.9  & 89.2  & 91.1  & 90.1  & 90.0 \\
		\end{tabular}%
	\end{adjustbox}
	\label{tab:200_70_2}%
\end{table}%

\begin{table}[H]
	\centering
	\caption{$n = 400,  \tau = 0.5, \pi=0.7$}
	\begin{adjustbox}{max width=\textwidth}
		\begin{tabular}{c|l|cccccc|c|cccccc}
			\multicolumn{8}{c|}{$H_0$}                                        &       & \multicolumn{6}{c}{$H_1$} \\ \hline 
			\multicolumn{1}{c}{M} & \multicolumn{1}{c}{A}   & \multicolumn{1}{l}{s/naive} & \multicolumn{1}{c}{s/adj} & \multicolumn{1}{c}{s/W} & \multicolumn{1}{c}{ipw/W} & \multicolumn{1}{c}{s/CA } & \multicolumn{1}{c|}{ipw/CA} & & \multicolumn{1}{l}{s/naive} & \multicolumn{1}{c}{s/adj} & \multicolumn{1}{c}{s/W} & \multicolumn{1}{c}{ipw/W} & \multicolumn{1}{c}{s/CA } & \multicolumn{1}{c}{ipw/CA} \\ \hline 		
			1     & SRS   & 4.4   & 4.4   & 5.1   & 3.9   & 4.8   & 3.7   &       & 18.4  & 18.4  & 18.7  & 47.9  & 19.4  & 46.6 \\
			& SBR   & 0.1   & 0.2   & 0     & 3.9   & 3.5   & 4     &       & 4.2   & 22    & 5.9   & 49.8  & 50.5  & 48.2 \\ \hline
			2     & SRS   & 0.7   & 0.7   & 3.9   & 4.2   & 4.2   & 4.7   &       & 86.7  & 86.7  & 93.9  & 93.3  & 94.1  & 93.6 \\
			& SBR   & 0.6   & 0.6   & 3.5   & 3.6   & 3.7   & 3.7   &       & 88.3  & 88.8  & 94.8  & 95.2  & 95.5  & 95.2 \\
		\end{tabular}%
	\end{adjustbox}
	\label{tab:400_70_2}%
\end{table}%

\subsection{Difference between Two QTEs}
Last, we consider to infer $q(0.25)-q(0.75)$ when $\pi=0.5$:
\begin{align*}
H_0:~q(0.25)-q(0.75) = \text{the true value} \quad \text{v.s.} \quad H_1:~q(0.25)-q(0.75) = \text{the true value}+\mu,
\end{align*}
where $\mu = 1$ and $0.75$ for sample sizes $200$ and $400$, respectively. The two estimators for QTEs at $\tau=0.25$ and $0.75$ are correlated. We can compute the naive and adjusted standard errors for the SQR estimator by taking this covariance structure into account.\footnote{The formulas for the covariances can be found in the proofs of Theorems \ref{thm:qr} and \ref{thm:ipw}.} On the other hand, in addition to avoiding the tuning parameters, another advantage of the bootstrap inference is it does not require the knowledge of this complicated covariance structure. Researchers may construct the t-statistic using the difference of two QTE estimators with the corresponding weighted and covariate-adaptive bootstrap standard errors, which are calculated using the exact same procedure as in Sections \ref{sec:SBI} and \ref{sec:CABI}. Taking the SQR estimator as an example, we estimate $q(0.25)-q(0.75)$ via $\hat{\beta}_1(0.25) - \hat{\beta}_1(0.75)$ and the corresponding covariate-adaptive bootstrap standard error is 
\begin{align*}
\hat{\sigma}_{CA} = \frac{\hat{Q}(0.975) - \hat{Q}(0.025)}{z_{0.975} - z_{0.025}},
\end{align*}
where  $\hat{Q}(\alpha)$ is the $\alpha$-th empirical quantile of $\{\hat{\beta}^*_{1,b}(0.25) - \hat{\beta}^*_{1,b}(0.75)\}_{b=1}^B$.

Based on the rejection rates reported in Tables \ref{tab:200_50_diff} and \ref{tab:400_50_diff}, the general observations for the previous simulation results still apply. Although under the null, the rejection rates for ``ipw/W", ``S/CA", ``ipw/CA" in DGP2 are below the nominal 5\%, they gradually increase as the sample size increases from 200 to 400.  
\begin{table}[H]
	\centering
	\caption{$n = 200,q(0.25)-q(0.75)$}
	\begin{adjustbox}{max width=\textwidth}
		\begin{tabular}{c|l|cccccc|c|cccccc}
			\multicolumn{8}{c|}{$H_0$}                                        &       & \multicolumn{6}{c}{$H_1$} \\ \hline 
			\multicolumn{1}{c}{M} & \multicolumn{1}{c}{A}   & \multicolumn{1}{l}{s/naive} & \multicolumn{1}{c}{s/adj} & \multicolumn{1}{c}{s/W} & \multicolumn{1}{c}{ipw/W} & \multicolumn{1}{c}{s/CA } & \multicolumn{1}{c|}{ipw/CA} & & \multicolumn{1}{l}{s/naive} & \multicolumn{1}{c}{s/adj} & \multicolumn{1}{c}{s/W} & \multicolumn{1}{c}{ipw/W} & \multicolumn{1}{c}{s/CA } & \multicolumn{1}{c}{ipw/CA} \\ \hline 		
			
			1     & SRS   & 4.0   & 4.0   & 3.6   & 3.8   & 3.5   & 3.5   &       & 15.6  & 15.6  & 14.9  & 19.4  & 16.0  & 19.4 \\
			& WEI   & 2.3   & 4.9   & 2.0   & 4.0   & 5.1   & 3.9   &       & 11.3  & 17.9  & 11.0  & 19.0  & 16.0  & 18.6 \\
			& BCD   & 1.0   & 4.1   & 1.1   & 4.4   & 3.7   & 4.2   &       & 9.9   & 20.7  & 10.1  & 22.0  & 20.6  & 21.4 \\
			& SBR   & 1.1   & 4.3   & 0.9   & 4.1   & 4.1   & 4.2   &       & 9.4   & 21.8  & 8.7   & 17.3  & 20.0  & 17.2 \\ \hline 
			2     & SRS   & 5.0   & 5.0   & 3.1   & 3.1   & 3.1   & 3.1   &       & 53.7  & 53.7  & 47.1  & 48.4  & 47.8  & 48.2 \\
			& WEI   & 3.6   & 3.6   & 2.1   & 2.8   & 2.9   & 2.9   &       & 57.0  & 57.7  & 47.6  & 49.8  & 50.3  & 50.0 \\
			& BCD   & 4.2   & 4.8   & 2.4   & 2.5   & 3.6   & 2.7   &       & 58.0  & 59.4  & 49.1  & 52.0  & 52.8  & 50.8 \\
			& SBR   & 5.1   & 5.3   & 2.4   & 3.4   & 4.1   & 3.4   &       & 55.5  & 57.0  & 46.5  & 46.5  & 50.5  & 45.6 \\
		\end{tabular}%
	\end{adjustbox}
	\label{tab:200_50_diff}%
\end{table}%

\begin{table}[H]
	\centering
	\caption{$n = 400,q(0.25)-q(0.75)$}
	\begin{adjustbox}{max width=\textwidth}
		\begin{tabular}{c|l|cccccc|c|cccccc}
			\multicolumn{8}{c|}{$H_0$}                                        &       & \multicolumn{6}{c}{$H_1$} \\ \hline 
			\multicolumn{1}{c}{M} & \multicolumn{1}{c}{A}   & \multicolumn{1}{l}{s/naive} & \multicolumn{1}{c}{s/adj} & \multicolumn{1}{c}{s/W} & \multicolumn{1}{c}{ipw/W} & \multicolumn{1}{c}{s/CA } & \multicolumn{1}{c|}{ipw/CA} & & \multicolumn{1}{l}{s/naive} & \multicolumn{1}{c}{s/adj} & \multicolumn{1}{c}{s/W} & \multicolumn{1}{c}{ipw/W} & \multicolumn{1}{c}{s/CA } & \multicolumn{1}{c}{ipw/CA} \\ \hline 		
			1     & SRS   & 3.8   & 3.8   & 3.9   & 5.1   & 3.7   & 5.0   &       & 17.2  & 17.2  & 15.9  & 21.5  & 16.8  & 21.2 \\
			& WEI   & 2.0   & 4.2   & 2.4   & 3.3   & 4.4   & 3.5   &       & 11.8  & 20.2  & 11.5  & 21.4  & 20.2  & 20.7 \\
			& BCD   & 1.4   & 4.4   & 1.4   & 4.3   & 4.4   & 4.1   &       & 10.5  & 21.8  & 10.2  & 20.7  & 21.5  & 20.6 \\
			& SBR   & 0.8   & 3.8   & 0.8   & 3.9   & 3.7   & 3.8   &       & 12.1  & 25.0  & 12.6  & 21.8  & 23.7  & 22.3 \\ \hline 
			2     & SRS   & 5.3   & 5.3   & 3.9   & 4.7   & 4.3   & 4.8   &       & 63.2  & 63.2  & 55.7  & 57.7  & 56.8  & 57.6 \\
			& WEI   & 5.4   & 5.8   & 3.4   & 3.7   & 4.1   & 3.5   &       & 63.6  & 64.4  & 55.6  & 58.0  & 58.0  & 58.5 \\
			& BCD   & 4.0   & 4.3   & 2.6   & 2.8   & 3.1   & 3.1   &       & 62.1  & 63.3  & 54.7  & 55.7  & 57.4  & 56.0 \\
			& SBR   & 5.1   & 5.7   & 4.0   & 4.5   & 4.4   & 4.5   &       & 61.1  & 62.0  & 52.4  & 51.3  & 56.0  & 53.0 \\
		\end{tabular}%
	\end{adjustbox}
	\label{tab:400_50_diff}%
\end{table}%

\section{Guidance for Practitioners}
\label{sec:guide}
We recommend employing the t-statistic (or equivalently, the confidence interval) constructed using the IPW estimator and its weighted bootstrap standard error for inference in covariate-adaptive randomization, for the following four reasons. First, its asymptotic size equals the nominal level. Second, the IPW estimator has a smaller asymptotic variance than the SQR estimator when the treatment assignment rule does not achieve strong balance and the stratification is relevant.\footnote{In this case, for ATE, the IPW estimator also has a strictly smaller asymptotic variance than the strata fixed effects estimator studied in \cite{BCS17}.} Third, compared with the covariate-adaptive bootstrap, the validity of the weighted bootstrap requires a weaker condition that $\sup_{s \in \mathcal{S} }|D_n(s)/n(s)| = o_p(1)$. Fourth, this method does not require the knowledge of the exact treatment assignment rule, thus is suitable in settings where such information is lacking, e.g., using someone else's RCT or subsample analysis. When the treatment assignment rule achieves strong balance, SQR estimator can also be used. But in this case, only the covariate-adaptive bootstrap standard error is valid. Last, the Wald test using SQR estimator and the weighted bootstrap standard error is not recommended, as it is conservative when the treatment assignment rule introduces negative dependence (i.e., $\gamma(s) < \pi(1-\pi)$) such as WEI, BCD, and SBR.

\section{Empirical Application}
\label{sec:app}

We illustrate our methods by estimating and inferring the average and quantile treatment effects of iron efficiency on educational attainment. The dataset we use is the same as the one analyzed by \citet{CCFNT16} and \citet{BCS17}. 

\subsection{Data Description}

The dataset consists of 215 students from one Peruvian secondary school during the 2009 school year. About two thirds of students were assigned to the treatment group ($A=1$ or $A=2$). The other one third of students were assigned to the control group ($A=0$). One half of the students in the treatment group were shown a video in which a physician encouraged iron supplements ($A=1$) and the other half were shown the same encouragement from a popular soccer player ($A=2$). Those assignments were stratified by the number of years of secondary school completed ($\mathcal{S}=\{1,\cdots,5\}$). The field experiment used a stratified block randomization scheme with fractions $(1/3, 1/3, 1/3)$ for each group, which achieves strong balance ($\gamma(s)=0$). 

In the following, we focus on the observations with $A=0$ and $A=1$, and estimate the treatment effect of the exposure to a video of encouraging iron supplements by a physician only. This practice was also implemented in \cite{BCS17}. In this case, the target proportions of treatment is $\pi = 1/2$. As in \cite{CCFNT16}, it is also possible to combine the two treatment groups, i.e., $A=1$ and $A=2$ and compute the treatment effects of exposure to a video of encouraging iron supplements by either a physician or a popular soccer player. Last, one can use the method developed in \cite{BCS18} to estimate the ATEs under multiple treatment status. However, in this setting, the estimation of QTE and the validity of bootstrap inference have not been investigated yet and are interesting topics for future research.

For each observation, we have three outcome variables: number of pills taken, grade point average, and cognitive ability measured by the average score across different Nintendo Wii games. For more details about the outcome variables, we refer interested readers to \cite{CCFNT16}. In the following, we focus on the grade point average only as the other two outcomes are discrete.

\subsection{Computation}
\label{sec:comp}
We consider three pairs of point estimates and their corresponding non-conservative standard errors: (1) the SQR estimator with the covariate-adaptive bootstrap standard error, (2) the IPW estimator with the covariate-adaptive bootstrap standard error, and (3) the IPW estimator with the weighted bootstrap standard error. We denote them as ``s/CA", ``ipw/CA", and ``ipw/W", respectively. For comparison, we also compute the SQR estimator with its  weighted bootstrap standard error, which is denoted as ``s/W" . The SQR estimator for the $\tau$-th QTE refers to $\hat{\beta}_1(\tau)$ as the second element of $\hat{\beta}(\tau) = (\hat{\beta}_0(\tau),\hat{\beta}_1(\tau))$, where 
\begin{align*}
\hat{\beta}(\tau) = \argmin_{b = (b_0,b_1)^\prime \in \Re^2} \sum_{i=1}^{n}\rho_\tau\left(Y_i - \dot{A}_i'b\right),
\end{align*}
$\dot{A}_i = (1,A_i)^\prime$, and $\rho_\tau(u) = u(\tau - 1\{u\leq 0\})$ is the standard check function. It is also just the difference between the $\tau$-th empirical quantiles of treatment and control groups. The IPW estimator refers to $\hat{q}(\tau) = \hat{q}_1(\tau) - \hat{q}_0(\tau)$, where
\begin{align*}
\hat{q}_1(\tau) = \argmin_q \frac{1}{n}\sum_{i=1}^n\frac{A_i}{\hat{\pi}(S_i)}\rho_\tau(Y_i - q), \quad \hat{q}_0(\tau) = \argmin_q\frac{1}{n}\sum_{i=1}^n \frac{1-A_i}{1-\hat{\pi}(S_i)}\rho_\tau(Y_i - q), 
\end{align*}
$\hat{\pi}(\cdot)$ denotes the propensity score estimator, $\hat{\pi}(s) = n_1(s)/n(s)$, $n_1(s) = \sum_{i=1}^n A_i 1\{S_i = s\}$, and $n(s) = \sum_{i=1}^n 1\{S_i = s\}$. The covariate-adaptive bootstrap standard error (``CA") refers to the standard error computed in Section \ref{sec:CABI}. In particular, we can draw the covariate-adaptive bootstrap sample $(Y_i^*,A_i^*,S_i^*)_{i=1}^n$ following the procedure in Section \ref{sec:CABI}. We then recompute the SQR and IPW estimates using the bootstrap sample. We repeat the bootstrap estimation $B$ times, and obtain 
$\{\hat{\beta}_{b,1}^*(\tau),\hat{q}_b^*(\tau)\}_{b=1}^B$. The standard errors for SQR and IPW estimates are computed as 
\begin{align*}
\hat{\sigma}_{sqr}(\tau) = \frac{\hat{Q}_{sqr}(0.975) - \hat{Q}_{sqr}(0.025)}{z_{0.975} - z_{0.025}} \quad \text{and} \quad \hat{\sigma}_{ipw}(\tau) = \frac{\hat{Q}_{ipw}(0.975) - \hat{Q}_{ipw}(0.025)}{z_{0.975} - z_{0.025}},
\end{align*}
respectively, where $\hat{Q}_{sqr}(\alpha)$ and $\hat{Q}_{ipw}(\alpha)$ are the $\alpha$-th empirical quantiles of $\{\hat{\beta}_{b,1}^*(\tau)\}_{b=1}^B$ and $\{\hat{q}_b^*(\tau)\}_{b=1}^B$, respectively, and $z_{\alpha}$ is the $\alpha$-th percentile of the standard normal distribution, i.e., $z_{0.975} \approx 1.96$ and $z_{0.025} \approx -1.96$. The weighted bootstrap standard error for the IPW estimate can be computed in the same manner with only one difference, the covariate-adaptive bootstrap estimator $\{\hat{q}_b^*(\tau)\}_{b=1}^B$ is replaced by the weighted bootstrap estimator $\{\hat{q}_b^w(\tau)\}_{b=1}^B$, where for the $b$-th replication,  $\hat{q}_b^w(\tau) = \hat{q}_{b,1}^w(\tau) - \hat{q}_{b,0}^w(\tau)$,
\begin{align*}
\hat{q}_{b,1}^w(\tau) = \argmin_q \frac{1}{n}\sum_{i=1}^n\frac{\xi_i^bA_i}{\hat{\pi}^w(S_i)}\rho_\tau(Y_i - q), \hat{q}_{b,1}^w(\tau) = \argmin_q\frac{1}{n}\sum_{i=1}^n \frac{\xi_i^b(1-A_i)}{1-\hat{\pi}^w(S_i)}\rho_\tau(Y_i - q),
\end{align*}
$\{\xi_i^b\}_{i=1}^n$ is a sequence of i.i.d. standard exponentially distributed random variables, $\hat{\pi}^w(s) = n_1^w(s)/n^w(s)$, $n_1^w(s) = \sum_{i =1}^n\xi_iA_i1\{S_i=s\}$, and $n^w(s) = \sum_{i =1}^n \xi_i 1\{S_i=s\}$. Similarly, we compute the weighted bootstrap SQR estimates $\{\beta_{b,1}^w(\tau)\}_{b=1}^B$ as the second element of $\{\beta_{b}^w(\tau)\}_{b=1}^B$, where 
\begin{align*}
\beta_{b}^w(\tau) = \argmin_{b = (b_0,b_1)^\prime \in \Re^2} \frac{1}{n}\sum_{i=1}^n \xi_i^b\rho_\tau(Y_i - \dot{A}_i'b).
\end{align*}

For the ATEs, we also compute the SQR estimator with the adjusted standard error based on the analytical formula derived by \cite{BCS17}, i.e., ``s/adj". For QTE estimates, we consider quantile indexes $\{0.1,0.15,\cdots,0.90\}$. The number of replications for the two bootstrap methods is $B=1000$.

\subsection{Main Results}
Table \ref{tab:gradesATE} shows the estimates with the corresponding standard errors in parentheses. From the table, we can make several remarks. First, for both ATE and QTE, the SQR and IPW estimates are very close to each other and so do their standard errors computed via the analytical formula, weighted bootstrap, and covariate-adaptive bootstrap. This is consistent with our theory that, under strong balance, the two estimators are first-order equivalent. Second, although in theory, the weighted bootstrap standard errors for the SQR estimators should be larger than those computed via the covariate-adaptive bootstrap, in this application, they are very close. This is consistent with the finding in \cite{BCS17} that their adjusted p-value for the ATE estimate is close to the naive one. It implies the stratification may be irrelevant for the full-sample analysis. Third, we do not compute the adjusted standard error for the QTEs as it requires tuning parameters. Fourth, the QTEs provide us a new insight that the impact of supplementation on grade promotion is only significantly positive at 25\% among the three quantiles. This may imply that the policy of reducing iron deficits is more effective for lower-ranked students.
\begin{table}[H]
	\centering
	\caption{Grades Points Average}
	\medskip
	\begin{tabular}{c|ccccc}\hline 
		& s/adj & s/W & s/CA&  ipw/W & ipw/CA  \\ \hline 
		ATE	&     0.35$^{**}$(0.16) &     0.35$^{**}$(0.16) &     0.35$^{**}$(0.17) &  0.37$^{**}$(0.16)     &     0.37$^{**}$(0.17)     \\ 
		QTE,25\%  &  &     0.43$^{***}$(0.15) &     0.43$^{***}$(0.15) &     0.43$^{***}$(0.15) &     0.43$^{***}$(0.15)  \\
		QTE,50\%  &  &     0.29(0.22) &     0.29(0.23) &     0.29(0.22) &     0.29(0.24) \\ 
		QTE,75\% &  &     0.35(0.25) &     0.35(0.24) &     0.36(0.25) &     0.36(0.25) \\  \hline
	\end{tabular}\\
	$^{*}$ $p<0.1$, $^{**}$ $p<0.05$, $^{***}$ $p<0.01$.
	\label{tab:gradesATE}%
\end{table}%

In order to provide more details on the QTE estimates, we plot the 95\% point-wise confidence band in Figure \ref{fig:quantileCI} with quantile index ranging from 0.1 to 0.9. The solid line and the shadow area represent the point estimate and its 95\% point-wise confidence interval, respectively. The confidence interval is constructed by $$[\hat{\beta}-1.96\hat{\sigma}(\hat{\beta}),\hat{\beta}+1.96\hat{\sigma}(\hat{\beta})],$$
where $\hat{\beta}$ and $\hat{\sigma}(\hat{\beta})$ are the point estimates and the corresponding standard errors described above. As we expected, all the four findings look the same and the estimates are only significantly positive at low quantiles (15\%--30\%). 

\begin{figure}[H]
	\centering
	\includegraphics[height=9cm,width=16cm]{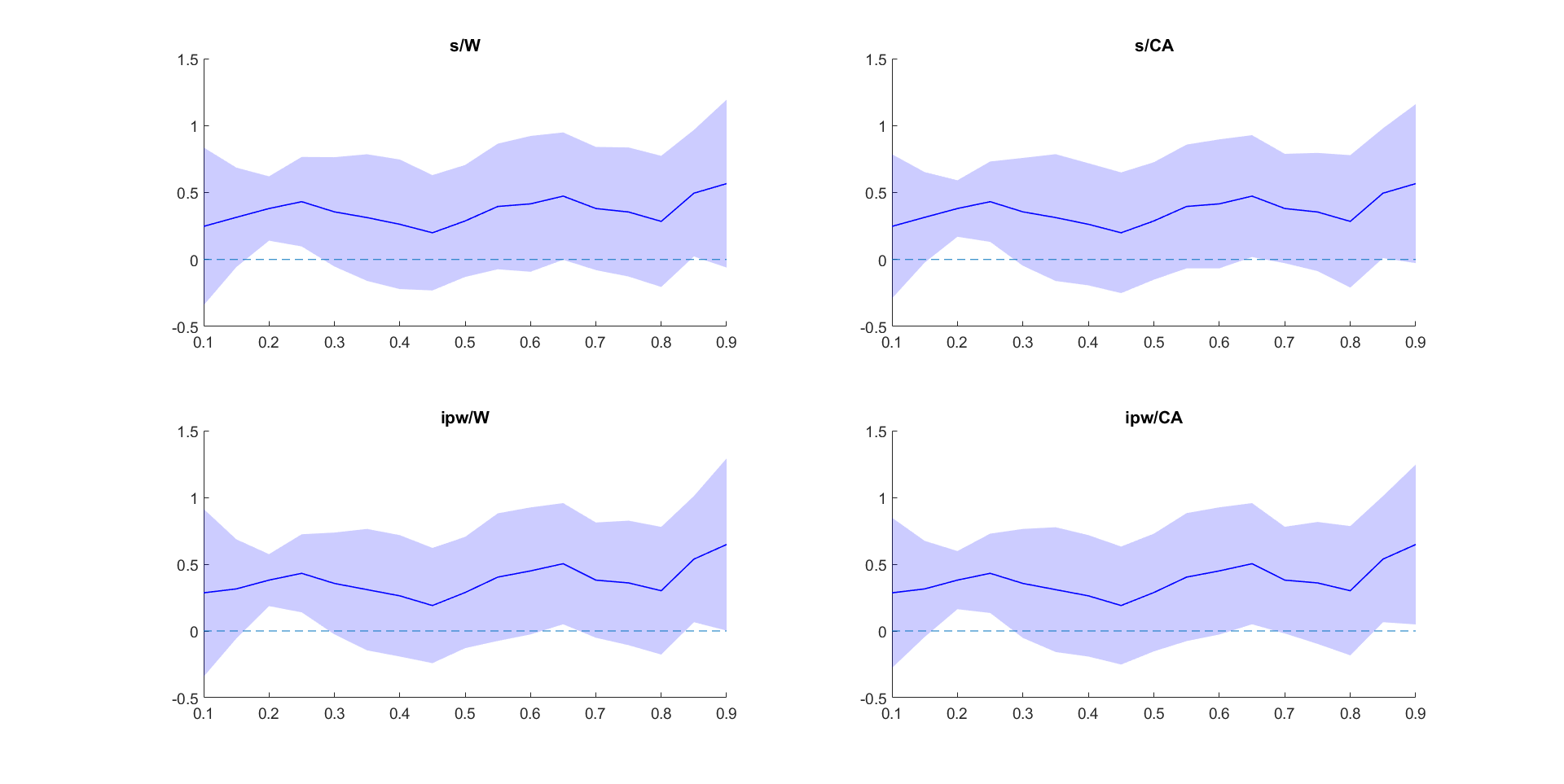}
	\begin{minipage}{0.75\linewidth}
		\footnotesize
	\end{minipage}
	\caption{95\% Point-wise Confidence Interval for Quantile Treatment Effects}
	\label{fig:quantileCI}
\end{figure}

\subsection{Subsample Results}
\label{sec:apply_sub}
Following \cite{CCFNT16}, we further split the sample into two based on whether the student is anemic, i.e., $Anem_i = 0$ or 1. We anticipate that there is no treatment effect for the nonanemic students and positive effects for anemic ones. In this subsample analysis, the covariate-adaptive bootstrap is infeasible, as in each sub-group, the strong-balance condition may be lost and the treatment assignment rule is not necessarily SBR and is generally unknown.\footnote{As the anonymous referee pointed out, it is possible to implement the covariate-adaptive bootstrap on the full sample and pick out the observations in the subsample to construct a bootstrap subsample. The analysis can then be repeated on this covariate-adaptive bootstrap subsample. Establishing the validity of this procedure is left as a topic for future research. } However, the weighted bootstrap is still feasible as it does not require the knowledge of the treatment assignment rule. According to Theorem \ref{thm:b}, the IPW estimator paired with the weighted bootstrap standard error is valid if
\begin{align}
\label{eq:D1}
\sup_{s\in\mathcal{S}}\left|\frac{D_n^{(1)}(s)}{n^{(1)}(s)}\right|\equiv&	\sup_{s\in\mathcal{S}}\left|\frac{\sum_{i =1}^n(A_i-\pi)1\{S_i=s\}1\{Anem_i=1\}}{\sum_{i =1}^n1\{S_i=s\}1\{Anem_i=1\}}\right|=o_p(1)
\end{align}
and 
\begin{align}
\label{eq:D0}
\sup_{s\in\mathcal{S}}\left|\frac{D_n^{(0)}(s)}{n^{(0)}(s)}\right|\equiv&	\sup_{s\in\mathcal{S}}\left|\frac{\sum_{i =1}^n(A_i-\pi)1\{S_i=s\}1\{Anem_i=0\}}{\sum_{i =1}^n1\{S_i=s\}1\{Anem_i=0\}}\right|=o_p(1).
\end{align}
We maintain this mild condition in this section. In our sample, 
\begin{align*}
\sup_{s\in\mathcal{S}}\left|\frac{D_n^{(1)}(s)}{n^{(1)}(s)}\right| = 0 \quad \text{and} \quad 	\sup_{s\in\mathcal{S}}\left|\frac{D_n^{(0)}(s)}{n^{(0)}(s)}\right| = 0.071, 
\end{align*}  
which indicate that \eqref{eq:D1} and \eqref{eq:D0} are plausible. 

\begin{table}[H]
	\centering
	\caption{Grades Points Average for Subsamples}
	\medskip
	\begin{tabular}{l|cc|cc}\hline 
		& \multicolumn{2}{c|}{Anemic} & \multicolumn{2}{c}{Nonanemic}  \\ \hline             
		& s/W & ipw/W  & s/W & ipw/W \\ \hline 
		ATE &     0.67$^{***}$(0.23) &     0.69$^{***}$(0.20) &     0.13(0.23) &     0.19(0.20)  \\  
		QTE, 25\%  &     0.74$^{***}$(0.24) &     0.76$^{***}$(0.22) &     0.14(0.28) &     0.22(0.26)  \\ 
		QTE, 50\%  &     1.05$^{***}$(0.29) &     1.05$^{***}$(0.27) &    -0.14(0.29) &    -0.14(0.27) \\ 
		QTE, 75\% &     0.71$^{**}$(0.36) &     0.76$^{**}$(0.32) &     0.14(0.39) &     0.14(0.37) \\    \hline  
	\end{tabular}\\
	$^{*}$ $p<0.1$, $^{**}$ $p<0.05$, $^{***}$ $p<0.01$.
	\label{tab:gradesATE'}%
\end{table}

From Table \ref{tab:gradesATE'} and Figure \ref{fig:quantileCIn}, we see that the QTE estimates are significantly positive for the anemic students when the quantile index is between around 20\%--75\%, but are insignificant for nonanemic students. The lack of significance at very low and high quantiles for the anemic subsample may be due to a poor asymptotic normal approximation at extreme quantiles. To extend the inference of extremal QTEs in \cite{Z18} to the context of covariate-adaptive randomization is an interesting topic for future research.  We also note that for both subsamples, the weighted bootstrap standard errors for the SQR estimators are larger than those for the IPW estimators, which is consistent with Theorem \ref{thm:b}. It implies, for both sub-groups, the stratification is relevant.

\begin{figure}[H]
	\centering
	\includegraphics[height=6cm,width=16cm]{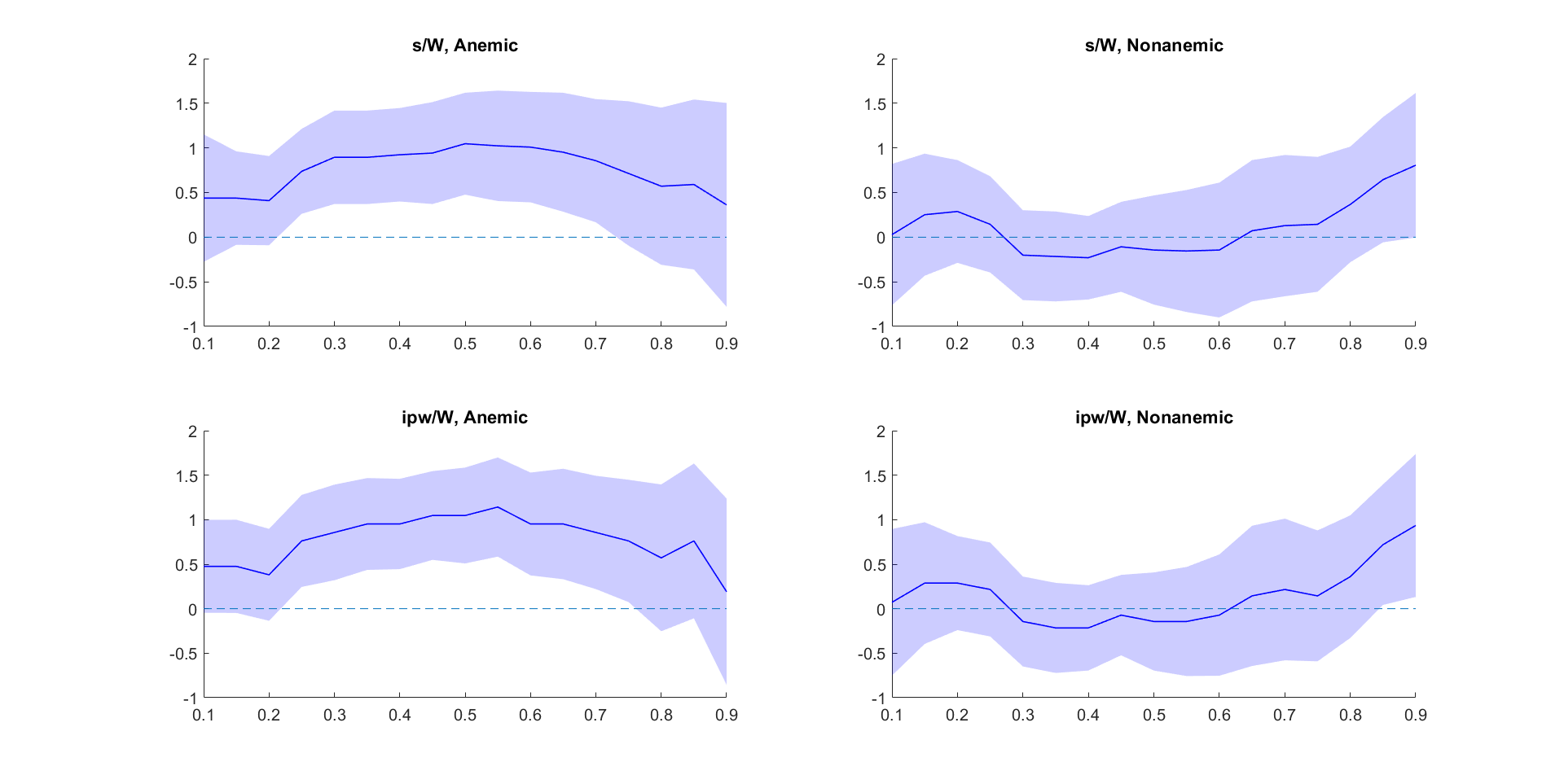}
	\caption{95\% Point-wise Confidence Interval for Anemic and Nonanemic Students}
	\label{fig:quantileCIn}
\end{figure}

\section{Conclusion}
\label{sec:concl}
This paper studies the estimation and bootstrap inference for QTEs under covariate-adaptive randomization. We show that the weighted bootstrap standard error is only valid for the IPW estimator while the covariate-adaptive bootstrap standard error is valid for both SQR and IPW estimators. In the empirical application, we find that the QTE of iron supplementation on grade promotion is trivial for nonanemic students, while the impact is significantly positive for middle-ranked anemic students.

\newpage
\appendix
\begin{center}
	\Large{Supplement to ``Quantile Treatment Effects and Bootstrap Inference under Covariate-Adaptive Randomization"}
\end{center}
\begin{abstract}
This paper  gathers the supplementary material to the original paper. Sections \ref{sec:qr}, \ref{sec:ipw_pf}, \ref{sec:b}, and \ref{sec:cab} contain the proofs for Theorems \ref{thm:qr}, \ref{thm:ipw}, \ref{thm:b}, and \ref{thm:cab}, respectively. Section \ref{sec:lem} contains the proofs for the technical lemmas. A separate supplement contains the analysis of strata fixed effects quantile regression estimator as well as additional simulation results.
\end{abstract}
\setcounter{page}{1} \renewcommand\thesection{\Alph{section}} %

\renewcommand{\thefootnote}{\arabic{footnote}} \setcounter{footnote}{0}

\setcounter{equation}{0}

\section{Proof of Theorem \ref{thm:qr}}
\label{sec:qr}
Let $u=(u_0, u_1)' \in \Re^2$ and 
\begin{align*}
L_n(u,\tau)  = \sum_{i=1}^n\left[\rho_\tau(Y_i - \dot{A}_i'\beta(\tau) - \dot{A}_i'u/\sqrt{n}) - \rho_\tau(Y_i - \dot{A}_i'\beta(\tau))\right].
\end{align*}  
Then, by the change of variable, we have that 
\begin{align*}
\sqrt{n}(\hat{\beta}(\tau) - \beta(\tau)) = \argmin_u L_n(u,\tau).
\end{align*}
Notice that $L_n(u,\tau)$ is convex in $u$ for each $\tau$ and bounded in $\tau$ for each $u$. In the following, we aim to show that there exists 
$$g_n(u,\tau) = - u'W_n(\tau) + \frac{1}{2}u'Q(\tau)u$$ 
such that (1) for each $u$, 
\begin{align*}
\sup_{\tau \in \Upsilon}|L_n(u,\tau) - g_n(u,\tau)| \convP 0;
\end{align*}
(2) the maximum eigenvalue of $Q(\tau)$ is bounded from above and the minimum eigenvalue of $Q(\tau)$ is bounded away from $0$, uniformly over $\tau \in \Upsilon$; (3) $W_n(\tau) \convD \tilde{\mathcal{B}}(\tau)$ uniformly over $\tau \in \Upsilon$, in which $\tilde{\mathcal{B}}(\cdot)$ is some Gaussian process. Then by \citet[Theorem 2]{K09}, we have 
\begin{align*}
\sqrt{n}(\hat{\beta}(\tau) - \beta(\tau)) = [Q(\tau)]^{-1}W_n(\tau) + r_n(\tau),
\end{align*}
where $\sup_{\tau \in \Upsilon}||r_n(\tau)|| = o_p(1)$. In addition, by (3), we have, uniformly over $\tau \in \Upsilon$, 
\begin{align*}
\sqrt{n}(\hat{\beta}(\tau) - \beta(\tau)) \convD [Q(\tau)]^{-1}\tilde{\mathcal{B}}(\tau) \equiv \mathcal{B}(\tau).
\end{align*}
The second element of $\mathcal{B}(\tau)$ is $\mathcal{B}_{sqr}(\tau)$ stated in Theorem \ref{thm:qr}. In the following, we prove requirements (1)--(3) in three steps. 

\textbf{Step 1.} By Knight's identity (\citep{K98}), we have 
\begin{align*}
& L_n(u,\tau) \\
= & -u'\sum_{i=1}^{n}\frac{1}{\sqrt{n}}\dot{A}_i\left(\tau- 1\{Y_i\leq \dot{A}_i^{\prime}\beta(\tau)\}\right) + \sum_{i=1}^{n}\int_0^{\frac{\dot{A}_iu}{\sqrt{n}}}\left(1\{Y_i -  \dot{A}_i^{\prime}\beta(\tau)\leq v\} - 1\{Y_i -  \dot{A}_i^{\prime}\beta(\tau)\leq 0\} \right)dv \\
\equiv & -u'W_n(\tau) +Q_n(u,\tau),
\end{align*}
where 
\begin{align*}
W_n(\tau) = \sum_{i=1}^{n}\frac{1}{\sqrt{n}}\dot{A}_i\left(\tau- 1\{Y_i\leq \dot{A}_i^{\prime}\beta(\tau)\}\right)
\end{align*}
and 
\begin{align*}
Q_n(u,\tau) = & \sum_{i=1}^{n}\int_0^{\frac{\dot{A}_i^\prime u}{\sqrt{n}}}\left(1\{Y_i -  \dot{A}_i^{\prime}\beta(\tau)\leq v\} - 1\{Y_i -  \dot{A}_i^{\prime}\beta(\tau)\leq 0\} \right)dv \\
= & \sum_{i=1}^n A_i\int_0^{\frac{u_0+u_1}{\sqrt{n}}}\left(1\{Y_i(1) - q_1(\tau)\leq v\} - 1\{Y_i(1) - q_1(\tau)\leq 0\}  \right)dv \\
& + \sum_{i=1}^n (1-A_i)\int_0^{\frac{u_0}{\sqrt{n}}}\left(1\{Y_i(0) - q_0(\tau)\leq v\} - 1\{Y_i(0) - q_0(\tau)\leq 0\}  \right)dv \\
\equiv & Q_{n,1}(u,\tau) + Q_{n,0}(u,\tau).
\end{align*}
We first consider $Q_{n,1}(u,\tau)$. Following \cite{BCS17}, we define $\{(Y_i^s(1),Y_i^s(0)): 1\leq i \leq n\}$ as a sequence of i.i.d. random variables with marginal distributions equal to the distribution of $(Y_i(1),Y_i(0))|S_i = s$. The distribution of $Q_{n,1}(u,\tau)$ is the same as the counterpart with units ordered by strata and then ordered by $A_i = 1$ first and $A_i = 0$ second within each stratum, i.e., 
\begin{align}
\label{eq:Qn1}
Q_{n,1}(u,\tau) \stackrel{d}{=} & \sum_{s \in \mathcal{S}} \sum_{i=N(s)+1}^{N(s)+n_1(s)}\int_0^{\frac{u_0+u_1}{\sqrt{n}}}\biggl(1\{Y_i^s(1) - q_1(\tau)\leq v\} - 1\{Y_i^s(1) - q_1(\tau)\leq 0\}  \biggr)dv \notag \\
= & \sum_{s \in \mathcal{S}}\biggl[\Gamma^s_{n}(N(s) + n_1(s),\tau) - \Gamma^s_{n}(N(s),\tau) \biggr],
\end{align} 
where $N(s) = \sum_{i=1}^n 1\{S_i < s\}$, $n_1(s) = \sum_{i=1}^n 1\{S_i = s\}A_i$, and
\begin{align*}
\Gamma^s_{n}(k,\tau) = \sum_{i=1}^{k}\int_0^{\frac{u_0+u_1}{\sqrt{n}}}\biggl(1\{Y_i^s(1) - q_1(\tau)\leq v\} - 1\{Y_i^s(1) - q_1(\tau)\leq 0\}  \biggr)dv.
\end{align*}

In addition, note that 
\begin{align}
\label{eq:Gamma}
& \mathbb{P}(\sup_{t \in (0,1),\tau \in \Upsilon}|\Gamma^s_{n}(\lfloor nt \rfloor,\tau) - \mathbb{E}\Gamma^s_{n}(\lfloor nt \rfloor,\tau)|> \eps) \notag \\
= & \mathbb{P}(\max_{1 \leq k \leq n}\sup_{\tau \in \Upsilon}|\Gamma^s_{n}(k,\tau) - \mathbb{E}\Gamma^s_{n}(k,\tau)|> \eps)  \notag \\
\leq & 3 \max_{1 \leq k \leq n}\mathbb{P}(\sup_{\tau \in \Upsilon}|\Gamma^s_{n}(k,\tau) - \mathbb{E}\Gamma^s_{n}(k,\tau)|> \eps/3) \notag \\
\leq & 9 \mathbb{P}(\sup_{\tau \in \Upsilon}|\Gamma^s_{n}(n,\tau) - \mathbb{E}\Gamma^s_{n}(n,\tau)|> \eps/30) \notag \\
\leq & \frac{270 \mathbb{E}\sup_{\tau \in \Upsilon}|\Gamma^s_{n}(n,\tau) - \mathbb{E}\Gamma^s_{n}(n,\tau)|}{\eps} = o(1).
\end{align}
The first inequality holds due to Lemma \ref{lem:S} with $S_k = \Gamma^s_{n}(k,\tau) - \mathbb{E}\Gamma^s_{n}(k,\tau)$ and $||S_k|| = \sup_{\tau \in \Upsilon} |\Gamma^s_{n}(k,\tau) - \mathbb{E}\Gamma^s_{n}(k,\tau)|.$ The second inequality holds due to \citet[Theorem 1]{M93}. To derive the last equality of \eqref{eq:Gamma}, we consider the class of functions
\begin{align*}
\mathcal{F} = \left\{\int_0^{\frac{u_0+u_1}{\sqrt{n}}}\biggl(1\{Y_i^s(1) - q_1(\tau)\leq v\} - 1\{Y_i^s(1) - q_1(\tau)\leq 0\}  \biggr)dv: \tau\in \Upsilon \right\}
\end{align*}
with envelope $\frac{|u_0+u_1|}{\sqrt{n}}$ and 
\begin{align*}
\sup_{f \in \mathcal{F}}\mathbb{E}f^2 \leq \sup_{\tau \in \Upsilon}\mathbb{E} \left[\frac{u_0+u_1}{\sqrt{n}}1\left\{|Y_i^s(1) - q_1(\tau)|\leq \frac{u_0+u_1}{\sqrt{n}}\right\} \right]^2 \lesssim n^{-3/2}. 
\end{align*}
Note that $\mathcal{F}$ is a VC-class with a fixed VC index. Therefore, by \citet[Corollary 5.1]{CCK14}, 
\begin{align*}
\mathbb{E}\sup_{\tau \in \Upsilon}|\Gamma^s_{n}(n,\tau) - \mathbb{E}\Gamma^s_{n}(n,\tau)| = n||\mathbb{P}_n  - \mathbb{P}||_\mathcal{F} \lesssim n\left[\sqrt{\frac{\log(n)}{n^{5/2}}} + \frac{\log(n)}{n^{3/2}}\right] = o(1). 
\end{align*} 
Then, \eqref{eq:Gamma} implies that 
\begin{align*}
\sup_{\tau \in \Upsilon}\left|Q_{n,1}(u,\tau) - \sum_{s\in \mathcal{S}}\mathbb{E}\biggl[\Gamma^s_{n}(\lfloor n(N(s)/n + n_1(s)/n)\rfloor,\tau) - \Gamma^s_{n}(\lfloor n(N(s)/n)\rfloor,\tau) \biggr]\right| = o_p(1),
\end{align*}
where following the convention in the empirical process literature,
\begin{align*}
\mathbb{E}\biggl[\Gamma^s_{n}(\lfloor n(N(s)/n + n_1(s)/n)\rfloor,\tau) - \Gamma^s_{n}(\lfloor n(N(s)/n)\rfloor,\tau) \biggr]
\end{align*}
is interpreted as 
\begin{align*}
\mathbb{E}\biggl[\Gamma^s_{n}(\lfloor nt_2\rfloor,\tau) - \Gamma^s_{n}(\lfloor nt_1\rfloor,\tau) \biggr]_{t_2 = \frac{N(s)}{n},t_2 = \frac{N(s) + n_1(s)}{n}}.
\end{align*} 
In addition,  $N(s)/n \convP F(s) = F(S_i < s)$ and $n_1(s)/n \convP \pi p(s)$. Thus, uniformly over $\tau \in \Upsilon$, 
\begin{align*}
& \mathbb{E}\biggl[\Gamma^s_{n}(\lfloor n(N(s)/n + n_1(s)/n)\rfloor,\tau) - \Gamma^s_{n}(\lfloor n(N(s)/n)\rfloor,\tau) \biggr] \\
= & n_1(s) \int_0^{\frac{u_0 + u_1}{\sqrt{n}}}(F_1(q_1(\tau)+v|s) - F_1(q_1(\tau)|s))dv \\
\convP & \frac{\pi p(s) f_1(q_1(\tau)|s)(u_0+u_1)^2}{2},
\end{align*}
where $F_1(\cdot|s)$ and $f_1(\cdot|s)$ are the conditional CDF and PDF of $Y_1$ given $S=s$, respectively. Then, uniformly over $\tau \in \Upsilon$,
\begin{align*}
Q_{n,1}(u,\tau) \convP \sum_{s \in \mathcal{S}}\frac{\pi p(s) f_1(q_1(\tau)|s)(u_0+u_1)^2}{2} = \frac{\pi f_1(q_1(\tau))(u_0+u_1)^2}{2}.
\end{align*}
Similarly, we can show that, uniformly over $\tau \in \Upsilon$, 
\begin{align*}
Q_{n,0}(u,\tau) \convP \frac{(1-\pi) f_0(q_0(\tau))u_0^2}{2}, 
\end{align*}
and thus
\begin{align*}
Q_n(u,\tau) \convP  \frac{1}{2}u'Q(\tau)u,
\end{align*}
where 
\begin{align}
\label{eq:Qqr}
Q(\tau) = \begin{pmatrix}
\pi f_1(q_1(\tau)) + (1-\pi)f_0(q_0(\tau)) & \pi f_1(q_1(\tau)) \\
\pi f_1(q_1(\tau)) & \pi f_1(q_1(\tau))
\end{pmatrix}.
\end{align}
Then, 
\begin{align*}
\sup_{\tau \in \Upsilon}|L_n(u,\tau) - g_n(u,\tau)| = \sup_{\tau \in \Upsilon}|Q_n(u,\tau) - \frac{1}{2}u'Q(\tau)u| = o_p(1). 
\end{align*}
This concludes the first step. 

\textbf{Step 2}. Note that $\text{det}(Q(\tau)) = \pi(1-\pi)f_1(q_1(\tau))f_0(q_0(\tau))$, which is bounded and bounded away from zero. In addition, it can be shown that the two eigenvalues of $Q$ are nonnegative. This leads to the desired result.

\textbf{Step 3.} Let $e_1 = (1,1)^\prime$ and $e_0 = (1,0)^{\prime}$. Then, we have 
\begin{align*}
W_n(\tau) = &e_1 \sum_{s \in \mathcal{S}} \sum_{i=1}^n \frac{1}{\sqrt{n}}A_i 1\{S_i = s\}(\tau - 1\{Y_i(1) \leq q_1(\tau)\})\\
& + e_0\sum_{s \in \mathcal{S}} \sum_{i=1}^n \frac{1}{\sqrt{n}} (1-A_i) 1\{S_i = s\}(\tau - 1\{Y_i(0) \leq q_0(\tau)\}).
\end{align*}
Let $m_j(s,\tau) = \mathbb{E}(\tau - 1\{Y_i(j)\leq q_j(\tau)\}|S_i = s)$ and $\eta_{i,j}(s,\tau) = (\tau - 1\{Y_i(j) \leq q_j(\tau)\}) - m_j(s,\tau)$, $j = 0,1$. Then, 
\begin{align}
\label{eq:W}
W_n(\tau) = & \biggl[e_1 \sum_{s \in \mathcal{S}} \sum_{i=1}^n \frac{1}{\sqrt{n}}A_i 1\{S_i = s\}\eta_{i,1}(s,\tau) + e_0\sum_{s \in \mathcal{S}} \sum_{i=1}^n \frac{1}{\sqrt{n}} (1-A_i) 1\{S_i = s\}\eta_{i,0}(s,\tau)\biggr] \notag \\
& + \biggl[e_1 \sum_{s \in \mathcal{S}} \sum_{i=1}^n \frac{1}{\sqrt{n}}(A_i-\pi) 1\{S_i = s\}m_1(s,\tau) - e_0\sum_{s \in \mathcal{S}} \sum_{i=1}^n \frac{1}{\sqrt{n}} (A_i-\pi) 1\{S_i = s\}m_0(s,\tau)\biggr] \notag \\
& + \biggl[e_1 \sum_{s \in \mathcal{S}} \sum_{i=1}^n \frac{1}{\sqrt{n}}\pi 1\{S_i = s\}m_1(s,\tau) + e_0\sum_{s \in \mathcal{S}} \sum_{i=1}^n \frac{1}{\sqrt{n}} (1-\pi) 1\{S_i = s\}m_0(s,\tau) \biggr] \notag \\
\equiv &  W_{n,1}(\tau) + W_{n,2}(\tau) + W_{n,3}(\tau).
\end{align}
By Lemma \ref{lem:Q}, uniformly over $\tau \in \Upsilon$, 
\begin{align*}
(W_{n,1}(\tau), W_{n,2}(\tau), W_{n,3}(\tau)) \convD (\mathcal{B}_1(\tau),\mathcal{B}_2(\tau),\mathcal{B}_3(\tau)),
\end{align*}
where $(\mathcal{B}_1(\tau),\mathcal{B}_2(\tau),\mathcal{B}_3(\tau))$ are three independent two-dimensional Gaussian processes with covariance kernels $\Sigma_1(\tau_1,\tau_2)$, $\Sigma_2(\tau_1,\tau_2)$, and $\Sigma_3(\tau_1,\tau_2)$, respectively. Therefore, uniformly over $\tau \in \Upsilon$,
\begin{align*}
W_n(\tau) \convD \tilde{\mathcal{B}}(\tau),
\end{align*}
where $\tilde{\mathcal{B}}(\tau)$ is a two-dimensional Gaussian process with covariance kernel
\begin{align*}
\tilde{\Sigma}(\tau_1,\tau_2) = \sum_{j=1}^3\Sigma_j(\tau_1,\tau_2). 
\end{align*}
Consequently, 
\begin{align*}
\sqrt{n}(\hat{\beta}(\tau) - \beta(\tau)) \convD [Q(\tau)]^{-1}\tilde{\mathcal{B}}(\tau) \equiv \mathcal{B}(\tau),
\end{align*}
where $\mathcal{B}(\tau)$ is a two-dimensional Gaussian process with covariance kernel
\begin{align*}
\Sigma(\tau_1,\tau_2) = & [Q(\tau_1)]^{-1}\tilde{\Sigma}(\tau_1,\tau_2)  [Q(\tau_2)]^{-1} \\
= & \frac{1}{\pi f_1(q_1(\tau_1))f_1(q_1(\tau_2))}[\min(\tau_1,\tau_2) - \tau_1\tau_2 - \mathbb{E}m_1(S,\tau_1)m_1(S,\tau_2)]\begin{pmatrix}
0 & 0 \\
0 & 1
\end{pmatrix} 
\\
& + \frac{1}{(1-\pi) f_0(q_0(\tau_1))f_0(q_0(\tau_2))}[\min(\tau_1,\tau_2) - \tau_1\tau_2 - \mathbb{E}m_0(S,\tau_1)m_0(S,\tau_2)]\begin{pmatrix}
1 & -1 \\
-1 & 1
\end{pmatrix} \\
& + \sum_{s \in \mathcal{S}}p(s)\gamma(s)\biggl[\frac{m_1(s,\tau_1)m_1(s,\tau_2)}{\pi^2 f_1(q_1(\tau_1))f_1(q_1(\tau_2))}\begin{pmatrix}
0 & 0 \\
0& 1
\end{pmatrix} - \frac{m_1(s,\tau_1)m_0(s,\tau_2)}{\pi(1-\pi) f_1(q_1(\tau_1))f_0(q_0(\tau_2))}\begin{pmatrix}
0 & 0 \\
1 & -1
\end{pmatrix} \\
& - \frac{m_0(s,\tau_1)m_1(s,\tau_2)}{\pi(1-\pi) f_0(q_0(\tau_1))f_1(q_1(\tau_2))}\begin{pmatrix}
0 & 1 \\
0 & -1
\end{pmatrix} + \frac{m_0(s,\tau_1)m_0(s,\tau_2)}{(1-\pi)^2 f_0(q_0(\tau_1))f_0(q_0(\tau_2))}\begin{pmatrix}
1 & -1 \\
-1 & 1
\end{pmatrix} \biggr] \\
& + \frac{\mathbb{E}m_1(S,\tau_1)m_1(S,\tau_2)}{f_1(q_1(\tau_1))f_1(q_1(\tau_2))}\begin{pmatrix}
0 & 0 \\
0 & 1
\end{pmatrix} + \frac{\mathbb{E}m_1(S,\tau_1)m_0(S,\tau_2)}{f_1(q_1(\tau_1))f_0(q_0(\tau_2))}\begin{pmatrix}
0 & 0 \\
1 & -1
\end{pmatrix} \\
& + \frac{\mathbb{E}m_0(S,\tau_1)m_1(S,\tau_2)}{f_0(q_0(\tau_1))f_1(q_1(\tau_2))}\begin{pmatrix}
0 & 1 \\
0 & -1
\end{pmatrix}+\frac{\mathbb{E}m_0(S,\tau_1)m_0(S,\tau_2)}{f_0(q_0(\tau_1))f_0(q_0(\tau_2))}\begin{pmatrix}
1 & -1 \\
-1 & 1
\end{pmatrix}.
\end{align*}

Focusing on the $(2,2)$-element of $\Sigma(\tau_1,\tau_2)$, we can conclude that 
\begin{align*}
\sqrt{n}(\hat{\beta}_1(\tau) - q(\tau)) \convD \mathcal{B}_{sqr}(\tau),
\end{align*}
where the Gaussian process $\mathcal{B}_{sqr}(\tau)$ has a covariance kernel 
\begin{align*}
& \Sigma_{sqr}(\tau_1,\tau_2) \\
= & \frac{\min(\tau_1,\tau_2) - \tau_1\tau_2 - \mathbb{E}m_1(S,\tau_1)m_1(S,\tau_2)}{\pi f_1(q_1(\tau_1))f_1(q_1(\tau_2))} + \frac{\min(\tau_1,\tau_2) - \tau_1\tau_2 - \mathbb{E}m_0(S,\tau_1)m_0(S,\tau_2)}{(1-\pi) f_0(q_0(\tau_1))f_0(q_0(\tau_2))} \\
& + \mathbb{E}\gamma(S)\biggl[\frac{m_1(S,\tau_1)m_1(S,\tau_2)}{\pi^2 f_1(q_1(\tau_1))f_1(q_1(\tau_2))}+\frac{m_1(S,\tau_1)m_0(S,\tau_2)}{\pi(1-\pi) f_1(q_1(\tau_1))f_0(q_0(\tau_2))} \\
&+ \frac{m_0(S,\tau_1)m_1(S,\tau_2)}{\pi(1-\pi) f_0(q_0(\tau_1))f_1(q_1(\tau_2))} +  \frac{m_0(S,\tau_1)m_0(S,\tau_2)}{(1-\pi)^2 f_0(q_0(\tau_1))f_0(q_0(\tau_2))}\biggr] \\
& +\mathbb{E}\biggl[\frac{m_1(S,\tau_1)}{f_1(q_1(\tau_1))} - \frac{m_0(S,\tau_1)}{f_0(q_0(\tau_1))}\biggr] \biggl[\frac{m_1(S,\tau_2)}{f_1(q_1(\tau_2))} - \frac{m_0(S,\tau_2)}{f_0(q_0(\tau_2))}\biggr].
\end{align*}

\section{Proof of Theorem \ref{thm:ipw}}
\label{sec:ipw_pf}
By Knight's identity, we have 
\begin{align*}
\sqrt{n}(\hat{q}_1(\tau) - q_1(\tau)) = & \argmin_u L_n(u,\tau),
\end{align*}
where 
\begin{align*}
L_n(u,\tau) \equiv & \sum_{i =1}^n \frac{A_i}{\hat{\pi}(S_i)}\left[\rho_\tau(Y_i - q_1(\tau) - \frac{u}{\sqrt{n}}) - \rho_\tau(Y_i - q_1(\tau))\right] \\
= & - L_{1,n}(\tau)u + L_{2,n}(u,\tau),
\end{align*}
\begin{align*}
L_{1,n}(\tau) = \frac{1}{\sqrt{n}}\sum_{i=1}^n \frac{A_i}{\hat{\pi}(S_i)}(\tau - 1\{Y_i \leq q_1(\tau)\})
\end{align*}
and 
\begin{align*}
L_{2,n}(u,\tau) = \sum_{i=1}^n \frac{A_i}{\hat{\pi}(S_i)}\int_0^{\frac{u}{\sqrt{n}}}(1\{Y_i \leq q_1(\tau)+v\} - 1\{Y_i \leq q_1(\tau)\})dv.
\end{align*}
We aim to show that there exists 
\begin{align}
\label{eq:gipw}
g_{ipw,n}(u,\tau) = - W_{ipw,n}(\tau) u + \frac{1}{2}Q_{ipw}(\tau)u^2
\end{align}
such that (1) for each $u$, 
\begin{align*}
\sup_{\tau \in \Upsilon}|L_n(u,\tau) - g_{ipw,n}(u,\tau)|\convP 0;
\end{align*}
(2) $Q_{ipw}(\tau)$ is bounded and bounded away from zero uniformly over $\tau \in \Upsilon$. In addition, as a corollary of claim (3) below, $\sup_{\tau \in \Upsilon}|W_{ipw,1,n}(\tau)| = O_p(1)$. Therefore, by \citet[Theorme 2]{K09}, we have 
\begin{align*}
\sqrt{n}(\hat{q}_1(\tau) - q_1(\tau)) = Q_{ipw,1}^{-1}(\tau)W_{ipw,1,n}(\tau) + R_{ipw,1,n}(\tau),
\end{align*}
where $\sup_{\tau \in \Upsilon}|R_{ipw,1,n}(\tau)| = o_p(1)$. Similarly, we can show that 
\begin{align*}
\sqrt{n}(\hat{q}_0(\tau) - q_0(\tau)) = Q_{ipw,0}^{-1}(\tau)W_{ipw,0,n}(\tau) + R_{ipw,0,n}(\tau),
\end{align*} 
where $\sup_{\tau \in \Upsilon}|R_{ipw,0,n}(\tau)| = o_p(1)$. Then, 
\begin{align*}
\sqrt{n}(\hat{q}(\tau) - q(\tau)) =  Q_{ipw,1}^{-1}(\tau)W_{ipw,1,n}(\tau) - Q_{ipw,0}^{-1}(\tau)W_{ipw,0,n}(\tau) + R_{ipw,1,n}(\tau) - R_{ipw,0,n}(\tau).
\end{align*}
Last, we aim to show that, (3) uniformly over $\tau \in \Upsilon$, 
\begin{align*}
Q_{ipw,1}^{-1}(\tau)W_{ipw,1,n}(\tau) - Q_{ipw,0}^{-1}(\tau)W_{ipw,0,n}(\tau) \convD \mathcal{B}_{ipw}(\tau),
\end{align*}
where $\mathcal{B}_{ipw}(\tau)$ is a scalar Gaussian process with covariance kernel $\Sigma_{ipw}(\tau_1,\tau_2)$. We prove claims (1)--(3) in three steps. 

\textbf{Step 1.} For $L_{1,n}(\tau)$, we have 
\begin{align*}
L_{1,n}(\tau) = & \frac{1}{\sqrt{n}}\sum_{i=1}^n \sum_{s \in \mathcal{S}}\frac{A_i}{\pi}1\{S_i=s\}(\tau - 1\{Y_i(1) \leq q_1(\tau)\}) \\
& - \sum_{i=1}^n \sum_{s \in \mathcal{S}}\frac{A_i1\{S_i = s\}(\hat{\pi}(s) - \pi)}{\sqrt{n}\hat{\pi}(s)\pi}(\tau - 1\{Y_i(1) \leq q_1(\tau)\}) \\
= & \frac{1}{\sqrt{n}}\sum_{i=1}^n \sum_{s \in \mathcal{S}}\frac{A_i}{\pi}1\{S_i=s\}(\tau - 1\{Y_i(1) \leq q_1(\tau)\}) \\
& - \sum_{i=1}^n \sum_{s \in \mathcal{S}}\frac{A_i1\{S_i = s\}D_n(s)}{n(s)\sqrt{n}\hat{\pi}(s)\pi}\eta_{i,1}(s,\tau) -  \sum_{s \in \mathcal{S}} \frac{D_n(s)m_1(s,\tau)}{n(s)\sqrt{n}\hat{\pi}(s)\pi} D_n(s) - \sum_{s \in \mathcal{S}} \frac{D_n(s)m_1(s,\tau)}{\sqrt{n}\hat{\pi}(s)} \\ 
= & \sum_{s \in \mathcal{S}} \frac{1}{\sqrt{n}}\sum_{i =1}^n \frac{A_i 1\{S_i = s\}}{\pi}\eta_{i,1}(s,\tau) + \sum_{s \in \mathcal{S}} \frac{D_n(s)}{\sqrt{n}\pi}m_{1}(s,\tau) + \sum_{i=1}^n \frac{m_1(S_i,\tau)}{\sqrt{n}} \\
& - \sum_{i=1}^n \sum_{s \in \mathcal{S}}\frac{A_i1\{S_i = s\}D_n(s)}{n(s)\sqrt{n}\hat{\pi}(s)\pi}\eta_{i,1}(s,\tau) -  \sum_{s \in \mathcal{S}} \frac{D_n(s)m_1(s,\tau)}{n(s)\sqrt{n}\hat{\pi}(s)\pi} D_n(s) - \sum_{s \in \mathcal{S}} \frac{D_n(s)m_1(s,\tau)}{\sqrt{n}\hat{\pi}(s)} \\
= & W_{ipw,1,n}(\tau)+ R_{ipw}(\tau),
\end{align*}
where 
\begin{align}
\label{eq:wipw}
W_{ipw,1,n}(\tau) =  \sum_{s \in \mathcal{S}} \frac{1}{\sqrt{n}}\sum_{i =1}^n \frac{A_i 1\{S_i = s\}}{\pi}\eta_{i,1}(s,\tau)  + \sum_{i=1}^n \frac{m_1(S_i,\tau)}{\sqrt{n}} 
\end{align}
and 
\begin{align*}
& R_{ipw}(\tau) \\
= & - \sum_{i=1}^n \sum_{s \in \mathcal{S}}\frac{A_i1\{S_i = s\}D_n(s)}{n(s)\sqrt{n}\hat{\pi}(s)\pi}\eta_{i,1}(s,\tau) -  \sum_{s \in \mathcal{S}} \frac{D_n(s)m_1(s,\tau)}{n(s)\sqrt{n}\hat{\pi}(s)\pi} D_n(s) + \sum_{s \in \mathcal{S}} \frac{D_n(s)m_1(s,\tau)}{\sqrt{n}}\left(\frac{1}{\pi} - \frac{1}{\hat{\pi}(s)}\right) \\
= & - \sum_{i=1}^n \sum_{s \in \mathcal{S}}\frac{A_i1\{S_i = s\}D_n(s)}{n(s)\sqrt{n}\hat{\pi}(s)\pi}\eta_{i,1}(s,\tau),
\end{align*}
where we use the fact that $\hat{\pi}(s) - \pi = \frac{D_n(s)}{n(s)}$. By the same argument in Claim (1) of the proof of Lemma \ref{lem:Q}, we have, for every $s\in\mathcal{S}$,
\begin{align}
\label{eq:eta1}
\sup_{\tau \in \Upsilon}\left|\frac{1}{\sqrt{n}}\sum_{i =1}^n A_i 1\{S_i=s\}\eta_{i,1}(s,\tau)\right| \stackrel{d}{=} \sup_{\tau \in \Upsilon}\left|\frac{1}{\sqrt{n}}\sum_{i =N(s)+1}^{N(s)+n(s)} \tilde{\eta}_{i,1}(s,\tau)\right| = O_p(1),
\end{align}
where $\tilde{\eta}_{i,j}(s,\tau) = \tau - 1\{Y_i^s(j) \leq q_j(\tau) \} - m_j(s,\tau)$, for $j = 0,1$, where $\{Y_i^s(0),Y_i^s(1)\}_{i \geq 1}$ are the same as defined in Step 1 in the proof of Theorem \ref{thm:qr}.

Because of \eqref{eq:eta1} and the fact that $\frac{D_n(s)}{n(s)} = o_p(1)$, we have 
\begin{align*}
\sup_{\tau \in \Upsilon}|R_{ipw}(\tau)| = o_p(1).
\end{align*}
For $L_{2,n}(u,\tau)$, we have 
\begin{align*}
L_{2,n}(u,\tau) = & \sum_{s \in \mathcal{S}} \frac{1}{\hat{\pi}(s)}\sum_{i=N(s)+1}^{N(s)+n_1(s)} \int_0^{\frac{u}{\sqrt{n}}}(1\{Y_i^s(1) \leq q_1(\tau)+v\} - 1\{Y_i^s(1) \leq q_1(\tau)+v\})dv \\
= & \sum_{s \in \mathcal{S}} \frac{1}{\hat{\pi}(s)} \left[\Gamma_n^s(N(s)+n_1(s),\tau) - \Gamma_n^s(N(s),\tau) \right],
\end{align*}
where 
\begin{align*}
\Gamma_n^s(k,\tau) = \sum_{i=1}^k \int_0^{\frac{u}{\sqrt{n}}}(1\{Y_i^s(1) \leq q_1(\tau)+v\} - 1\{Y_i^s(1) \leq q_1(\tau)+v\})dv.
\end{align*}
By the same argument in \eqref{eq:Gamma}, we can show that 
\begin{align*}
\sup_{t \in (0,1),\tau \in \Upsilon}|\Gamma_n^s(\lfloor nt \rfloor,\tau) - \mathbb{E}\Gamma_n^s(\lfloor nt \rfloor,\tau)| = o_p(1).
\end{align*}
In addition,  
\begin{align*}
\mathbb{E}\Gamma_n^s(N(s)+n_1(s),\tau) - \mathbb{E}\Gamma_n^s(N(s),\tau) \convP \frac{\pi p(s)f_1(q_1(\tau)|s)u^2}{2}.
\end{align*}
Therefore, 
\begin{align*}
\sup_{\tau \in \Upsilon}\left|L_{2,n}(u,\tau) - \frac{f_1(q_1(\tau))u^2}{2}\right| = o_p(1),
\end{align*}
where we use the fact that $\hat{\pi}(s)-\pi = \frac{D_{n}(s)}{n(s)} = o_p(1)$ and 
\begin{align*}
\sum_{s \in \mathcal{S}}p(s)f_1(q_1(\tau)|s) = f_1(q_1(\tau)). 
\end{align*}
This establishes \eqref{eq:gipw} with $Q_{ipw,1}(\tau) = f_1(q_1(\tau))$ and $W_{ipw,n}(\tau)$ defined in \eqref{eq:wipw}.

\textbf{Step 2.} Statement (2) holds by Assumption \ref{ass:tau}. 

\textbf{Step 3.} By a similar argument in Step 1, we have 
\begin{align*}
W_{ipw,0,n}(\tau) = \sum_{s \in \mathcal{S}} \frac{1}{\sqrt{n}}\sum_{i =1}^n \frac{(1-A_i) 1\{S_i = s\}}{1-\pi}\eta_{i,0}(s,\tau)  + \sum_{i=1}^n \frac{m_0(S_i,\tau)}{\sqrt{n}} 
\end{align*}
and $Q_{ipw,0}(\tau) = f_0(q_0(\tau))$. Therefore, 
\begin{align}
\label{eq:wipw2}
\sqrt{n}(\hat{q} - q) = & \frac{1}{\sqrt{n}}\sum_{s \in \mathcal{S}} \sum_{i =1}^n \left[\frac{A_i 1\{S_i = s\}\eta_{i,1}(s,\tau)}{\pi f_1(q_1(\tau))} - \frac{(1-A_i)1\{S_i = s\}\eta_{i,0}(s,\tau)}{(1-\pi) f_0(q_0(\tau))}\right] \notag \\
& + \left[\frac{1}{\sqrt{n}}\sum_{i=1}^n\left(\frac{m_1(S_i,\tau)}{f_1(q_1(\tau))} -  \frac{m_0(S_i,\tau)}{f_0(q_0(\tau))}\right)\right] + R_{ipw,n}(\tau) \notag \\
= & \mathcal{W}_{n,1}(\tau)+\mathcal{W}_{n,2}(\tau)+ R_{ipw,n}(\tau)
\end{align}
where $\sup_{\tau \in \Upsilon}|R_{ipw,n}(\tau)| = o_p(1)$. Last, Lemma \ref{lem:Qipw} establishes that 
\begin{align*}
(\mathcal{W}_{n,1}(\tau),\mathcal{W}_{n,2}(\tau)) \convD (\mathcal{B}_{ipw,1}(\tau),\mathcal{B}_{ipw,2}(\tau)),
\end{align*}
where $(\mathcal{B}_{ipw,1}(\tau),\mathcal{B}_{ipw,2}(\tau))$ are two mutually independent scalar Gaussian processes with covariance kernels 
\begin{align*}
\Sigma_{ipw,1}(\tau_1,\tau_2) = \frac{\min(\tau_1,\tau_2) - \tau_1\tau_2 - \mathbb{E}m_1(S,\tau_1)m_1(S,\tau_2)  }{\pi f_1(q_1(\tau_1))f_1(q_1(\tau_2))} + \frac{\min(\tau_1,\tau_2) - \tau_1\tau_2 - \mathbb{E}m_0(S,\tau_1)m_0(S,\tau_2)}{(1-\pi)f_0(q_0(\tau_1))f_0(q_0(\tau_2)) }
\end{align*}
and 
\begin{align*}
\Sigma_{ipw,2}(\tau_1,\tau_2) = \mathbb{E}\left(\frac{m_1(S,\tau_1)}{f_1(q_1(\tau_1))} - \frac{m_0(S,\tau_1)}{f_0(q_0(\tau_1))}\right)\left(\frac{m_1(S,\tau_2)}{f_1(q_1(\tau_2))} - \frac{m_0(S,\tau_2)}{f_0(q_0(\tau_2))}\right),
\end{align*}
respectively. In particular, the asymptotic variance for $\hat{q}$ is 
\begin{align*}
\zeta_Y^2(\pi,\tau) + \zeta_S^2(\tau), 
\end{align*}
where $\zeta_Y^2(\pi,\tau)$ and $\zeta_S^2(\tau)$ are the same as those in the proof of Theorem \ref{thm:qr}.

\section{Proof of Theorem \ref{thm:b}}
\label{sec:b}
First, we consider the weighted bootstrap for the SQR estimator. Note that 
\begin{align*}
\sqrt{n}(\hat{\beta}^w(\tau) - \beta(\tau)) = \argmin_{u}L_n^w(u,\tau),
\end{align*}
where 
\begin{align*}
L_n^w(u,\tau) = \sum_{i=1}^n \xi_i\left[\rho_\tau(Y_i - \dot{A}_i'\beta(\tau) - \dot{A}_i'u/\sqrt{n}) - \rho_\tau(Y_i - \dot{A}_i'\beta(\tau))\right]. 
\end{align*}
Similar to the proof of Theorem \ref{thm:qr}, we can show that 
\begin{align*}
\sup_{ \tau \in \Upsilon}|L_n^w(u,\tau) - g_n^w(u,\tau)| \rightarrow 0, 
\end{align*}
where 
\begin{align*}
g_n^w(u,\tau) = - u'W_n^w(\tau) + \frac{1}{2}u'Q(\tau)u, 
\end{align*}
\begin{align*}
W_n^w(\tau) = \sum_{i=1}^n \frac{\xi_{i}}{\sqrt{n}}\dot{A}_i\left(\tau - 1\{Y_i \leq \dot{A}'\beta(\tau)\}\right),
\end{align*}
and $Q(\tau)$ is defined in \eqref{eq:Qqr}. Therefore, by \citet[Theorem 2]{K09}, we have 
\begin{align*}
\sqrt{n}(\hat{\beta}^w(\tau) - \beta(\tau)) = [Q(\tau)]^{-1}W_n^w(\tau) + r_n^w(\tau),
\end{align*}
where $\sup_{\tau \in \Upsilon}||r_n^w(\tau)|| = o_p(1)$. By Theorem \ref{thm:qr}, 
\begin{align*}
\sqrt{n}(\hat{\beta}^w(\tau) - \hat{\beta}(\tau)) =  [Q(\tau)]^{-1}\sum_{i=1}^n \frac{\xi_{i}-1}{\sqrt{n}}\dot{A}_i\left(\tau - 1\{Y_i \leq \dot{A}'\beta(\tau)\}\right) + o_p(1),
\end{align*}
where the $o_p(1)$ term holds uniformly over $\tau \in \Upsilon$. In addition, Lemma \ref{lem:Qwqr} shows that, conditionally on data, the second element of 
$[Q(\tau)]^{-1}\sum_{i=1}^n \frac{\xi_{i}-1}{\sqrt{n}}\dot{A}_i\left(\tau - 1\{Y_i \leq \dot{A}'\beta(\tau)\}\right)$ converges to $\tilde{\mathcal{B}}_{sqr}(\tau)$ uniformly over $\tau \in \Upsilon$. This leads to the desired result for the weighted bootstrap simple quantile regression estimator.

Next, we turn to the IPW estimator. Denote $\hat{q}_j^w(\tau)$, $j=0,1$ the weighted bootstrap counterpart of $\hat{q}_j(\tau)$. We have
\begin{align*}
\sqrt{n}(\hat{q}_1^w(\tau) - q_1(\tau)) = \argmin_u L_n^w(u,\tau),
\end{align*}
where 
\begin{align*}
L_n^w(u,\tau) = & \sum_{i =1}^n \frac{\xi_iA_i}{\hat{\pi}^w(S_i)}\left[\rho_\tau(Y_i - q_1(\tau) - \frac{u}{\sqrt{n}}) - \rho_\tau(Y_i -q_1(\tau))\right] \\
\equiv & - L_{1,n}^w(\tau) u + L_{2,n}^w(u,\tau),
\end{align*}
where 
\begin{align*}
L_{1,n}^w(\tau) = \frac{1}{\sqrt{n}}\sum_{i =1}^n \frac{\xi_iA_i}{\hat{\pi}^w(S_i)}(\tau - 1\{Y_i \leq q_1(\tau)\})
\end{align*}
and
\begin{align*}
L_{2,n}^w(\tau) = \sum_{i=1}^n \frac{\xi_iA_i}{\hat{\pi}^w(S_i)}\int_0^{\frac{u}{\sqrt{n}}}(1\{Y_i \leq q_1(\tau)+v\} - 1\{Y_i \leq q_1(\tau)\})dv.
\end{align*}
Recall 
$$D_n^w(s) = \sum_{i =1}^n \xi_i(A_i - \pi)1\{S_i = s\}, \quad n^w(s) = \sum_{i =1}^n \xi_i1\{S_i = s\},$$ 
and 
$$\hat{\pi}^w(s) = \frac{\sum_{i=1}^n \xi_iA_i1\{S_i=s\}}{n^w(s)} = \pi + \frac{D_n^w(s)}{n^w(s)}.$$ 
Then, for $L_{1,n}^w(\tau)$, we have 
\begin{align*}
L_{1,n}^w(\tau) = & \frac{1}{\sqrt{n}}\sum_{i=1}^n \sum_{s \in \mathcal{S}}\frac{\xi_iA_i}{\pi}1\{S_i=s\}(\tau - 1\{Y_i(1) \leq q_1(\tau)\}) \\
& - \sum_{i=1}^n \sum_{s \in \mathcal{S}}\frac{\xi_iA_i1\{S_i = s\}(\hat{\pi}^w(s) - \pi)}{\sqrt{n}\hat{\pi}^w(s)\pi}(\tau - 1\{Y_i(1) \leq q_1(\tau)\}) \\
= & \frac{1}{\sqrt{n}}\sum_{i=1}^n \sum_{s \in \mathcal{S}}\frac{\xi_iA_i}{\pi}1\{S_i=s\}(\tau - 1\{Y_i(1) \leq q_1(\tau)\}) \\
& - \sum_{i=1}^n \sum_{s \in \mathcal{S}}\frac{\xi_iA_i1\{S_i = s\}D^w_n(s)}{n^w(s)\sqrt{n}\hat{\pi}(s)\pi}\eta_{i,1}(s,\tau) -  \sum_{s \in \mathcal{S}} \frac{D^w_n(s)m_1(s,\tau)}{n^w(s)\sqrt{n}\hat{\pi}^w(s)\pi} D^w_n(s) - \sum_{s \in \mathcal{S}} \frac{D^w_n(s)m_1(s,\tau)}{\sqrt{n}\hat{\pi}^w(s)} \\ 
= & \sum_{s \in \mathcal{S}} \frac{1}{\sqrt{n}}\sum_{i =1}^n \frac{\xi_iA_i 1\{S_i = s\}}{\pi}\eta_{i,1}(s,\tau) + \sum_{s \in \mathcal{S}} \frac{D^w_n(s)}{\sqrt{n}\pi}m_{1}(s,\tau) + \sum_{i=1}^n \frac{\xi_im_1(S_i,\tau)}{\sqrt{n}} \\
& - \sum_{s \in \mathcal{S}}D^w_n(s)\sum_{i=1}^n \frac{\xi_iA_i1\{S_i = s\}}{n^w(s)\sqrt{n}\hat{\pi}^w(s)\pi}\eta_{i,1}(s,\tau) -  \sum_{s \in \mathcal{S}} \frac{D^w_n(s)m_1(s,\tau)}{n^w(s)\sqrt{n}\hat{\pi}^w(s)\pi} D^w_n(s) - \sum_{s \in \mathcal{S}} \frac{D^w_n(s)m_1(s,\tau)}{\sqrt{n}\hat{\pi}^w(s)} \\
= & W_{ipw,1,n}^w(\tau)+ R_{ipw}^w(\tau),
\end{align*}
where 
\begin{align}
\label{eq:wipwc}
W_{ipw,1,n}^w(\tau) =  \sum_{s \in \mathcal{S}} \frac{1}{\sqrt{n}}\sum_{i =1}^n \frac{\xi_iA_i 1\{S_i = s\}}{\pi}\eta_{i,1}(s,\tau)  + \sum_{i=1}^n \frac{\xi_im_1(S_i,\tau)}{\sqrt{n}} 
\end{align}
and 
\begin{align*}
& R_{ipw}^w(\tau) \\
= & - \sum_{s \in \mathcal{S}}D^w_n(s)\sum_{i=1}^n \frac{\xi_iA_i1\{S_i = s\}}{n^w(s)\sqrt{n}\hat{\pi}^w(s)\pi}\eta_{i,1}(s,\tau) -  \sum_{s \in \mathcal{S}} \frac{D^w_n(s)m_1(s,\tau)}{n^w(s)\sqrt{n}\hat{\pi}^w(s)\pi} D^w_n(s) + \sum_{s \in \mathcal{S}} \frac{D^w_n(s)m_1(s,\tau)}{\sqrt{n}}(\frac{1}{\pi} - \frac{1}{\hat{\pi}^w(s)}) \\
= & - \sum_{s \in \mathcal{S}}D^w_n(s)\sum_{i=1}^n \frac{\xi_iA_i1\{S_i = s\}}{n^w(s)\sqrt{n}\hat{\pi}^w(s)\pi}\eta_{i,1}(s,\tau).
\end{align*}
In the following, we aim to show $D_n^w(s)/n^w(s) = o_p(1)$ and 
\begin{align*}
\sup_{\tau \in \Upsilon, s \in \mathcal{S}}|\sum_{i=1}^{n}\xi_iA_i1\{S_i =s\}\eta_{i,1}(s,\tau)| = O_p(\sqrt{n}). 
\end{align*}
For the first claim, we note that $n^w(s)/n(s) \convP 1$ and $D_n(s)/n(s) \convP 0$. Therefore, we only need to show 
\begin{align*}
\frac{D_n^w(s) - D_n(s)}{n(s)} = \sum_{i=1}^n\frac{(\xi_i-1)(A_i-\pi)1\{S_i=s\}}{n(s)} \convP 0. 
\end{align*}
As $n(s) \rightarrow \infty$ a.s., given data, 

\begin{align*}
\frac{1}{n(s)}\sum_{i=1}^n(A_i - \pi)^21\{S_i = s\} = & \frac{1}{n}\sum_{i=1}^n\left(A_i - \pi - 2\pi(A_i - \pi) + \pi - \pi^2\right)1\{S_i=s\} \\
= & \frac{D_n(s) - 2\pi D_n(s)}{n(s)} + \pi(1-\pi) \convP \pi (1-\pi).  
\end{align*}
Then, by the Lindeberg CLT, conditionally on data, 
\begin{align*}
\frac{1}{\sqrt{n(s)}} \sum_{i =1}^n (\xi_i - 1)(A_i - \pi)1\{S_i = s\} \convD N(0,\pi(1-\pi)) = O_p(1),  
\end{align*}
and thus 
\begin{align*}
\frac{D_n^w(s) - D_n(s)}{n(s)} = O_p(n^{-1/2}(s)) = o_p(1). 
\end{align*}
This leads to the first claim. For the second claim, we note that 
\begin{align*}
\sum_{i=1}^{n}\xi_iA_i1\{S_i =s\}\eta_{i,1}(s,\tau) = \sum_{i = N(s)+1}^{N(s)+n_1(s)}\xi_i \tilde{\eta}_{i,1}(s,\tau). 
\end{align*}
We can show the RHS of the above display is $O_p(\sqrt{n})$ for all $s \in \mathcal{S}$ following the same argument used in Claim (1) of the proof of Lemma \ref{lem:Q}. Given these two claims and by noticing that 
\begin{align*}
\hat{\pi}^w(s) - \pi = \frac{D^w_n(s)}{n^w(s)} = o_p(1),
\end{align*}
we have 
\begin{align*}
\sup_{\tau \in \Upsilon}|R_{ipw}^w(\tau)| = o_p(1). 
\end{align*}

Similar to the argument used to derive the limit of $L_{2,n}(\tau)$ in the proof of Theorem \ref{thm:ipw}, we can show that 
\begin{align*}
\sup_{ \tau \in \Upsilon}|L_{2,n}^w(u,\tau) - \frac{f_1(q_1(\tau))u^2}{2}| = o_p(1). 
\end{align*}
Therefore, 
\begin{align*}
\sqrt{n}(\hat{q}_1^w(\tau) - q_1(\tau)) = \frac{W_{ipw,1,n}^w(\tau)}{f_1(q_1(\tau))} + R_1^w(\tau),
\end{align*}
where $\sup_{ \tau \in \Upsilon}|R_1^w(\tau)| = o_p(1)$. Similarly, 
\begin{align*}
\sqrt{n}(\hat{q}_0^w(\tau) - q_0(\tau)) = \frac{W_{ipw,0,n}^w(\tau)}{f_0(q_0(\tau))} + R_0^w(\tau),
\end{align*}
where 
$$W_{ipw,0,n}^w(\tau) = \sum_{s \in \mathcal{S}} \frac{1}{\sqrt{n}}\sum_{i =1}^n \frac{\xi_i(1-A_i)1\{S_i = s\}}{1-\pi}\eta_{i,0}(s,\tau) + \sum_{i =1}^n\frac{\xi_im_0(S_i,\tau)}{\sqrt{n}}$$
and $\sup_{ \tau \in \Upsilon}|R_0^w(\tau)| = o_p(1)$. Therefore, 
\begin{align*}
& \sqrt{n}(\hat{q}^w(\tau) - \hat{q}(\tau)) \\
= & \sum_{s \in \mathcal{S}}\frac{1}{\sqrt{n}}\sum_{i =1}^n (\xi_i-1)\biggl\{\frac{A_i1\{S_i = s\}\eta_{i,1}(s,\tau)}{\pi f_1(q_1(\tau))}-\frac{(1-A_i)1\{S_i = s\}\eta_{i,0}(s,\tau)}{(1-\pi)f_0(q_0(\tau))}  \\
& + \left[\frac{m_1(s,\tau)}{f_1(q_1(\tau))} - \frac{m_0(s,\tau)}{f_0(q_0(\tau))}\right]1\{S_i=s\}\biggr\} + o_p(1),
\end{align*}
where the $o_p(1)$ term holds uniformly over $\tau \in \Upsilon.$ In order to show the conditional weak convergence, we only need to show the conditionally stochastic equicontinuity and finite-dimensional convergence. The former can be shown in the same manner as Lemma \ref{lem:Qwqr}. For the latter, we note that 
\begin{align*}
& \frac{1}{n}\sum_{s \in \mathcal{S}}\sum_{i =1}^n \biggl\{\frac{A_i1\{S_i = s\}\eta_{i,1}(s,\tau)}{\pi f_1(q_1(\tau))}-\frac{(1-A_i)1\{S_i = s\}\eta_{i,0}(s,\tau)}{(1-\pi)f_0(q_0(\tau))} + \left[\frac{m_1(s,\tau)}{f_1(q_1(\tau))} - \frac{m_0(s,\tau)}{f_0(q_0(\tau))}\right]1\{S_i=s\}\biggr\}^2 \\
= & \sum_{s \in \mathcal{S}}\frac{1}{n}\sum_{i =1}^n \biggl\{\frac{(1-A_i)1\{S_i = s\}\eta_{i,0}(s,\tau)}{(1-\pi)f_0(q_0(\tau))}\biggr\}^2 + \sum_{s \in \mathcal{S}}\frac{1}{n}\sum_{i =1}^n \biggl\{\frac{A_i1\{S_i = s\}\eta_{i,1}(s,\tau)}{\pi f_1(q_1(\tau))}\biggr\}^2 \\
& + \sum_{s \in \mathcal{S}}\frac{1}{n}\sum_{i =1}^n  \biggl\{\left[\frac{m_1(s,\tau)}{f_1(q_1(\tau))} - \frac{m_0(s,\tau)}{f_0(q_0(\tau))}\right]1\{S_i=s\}\biggr\}^2 \\
& + \sum_{s \in \mathcal{S}}\frac{2}{n}\sum_{i =1}^n \biggl\{\frac{A_i1\{S_i = s\}\eta_{i,1}(s,\tau)}{\pi f_1(q_1(\tau))}\biggr\}\left[\frac{m_1(s,\tau)}{f_1(q_1(\tau))} - \frac{m_0(s,\tau)}{f_0(q_0(\tau))}\right] \\
& - \sum_{s \in \mathcal{S}}\frac{2}{n}\sum_{i =1}^n\biggl\{\frac{(1-A_i)1\{S_i = s\}\eta_{i,0}(s,\tau)}{(1-\pi)f_0(q_0(\tau))}\biggr\} \left[\frac{m_1(s,\tau)}{f_1(q_1(\tau))} - \frac{m_0(s,\tau)}{f_0(q_0(\tau))}\right] \\
\convP & \zeta_Y^2(\pi,\tau) + \zeta_S^2(\tau). 
\end{align*}
Note that the RHS of the above display is the same as the asymptotic variance of the original estimator $\hat{q}(\tau)$. By the CLT conditional on data, we can establish the one-dimensional weak convergence. Then, by the Cram\'{e}r-Wold Theorem, we can extend such result to any finite dimension. This concludes the proof. 

\section{Proof of Theorem \ref{thm:cab}}
\label{sec:cab}
It suffices to prove the theorem with 
\begin{align*}
\tilde{q}(\tau)  =  &  q(\tau)+\biggl[\sum_{s \in \mathcal{S}} \sum_{i = \lfloor n F(s) \rfloor+1}^{\lfloor n (F(s)+\pi p(s)) \rfloor}\frac{\tilde{\eta}_{i,1}(s,\tau)}{n\pi f_1(q_1(\tau))} -  \sum_{s \in \mathcal{S}} \sum_{i = \lfloor n (F(s)+\pi p(s)) \rfloor+1}^{\lfloor n (F(s)+ p(s)) \rfloor}\frac{\tilde{\eta}_{i,0}(s,\tau)}{n(1-\pi)f_0(q_0(\tau))}\biggr] \\
& + \biggl[\sum_{i=1}^n \frac{1}{n} \left(\frac{m_1(S_i,\tau) }{f_1(q_1(\tau))} - \frac{m_0(S_i,\tau) }{f_0(q_0(\tau))}\right) \biggr],
\end{align*}
as we have shown in Theorem \ref{thm:ipw} that 
\begin{align*}
\sup_{ \tau \in \Upsilon}|\tilde{q}(\tau) - \hat{q}(\tau)| = o_p(1/\sqrt{n}). 
\end{align*} 

We first consider the SQR estimator. Note that 
\begin{align*}
\sqrt{n}(\hat{\beta}^*(\tau) - \beta(\tau)) = \argmin_{u} L^*_n(u,\tau),
\end{align*}
where $L^*_n(u,\tau) = \sum_{i=1}^n \left[\rho_\tau(Y_i^* - \dot{A}^{*'}_i\beta(\tau) - \dot{A}^{*'}_iu/\sqrt{n}) - \rho_\tau(Y_i^* - \dot{A}^{*'}_i\beta(\tau))\right]$. Then, $\hat{\beta}_1^*(\tau)$, the bootstrap counterpart of the SQR estimator, is just the second element of $\hat{\beta}^*(\tau)$. Similar to the proof of Theorem \ref{thm:qr},
\begin{align*}
L^*_n(u,\tau) = -u'W_n^*(\tau) + Q_n^*(u,\tau),
\end{align*}
where 
\begin{align*}
W_n^*(\tau) = \sum_{i =1}^n \frac{1}{\sqrt{n}}\dot{A}_i^*(\tau - 1\{Y_i^* \leq \dot{A}_i^{*'}\beta(\tau) \})
\end{align*}
and 
\begin{align}
\label{eq:Qstar}
Q_n^*(u,\tau) = & \sum_{i=1}^{n}\int_0^{\frac{\dot{A}_i^{*\prime} u}{\sqrt{n}}}\left(1\{Y_i^* -  \dot{A}_i^{*\prime}\beta(\tau)\leq v\} - 1\{Y_i^* -  \dot{A}_i^{*\prime}\beta(\tau)\leq 0\} \right)dv \notag \\
= & \sum_{i=1}^n A^*_i\int_0^{\frac{u_0+u_1}{\sqrt{n}}}\left(1\{Y_i^*(1) - q_1(\tau)\leq v\} - 1\{Y_i^*(1) - q_1(\tau)\leq 0\}  \right)dv \notag \\
& + \sum_{i=1}^n (1-A^*_i)\int_0^{\frac{u_0}{\sqrt{n}}}\left(1\{Y_i^*(0) - q_0(\tau)\leq v\} - 1\{Y_i^*(0) - q_0(\tau)\leq 0\}  \right)dv \notag \\
\equiv & Q^*_{n,1}(u,\tau) + Q^*_{n,0}(u,\tau).
\end{align} 

Define $\eta_{i,j}^*(s,\tau) = (\tau - 1\{Y_i^*(j)\leq q_j(\tau) \}) - m_j(s,\tau)$ and $\tilde{\eta}_{i,j}(s,\tau) = \tau - 1\{Y^s_i(j) \leq q_j(\tau)\} - m_j(s,\tau)$, $j = 0,1$, where $Y^s_i(j)$ is defined in the proof of Theorem \ref{thm:qr}. Then, we have 
\begin{align*}
W_n^*(\tau) = &e_1 \sum_{s \in \mathcal{S}} \sum_{i=1}^n \frac{1}{\sqrt{n}}A^*_i 1\{S^*_i = s\}(\tau - 1\{Y^*_i(1) \leq q_1(\tau)\})\\
& + e_0\sum_{s \in \mathcal{S}} \sum_{i=1}^n \frac{1}{\sqrt{n}} (1-A^*_i) 1\{S^*_i = s\}(\tau - 1\{Y^*_i(0) \leq q_0(\tau)\}) \\
= & \biggl[e_1 \sum_{s \in \mathcal{S}} \sum_{i=1}^n \frac{1}{\sqrt{n}}A^*_i 1\{S^*_i = s\}\eta^*_{i,1}(s,\tau) + e_0\sum_{s \in \mathcal{S}} \sum_{i=1}^n \frac{1}{\sqrt{n}} (1-A^*_i) 1\{S^*_i = s\}\eta^*_{i,0}(s,\tau)\biggr] \notag \\
& + \biggl[e_1 \sum_{s \in \mathcal{S}} \sum_{i=1}^n \frac{1}{\sqrt{n}}(A^*_i-\pi) 1\{S^*_i = s\}m_1(s,\tau) - e_0\sum_{s \in \mathcal{S}} \sum_{i=1}^n \frac{1}{\sqrt{n}} (A_i^*-\pi) 1\{S_i^* = s\}m_0(s,\tau)\biggr] \notag \\
& + \biggl[e_1 \sum_{s \in \mathcal{S}} \sum_{i=1}^n \frac{1}{\sqrt{n}}\pi 1\{S^*_i = s\}m_1(s,\tau) + e_0\sum_{s \in \mathcal{S}} \sum_{i=1}^n \frac{1}{\sqrt{n}} (1-\pi) 1\{S^*_i = s\}m_0(s,\tau) \biggr] \notag \\
\equiv &  W^*_{n,1}(\tau) + W^*_{n,2}(\tau) + W^*_{n,3}(\tau).
\end{align*}

By Lemma \ref{lem:Rsfestar}, there exists a sequence of independent Poisson(1) random variables $\{\xi_i^s\}_{i \geq 1,s\in \mathcal{S}}$ such that $\{\xi_i^s\}_{i \geq 1,s\in \mathcal{S}} \indep \{A_i^*,S_i^*,Y_i,A_i,S_i\}_{i\geq 1}$,
\begin{align*}
\sum_{i =1}^nA_i^*1\{S_i^*=s\}\eta_{i,1}^*(s,\tau) = \sum_{i = N(s)+1}^{N(s)+n_1(s)}\xi_i^s\tilde{\eta}_{i,1}(s,\tau) + R^*_1(s,\tau),
\end{align*}
and
\begin{align*}
\sum_{i =1}^n(1-A_i^*)1\{S_i^*=s\}\eta_{i,1}^*(s,\tau) = \sum_{i = N(s)+n_1(s)+1}^{N(s)+n(s)}\xi_i^s\tilde{\eta}_{i,0}(s,\tau) + R^*_0(s,\tau),
\end{align*}
where $\sup_{\tau \in \Upsilon}(|R^*_1(s,\tau)| + |R^*_0(s,\tau)|) = o_p(\sqrt{n(s)}) = o_p(\sqrt{n})$ for all $s \in \mathcal{S}$. Therefore, 
\begin{align*}
(W^*_{n,1}(\tau),W^*_{n,2}(\tau),W^*_{n,3}(\tau)) \stackrel{d}{=} (\tilde{W}^*_{n,1}(\tau) + R(\tau),W^*_{n,2}(\tau), W^*_{n,3}(\tau))
\end{align*}
where $\sup_{\tau \in \Upsilon}||R(\tau)|| = o_p(1)$ and 
\begin{align*}
\tilde{W}^*_{n,1}(\tau) =  e_1 \sum_{s \in \mathcal{S}} \sum_{i = N(s)+1}^{N(s)+n_1(s)}\frac{\xi_i^s}{\sqrt{n}}\tilde{\eta}_{i,1}(s,\tau) + e_0\sum_{s \in \mathcal{S}} \sum_{i = N(s)+n_1(s)+1}^{N(s)+n(s)}\frac{\xi_i^s}{\sqrt{n}}\tilde{\eta}_{i,0}(s,\tau)
\end{align*}

In addition, following the same argument in the proof of Lemma \ref{lem:Q}, we can further show that 
\begin{align*}
\tilde{W}^*_{n,1}(\tau) = W^{*\ast}_{n,1}(\tau) + R^\ast_n(\tau),
\end{align*}
where $\sup_{\tau \in \Upsilon}||R^\ast_n(\tau)||= o_p(1)$ and 
\begin{align*}
W^{*\ast}_{n,1}(\tau) = e_1 \sum_{s \in \mathcal{S}} \sum_{i = \lfloor n F(s) \rfloor+1}^{\lfloor n (F(s)+\pi p(s)) \rfloor}\frac{\xi_i^s}{\sqrt{n}}\tilde{\eta}_{i,1}(s,\tau) + e_0\sum_{s \in \mathcal{S}} \sum_{i = \lfloor n (F(s)+\pi p(s)) \rfloor+1}^{\lfloor n (F(s)+ p(s)) \rfloor}\frac{\xi_i^s}{\sqrt{n}}\tilde{\eta}_{i,0}(s,\tau).
\end{align*}
By construction, $W^{*\ast}_{n,1}(\tau) \indep (W^{*}_{n,2}(\tau),W^{*}_{n,3}(\tau))$. Also note that $\{S_i^*\}_{i=1}^n$ are the nonparametric bootstrap draws based on the empirical CDF of $\{S_i\}_{i=1}^n$. Then, by \citet[Section 3.6]{VW96}, there exists a sequence of independent Poisson(1) random variables $\{\tilde{\xi}_i\}_{i \geq 1}$ that is independent of data, $\{A_i^*\}$ and $\{\xi^s_i\}_{i \geq 1,s\in \mathcal{S}}$ such that 
\begin{align*}
\sup_{\tau \in \Upsilon}||W_{n,3}^*(\tau) - W_{n,3}^{**}(\tau)|| = o_p(1),
\end{align*}
where 
\begin{align*}
W_{n,3}^{**}(\tau) = e_1 \sum_{s \in \mathcal{S}} \sum_{i=1}^n \frac{\tilde{\xi}_i}{\sqrt{n}}\pi 1\{S_i = s\}m_1(s,\tau) + e_0\sum_{s \in \mathcal{S}} \sum_{i=1}^n \frac{\tilde{\xi}_i}{\sqrt{n}} (1-\pi) 1\{S_i = s\}m_0(s,\tau) 
\end{align*}
By Lemma \ref{lem:Qstar}, 
\begin{align*}
Q_n^*(u,\tau) \convP \frac{1}{2}u'Q(\tau)u,
\end{align*}
where $Q(\tau)$ is defined in \eqref{eq:Qqr}. Then, by the same argument in the proof of Theorem \ref{thm:qr}, we have 
\begin{align*}
\sqrt{n}(\hat{\beta}^*(\tau) - \beta(\tau)) = Q^{-1}(\tau)(W^{*\ast}_{n,1}(\tau)+W^{*}_{n,2}(\tau)+W^{**}_{n,3}(\tau)) + R^*(\tau),
\end{align*}
where $\sup_{\tau \in \Upsilon}||R^*(\tau)|| = o_p(1)$. Focusing on the second element of $\hat{\beta}^*(\tau)$, we have 
\begin{align*}
\sqrt{n}(\hat{\beta}^*_1(\tau) - q(\tau)) = & \biggl[\sum_{s \in \mathcal{S}} \sum_{i = \lfloor n F(s) \rfloor+1}^{\lfloor n (F(s)+\pi p(s)) \rfloor}\frac{\xi_i^s\tilde{\eta}_{i,1}(s,\tau)}{\sqrt{n}\pi f_1(q_1(\tau))} -  \sum_{s \in \mathcal{S}} \sum_{i = \lfloor n (F(s)+\pi p(s)) \rfloor+1}^{\lfloor n (F(s)+ p(s)) \rfloor}\frac{\xi_i^s\tilde{\eta}_{i,0}(s,\tau)}{\sqrt{n}(1-\pi)f_0(q_0(\tau))}\biggr] \\
& + \biggl[ \sum_{s \in \mathcal{S}} \frac{D_n^*(s)}{\sqrt{n}}\left(\frac{m_1(s,\tau)}{\pi f_1(q_1(\tau))} + \frac{m_0(s,\tau)}{\pi f_0(q_0(\tau))}\right) \biggr] \\
& + \biggl[\sum_{i=1}^n \frac{\tilde{\xi}_i}{\sqrt{n}} \left(\frac{m_1(S_i,\tau) }{f_1(q_1(\tau))} - \frac{m_0(S_i,\tau) }{f_0(q_0(\tau))}\right) \biggr] + R_1^*(\tau),
\end{align*}
where $\sup_{ \tau \in \Upsilon}|R_1^*(\tau)| = o_p(1)$. In addition, by definition, we have 
\begin{align*}
\sqrt{n}(\tilde{q}(\tau) - q(\tau)) =  & \biggl[\sum_{s \in \mathcal{S}} \sum_{i = \lfloor n F(s) \rfloor+1}^{\lfloor n (F(s)+\pi p(s)) \rfloor}\frac{\tilde{\eta}_{i,1}(s,\tau)}{\sqrt{n}\pi f_1(q_1(\tau))} -  \sum_{s \in \mathcal{S}} \sum_{i = \lfloor n (F(s)+\pi p(s)) \rfloor+1}^{\lfloor n (F(s)+ p(s)) \rfloor}\frac{\tilde{\eta}_{i,0}(s,\tau)}{\sqrt{n}(1-\pi)f_0(q_0(\tau))}\biggr] \\
& + \biggl[\sum_{i=1}^n \frac{1}{\sqrt{n}} \left(\frac{m_1(S_i,\tau) }{f_1(q_1(\tau))} - \frac{m_0(S_i,\tau) }{f_0(q_0(\tau))}\right) \biggr].
\end{align*}
By taking difference of the two displays above, we have 
\begin{align}
\label{eq:betastar}
\sqrt{n}(\hat{\beta}^*_1(\tau) - \tilde{q}(\tau)) = & \biggl[\sum_{s \in \mathcal{S}} \sum_{i = \lfloor n F(s) \rfloor+1}^{\lfloor n (F(s)+\pi p(s)) \rfloor}\frac{(\xi_i^s-1)\tilde{\eta}_{i,1}(s,\tau)}{\sqrt{n}\pi f_1(q_1(\tau))} -  \sum_{s \in \mathcal{S}} \sum_{i = \lfloor n (F(s)+\pi p(s)) \rfloor+1}^{\lfloor n (F(s)+ p(s)) \rfloor}\frac{(\xi_i^s-1)\tilde{\eta}_{i,0}(s,\tau)}{\sqrt{n}(1-\pi)f_0(q_0(\tau))}\biggr] \notag \\
& + \biggl[ \sum_{s \in \mathcal{S}} \frac{D_n^*(s)}{\sqrt{n}}\left(\frac{m_1(s,\tau)}{\pi f_1(q_1(\tau))} + \frac{m_0(s,\tau)}{\pi f_0(q_0(\tau))}\right) \biggr] \notag \\
& + \biggl[\sum_{i=1}^n \frac{\tilde{\xi}_i-1}{\sqrt{n}} \left(\frac{m_1(S_i,\tau) }{f_1(q_1(\tau))} - \frac{m_0(S_i,\tau) }{f_0(q_0(\tau))}\right) \biggr] + R_1^*(\tau).
\end{align}
Note that, conditionally on data, the first and third brackets on the RHS of the above display converge to Gaussian processes with covariance kernels 
\begin{align*}
\frac{\min(\tau_1,\tau_2) - \tau_1\tau_2 - \mathbb{E}m_1(S,\tau_1)m_1(S,\tau_2)}{\pi f_1(q_1(\tau_1))f_1(q_1(\tau_2))} + \frac{\min(\tau_1,\tau_2) - \tau_1\tau_2 - \mathbb{E}m_0(S,\tau_1)m_0(S,\tau_2)}{(1-\pi) f_0(q_0(\tau_1))f_0(q_0(\tau_2))}
\end{align*}
and
\begin{align*}
\mathbb{E}\biggl[\frac{m_1(S,\tau_1)}{f_1(q_1(\tau_1))} - \frac{m_0(S,\tau_1)}{f_0(q_0(\tau_1))}\biggr] \biggl[\frac{m_1(S,\tau_2)}{f_1(q_1(\tau_2))} - \frac{m_0(S,\tau_2)}{f_0(q_0(\tau_2))}\biggr], 
\end{align*}
uniformly over $\tau \in \Upsilon$, respectively. In addition, by Assumption \ref{ass:bassignment}.1, conditionally data (and thus $\{S_i\}_{i=1}^n$), the second bracket on the RHS of \eqref{eq:betastar} converges to a Gaussian process with a covariance kernel 
\begin{align*}
\mathbb{E}\gamma(S)\biggl[\frac{m_1(S,\tau_1)m_1(S,\tau_2)}{\pi^2 f_1(q_1(\tau_1))f_1(q_1(\tau_2))}+\frac{m_1(S,\tau_1)m_0(S,\tau_2)}{\pi(1-\pi) f_1(q_1(\tau_1))f_0(q_0(\tau_2))}\biggr],
\end{align*}
uniformly over $\tau \in \Upsilon$. Furthermore, we notice that these three Gaussian processes are independent. Therefore, we have, conditionally on data and uniformly over $\tau \in \Upsilon$, 
\begin{align*}
\sqrt{n}(\hat{\beta}^*_1(\tau) - \tilde{q}(\tau)) \convD \mathcal{B}_{sqr}(\tau),
\end{align*}
where $\mathcal{B}_{sqr}(\tau)$ is defined in Theorem \ref{thm:qr}. This leads to the desired result for the simple quantile regression estimator. 

Next, we briefly describe the derivation for the IPW estimator. Following the proof of Theorem \ref{thm:ipw},  we have 
\begin{align*}
\sqrt{n}(\hat{q}^*_1(\tau) - q_1(\tau)) = & \argmin_u L^*_n(u,\tau),
\end{align*}
where 
\begin{align*}
L^*_n(u,\tau) \equiv & \sum_{i =1}^n \frac{A^*_i}{\hat{\pi}^*(S^*_i)}\left[\rho_\tau(Y^*_i - q_1(\tau) - \frac{u}{\sqrt{n}}) - \rho_\tau(Y^*_i - q_1(\tau))\right] \\
= & - L^*_{1,n}(\tau)u + L^*_{2,n}(u,\tau),
\end{align*}
and $\hat{\pi}^*(s) = \frac{n_1^*(s)}{n^*(s)}$. Then, we have 
\begin{align*}
L^*_{1,n}(\tau) = W_{ipw,1,n}^*(\tau) + R_{ipw,1}^*(\tau),
\end{align*}
where 
\begin{align*}
W_{ipw,1,n}^*(\tau) = \sum_{s\in \mathcal{S}}\frac{1}{\sqrt{n}}\sum_{i=1}^n \frac{A_i^*1\{S_i^*=s\}\eta_{i,1}^*(s,\tau)}{\pi} + \sum_{i=1}^n \frac{m_1(S_i^*,\tau)}{\sqrt{n}},
\end{align*} 
and 
\begin{align*}
R_{ipw,1}^*(\tau) = - \sum_{i=1}^n \sum_{s \in \mathcal{S}}\frac{A_i^* 1\{S_i^*=s\}D_n^*(s)}{n^*(s)\sqrt{n}\hat{\pi}^*(s)\pi}\eta_{i,1}^*(s,\tau). 
\end{align*}
By Lemma \ref{lem:Rsfestar}, $\sup_{ \tau \in \Upsilon}|R_{ipw,1}^*(\tau)| = o_p(1)$. In addition, same as above, we can show that 
\begin{align*}
\sup_{ \tau \in \Upsilon}|W_{ipw,1,n}^*(\tau) - W_{ipw,1,n}^{**}(\tau)| = o_p(1),
\end{align*}
where 
\begin{align*}
W_{ipw,1,n}^{**}(\tau) =  & \sum_{s \in \mathcal{S}} \sum_{i = \lfloor n F(s) \rfloor+1}^{\lfloor n (F(s)+\pi p(s)) \rfloor}\frac{\xi_i^s\tilde{\eta}_{i,1}(s,\tau)}{\sqrt{n}\pi} + \sum_{i=1}^n \frac{\tilde{\xi}_im_1(S_i,\tau) }{\sqrt{n}}.  
\end{align*}
Similar to Lemma \ref{lem:Qstar}, we can show that, uniformly over $\tau \in \Upsilon$, 
\begin{align*}
L_{2,n}^*(\tau) \convP \frac{f_1(q_1(\tau))u^2}{2}.
\end{align*}
Therefore, 
\begin{align*}
\sqrt{n}(\hat{q}^*_1(\tau) - q_1(\tau)) = \frac{W_{ipw,1,n}^{**}(\tau)}{f_1(q_1(\tau))} + R_{ipw,1}^{**}(\tau),
\end{align*}
where $\sup_{ \tau \in \Upsilon}|R_{ipw,1}^{**}(\tau)| = o_p(1)$. Similarly, we can show 
\begin{align*}
\sqrt{n}(\hat{q}^*_0(\tau) - q_0(\tau)) = \frac{W_{ipw,0,n}^{**}(\tau)}{f_0(q_0(\tau))} + R_{ipw,0}^{**}(\tau),
\end{align*}
where $\sup_{ \tau \in \Upsilon}|R_{ipw,0}^{**}(\tau)| = o_p(1)$ and 
\begin{align*}
W_{ipw,0,n}^{**}(\tau) =  & \sum_{s \in \mathcal{S}} \sum_{i = \lfloor n (F(s) \pi p(s))\rfloor+1}^{\lfloor n (F(s)+p(s)) \rfloor}\frac{\xi_i^s\tilde{\eta}_{i,0}(s,\tau)}{\sqrt{n}\pi} + \sum_{i=1}^n \frac{\tilde{\xi}_im_0(S_i,\tau) }{\sqrt{n}}.  
\end{align*}
Therefore, 
\begin{align*}
\sqrt{n}(\hat{q}^*(\tau) - \tilde{q}(\tau)) = & \biggl[\sum_{s \in \mathcal{S}} \sum_{i = \lfloor n F(s) \rfloor+1}^{\lfloor n (F(s)+\pi p(s)) \rfloor}\frac{(\xi_i^s-1)\tilde{\eta}_{i,1}(s,\tau)}{\sqrt{n}\pi f_1(q_1(\tau))} -  \sum_{s \in \mathcal{S}} \sum_{i = \lfloor n (F(s)+\pi p(s)) \rfloor+1}^{\lfloor n (F(s)+ p(s)) \rfloor}\frac{(\xi_i^s-1)\tilde{\eta}_{i,0}(s,\tau)}{\sqrt{n}(1-\pi)f_0(q_0(\tau))}\biggr] \notag \\
& + \biggl[\sum_{i=1}^n \frac{\tilde{\xi}_i-1}{\sqrt{n}} \left(\frac{m_1(S_i,\tau) }{f_1(q_1(\tau))} - \frac{m_0(S_i,\tau) }{f_0(q_0(\tau))}\right) \biggr] + R_{ipw}^*(\tau),
\end{align*}
where $\sup_{ \tau \in \Upsilon}|R_{ipw}^{*}(\tau)| = o_p(1)$. Last, we can show that, conditionally on data and uniformly over $\tau \in \Upsilon$, the RHS of the above display weakly converges to the Gaussian process $\mathcal{B}_{ipw}(\tau)$, where $\mathcal{B}_{ipw}(\tau)$ is defined in Theorem \ref{thm:ipw}.

\section{Technical Lemmas}
\label{sec:lem}
\begin{lem}
	\label{lem:S}
	Let $S_k$ be the $k$-th partial sum of Banach space valued independent identically distributed random variables, then 
	\begin{align*}
	\mathbb{P}(\max_{1 \leq k \leq n}||S_k|| \geq \eps) \leq 3\max_{1 \leq k \leq n}\mathbb{P}(||S_k||\geq \eps/3).
	\end{align*}
\end{lem}
When $S_k$ takes values on $\Re$, Lemma \ref{lem:S} is \citet[Exercise 2.3]{PLS08}.

\begin{proof}
	First suppose $\max_k \mathbb{P}(||S_n - S_k||\geq 2\eps/3)\leq 2/3$. In addition, define 
	\begin{align*}
	A_k = \{ ||S_k|| \geq \eps, ||S_j||< \eps, 1\leq j < k\}. 
	\end{align*}
	
	Then, 
	\begin{align*}
	\mathbb{P}(\max_k ||S_k|| \geq \eps) \leq & \mathbb{P}(||S_n|| \geq \eps/3) +  \sum_{k=1}^n\mathbb{P}(||S_n|| \leq \eps/3,A_k) \\
	\leq & \mathbb{P}(||S_n|| \geq \eps/3) +  \sum_{k=1}^n\mathbb{P}(||S_n-S_k|| \geq 2\eps/3)\mathbb{P}(A_k) \\
	\leq & \mathbb{P}(||S_n|| \geq \eps/3) + \frac{2}{3}\mathbb{P}(\max_k ||S_k|| \geq \eps).
	\end{align*}
	This implies, 
	\begin{align*}
	\mathbb{P}(\max_k ||S_k|| \geq \eps) \leq 3\mathbb{P}(||S_n|| \geq \eps/3).
	\end{align*}
	On the other hand, if $\max_k \mathbb{P}(||S_n - S_k||\geq 2\eps/3) > 2/3$, then there exists $k_0$ such that $ \mathbb{P}(||S_n - S_{k_0}||\geq 2\eps/3) > 2/3$. Thus, 
	\begin{align*}
	\mathbb{P}(||S_n|| \geq \eps/3) + \mathbb{P}(||S_{k_0}|| \geq \eps/3) \geq 2/3. 
	\end{align*}
	This implies, 
	\begin{align*}
	3\max_{1 \leq k \leq n}\mathbb{P}(||S_k|| \geq \eps/3) \geq 3 \max(\mathbb{P}(||S_n|| \geq \eps/3), \mathbb{P}(||S_{k_0}|| \geq \eps/3)) \geq 1 \geq \mathbb{P}(\max_{1 \leq k \leq n}||S_k|| \geq \eps).
	\end{align*}
	This concludes the proof. 
\end{proof}

\begin{lem}
	\label{lem:Q}
	Let $W_{n,j}(\tau)$, $j = 1,2,3$ be defined as in \eqref{eq:W}. If Assumptions in Theorem \ref{thm:qr} hold, then uniformly over $\tau \in \Upsilon$, 
	\begin{align*}
	(W_{n,1}(\tau), W_{n,2}(\tau), W_{n,3}(\tau)) \convD (\mathcal{B}_1(\tau),\mathcal{B}_2(\tau),\mathcal{B}_3(\tau)),
	\end{align*}
	where $(\mathcal{B}_1(\tau),\mathcal{B}_2(\tau),\mathcal{B}_3(\tau))$ are three independent two-dimensional Gaussian processes with covariance kernels $\Sigma_1(\tau_1,\tau_2)$, $\Sigma_2(\tau_1,\tau_2)$, and $\Sigma_3(\tau_1,\tau_2)$, respectively. The expressions for the three kernels are derived in the proof below. 
\end{lem}
\begin{proof}
	We follow the general argument in the proof of \citet[Lemma B.2]{BCS17}. We divide the proof into two steps. In the first step, we show that 
	\begin{align*}
	(W_{n,1}(\tau), W_{n,2}(\tau), W_{n,3}(\tau)) \stackrel{d}{=}( W^\star_{n,1}(\tau), W_{n,2}(\tau), W_{n,3}(\tau)) + o_p(1), 
	\end{align*}
	where the $o_p(1)$ term holds uniformly over $\tau \in \Upsilon$, $ W^\star_{n,1}(\tau) \indep (W_{n,2}(\tau),W_{n,3}(\tau))$, and, uniformly over $\tau \in \Upsilon$, 
	\begin{align*}
	W^\star_{n,1}(\tau) \convD \mathcal{B}_1(\tau).
	\end{align*}
	In the second step, we show that 
	\begin{align*}
	(W_{n,2}(\tau),W_{n,3}(\tau)) \convD (\mathcal{B}_2(\tau),\mathcal{B}_3(\tau))
	\end{align*}
	uniformly over $\tau \in \Upsilon$ and $\mathcal{B}_2(\tau) \indep \mathcal{B}_3(\tau)$. 
	
	\textbf{Step 1.} Let $\tilde{\eta}_{i,j}(s,\tau) = \tau - 1\{Y_i^s(j) \leq q_j(\tau) \} - m_j(s,\tau)$, for $j = 0,1$, where $\{Y_i^s(0),Y_i^s(1)\}_{i \geq 1}$ are the same as defined in Step 1 in the proof of Theorem \ref{thm:qr}. In addition, denote 
	\begin{align*}
	\tilde{W}_{n,1}(\tau) =  e_1 \sum_{s \in \mathcal{S}} \sum_{i=N(s) + 1}^{N(s) + n_1(s)} \frac{1}{\sqrt{n}}\tilde{\eta}_{i,1}(s,\tau) + e_0\sum_{s \in \mathcal{S}} \sum_{i=N(s) + n_1(s)+1}^{N(s) + n(s)} \frac{1}{\sqrt{n}}\tilde{\eta}_{i,0}(s,\tau).
	\end{align*}
	Then, we have 
	\begin{align*}
	\{W_{n,1}(\tau)|\{A_i,S_i\}_{i=1}^n\} \stackrel{d}{=} \{\tilde{W}_{n,1}(\tau)|\{A_i,S_i\}_{i=1}^n\}. 
	\end{align*}
	Because both $W_{n,2}(\tau)$ and $W_{n,3}(\tau)$ are only functions of $\{A_i,S_i\}_{i=1}^n$, we have 
	\begin{align*}
	(W_{n,1}(\tau),W_{n,2}(\tau),W_{n,3}(\tau)) \stackrel{d}{=} 
	(\tilde{W}_{n,1}(\tau),W_{n,2}(\tau),W_{n,3}(\tau)).
	\end{align*}
	Let 
	\begin{align*}
	W_{n,1}^\star(\tau) =  e_1 \sum_{s \in \mathcal{S}} \sum_{i=\lfloor nF(s) \rfloor + 1}^{\lfloor n(F(s)+\pi p(s)) \rfloor} \frac{1}{\sqrt{n}}\tilde{\eta}_{i,1}(s,\tau) + e_0\sum_{s \in \mathcal{S}} \sum_{i=\lfloor n(F(s)+\pi p(s)) \rfloor+1}^{\lfloor n(F(s)+ p(s)) \rfloor} \frac{1}{\sqrt{n}}\tilde{\eta}_{i,0}(s,\tau).
	\end{align*}
	Note that $W_{n,1}^\star(\tau)$ is a function of $(Y_i^s(1),Y_i^s(0))_{i \geq 1}$ only, which is independent of $\{A_i,S_i\}_{i=1}^n$ by construction. Therefore, $W_{n,1}^\star(\tau) \indep (W_{n,2}(\tau),W_{n,3}(\tau))$. 
	
	Furthermore, note that 
	\begin{align*}
	\frac{N(s)}{n} \convP F(s), \quad \frac{n_1(s)}{n} \convP \pi p(s), \quad \text{and} \quad \frac{n(s)}{n} \convP p(s).
	\end{align*}
	Denote $\Gamma_{n,j}(s,t,\tau) = \sum_{i =1}^{\lfloor nt \rfloor}\frac{1}{\sqrt{n}}\tilde{\eta}_{i,j}(s,\tau)$. In order to show $\sup_{\tau \in \Upsilon}|\tilde{W}_{n,1}(\tau) -  W^\star_{n,1}(\tau)| = o_p(1)$ and $ W^\star_{n,1}(\tau) \convD \mathcal{B}_1(\tau)$, it suffices to show that, (1) for $j = 0,1$ and $s \in \mathcal{S}$, the stochastic processes 
	\begin{align*}
	\{\Gamma_{n,j}(s,t,\tau): t\in (0,1), \tau \in \Upsilon \}
	\end{align*}
	in stochastically equicontinuous; and (2) $ W^\star_{n,1}(\tau)$ converges to $\mathcal{B}_1(\tau)$ in finite dimension. 
	
	\textbf{Claim (1).} We want to bound
	\begin{align*}
	\sup|\Gamma_{n,j}(s,t_2,\tau_2) - \Gamma_{n,j}(s,t_1,\tau_1)|,
	\end{align*}
	where supremum is taken over $0 < t_1<t_2<t_1+\eps < 1$ and $\tau_1<\tau_2<\tau_1+\eps$ such that $\tau_1,\tau_1+\eps \in \Upsilon.$ Note that,
	\begin{align}
	\label{eq:GW}
	& \sup|\Gamma_{n,j}(s,t_2,\tau_2) - \Gamma_{n,j}(s,t_1,\tau_1)| \notag \\
	\leq & \sup_{0 < t_1<t_2<t_1+\eps < 1,\tau \in \Upsilon}|\Gamma_{n,j}(s,t_2,\tau) - \Gamma_{n,j}(s,t_1,\tau)| +  \sup_{t \in (0,1),\tau_1,\tau_2 \in \Upsilon, \tau_1 < \tau_2 < \tau_1+\eps}|\Gamma_{n,j}(s,t,\tau_2) - \Gamma_{n,j}(s,t,\tau_1)|. 
	\end{align}
	Let $m = \lfloor nt_2 \rfloor - \lfloor nt_1 \rfloor \leq \lfloor n\eps \rfloor+1$. Then,  for an arbitrary $\delta>0$, by taking $\eps = \delta^4$,  we have
	\begin{align*}
	& \mathbb{P}(\sup_{0 < t_1<t_2<t_1+\eps < 1,\tau \in \Upsilon}|\Gamma_{n,j}(s,t_2,\tau) - \Gamma_{n,j}(s,t_1,\tau)| \geq \delta) \\
	= & \mathbb{P}(\sup_{0 < t_1<t_2<t_1+\eps < 1,\tau \in \Upsilon} |\sum_{i=\lfloor nt_1 \rfloor+1}^{i=\lfloor nt_2 \rfloor}\tilde{\eta}_{i,j}(s,\tau)| \geq \sqrt{n}\delta) \\
	= & \mathbb{P}(\sup_{0 < t \leq \eps,\tau \in \Upsilon} |\sum_{i=1}^{\lfloor nt \rfloor}\tilde{\eta}_{i,j}(s,\tau)| \geq \sqrt{n}\delta) \\
	\leq & \mathbb{P}(\max_{1 \leq k \leq \lfloor n\eps \rfloor}\sup_{\tau \in \Upsilon} |S_k(\tau)| \geq \sqrt{n}\delta) \\
	\leq & \frac{270\mathbb{E}\sup_{\tau \in \Upsilon} |\sum_{i=1}^{\lfloor n\eps \rfloor}\tilde{\eta}_{i,j}(s,\tau)|}{\sqrt{n}\delta} \\
	\lesssim & \frac{\sqrt{n\eps}}{\sqrt{n} \delta} \lesssim \delta,
	\end{align*}
	where in the first inequality, $S_k(\tau) = \sum_{i=1}^k\tilde{\eta}_{i,j}(s,\tau)$ and the second inequality holds due to the same argument in \eqref{eq:Gamma}. For the third inequality, denote 
	\begin{align*}
	\mathcal{F} = \{\tilde{\eta}_{i,j}(s,\tau): \tau \in \Upsilon \}
	\end{align*}
	with an envelope function $F = 2$. In addition, because $\mathcal{F}$ is a VC-class with a fixed VC-index, we have
	\begin{align*}
	J(1,\mathcal{F}) < \infty,
	\end{align*}
	where 
	\begin{align*}
	J(\delta,\mathcal{F}) = \sup_Q \int_0^\delta \sqrt{1 + \log N(\eps||F||_{Q,2},\mathcal{F},L_2(Q))}d\eps,
	\end{align*} 
	$N(\eps||F||_{Q,2},\mathcal{F},L_2(Q))$ is the covering number, and the supremum is taken over all discrete probability measures $Q$. Therefore, by \citet[Theorem 2.14.1]{VW96}
	\begin{align*}
	\frac{270\mathbb{E}\sup_{\tau \in \Upsilon} |\sum_{i=1}^{\lfloor n\eps \rfloor}\tilde{\eta}_{i,j}(s,\tau)|}{\sqrt{n}\delta} \lesssim \frac{\sqrt{\lfloor n\eps \rfloor}\left[ \mathbb{E}\sqrt{\lfloor n\eps \rfloor}||\mathbb{P}_{\lfloor n\eps \rfloor} - \mathbb{P}||_{\mathcal{F}}\right]}{\sqrt{n}\delta} \lesssim \frac{\sqrt{\lfloor n\eps \rfloor}J(1,\mathcal{F})}{\sqrt{n}\delta}.
	\end{align*}
	For the second term on the RHS of \eqref{eq:GW}, by taking $\eps = \delta^4$, we have 
	\begin{align*}
	& \mathbb{P}(\sup_{t \in (0,1),\tau_1,\tau_2 \in \Upsilon, \tau_1 < \tau_2 < \tau_1+\eps}|\Gamma_{n,j}(s,t,\tau_2) - \Gamma_{n,j}(s,t,\tau_1)| \geq \delta) \\
	= & \mathbb{P}(\max_{1 \leq k \leq n}\sup_{\tau_1,\tau_2 \in \Upsilon, \tau_1 < \tau_2 < \tau_1+\eps}|S_k(\tau_1,\tau_2)| \geq \sqrt{n}\delta) \\
	\leq & \frac{270\mathbb{E}\sup_{\tau_1,\tau_2 \in \Upsilon, \tau_1 < \tau_2 < \tau_1+\eps} |\sum_{i=1}^{n}(\tilde{\eta}_{i,j}(s,\tau_2) - \tilde{\eta}_{i,j}(s,\tau_1))|}{\sqrt{n}\delta} \lesssim \delta \sqrt{\log(\frac{C}{\delta^2})},
	\end{align*}
	where in the first equality, $S_k(\tau_1,\tau_2) = \sum_{i=1}^k (\tilde{\eta}_{i,j}(s,\tau_2) - \tilde{\eta}_{i,j}(s,\tau_1))$ and the first inequality follows the same argument as in \eqref{eq:Gamma}. For the last inequality, denote 
	\begin{align*}
	\mathcal{F} = \{\tilde{\eta}_{i,j}(s,\tau_2) - \tilde{\eta}_{i,j}(s,\tau_1): \tau_1,\tau_2 \in \Upsilon, \tau_1 < \tau_2 < \tau_1+\eps\}
	\end{align*}
	with a constant envelope function $F = C$ and 
	\begin{align*}
	\sigma^2 = \sup_{f \in \mathcal{F}}\mathbb{E}f^2 \in [c_1\eps,c_2\eps], 
	\end{align*}
	for some constant $0<c_1<c_2<\infty$. Last, $\mathcal{F}$ is nested by some VC class with a fixed VC index. Therefore, by \citet[Corollary 5.1]{CCK14}, 
	\begin{align*}
	& \frac{270\mathbb{E}\sup_{\tau_1,\tau_2 \in \Upsilon, \tau_1 < \tau_2 < \tau_1+\eps} |\sum_{i=1}^{n}(\tilde{\eta}_{i,j}(s,\tau_2) - \tilde{\eta}_{i,j}(s,\tau_1))|}{\sqrt{n}\delta} \\
	\lesssim & \frac{\sqrt{n}\mathbb{E}||\mathbb{P}_n - \mathbb{P}||_\mathcal{F}}{\delta} \lesssim \sqrt{\frac{\sigma^2 \log(\frac{C}{\sigma})}{ \delta^2}} + \frac{C\log(\frac{C}{\sigma})}{\sqrt{n}\delta} \lesssim \delta \sqrt{\log(\frac{C}{\delta^2})},
	\end{align*}
	where the last inequality holds by letting $n$ be sufficiently large. Note that $\delta \sqrt{\log(\frac{C}{\delta^2})} \rightarrow 0$ as $\delta \rightarrow 0$. This concludes the proof of Claim (1). 
	
	\textbf{Claim (2).} For a single $\tau$, by the triangular CLT, 
	\begin{align*}
	W_{n,1}^\star(\tau) \convD N(0,\Sigma_1(\tau)),
	\end{align*}
	where $\Sigma_1(\tau) = \pi[\tau(1-\tau) - \mathbb{E}m^2_1(S,\tau)] e_1  e_1^\prime +  (1-\pi)[\tau(1-\tau) -\mathbb{E}m^2_0(S,\tau)]e_0e_0^\prime.$ The convergence in finite dimension can be proved by using the Cram\'{e}r-Wold device. In particular, we can show that the covariance kernel is 
	\begin{align*}
	\Sigma_1(\tau_1,\tau_2) = & \pi[\min(\tau_1,\tau_2) - \tau_1\tau_2 - \mathbb{E} m_1(S,\tau_1)m_1(S,\tau_2)] e_1  e_1^\prime\\
	& + (1-\pi)[\min(\tau_1,\tau_2) - \tau_1\tau_2 - \mathbb{E}m_0(S,\tau_1)m_0(S,\tau_2)]e_0e_0^\prime.
	\end{align*}
	This concludes the proof of Claim (2), and thus leads to the desired results in Step 1.   
	
	\textbf{Step 2.}
	We first consider the marginal distributions for $W_{n,2}(\tau)$ and $W_{n,3}(\tau)$. For $W_{n,2}(\tau)$, by Assumption \ref{ass:assignment1} and the fact that $m_j(s,\tau)$ is continuous in $\tau \in \Upsilon$ $j=0,1$, we have, conditionally on $\{S_i\}_{i=1}^n$, 
	\begin{align}
	\label{eq:B2}
	W_{n,2}(\tau) =  \sum_{s \in \mathcal{S}}\frac{D_n(s)}{\sqrt{n}}[ e_1 m_1(s,\tau) - e_0m_0(s,\tau)] \convD \mathcal{B}_2(\tau),
	\end{align}
	where $\mathcal{B}_2(\tau)$ is a two-dimensional Gaussian process with covariance kernel 
	\begin{align*}
	& \Sigma_2(\tau_1,\tau_2) \\
	= &\sum_{s \in \mathcal{S}}p(s)\gamma(s)\biggl[ e_1 e_1^\prime m_1(s,\tau_1)m_1(s,\tau_2) -  e_1 e_0^\prime m_1(s,\tau_1)m_0(s,\tau_2) \\
	&- e_0  e_1^\prime m_0(s,\tau_1)m_1(s,\tau_2)+e_0e_0^\prime m_0(s,\tau_1)m_0(s,\tau_2) \biggr].
	\end{align*}
	
	For $W_{n,3}(\tau)$, by the fact that $m_j(s,\tau)$ is continuous in $\tau \in \Upsilon$ $j=0,1$, we have that, uniformly over $\tau \in \Upsilon$, 
	\begin{align}
	\label{eq:B3}
	W_{n,3}(\tau) =  \frac{1}{\sqrt{n}}\sum_{i=1}^n[ e_1 \pi m_1(S_i,\tau) + e_0 (1-\pi)m_0(S_i,\tau)] \convD \mathcal{B}_3(\tau),
	\end{align}
	where $\mathcal{B}_3(\tau)$ a two-dimensional Gaussian process with covariance kernel
	\begin{align*}
	\Sigma_3(\tau_1,\tau_2) = &  e_1  e_1^\prime \pi^2 \mathbb{E}m_1(S,\tau_1)m_1(S,\tau_2) +  e_1 e_0^\prime \pi(1-\pi) \mathbb{E}m_1(S,\tau_1)m_0(S,\tau_2) \\
	&+ e_0  e_1^\prime \pi(1-\pi) \mathbb{E}m_0(S,\tau_1)m_1(S,\tau_2) + e_0 e_0^\prime (1-\pi)^2 \mathbb{E}m_0(S,\tau_1)m_0(S,\tau_2).   
	\end{align*}
	In addition, we note that, for any fixed $\tau$, 
	\begin{align*}
	\mathbb{P}(W_{n,2}(\tau) \leq w_1,W_{n,3}(\tau) \leq w_2) = &  \mathbb{E}\mathbb{P}(W_{n,2}(\tau) \leq w_1|\{S_i\}_{i=1}^n)1\{W_{n,3}(\tau) \leq w_2 \} \\
	= & \mathbb{E}\mathbb{P}(N(0,\Sigma_2(\tau,\tau))\leq w_1)1\{W_{n,3}(\tau) \leq w_2 \} + o(1) \\
	= & \mathbb{P}(N(0,\Sigma_3(\tau,\tau))\leq w_2)\mathbb{P}(N(0,\Sigma_2(\tau,\tau))\leq w_1) + o(1).
	\end{align*}
	This implies $\mathcal{B}_2(\tau) \indep \mathcal{B}_3(\tau)$. By the Cram\'{e}r-Wold device, we can show that 
	\begin{align*}
	(W_{n,2}(\tau),W_{n,3}(\tau)) \convD (\mathcal{B}_2(\tau),\mathcal{B}_3(\tau))
	\end{align*}
	jointly in finite dimension, where by an abuse of notation, $\mathcal{B}_2(\tau)$ and $\mathcal{B}_3(\tau)$ have the same marginal distributions of those in \eqref{eq:B2} and \eqref{eq:B3}, respectively, and $\mathcal{B}_2(\tau) \indep \mathcal{B}_3(\tau)$. Last, because both $W_{n,2}(\tau)$ and $W_{n,3}(\tau)$ are tight marginally, so be the joint process $(W_{n,2}(\tau),W_{n,3}(\tau))$. This concludes the proof of Step 2, and thus the whole lemma. 
\end{proof}

\begin{lem}
	\label{lem:Qipw}
	Let $\mathcal{W}_{n,j}(\tau)$, $j = 1,2$ be defined as in \eqref{eq:wipw2}. If Assumptions in Theorem \ref{thm:ipw} hold, then uniformly over $\tau \in \Upsilon$, 
	\begin{align*}
	(\mathcal{W}_{n,1}(\tau),\mathcal{W}_{n,2}(\tau)) \convD (\mathcal{B}_{ipw,1}(\tau),\mathcal{B}_{ipw,2}(\tau)),
	\end{align*}
	where $(\mathcal{B}_{ipw,1}(\tau),\mathcal{B}_{ipw,2}(\tau))$ are two independent two-dimensional Gaussian processes with covariance kernels $\Sigma_{ipw,1}(\tau_1,\tau_2)$ and $\Sigma_{ipw,2}(\tau_1,\tau_2)$, respectively. The expressions for  $\Sigma_{ipw,1}(\tau_1,\tau_2)$ and $\Sigma_{ipw,2}(\tau_1,\tau_2)$ are derived in the proof below. 
\end{lem}
\begin{proof}
	The proofs of weak convergence and the independence between $(\mathcal{B}_{ipw,1}(\tau),\mathcal{B}_{ipw,2}(\tau))$ are similar to that in Lemma \ref{lem:Q}, and thus, are omitted. Next, we focus on deriving the covariance
	kernels. 
	
	First, similar to the argument in the proof of Lemma \ref{lem:Q}, 
	\begin{align*}
	\mathcal{W}_{n,1}(\tau) \stackrel{d}{=}\sum_{s \in \mathcal{S}} \sum_{i=N(s)+1}^{N(s)+n_1(s)}\frac{1}{\sqrt{n}f_1(q_1(\tau))}\tilde{\eta}_{i,1}(s,\tau)-\sum_{s \in \mathcal{S}} \sum_{i=N(s)+n_1(s)+1}^{N(s)+n(s)}\frac{1}{\sqrt{n}f_0(q_0(\tau))}\tilde{\eta}_{i,0}(s,\tau).
	\end{align*}
	Because $(\tilde{\eta}_{i,1}(s,\tau),\tilde{\eta}_{i,0}(s,\tau))$ are independent across $i$, $n_1(s)/n \convP \pi p(s)$, and $(n(s) - n_1(s))/n \convP (1-\pi)p(s)$, we have 
	\begin{align*}
	\Sigma_{ipw,1}(\tau_1,\tau_2) = \frac{\min(\tau_1,\tau_2) - \tau_1\tau_2 - \mathbb{E}m_1(S,\tau_1)m_1(S,\tau_2)  }{\pi f_1(q_1(\tau_1))f_1(q_1(\tau_2))} + \frac{\min(\tau_1,\tau_2) - \tau_1\tau_2 - \mathbb{E}m_0(S,\tau_1)m_0(S,\tau_2)}{(1-\pi)f_0(q_0(\tau_1))f_0(q_0(\tau_2)) }.
	\end{align*} 
	Obviously, 
	\begin{align*}
	\Sigma_{ipw,2}(\tau_1,\tau_2) = \mathbb{E}\left(\frac{m_1(S,\tau_1)}{f_1(q_1(\tau_1))} - \frac{m_0(S,\tau_1)}{f_0(q_0(\tau_1))}\right)\left(\frac{m_1(S,\tau_2)}{f_1(q_1(\tau_2))} - \frac{m_0(S,\tau_2)}{f_0(q_0(\tau_2))}\right),
	\end{align*}
\end{proof}

\begin{lem}
	\label{lem:Qwqr}
	If Assumptions \ref{ass:assignment1} and \ref{ass:tau} hold, then conditionally on data,  the second element of $[Q(\tau)]^{-1}\sum_{i=1}^n \frac{\xi_{i}-1}{\sqrt{n}}\dot{A}_i\left(\tau - 1\{Y_i \leq \dot{A}'\beta(\tau)\}\right)$ weakly converges to $\tilde{\mathcal{B}}_{sqr}(\tau)$, 	where $\tilde{\mathcal{B}}_{sqr}(\tau)$ is a Gaussian process with covariance kernel $\tilde{\Sigma}_{sqr}(\cdot,\cdot)$ defined in Theorem \ref{thm:b}. 
\end{lem}
\begin{proof}
	We denote the second element of 	$[Q(\tau)]^{-1}\sum_{i=1}^n \frac{\xi_{i}-1}{\sqrt{n}}\dot{A}_i\left(\tau - 1\{Y_i \leq \dot{A}'\beta(\tau)\}\right)$ as 
	\begin{align*}
	\frac{1}{\sqrt{n}}\sum_{i=1}^n (\xi_i - 1)\mathcal{J}_i(s,\tau), 
	\end{align*}
	where 
	\begin{align*}
	\mathcal{J}_i(s,\tau) = \mathcal{J}_{i,1}(s,\tau) + \mathcal{J}_{i,2}(s,\tau)  + \mathcal{J}_{i,3}(s,\tau), 
	\end{align*}
	\begin{align*}
	\mathcal{J}_{i,1}(s,\tau) = \frac{A_i 1\{S_i = s\}\eta_{i,1}(s,\tau)}{\pi f_1(q_1(\tau))} - \frac{(1-A_i)1\{S_i = s\}\eta_{i,0}(s,\tau)}{(1-\pi) f_0(q_0(\tau))},
	\end{align*}
	\begin{align*}
	\mathcal{J}_{i,2}(s,\tau) = F_1(s,\tau) (A_i - \pi)1\{S_i = s\},
	\end{align*}
	$$F_1(s,\tau) =  \frac{m_1(s,\tau)}{\pi f_1(q_1(\tau))}+\frac{m_0(s,\tau)}{(1-\pi) f_0(q_0(\tau))},$$
	and
	\begin{align*}
	\mathcal{J}_{i,3}(s,\tau) =  \left(\frac{m_1(s,\tau)}{f_1(q_1(\tau))} - \frac{m_0(s,\tau)}{f_0(q_0(\tau))}\right)1\{S_i = s\}.
	\end{align*}

	In order to show the weak convergence, we only need to show (1) conditionally stochastic equicontinuity and (2) conditional convergence in finite dimension. We divide the proof into two steps accordingly. 
	
	\textbf{Step 1.} In order to show the conditionally stochastic equicontinuity, it suffices to show that, for any $\eps>0$, as $n \rightarrow \infty$ followed by $\delta \rightarrow 0$, 
	\begin{align*}
	\mathbb{P}_\xi\left( \sup_{\tau_1,\tau_2 \in \Upsilon, \tau_1 < \tau_2 < \tau_1+\delta,s \in \mathcal{S}}\left|\frac{1}{\sqrt{n}}\sum_{i=1}^n (\xi_i - 1)(\mathcal{J}_i(s,\tau_2) - \mathcal{J}_i(s,\tau_1))\right| \geq \eps\right) \convP 0,
	\end{align*}
	where $\mathbb{P}_\xi(\cdot)$ means that the probability operator is with respect to $\xi_1,\cdots,\xi_n$ and conditional on data. Note
	\begin{align*}
	& \mathbb{E}\mathbb{P}_\xi\left( \sup_{\tau_1,\tau_2 \in \Upsilon, \tau_1 < \tau_2 < \tau_1+\delta,s \in \mathcal{S}}\left|\frac{1}{\sqrt{n}}\sum_{i=1}^n (\xi_i - 1)(\mathcal{J}_i(s,\tau_1) - \mathcal{J}_i(s,\tau_1))\right| \geq \eps\right) \\
	= & \mathbb{P}\left( \sup_{\tau_1,\tau_2 \in \Upsilon, \tau_1 < \tau_2 < \tau_1+\delta,s \in \mathcal{S}}\left|\frac{1}{\sqrt{n}}\sum_{i=1}^n (\xi_i - 1)(\mathcal{J}_i(s,\tau_2) - \mathcal{J}_i(s,\tau_1))\right| \geq \eps\right)\\
	\leq &  \mathbb{P}\left( \sup_{\tau_1,\tau_2 \in \Upsilon, \tau_1 < \tau_2 < \tau_1+\delta,s \in \mathcal{S}}\left|\frac{1}{\sqrt{n}}\sum_{i=1}^n (\xi_i - 1)(\mathcal{J}_{i,1}(s,\tau_2) - \mathcal{J}_{i,1}(s,\tau_1))\right| \geq \eps/3\right) \\
	& + \mathbb{P}\left( \sup_{\tau_1,\tau_2 \in \Upsilon, \tau_1 < \tau_2 < \tau_1+\delta,s \in \mathcal{S}}\left|\frac{1}{\sqrt{n}}\sum_{i=1}^n (\xi_i - 1)(\mathcal{J}_{i,2}(s,\tau_2) - \mathcal{J}_{i,2}(s,\tau_1))\right| \geq \eps/3\right)\\
	& + \mathbb{P}\left( \sup_{\tau_1,\tau_2 \in \Upsilon, \tau_1 < \tau_2 < \tau_1+\delta,s \in \mathcal{S}}\left|\frac{1}{\sqrt{n}}\sum_{i=1}^n (\xi_i - 1)(\mathcal{J}_{i,3}(s,\tau_2) - \mathcal{J}_{i,3}(s,\tau_1))\right| \geq \eps/3\right).
	\end{align*}
	
	Further note that 
	\begin{align*}
	\sum_{i =1}^n (\xi_i-1)\mathcal{J}_{i,1}(s,\tau) \stackrel{d}{=} \sum_{i=N(s) + 1}^{N(s)+n_1(s)}\frac{(\xi_i-1)\tilde{\eta}_{i,1}(s,\tau)}{\pi f_1(q_1(\tau))} - \sum_{i=n(s) +n_1(s)+ 1}^{N(s)+n(s)}\frac{(\xi_i-1)\tilde{\eta}_{i,0}(s,\tau)}{(1-\pi) f_0(q_0(\tau))}
	\end{align*}
	By the same argument in Claim (1) in the proof of Lemma \ref{lem:Q}, we have 
	\begin{align*}
	& \mathbb{P}\left( \sup_{\tau_1,\tau_2 \in \Upsilon, \tau_1 < \tau_2 < \tau_1+\delta,s \in \mathcal{S}}\left|\frac{1}{\sqrt{n}}\sum_{i=1}^n (\xi_i - 1)(\mathcal{J}_{i,1}(s,\tau_2) - \mathcal{J}_i(s,\tau_1))\right| \geq \eps/3\right) \\
	\leq & \frac{3\mathbb{E}\sup_{\tau_1,\tau_2 \in \Upsilon, \tau_1 < \tau_2 < \tau_1+\delta,s \in \mathcal{S}}\left|\frac{1}{\sqrt{n}}\sum_{i=1}^n (\xi_i - 1)(\mathcal{J}_{i,1}(s,\tau_2) - \mathcal{J}_{i,1}(s,\tau_1))\right|}{\eps} \\
	\leq & \frac{3\sqrt{c_2\delta\log(\frac{C}{c_1\delta})} + \frac{3C\log(\frac{C}{c_1\delta})}{\sqrt{n}}}{\eps},
	\end{align*}	
	where $C$, $c_1< c_2$  are some positive constants that are independent of $(n,\eps,\delta)$. By letting $n\rightarrow \infty$ followed by $\delta \rightarrow 0$, the RHS vanishes. 	
	
	For $\mathcal{J}_{i,2}$, we note that $F_1(s,\tau)$ is Lipschitz in $\tau$. Therefore, 
	\begin{align*}
	& \mathbb{P}\left( \sup_{\tau_1,\tau_2 \in \Upsilon, \tau_1 < \tau_2 < \tau_1+\delta,s \in \mathcal{S}}\left|\frac{1}{\sqrt{n}}\sum_{i=1}^n (\xi_i - 1)(\mathcal{J}_{i,2}(s,\tau_2) - \mathcal{J}_{i,2}(s,\tau_1))\right| \geq \eps/3\right) \\
	\leq & \sum_{s \in \mathcal{S}}\mathbb{P}\left(C\delta\left|\frac{1}{\sqrt{n}} \sum_{i =1}^n (\xi_i - 1)(A_i - \pi)1\{S_i = s\}\right| \geq \eps/3 \right) \rightarrow 0 
	\end{align*}
	as $n \rightarrow \infty$ followed by $\delta \rightarrow 0$, where we use the fact that
	\begin{align*}
	\sup_{s \in \mathcal{S} }\left|\frac{1}{\sqrt{n}} \sum_{i =1}^n (\xi_i - 1)(A_i - \pi)1\{S_i = s\}\right| = O_p(1). 
	\end{align*}
	To see this claim, we note that, conditionally on data, 
	\begin{align*}
	\frac{1}{n}\sum_{i=1}^n(A_i - \pi)^21\{S_i = s\} = & \frac{1}{n}\sum_{i=1}^n\left(A_i - \pi - 2\pi(A_i - \pi) + \pi - \pi^2\right)1\{S_i=s\} \\
	= & \frac{D_n(s) - 2\pi D_n(s)}{n} + \pi(1-\pi)\frac{n(s)}{n} \convP \pi (1-\pi)p(s).  
	\end{align*}
	Then, by the Lindeberg CLT, conditionally on data, 
	\begin{align*}
	\frac{1}{\sqrt{n}} \sum_{i =1}^n (\xi_i - 1)(A_i - \pi)1\{S_i = s\} \convD N(0,\pi(1-\pi)p(s)) = O_p(1). 
	\end{align*}
	Last, by the standard maximal inequality (e.g., \citet[Theorem 2.14.1]{VW96}) and the fact that 
	\begin{align*}
	\left(\frac{m_1(s,\tau)}{f_1(q_1(\tau))} - \frac{m_0(s,\tau)}{f_0(q_0(\tau))}\right)
	\end{align*}
	is Lipschitz in $\tau$, we have, as $n \rightarrow \infty$ followed by $\delta \rightarrow 0$,  
	\begin{align*}
	\mathbb{P}\left( \sup_{\tau_1,\tau_2 \in \Upsilon, \tau_1 < \tau_2 < \tau_1+\delta,s \in \mathcal{S}}\left|\frac{1}{\sqrt{n}}\sum_{i=1}^n (\xi_i - 1)(\mathcal{J}_{i,3}(s,\tau_2) - \mathcal{J}_{i,3}(s,\tau_1))\right| \geq \eps/3\right) \rightarrow 0.
	\end{align*}
	This concludes the proof of the conditionally stochastic equicontinuity.

	\textbf{Step 2.}
	We focus on the one-dimension case and aim to show that, conditionally on data, for fixed $\tau \in \Upsilon$,  
	\begin{align*}
	\frac{1}{\sqrt{n}}\sum_{s \in \mathcal{S}} \sum_{i =1}^n (\xi_i-1) \mathcal{J}_i(s,\tau) \convD \N(0,\tilde{\Sigma}_{sqr}(\tau,\tau)).
	\end{align*}
	The finite-dimensional convergence can be established similarly by the Cram\'{e}r-Wold device. In view of Lindeberg-Feller central limit theorem, we only need to show that (1)
	\begin{align*}
	\frac{1}{n}\sum_{i =1}^n[\sum_{s \in \mathcal{S}}\mathcal{J}_i(s,\tau)]^2 \convP \zeta_Y^2(\pi,\tau) + \tilde{\xi}_A^{2}(\pi,\tau) + \xi_S^2(\pi,\tau)
	\end{align*}
	and (2)
	\begin{align*}
	\frac{1}{n}\sum_{i =1}^n[\sum_{s \in \mathcal{S}}\mathcal{J}_i(s,\tau)]^2 \mathbb{E}_\xi(\xi-1)^21\{|\sum_{s \in \mathcal{S}}(\xi_i - 1)\mathcal{J}_i(s,\tau)| \geq \eps \sqrt{n}\} \rightarrow 0.
	\end{align*}
	(2) is obvious as $|\mathcal{J}_i(s,\tau)|$ is bounded and $\max_i|\xi_i-1| \lesssim \log(n)$ as $\xi_i$ is sub-exponential. Next, we focus on (1). We have 
	\begin{align*}
	& \frac{1}{n}\sum_{i =1}^n[\sum_{s \in \mathcal{S}}\mathcal{J}_i(s,\tau)]^2 \\
	= & \frac{1}{n}\sum_{i=1}^n\sum_{s \in \mathcal{S}}\biggl\{\biggl[\frac{A_i 1\{S_i = s\}\eta_{i,1}(s,\tau)}{\pi f_1(q_1(\tau))} - \frac{(1-A_i)1\{S_i = s\}\eta_{i,0}(s,\tau)}{(1-\pi) f_0(q_0(\tau))}\biggr] \\
	& + F_1(s,\tau)(A_i -\pi)1\{S_i=s\} + \biggl[ \left(\frac{m_1(s,\tau)}{f_1(q_1(\tau))} - \frac{m_0(s,\tau)}{f_0(q_0(\tau))}\right)1\{S_i = s\}\biggr]\biggr\}^2 \\
	\equiv & \sigma_1^2 + \sigma_2^2 + \sigma_3^2 + 2\sigma_{12} + 2\sigma_{13} + 2 \sigma_{23},
	\end{align*}
	where 
	\begin{align*}
	\sigma_1^2 = \frac{1}{n}\sum_{s \in \mathcal{S}}\sum_{i =1}^n \biggl[\frac{A_i 1\{S_i = s\}\eta_{i,1}(s,\tau)}{\pi f_1(q_1(\tau))} - \frac{(1-A_i)1\{S_i = s\}\eta_{i,0}(s,\tau)}{(1-\pi) f_0(q_0(\tau))}\biggr]^2,
	\end{align*}
	\begin{align*}
	\sigma_2^2 = \frac{1}{n}\sum_{s \in \mathcal{S}}F^2_1(s,\tau)\sum_{i =1}^n (A_i-\pi)^21\{S_i = s\},
	\end{align*}
	\begin{align*}
	\sigma_3^2 = \frac{1}{n}\sum_{i =1}^n \biggl[ \left(\frac{m_1(S_i,\tau)}{f_1(q_1(\tau))} - \frac{m_0(S_i,\tau)}{f_0(q_0(\tau))}\right)\biggr]^2,
	\end{align*}
	\begin{align*}
	\sigma_{12} = \frac{1}{n}\sum_{i=1}^n\sum_{s \in \mathcal{S}}\biggl[\frac{A_i 1\{S_i = s\}\eta_{i,1}(s,\tau)}{\pi f_1(q_1(\tau))} - \frac{(1-A_i)1\{S_i = s\}\eta_{i,0}(s,\tau)}{(1-\pi) f_0(q_0(\tau))}\biggr]F_1(s,\tau)(A_i - \pi)1\{S_i=s\},
	\end{align*}
	
	\begin{align*}
	\sigma_{13} = \frac{1}{n}\sum_{i=1}^n\sum_{s \in \mathcal{S}}\biggl[\frac{A_i 1\{S_i = s\}\eta_{i,1}(s,\tau)}{\pi f_1(q_1(\tau))} - \frac{(1-A_i)1\{S_i = s\}\eta_{i,0}(s,\tau)}{(1-\pi) f_0(q_0(\tau))}\biggr] \biggl[ \left(\frac{m_1(s,\tau)}{f_1(q_1(\tau))} - \frac{m_0(s,\tau)}{f_0(q_0(\tau))}\right)\biggr],
	\end{align*}
	and 
	\begin{align*}
	\sigma_{23} = \sigma_{12} = \frac{1}{n}\sum_{i=1}^n\sum_{s \in \mathcal{S}}F_1(s,\tau)(A_i - \pi)1\{S_i=s\}\biggl[ \left(\frac{m_1(s,\tau)}{f_1(q_1(\tau))} - \frac{m_0(s,\tau)}{f_0(q_0(\tau))}\right)\biggr].
	\end{align*}
	For $\sigma_1^2$, we have 
	\begin{align*}
	\sigma_1^2 = & \frac{1}{n}\sum_{s \in \mathcal{S}}\sum_{i =1}^n \biggl[\frac{A_i 1\{S_i = s\}\eta^2_{i,1}(s,\tau)}{\pi^2 f^2_1(q_1(\tau))} - \frac{(1-A_i)1\{S_i = s\}\eta^2_{i,0}(s,\tau)}{(1-\pi)^2 f^2_0(q_0(\tau))}\biggr] \\
	\stackrel{d}{=} & \frac{1}{n}\sum_{s \in \mathcal{S}} \sum_{i=N(s)+1}^{N(s)+n_1(s)}\frac{\tilde{\eta}^2_{i,1}(s,\tau)}{\pi^2 f^2_1(q_1(\tau))} +  \frac{1}{n}\sum_{s \in \mathcal{S}} \sum_{i=N(s)+n_1(s)+1}^{N(s)+n(s)}\frac{\tilde{\eta}^2_{i,0}(s,\tau)}{(1-\pi)^2 f^2_0(q_0(\tau))} \\
	\convP & \frac{\tau(1-\tau) - \mathbb{E}m_1^s(S,\tau)}{\pi f_1^2(q_1(\tau))} + \frac{\tau(1-\tau) - \mathbb{E}m_0^s(S,\tau)}{(1-\pi) f_0^2(q_0(\tau))} = \zeta_Y^2(\pi,\tau),
	\end{align*}
	where the second equality holds due to the rearrangement argument in Lemma \ref{lem:Q} and the convergence in probability holds due to uniform convergence of the partial sum process. 
	
	For $\sigma_2^2$, by Assumption \ref{ass:assignment1},  
	\begin{align*}
	\sigma_2^2 = \frac{1}{n}\sum_{s \in \mathcal{S}}F_1^2(s,\tau)(D_n(s) - 2\pi D_n(s) + \pi(1-\pi)1\{S_i =s\}) \convP \pi(1-\pi)\mathbb{E}F_1^2(S_i,\tau) = \tilde{\xi}_{A}^{2}(\pi,\tau).
	\end{align*}
	
	For $\sigma_3^2$, by the law of large number, 
	\begin{align*}
	\sigma_3^2 \convP \mathbb{E}\biggl[ \left(\frac{m_1(S_i,\tau)}{f_1(q_1(\tau))} - \frac{m_0(S_i,\tau)}{f_0(q_0(\tau))}\right)\biggr]^2 = \xi_S^2(\pi,\tau).
	\end{align*}
	
	For $\sigma_{12}$, we have 
	\begin{align*}
	\sigma_{12} = & \frac{1}{n}\sum_{s \in \mathcal{S}}(1-\pi)F_1(s,\tau)\sum_{i=1}^n\frac{A_i 1\{S_i = s\}\eta_{i,1}(s,\tau)}{\pi f_1(q_1(\tau))} - \frac{1}{n}\sum_{s \in \mathcal{S}}\pi F_1(s,\tau)\sum_{i=1}^{n}\frac{(1-A_i)1\{S_i = s\}\eta_{i,0}(s,\tau)}{(1-\pi) f_0(q_0(\tau))} \\
	\stackrel{d}{=} & \frac{1}{n}\sum_{s \in \mathcal{S}}(1-\pi)F_1(s,\tau)\sum_{i=N(s)+1}^{N(s)+n_1(s)}\frac{\tilde{\eta}_{i,1}(s,\tau)}{\pi f_1(q_1(\tau))} - \frac{1}{n}\sum_{s \in \mathcal{S}}\pi F_1(s,\tau)\sum_{i=N(s)+n_1(s)+1}^{N(s)+n(s)}\frac{\tilde{\eta}_{i,0}(s,\tau)}{(1-\pi) f_0(q_0(\tau))} \convP 0,
	\end{align*}
	where the last convergence holds because by Lemma \ref{lem:Q}, 
	\begin{align*}
	\frac{1}{n}\sum_{i=N(s)+1}^{N(s)+n_1(s)}\tilde{\eta}_{i,1}(s,\tau) \convP 0, \quad \text{and} \quad \frac{1}{n}\sum_{i=N(s)+n_1(s)+1}^{N(s)+n(s)}\tilde{\eta}_{i,0}(s,\tau) \convP 0.
	\end{align*}
	By the same argument, we can show that 
	\begin{align*}
	\sigma_{13} \convP 0.
	\end{align*}
	
	Last, for $\sigma_{23}$, by Assumption \ref{ass:assignment1},  
	\begin{align*}
	\sigma_{23} = \sum_{s \in \mathcal{S}}F_1(s,\tau)\biggl[ \left(\frac{m_1(s,\tau)}{f_1(q_1(\tau))} - \frac{m_0(s,\tau)}{f_0(q_0(\tau))}\right)\biggr]\frac{D_n(s)}{n} \convP 0. 
	\end{align*}
	Therefore, conditionally on data, 
	\begin{align*}
	\frac{1}{n}\sum_{i =1}^n[\sum_{s \in \mathcal{S}}\mathcal{J}_i(s,\tau)]^2 \convP \zeta_Y^2(\pi,\tau) + \tilde{\xi}_A^{2}(\pi,\tau) + \xi_S^2(\pi,\tau). 
	\end{align*}
\end{proof}


\begin{lem}
	\label{lem:Rsfestar}
	If Assumptions \ref{ass:assignment1}.1 and \ref{ass:assignment1}.2 hold, $\sup_{ s \in \mathcal{S}}\frac{|D_n^*(s)|}{\sqrt{n^*(s)}} = O_p(1)$, $\sup_{ s \in \mathcal{S}}\frac{|D_n(s)|}{\sqrt{n(s)}} = O_p(1)$, and $n(s) \rightarrow \infty$ for all $s \in \mathcal{S}$, a.s., then there exists a sequence of Poisson(1) random variables $\{\xi_i^s\}_{i \geq 1,s\in \mathcal{S}}$ independent of $\{A_i^*,S_i^*,Y_i,A_i,S_i\}_{i\geq 1}$ such that
	\begin{align*}
	\sum_{i =1}^nA_i^*1\{S_i^*=s\}\eta_{i,1}^*(s,\tau) = \sum_{i = N(s)+1}^{N(s)+n_1(s)}\xi_i^s\tilde{\eta}_{i,1}(s,\tau) + R^*_1(s,\tau),
	\end{align*}
	where $\sup_{\tau \in \Upsilon, s \in \mathcal{S}}|R^*_1(s,\tau)/\sqrt{n(s)}| = o_p(1).$ In addition, 
	\begin{align}
	\label{eq:targetF9}
	\sup_{s \in \mathcal{S},\tau \in \Upsilon}|\sum_{i =1}^nA_i^*1\{S_i^*=s\}\eta_{i,1}^*(s,\tau)|/\sqrt{n(s)} = O_p(1).
	\end{align}
\end{lem}
\begin{proof}
	Recall $\{Y^s_i(0),Y^s_i(1)\}_{i=1}^n $ as defined in the proof of Theorem \ref{thm:qr} and 
	$$\tilde{\eta}_{i,j}(s,\tau) = \tau - 1\{Y^s_i(j) \leq q_j(\tau)\} - m_j(s,\tau),$$ 
	$j = 0,1$. In addition, let $\Psi_n = \{\eta_{i,1}(s,\tau)\}_{i=1}^n$, 
	$$\mathbb{N}_n = \{n(s)/n,n_1(s)/n, n^*(s)/n,n_1^*(s)/n\}_{s \in \mathcal{S}}$$
	and given $\mathbb{N}_n$, $\{M_{ni}\}_{i=1}^n$ be a sequence of random variables such that the $n_1(s) \times 1$ vector 
	$$M_n^1(s) = (M_{n,N(s)+1}, \cdots,M_{n,N(s)+n_1(s)})$$ 
	and the $(n(s) - n_1(s)) \times 1$ vector 
	$$M_n^0(s) = (M_{n,N(s)+n_1(s)+1}, \cdots,M_{n,N(s)+n(s)})$$ 
	satisfy: 
	\begin{enumerate}
		\item $M_n^1(s) = \sum_{i=1}^{n_1^*(s)}m_i$ and $M_n^0(s) = \sum_{i=1}^{n^*(s) - n_1^*(s)}m_i'$, where $\{m_i\}_{i=1}^{n_1^*(s)}$ and $\{m_i'\}_{i=1}^{n^*(s) - n_1^*(s)}$ are $n_1^*(s)$ i.i.d. multinomial$(1,n_1^{-1}(s),\cdots,n_1^{-1}(s))$ random vectors and $n^*(s) - n_1^*(s)$ i.i.d. multinomial $(1,(n(s) -n_1(s))^{-1},\cdots,(n(s) -n_1(s))^{-1})$ random vectors, respectively;
		\item $M_n^0(s) \indep M_n^1(s)|\mathbb{N}_n;$ and 
		\item $\{M_n^0(s),M_n^1(s)\}_{s \in \mathcal{S}}$ are independent across $s$ given  $\mathbb{N}_n$ and are independent of $\Psi_n$.
	\end{enumerate}
	
	Recall that, by \cite{BCS17}, the original observations can be rearranged according to $s \in \mathcal{S}$ and then within strata, treatment group first and then the control group. Then, given $\mathbb{N}_n$, Step 3 in Section \ref{sec:CABI} implies that the bootstrap observations $\{Y_i^*\}_{i=1}^n$ can be generated by drawing with replacement from the empirical distribution of the outcomes in each $(s,a)$ cell for $(s,a) \in \mathcal{S} \times \{0,1\}$, $n^*_{a}(s)$ times, $a = 0,1$, where $n_0^*(s) = n^*(s) - n_1^*(s)$. Therefore,  
	\begin{align}
	\label{eq:Astar}
	\sum_{i =1}^n A_i^*1\{S_i^*=s\}\eta^*_{i,1}(s,\tau) = \sum_{i=N(s)+1}^{N(s)+n_1(s)} M_{ni} \tilde{\eta}_{i,1}(s,\tau). 
	\end{align}
	
	Following the standard approach in dealing with the nonparametric bootstrap, we want to approximate $$M_{ni}, i=N(s)+1,\cdots,N(s)+n_1(s)$$ by a sequence of i.i.d. Poisson(1) random variables. We construct this sequence as follows. Let $\widetilde{M}_n^1(s) =\sum_{i=1}^{\widetilde{N}(n_1(s))}m_i$, where $\widetilde{N}(k)$ is a Poisson number with mean $k$ and is independent of $\mathbb{N}_n$. The $n_1(s)$ elements of vector $\widetilde{M}_n^1(s)$ is denoted as $\{\widetilde{M}_{ni}\}_{i=N(s)+1}^{N(s)+n_1(s)}$, which is a sequence of i.i.d. Poisson(1) random variables, given $\mathbb{N}_n$. Therefore, 
	\begin{align*}
	\{\widetilde{M}_{ni}, i=N(s)+1,\cdots,N(s)+n_1(s) |\mathbb{N}_n\} \equiv
	\{\xi_i^s, i=N(s)+1,\cdots,N(s)+n_1(s)|\mathbb{N}_n\}
	\end{align*}
	where $\{\xi^s_i\}_{i=1}^n$, $s \in \mathcal{S}$ are i.i.d. sequences of Poisson(1) random variables such that  $\{\xi^s_i\}_{i=1}^n$ are independent across $s \in \mathcal{S}$ and against $\mathbb{N}_n$. 
	
	Following the argument in \citet[Section 3.6]{VW96}, given $n_1(s)$, $n_1^*(s)$, and $\widetilde{N}(n_1(s)) = k$, $|\xi_i^s - M_{ni}|$ is binomially $(|k-n_1^*(s)|,n_1(s)^{-1})$-distributed. In addition, there exists a sequence $\ell_n = O(\sqrt{n(s)})$ such that 
	\begin{align*}
	\mathbb{P}(|\widetilde{N}(n_1(s)) - n_1^*(s)| \geq \ell_n) \leq & \mathbb{P}(|\widetilde{N}(n_1(s)) - n_1(s)| \geq \ell_n/3) + \mathbb{P}(|n_1^*(s) - n_1(s)| \geq 2\ell_n/3) \\
	\leq & \mathbb{E}\mathbb{P}(|N(n_1(s)) - n_1(s)| \geq \ell_n/3|n_1(s)) + \mathbb{P}(|n_1^*(s) - n_1(s)| \geq 2\ell_n/3) \\
	\leq & \eps/3 + \mathbb{P}(|n_1^*(s) - n_1(s)| \geq 2\ell_n/3)  \\
	\leq & \eps/3 + \mathbb{P}(|D_n^*(s)| + |D_n(s)| + \pi|n^*(s) - n(s)| \geq 2\ell_n/3) \\
	\leq & 2\eps/3 + \mathbb{P}(\pi|n^*(s) - n(s)| \geq \ell_n/3) \\
	\leq & \eps, 
	\end{align*} 
	where the first inequality holds due to the union bound inequality, the second inequality holds by the law of iterated expectation, the third inequality holds because (1) conditionally on data, $\widetilde{N}(n_1(s)) - n_1(s) = O_p(\sqrt{n_1(s)})$ and (2) $n_1(s)/n(s) = \pi+ \frac{D_n(s)}{n(s)} \rightarrow \pi>0$ as $n(s) \rightarrow \infty$ , the fourth inequality holds by the fact that 
	$$n_1^*(s) - n_1(s) = D_n^*(s) - D_n(s) + \pi(n^*(s) - n(s)),$$
	the fifth inequality holds because by Assumptions \ref{ass:assignment1} and \ref{ass:bassignment},  $|D_n^*(s)| + |D_n(s)| = O_p(\sqrt{n(s)})$, and the sixth inequality holds because $\{S_i^*\}_{i=1}^n$ is generated from $\{S_i\}_{i=1}^n$ by the standard bootstrap procedure, and thus, by \citet[Theorem 3.6.1]{VW96}, 
	\begin{align*}
	n^*(s) - n(s) = \sum_{i=1}^n(M^w_{ni}-1)(1\{S_i = s\} - p(s)) = O_p(\sqrt{n(s)}), 
	\end{align*}
	where $(M^w_{n1},\cdots,M^w_{nn})$ is independent of $\{S_i\}_{i=1}^n$ and multinomially distributed with parameters $n$ and (probabilities) $1/n,\cdots,1/n$. Therefore, by direct calculation, as $n \rightarrow \infty$,
	\begin{align*}
	& \mathbb{P}(\max_{N(s)+1 \leq i \leq N(s)+n_1(s)}|\xi_{i}^s - M_{ni}|>2) \notag \\
	\leq & \mathbb{P}(\max_{N(s)+1 \leq i \leq N(s)+n_1(s)}|\xi_{i}^s - M_{ni}|>2, n_1(s) \geq n(s)\eps) + \mathbb{P}(n_1(s) \leq n(s)\eps) \notag \\
	\leq & \eps + \mathbb{E}\sum_{i=N(s)+1}^{N(s)+n_1(s)}\mathbb{P}(|\xi_{i}^s - M_{ni}|>2,|N(n_1(s)) - n_1^*(s)| \leq \ell_n, n_1(s) \geq n(s) \eps|n_1(s),n_1^*(s),n(s)) + \eps  \notag \\
	\leq & 2\eps +  \mathbb{E}n_1(s)\mathbb{P}(\text{bin}(\ell_n,n^{-1}_1(s))>2|n_1(s),n_1^*(s),n(s))1\{n_1(s) \geq n(s) \eps\} \rightarrow 2\eps,  
	\end{align*} 
	where we use the fact that 
	\begin{align*}
	& n_1(s)\mathbb{P}(\text{bin}(\ell_n,n^{-1}_1(s))>2|n_1(s),n_1^*(s),n(s))1\{n_1(s) \geq n(s) \eps\} \\
	\lesssim & n_1(s)\left(\frac{\ell_n}{n(s)}\right)^3\left(\frac{n(s)}{n_1(s)}\right)^31\{n_1(s) \geq n(s) \eps\} \lesssim \frac{1}{\sqrt{n(s)}\eps^3} \rightarrow 0.
	\end{align*}
	Because $\eps$ is arbitrary, we have  
	\begin{align}
	\label{eq:F92}
	\mathbb{P}\left(\max_{N(s)+1 \leq i \leq N(s)+n_1(s)}|\xi_{i}^s - M_{ni}|>2\right) \rightarrow 0. 
	\end{align}
	
	Note that $|\xi_{i}^s - M_{ni}| = \sum_{j=1}^\infty 1\{|\xi_i^s - M_{ni}|\geq j\}$. Let $I_n^j(s)$ be the set of indexes $i \in \{N(s)+1,\cdots,N(s)+n_1(s)\}$ such that $|\xi_i^s - M_{ni}|\geq j$. Then, $\xi_i^s - M_{ni} = \text{sign}(\widetilde{N}(n_1(s))-n^*_1(s))\sum_{j=1}^\infty 1\{i \in I_n^j(s)\}$. Thus, 
	\begin{align}
	\label{eq:boot}
	\frac{1}{\sqrt{n(s)}}\sum_{i=N(s) + 1}^{N(s)+n_1(s)}(\xi_i^s-M_{ni})\tilde{\eta}_{i,1}(s,\tau) = \text{sign}(\widetilde{N}(n_1(s))-n^*_1(s))\sum_{j=1}^\infty \left[ \frac{\#I_n^j(s)}{\sqrt{n(s)}}\frac{1}{\#I_n^j(s)}\sum_{i \in I_n^j(s)} \tilde{\eta}_{i,1}(s,\tau)\right].
	\end{align}
	
	In the following, we aim to show that the RHS of \eqref{eq:boot} converges to zero in probability uniformly over $s \in \mathcal{S},\tau \in \Upsilon$. First, note that, by \eqref{eq:F92}, $\max_{N(s)+1 \leq i \leq N(s)+n_1(s)}|\xi_i^s - M_{ni}| \leq 2$ occurs with probability approaching one. In the event set that $\max_{N(s)+1 \leq i \leq N(s)+n_1(s)}|\xi_i^s - M_{ni}| \leq 2$, only the first two terms of the first summation on the RHS of \eqref{eq:boot} can be nonzero. In addition, for any $j$, we have $j (\#I_n^j(s)) \leq |\widetilde{N}(n_1(s)) - n_1(s)| = O_p(\sqrt{n(s)})$, and thus, $\frac{\#I_n^j(s)}{\sqrt{n(s)}} = O_p(1)$ for $j = 1,2$. Therefore, it suffices to show that, for $j = 1,2$,  
	
	\begin{align*}
	\sup_{s \in \mathcal{S},\tau \in \Upsilon}\left|\frac{1}{\#I_n^j(s)}\sum_{i \in I_n^j(s)} \tilde{\eta}_{i,1}(s,\tau)\right| = o_p(1). 
	\end{align*} 
	Note that 
	\begin{align}
	\label{eq:omega}
	\frac{1}{\#I_n^j(s)}\sum_{i \in I_n^j(s)} \tilde{\eta}_{i,1}(s,\tau) = \sum_{i=N(s) + 1}^{N(s)+n_1(s)}\omega_{ni}\tilde{\eta}_{i,1}(s,\tau),
	\end{align}
	where $\omega_{ni} = \frac{1\{|\xi_i^s - M_{ni}|\geq j\}}{\#I_n^j(s)}$, $i=N(s) + 1,\cdots,N(s)+n_1(s)$ and by construction, $\{\omega_{ni}\}_{i=N(s) + 1}^{N(s)+n_1(s)}$ is independent of $\{\eta_{i,1}(s,\tau)\}_{i=1}^{n}$. In addition, because $\{\omega_{ni}\}_{i=N(s) + 1}^{N(s)+n_1(s)}$ is exchangeable conditional on $\mathbb{N}_n$, so be it unconditionally. Third, $\sum_{i=N(s) + 1}^{N(s)+n_1(s)}\omega_{ni} = 1$ and $\max_{i=N(s) + 1,\cdots,N(s)+n_1(s)}|\omega_{ni}| \leq 1/\#I_n^j(s) \convP 0$. Then, by the same argument in the proof of \citet[Lemma 3.6.16]{VW96}, for some $r \in (0,1)$ and any $n_0 = N(s)+1,\cdots,N(s)+n_1(s)$, we have 
	\begin{align}
	\label{eq:3617}
	& \mathbb{E}\left(\sup_{\tau \in \Upsilon, s\in \mathcal{S}}\left|\sum_{i=N(s) + 1}^{N(s)+n_1(s)}\omega_{ni}\tilde{\eta}_{i,1}(s,\tau)\right|^r|\Psi_n,\mathbb{N}_n\right) \notag \\
	\leq & (n_0 -1)\mathbb{E}\left[\max_{N(s)+n_0 \leq i \leq N(s)+n_1(s)}\omega_{ni}^r|\mathbb{N}_n\right] \left[\frac{1}{n_1(s)}\sum_{i=N(s) + 1}^{N(s)+n_1(s)}\sup_{\tau \in \Upsilon, s\in \mathcal{S}}|\tilde{\eta}_{i,1}^r(s,\tau)|\right] \notag \\
	& + (n_1(s)\mathbb{E}(\omega_{ni}|\mathbb{N}_n))^r\max_{n_0 \leq k \leq n_1(s)}\mathbb{E}\left[\sup_{\tau \in \Upsilon, s\in \mathcal{S}}\left|\frac{1}{k}\sum_{j=N(s)+n_0}^{N(s)+k} \tilde{\eta}_{R_j(N(s),n_1(s)),1}(s,\tau)\right|^r|\mathbb{N}_n,\Psi_n\right],
	\end{align}
	where $(R_{k_1+1}(k_1,k_2),\cdots,R_{k_1+k_2}(k_1,k_2))$ is uniformly distributed on the set of all permutations of $k_1+1,\cdots,k_1+k_2$ and independent of $\mathbb{N}_n$ and $\Psi_n$. First note that $\sup_{s \in \mathcal{S},\tau \in \Upsilon}|\eta_{i,1}(s,\tau)|$ is bounded and $$\max_{N(s)+1 \leq i \leq N(s)+n_1(s)}\omega_{ni}^r \leq 1/(\#I_n^j(s))^r \convP 0.$$ Therefore, the first term on the RHS of \eqref{eq:3617} converges to zero in probability for every fixed $n_0$. For the second term, because $\omega_{ni}|\mathbb{N}_n$ is exchangeable, 
	\begin{align*}
	n_1(s)\mathbb{E}(\omega_{ni}|\mathbb{N}_n) = \sum_{i=N(s) + 1}^{N(s)+n_1(s)}\mathbb{E}(\omega_{ni}|\mathbb{N}_n) = 1.
	\end{align*}
	
	In addition, let $\mathbb{S}_n(k_1,k_2)$ be the $\sigma$-field generated by all functions of $\{\tilde{\eta}_{i,1}(s,\tau)\}_{i\geq 1}$ that are symmetric in their $k_1+1$ to $k_1+k_2$ arguments. Then, 
	\begin{align*}
	& \max_{n_0 \leq k \leq n_1(s)}\mathbb{E}\left[\sup_{\tau \in \Upsilon, s\in \mathcal{S}}\left|\frac{1}{k}\sum_{j=N(s)+n_0}^{N(s)+k} \tilde{\eta}_{R_j(N(s),n_1(s)),1}(s,\tau)\right|^r|\mathbb{N}_n, \Psi_n\right] \\
	= & \max_{n_0 \leq k \leq n_1(s)}\mathbb{E}\left[\sup_{\tau \in \Upsilon, s\in \mathcal{S}}\left|\frac{1}{k}\sum_{j=N(s)+n_0}^{N(s)+k} \tilde{\eta}_{j,1}(s,\tau)\right|^r|\mathbb{N}_n,\mathbb{S}_n(N(s),n_1(s))\right] \\
	\leq & 2\mathbb{E}\left\{\max_{n_0 \leq k }\left[\sup_{\tau \in \Upsilon, s\in \mathcal{S}}\left|\frac{1}{k}\sum_{j=N(s)+1}^{N(s)+k} \tilde{\eta}_{j,1}(s,\tau)\right|^r\right]|\mathbb{N}_n,\mathbb{S}_n(N(s),n_1(s))\right\} \\
	= & 2\mathbb{E}\left\{\max_{n_0 \leq k }\left[\sup_{\tau \in \Upsilon, s\in \mathcal{S}}\left|\frac{1}{k}\sum_{j=1}^{k} \tilde{\eta}_{j,1}(s,\tau)\right|^r\right]|\mathbb{N}_n,\mathbb{S}_n(0,n_1(s))\right\},
	\end{align*}
	where the inequality holds by the Jansen's inequality and the triangle inequality and the last equality holds because $\{\tilde{\eta}_{j,1}(s,\tau)\}_{j\geq 1}$ is an i.i.d. sequence. Apply expectation on both sides, we obtain that 
	\begin{align}
	\label{eq:kk}
	& \mathbb{E}\max_{n_0 \leq k \leq n_1(s)}\mathbb{E}\left[\sup_{\tau \in \Upsilon, s\in \mathcal{S}}\left|\frac{1}{k}\sum_{j=N(s)+n_0}^{N(s)+k} \tilde{\eta}_{R_j(N(s),n_1(s)),1}(s,\tau)\right|^r|\mathbb{N}_n, \Psi_n\right] \notag \\
	\leq & 2\mathbb{E}\max_{n_0 \leq k \leq n}\left[\sup_{\tau \in \Upsilon, s\in \mathcal{S}}\left|\frac{1}{k}\sum_{j=1}^{k} \tilde{\eta}_{j,1}(s,\tau)\right|^r\right].
	\end{align}
	By the usual maximal inequality, as $k \rightarrow \infty$,
	\begin{align*}
	\sup_{\tau \in \Upsilon, s\in \mathcal{S}}\left|\frac{1}{k}\sum_{j=1}^{k} \tilde{\eta}_{j,1}(s,\tau)\right| \stackrel{a.s.}{\longrightarrow} 0,
	\end{align*}
	which implies that as $n_0 \rightarrow \infty$
	\begin{align*}
	\max_{n_0 \leq k \leq n}\left[\sup_{\tau \in \Upsilon, s\in \mathcal{S}}\left|\frac{1}{k}\sum_{j=1}^{k} \tilde{\eta}_{j,1}(s,\tau)\right|^r\right] \leq \max_{n_0 \leq k }\left[\sup_{\tau \in \Upsilon, s\in \mathcal{S}}\left|\frac{1}{k}\sum_{j=1}^{k} \tilde{\eta}_{j,1}(s,\tau)\right|^r\right] \stackrel{a.s.}{\longrightarrow} 0.
	\end{align*}
	In addition, $\sup_{\tau \in \Upsilon, s\in \mathcal{S}}\left|\frac{1}{k}\sum_{j=1}^{k} \tilde{\eta}_{j,1}(s,\tau)\right|$ is bounded. Then, by the bounded convergence theorem, we have, as $n_0 \rightarrow \infty$, 
	\begin{align*}
	\mathbb{E}\max_{n_0 \leq k \leq n}\left[\sup_{\tau \in \Upsilon, s\in \mathcal{S}}\left|\frac{1}{k}\sum_{j=1}^{k} \tilde{\eta}_{j,1}(s,\tau)\right|^r\right] \rightarrow 0. 
	\end{align*}
	
	which implies that, 
	\begin{align*}
	\mathbb{E}\max_{n_0 \leq k \leq n_1(s)}\mathbb{E}\left[\sup_{\tau \in \Upsilon, s\in \mathcal{S}}\left|\frac{1}{k}\sum_{j=N(s)+n_0}^{N(s)+k} \tilde{\eta}_{R_j(N(s),n_1(s)),1}(s,\tau)\right|^r|\mathbb{N}_n, \Psi_n\right] \convP 0.
	\end{align*}
	Therefore, the second term on the RHS of \eqref{eq:3617} converges to zero in probability as $n_0 \rightarrow \infty$. Then, as $n \rightarrow \infty$ followed by $n_0 \rightarrow \infty$, 
	\begin{align*}
	\mathbb{E}\left(\sup_{\tau \in \Upsilon, s\in \mathcal{S}}\left|\sum_{i=N(s) + 1}^{N(s)+n_1(s)}\omega_{ni}\tilde{\eta}_{i,1}(s,\tau)\right|^r|\Psi_n,\mathbb{N}_n\right) \convP 0.
	\end{align*}
	Hence, by the Markov inequality and \eqref{eq:omega}, we have 
	\begin{align*}
	\sup_{s \in \mathcal{S},\tau \in \Upsilon}\left|\frac{1}{\#I_n^j(s)}\sum_{i \in I_n^j(s)} \tilde{\eta}_{i,1}(s,\tau)\right| \convP 0.
	\end{align*}
	
	Consequently, following \eqref{eq:boot}
	\begin{align}
	\label{eq:Mtilde-M}
	\sup_{s \in \mathcal{S},\tau \in \Upsilon}\left\vert\sum_{i=N(s) + 1}^{N(s)+n_1(s)}(\xi_i^s-M_{ni})\tilde{\eta}_{i,1}(s,\tau)\right \vert  = o_p(\sqrt{n(s)}). 
	\end{align}
	This concludes the first part of this Lemma. For the second part, we note 
	\begin{align*}
	\sum_{i=N(s) + 1}^{N(s)+n_1(s)}\widetilde{M}_{ni}\tilde{\eta}_{i,1}(s,\tau) \stackrel{d}{=} \sum_{i=N(s) + 1}^{N(s)+n_1(s)}\xi^s_i\tilde{\eta}_{i,1}(s,\tau) \stackrel{d}{=} \sum_{i= 1}^{n_1(s)}\xi^s_i\tilde{\eta}_{i,1}(s,\tau),
	\end{align*}
	where the second equality holds because $\{\xi_i^s,\tilde{\eta}_{i,1}(s,\tau)\}_{i \geq 1} \indep \{N(s),n_1(s),n(s)\}$. Then, conditionally on $\{N(s),n_1(s),n(s)\}$ and uniformly over $s \in \mathcal{S}$, the usual maximal inequality (\citet[Theorem 2.14.1]{VW96}) implies  
	\begin{align}
	\label{eq:xi}
	\sup_{\tau \in \Upsilon}|\sum_{i=N(s) + 1}^{N(s)+n_1(s)}\widetilde{M}_{ni}\tilde{\eta}_{i,1}(s,\tau)|  \stackrel{d}{=}  \sup_{\tau \in \Upsilon}|\sum_{i= 1}^{n_1(s)}\xi^s_i\tilde{\eta}_{i,1}(s,\tau)| = O_p(\sqrt{n(s)}).
	\end{align}
	
	Combining \eqref{eq:Astar}, \eqref{eq:Mtilde-M}, and \eqref{eq:xi}, we establish \eqref{eq:targetF9}. This concludes the proof. 
\end{proof}

\begin{lem}
	\label{lem:Qstar}
	If Assumptions \ref{ass:assignment1}.1 and \ref{ass:assignment1}.2 hold, $\sup_{ s \in \mathcal{S}}\frac{|D_n^*(s)|}{\sqrt{n^*(s)}} = O_p(1)$, $\sup_{ s \in \mathcal{S}}\frac{|D_n(s)|}{\sqrt{n(s)}} = O_p(1)$, and $n(s) \rightarrow \infty$ for all $s \in \mathcal{S}$, a.s., then, uniformly over $\tau \in \Upsilon$, 
	\begin{align*}
	Q^*_{n}(u,\tau) \convP \frac{1}{2}u'Qu.
	\end{align*}
\end{lem}
\begin{proof}
	Recall $Q^*_{n,1}(u,\tau)$ and $Q^*_{n,0}(u,\tau)$ defined in \eqref{eq:Qstar}. We focus on $Q^*_{n,1}(u,\tau)$. Recall the definition of $M_{ni}$ in the proof of Lemma \ref{lem:Rsfestar}. We have  
	\begin{align}
	\label{eq:l21nstar'qr}
	Q^*_{n,1}(u,\tau) = & \sum_{s \in \mathcal{S}} \sum_{i = N(s)+1}^{N(s)+n_1(s)}M_{ni}\int_0^{\frac{u_0 + u_1}{\sqrt{n}}}(1\{Y^s_i(1) - q_1(\tau) \leq v\} - 1\{Y^s_i(1) - q_1(\tau) \leq 0\})dv \notag \\
	= & \sum_{s \in \mathcal{S}} \sum_{i = N(s)+1}^{N(s)+n_1(s)}M_{ni}[\phi_i(u,\tau,s) - \mathbb{E}\phi_i(u,\tau,s)] + \sum_{s \in \mathcal{S}} \sum_{i = N(s)+1}^{N(s)+n_1(s)}M_{ni}\mathbb{E}\phi_i(u,\tau,s),
	\end{align}
	where $\phi_i(u,\tau,s) = \int_0^{\frac{u_0 + u_1}{\sqrt{n}}}(1\{Y^s_i(1) - q_1(\tau) \leq v\} - 1\{Y^s_i(1) - q_1(\tau) \leq 0\})dv$.
	
	Similar to \eqref{eq:Mtilde-M}, we have
	\begin{align}
	\label{eq:term1qr}
	& \sum_{s \in \mathcal{S}} \sum_{i=N(s) + 1}^{N(s)+n_1(s)}M_{ni}\left[\phi_i(u,\tau,s) - \mathbb{E}\phi_i(u,\tau,s) \right] \notag \\
	= & \sum_{s \in \mathcal{S}} \sum_{i=N(s) + 1}^{N(s)+n_1(s)}\xi_i^s\left[\phi_i(u,\tau,s) - \mathbb{E}\phi_i(u,\tau,s) \right] + \sum_{s \in \mathcal{S}}r_n(u,\tau,s),
	\end{align}
	where $\{\xi_i^s\}_{i=1}^n$ is a sequence of i.i.d. Poisson(1) random variables and is independent of everything else, and 
	\begin{align*}
	r_n(u,\tau,s) = \text{sign}(\widetilde{N}(n_1(s)) - n_1^*(s))\sum_{j=1}^\infty\frac{\#I_{n}^j(s)}{\sqrt{n(s)}}\frac{1}{\#I_{n}^j(s)}\sum_{i \in I_n^j(s)}\sqrt{n(s)}\left[\phi_i(u,\tau,s) - \mathbb{E}\phi_i(u,\tau,s) \right].
	\end{align*}
	We aim to show 
	\begin{align}
	\label{eq:rnqr}
	\sup_{\tau \in \Upsilon, s\in \mathcal{S}}|r_n(u,\tau,s)| = o_p(1), 
	\end{align}

	Recall that the proof of Lemma \ref{lem:Rsfestar} relies on \eqref{eq:kk} and the fact that 
	$$\mathbb{E}\sup_{n(s) \geq k \geq n_0}\sup_{\tau \in \Upsilon, s \in \mathcal{S}}\left|\frac{1}{k}\sum_{j=1}^k \tilde{\eta}_{j,1}(s,\tau)\right| \rightarrow 0.$$ 
	Using the same argument and replacing $\tilde{\eta}_{j,1}(s,\tau)$ by $\sqrt{n(s)}\left[\phi_i(u,\tau,s) - \mathbb{E}\phi_i(u,\tau,s) \right]$, in order to show \eqref{eq:rnqr}, we only need to verify that, as $n \rightarrow \infty$ followed by $n_0 \rightarrow \infty$, 
	\begin{align*}
	\mathbb{E}\sup_{n(s) \geq k \geq n_0}\sup_{\tau \in \Upsilon, s\in \mathcal{S}}\left|\frac{1}{k}\sum_{i=1}^k\sqrt{n(s)}\left[\phi_i(u,\tau,s) - \mathbb{E}\phi_i(u,\tau,s) \right]\right| \rightarrow 0.
	\end{align*}
	Note $\sup_{\tau \in \Upsilon, s\in \mathcal{S}}\left|\frac{1}{k}\sum_{i=1}^k\sqrt{n(s)}\left[\phi_i(u,\tau,s) - \mathbb{E}\phi_i(u,\tau,s) \right]\right|$ is bounded by $|u_0| + |u_1|$. It suffices to show that, for any $\eps>0$, as $n(s) \rightarrow \infty$ followed by $n_0 \rightarrow \infty$,  
	\begin{align}
	\label{eq:targetp_qr}
	\mathbb{P}\left(\sup_{n(s) \geq k \geq n_0}\sup_{\tau \in \Upsilon, s\in \mathcal{S}}\left|\frac{1}{k}\sum_{i=1}^k\sqrt{n(s)}\left[\phi_i(u,\tau,s) - \mathbb{E}\phi_i(u,\tau,s) \right]\right| \geq \eps \right) \rightarrow 0. 
	\end{align}
	Define the class of functions $\mathcal{F}_n$ as
	\begin{align*}
	\mathcal{F}_n = \{\sqrt{n(s)}\left[\phi_i(u,\tau,s) - \mathbb{E}\phi_i(u,\tau,s) \right]: \tau \in \Upsilon, s\in \mathcal{S} \}.
	\end{align*} 
	Then, $\mathcal{F}_n$ is nested by a VC-class with fixed VC-index. In addition, for fixed $u$, $\mathcal{F}_n$ has a bounded (and independent of $n$) envelope function $F = |u_0| + |u_1|$. Last, define $\mathcal{I}_l = \{2^l,2^l+1,\cdots,2^{l+1}-1\}$. Then, 
	\begin{align*}
	& \mathbb{P}\left(\sup_{n(s) \geq k \geq n_0}\sup_{\tau \in \Upsilon, s\in \mathcal{S}}\left|\frac{1}{k}\sum_{i=1}^k\sqrt{n(s)}\left[\phi_i(u,\tau,s,e) - \mathbb{E}\phi_i(u,\tau,s,e) \right]\right| \geq \eps \biggl|n(s)\right) \\
	\leq & \sum_{l=\lfloor \log_2(n_0) \rfloor}^{\lfloor \log_2(n(s)) \rfloor+1}\mathbb{P}\left(\sup_{k \in \mathcal{I}_l}\sup_{\tau \in \Upsilon, s\in \mathcal{S}}\left|\frac{1}{k}\sum_{i=1}^k\sqrt{n(s)}\left[\phi_i(u,\tau,s) - \mathbb{E}\phi_i(u,\tau,s) \right]\right| \geq \eps \biggl|n(s)\right) \\
	\leq & \sum_{l=\lfloor \log_2(n_0) \rfloor}^{\lfloor \log_2(n(s)) \rfloor+1}\mathbb{P}\left(\sup_{k \leq 2^{l+1}}\sup_{\tau \in \Upsilon, s\in \mathcal{S}}\left|\sum_{i=1}^k\sqrt{n(s)}\left[\phi_i(u,\tau,s) - \mathbb{E}\phi_i(u,\tau,s) \right]\right| \geq \eps 2^l \biggl| n(s)\right) \\ 
	\leq & \sum_{l=\lfloor \log_2(n_0) \rfloor}^{\lfloor \log_2(n(s)) \rfloor+1}9\mathbb{P}\left(\sup_{\tau \in \Upsilon, s\in \mathcal{S}}\left|\sum_{i=1}^{2^{l+1}}\sqrt{n(s)}\left[\phi_i(u,\tau,s) - \mathbb{E}\phi_i(u,\tau,s) \right]\right| \geq \eps 2^l/30 \biggl| n(s) \right) \\ 
	\leq & \sum_{l=\lfloor \log_2(n_0) \rfloor}^{\lfloor \log_2(n(s)) \rfloor+1}\frac{270 \mathbb{E}\left(\sup_{\tau \in \Upsilon, s\in \mathcal{S}}\left|\sum_{i=1}^{2^{l+1}}\sqrt{n(s)}\left[\phi_i(u,\tau,s) - \mathbb{E}\phi_i(u,\tau,s) \right] \right|\biggl| n(s)\right)}{\eps 2^l} \\
	\leq & \sum_{l=\lfloor \log_2(n_0) \rfloor}^{\lfloor \log_2(n(s)) \rfloor+1}\frac{C_1}{\eps 2^{l/2}} \\
	\leq & \frac{2C_1}{\eps \sqrt{n_0}} \rightarrow 0,
	\end{align*}
	where the first inequality holds by the union bound, the second inequality holds because on $\mathcal{I}_l$, $2^{l+1} \geq k \geq 2^l$, the third inequality follows the same argument in the proof of Theorem \ref{thm:qr}, the fourth inequality is due to the Markov inequality, the fifth inequality follows the standard maximal inequality such as \citet[Theorem 2.14.1]{VW96} and the constant $C_1$ is independent of $(l,\eps,n)$, and the last inequality holds by letting $n \rightarrow \infty$. Because $\eps$ is arbitrary, we have established \eqref{eq:targetp_qr}, and thus, \eqref{eq:rnqr}, which further implies that 
	\begin{align*}
	\sup_{\tau \in \Upsilon, s\in \mathcal{S}}|r_n(u,\tau,s)| = o_p(1). 
	\end{align*}
	In addition, for the leading term of \eqref{eq:term1qr}, we have 
	\begin{align*}
	& \sum_{s \in \mathcal{S}} \sum_{i=N(s) + 1}^{N(s)+n_1(s)}\xi_i^s\left[\phi_i(u,\tau,s) - \mathbb{E}\phi_i(u,\tau,s) \right] \\
	= & \sum_{s \in \mathcal{S}}\left[ \Gamma_n^{s*}(N(s)+n_1(s),\tau)- \Gamma_n^{s*}(N(s),\tau)\right],
	\end{align*}
	where 
	\begin{align*}
	\Gamma_n^{s*}(k,\tau,e) = & \sum_{i=1}^k \xi_i^s\int_0^{\frac{u_0+u_1}{\sqrt{n}}}\left(1\{Y_i^s(1) \leq q_1(\tau)+v\} - 1\{Y_i^s(1) \leq q_1(\tau)\} \right)dv \\
	& - k \mathbb{E}\left[\int_0^{\frac{u_0+u_1}{\sqrt{n}}}\left(1\{Y_i^s(1) \leq q_1(\tau)+v\} - 1\{Y_i^s(1) \leq q_1(\tau)\} \right)dv\right].
	\end{align*}
	By the same argument in \eqref{eq:Qn1}, we can show that 
	\begin{align*}
	\sup_{0 < t \leq 1,\tau \in \Upsilon}|\Gamma_n^{s*}(k,\tau,e) | = o_p(1), 
	\end{align*}
	where we need to use the fact that the Poisson(1) random variable has an exponential tail and thus
	\begin{align*}
	\mathbb{E}\sup_{i\in\{1,\cdots,n\},s\in \mathcal{S}} \xi_i^s = O(\log(n)). 
	\end{align*}
	Therefore, 
	\begin{align}
	\label{eq:l21nstar''qr}
	\sup_{\tau \in \Upsilon}\left|\sum_{s \in \mathcal{S}} \sum_{i=N(s) + 1}^{N(s)+n_1(s)}M_{ni}\left[\phi_i(u,\tau,s) - \mathbb{E}\phi_i(u,\tau,s) \right]\right| = o_p(1).
	\end{align}
	For the second term on the RHS of \eqref{eq:l21nstar'qr}, we have 
	\begin{align}
	\label{eq:l21nstar'''qr}
	\sum_{s \in \mathcal{S}} \sum_{i=N(s) + 1}^{N(s)+n_1(s)}M_{ni}\mathbb{E}\phi_i(u,\tau,s) = & \sum_{s \in \mathcal{S}}n_1^*(s)\mathbb{E}\phi_i(u,\tau,s) \notag \\
	= & \sum_{s \in \mathcal{S}} \pi p(s) \frac{f_1(q_1(\tau)|s)}{2}(u_0+u_1)^2 +o(1) \notag \\
	= &  \frac{\pi f_1(q_1(\tau))(u_0+u_1)^2}{2} +o(1),
	\end{align}
	where the $o(1)$ term holds uniformly over $\tau \in \Upsilon$, the first equality holds because $\sum_{i=N(s) + 1}^{N(s)+n_1(s)}M_{ni} = n^*_1(s)$ and the second equality holds by the same calculation in \eqref{eq:Qn1} and the facts that $n^*(s)/n \convP p(s)$ and 
	\begin{align*}
	\frac{n^*_1(s)}{n} = \frac{D_n^*(s)+\pi n^*(s)}{n} \convP \pi p(s).
	\end{align*}
	Combining \eqref{eq:l21nstar'qr}--\eqref{eq:rnqr}, \eqref{eq:l21nstar''qr}, and \eqref{eq:l21nstar'''qr}, we have 
	\begin{align*}
	Q^*_{n,1}(u,\tau) \convP \frac{\pi f_1(q_1(\tau))(u_0+u_1)^2}{2}, 
	\end{align*}
	uniformly over $\tau \in \Upsilon.$ By the same argument, we can show that, uniformly over $\tau \in \Upsilon$, 
	\begin{align*}
	Q^*_{n,0}(u,\tau) \convP \frac{(1-\pi) f_0(q_0(\tau))u^2_0}{2}. 
	\end{align*}
	This concludes the proof. 
\end{proof}

\newpage
\begin{center}
	\Large{Second Supplement to ``Quantile Treatment Effects and Bootstrap Inference under Covariate-Adaptive Randomization": Strata Fixed Effects Quantile Regression Estimation and Additional Simulation Results}
\end{center}
\begin{abstract}
This paper gathers the theories for the strata fixed effects quantile regression estimator and additional simulation results. Section \ref{sec:sfe} describes the estimation, weighted bootstrap, and covariate-adaptive bootstrap inference procedures for the strata fixed effects quantile regression estimator. Sections \ref{sec:proofsfe}--\ref{sec:proofcabsfe} prove Theorems \ref{thm:sfe}--\ref{thm:cabsfe}, respectively. Section \ref{sec:lem2} contains the proofs of the technical lemmas. Section \ref{sec:addsim} contains additional simulation results.  
\end{abstract}
\setcounter{page}{1} \renewcommand\thesection{\Alph{section}} %

\renewcommand{\thefootnote}{\arabic{footnote}} \setcounter{footnote}{0}

\setcounter{equation}{0}

\section{Quantile Regression with Strata Fixed Effects}
\label{sec:sfe}
The strata fixed effects estimator for the ATE is obtained by a linear regression of outcome $Y_i$ on the treatment status $A_i$, controlling for strata dummies $\{1\{S_i=s\}_{s \in \mathcal{S}}\}$. \cite{BCS17} point out that, due to the Frisch-Waugh-Lovell theorem, this estimator is equal to the linear coefficient in the regression of $Y_i$ on $\tilde{A}_i$, in which $\tilde{A}_i$ is the residual of the projection of $A_i$ on the strata dummies. Unlike the expectation, the quantile operator is nonlinear. Therefore, we cannot consistently estimate QTEs by a linear QR of $Y_i$ on $A_i$ and strata dummies. Instead, based on the equivalence relationship, we propose to run the QR of $Y_i$ on $\tilde{A}_i$. Formally, let $\tilde{A}_i = A_i - \hat{\pi}(S_i)$ and $\dot{\tilde{A}}_i = (1,\tilde{A}_i)'$, where $\hat{\pi}(s) = n_1(s)/n(s)$, $n_1(s) = \sum_{i=1}^n A_i 1\{S_i = s\}$, and $n(s) = \sum_{i=1}^n 1\{S_i = s\}$. Then, the strata fixed effects (SFE) estimator for the QTE is $\hat{\beta}_{sfe,1}(\tau)$, where   
\begin{align*}
\hat{\beta}_{sfe}(\tau) \equiv \left(\hat{\beta}_{sfe,0}(\tau),\hat{\beta}_{sfe,1}(\tau) \right)'= \argmin_{b = (b_0,b_1)^\prime \in \Re^2}\sum_{i =1}^n \rho_\tau\left(Y_i - \dot{\tilde{A}}_i^\prime b\right). 
\end{align*} 

\begin{thm}
	\label{thm:sfe}
	If Assumptions \ref{ass:assignment1}.1--\ref{ass:assignment1}.3 and \ref{ass:tau} hold and $p(s)>0$ for $s \in \mathcal{S}$, then, uniformly over $\tau \in \Upsilon$,  
	\begin{equation*}
	\sqrt{n}\left(\hat{\beta}_{sfe,1}(\tau)-q(\tau)\right)\convD \mathcal{B}_{sfe}(\tau),~\text{as}~n\rightarrow \infty,
	\end{equation*}
	where $\mathcal{B}_{sfe}(\cdot)$ is a Gaussian process with covariance kernel $\Sigma_{sfe}(\cdot,\cdot)$. The expression for $\Sigma_{sfe}(\cdot,\cdot)$ can be found in the proof of this theorem. 	
\end{thm}
In particular, the asymptotic variance for $\hat{\beta}_{sfe,1}(\tau)$ is 
\begin{align*}
\zeta_Y^2(\pi,\tau) + \zeta_A^{\prime 2}(\pi,\tau) + \zeta_S^2(\tau),
\end{align*}
where $\zeta_Y^2(\pi,\tau)$ and $\zeta_S^2(\tau)$ are the same as those defined below Theorem \ref{thm:qr},
\begin{align*}
\zeta_A^{\prime 2}(\pi,\tau)= & \mathbb{E}\gamma(S)\biggl[(m_1(S,\tau)-m_0(S,\tau))\left(\frac{1-\pi}{\pi f_1(q_1(\tau))} - \frac{\pi}{(1-\pi) f_0(q_0(\tau))} \right) \\
& + q(\tau)\left(\frac{f_1(q_1(\tau)|S)}{f_1(q_1(\tau))}-\frac{f_0(q_0(\tau)|S)}{f_0(q_0(\tau))}\right) \biggr]^2.
\end{align*}

Three remarks are in order. First, if the treatment assignment rule achieves strong balance, then $\zeta_A^{\prime 2}(\pi,\tau) = 0$ and the asymptotic variances for  $\hat{\beta}_{1}(\tau)$ and $\hat{\beta}_{sfe,1}(\tau)$ are the same. Second, if the treatment assignment rule does not achieve strong balance, then it is difficult to compare the asymptotic variances of $\hat{\beta}_{1}(\tau)$ and $\hat{\beta}_{sfe,1}(\tau)$. Based on our simulation results in Section \ref{sec:addsim}, the SFE estimator usually has a smaller standard error. Third, in order to analytically compute the asymptotic variance $\hat{\beta}_{sfe,1}(\tau)$, one needs to nonparametrically estimate not only the unconditional densities $f_j(\cdot)$ but also the conditional densities $f_j(\cdot|s)$ for $j = 0,1$ and $s \in \mathcal{S}.$ However, such difficulty can be avoided by the covariate-adaptive bootstrap inference considered in Section \ref{sec:CABI}.

We can compute the weighted bootstrap counterpart of strata fixed effects estimator: 
\begin{align*}
\hat{\beta}_{sfe}^w(\tau) = \argmin_b \sum_{i=1}^n \xi_i \rho_\tau\left(Y_i - \dot{\tilde{A}}_i^{w'}b\right),
\end{align*}
where $\dot{\tilde{A}}^w_i = (1,\tilde{A}^w_i)'$, $\tilde{A}^w_i = A_i - \hat{\pi}^w(S_i)$, and $\hat{\pi}^w(\cdot)$ is defined in Section \ref{sec:SBI}. The second element of $\hat{\beta}_{sfe}^w(\tau)$ is our bootstrap estimator of the QTE. 
\begin{thm}
	\label{thm:sfeb}
	If Assumptions \ref{ass:assignment1}--\ref{ass:weight} hold and $p(s)>0$ for all $s \in \mathcal{S}$, then uniformly over $\tau \in \Upsilon$ and conditionally on data,
	\begin{align*}
	\sqrt{n}\left(\hat{\beta}_{sfe,1}^w(\tau) - \hat{\beta}_{sfe,1}(\tau)\right) \convD \tilde{\mathcal{B}}_{sfe}(\tau),~\text{as}~n\rightarrow \infty,
	\end{align*}
	where $\tilde{\mathcal{B}}_{sfe}(\tau)$ is a Gaussian process with covariance kernel being equal to that of $\mathcal{B}_{sfe}(\tau)$ defined in Theorem \ref{thm:sfe} with $\gamma(s)$ being replaced by $\pi(1-\pi)$.
\end{thm}

Similar to the SQR estimator, the weighted bootstrap fails to capture the cross-sectional dependence due to the covariate-adaptive randomization, and thus, overestimates the asymptotic variance of the SFE estimator.

We can also implement the covariate-adaptive bootstrap. Let 
\begin{align*}
\hat{\beta}_{sfe}^*(\tau) = \argmin_b \sum_{i=1}^n \rho_\tau\left(Y_i^* - \dot{\tilde{A}}_i^{*'}b\right),
\end{align*}
where $\dot{\tilde{A}}_i^* = (1,\tilde{A}_i^*)'$, $\tilde{A}_i^* = A_i^* - \hat{\pi}^*(S_i^*)$, $\hat{\pi}^*(s) = \frac{n_1^*(s)}{n^*(s)}$, and $(Y_i^*,A_i^*,S_i^*)_{i=1}^n$ is the covariate-adaptive bootstrap sample generated via the procedure mentioned in Section \ref{sec:CABI}. The the second element $\hat{\beta}_{sfe,1}^*(\tau)$ of $\hat{\beta}_{sfe}^*(\tau)$ is the covariate-adaptive SFE estimator. 

\begin{thm}
	\label{thm:cabsfe}
	If Assumptions \ref{ass:assignment1}, \ref{ass:tau}, and \ref{ass:bassignment} hold and $p(s)>0$ for all $s \in \mathcal{S}$, then, uniformly over $\tau \in \Upsilon$ and conditionally on data,   
	\begin{align*}
	\sqrt{n}\left(\hat{\beta}_{sfe,1}^*(\tau) - \hat{q}(\tau)\right) \convD \mathcal{B}_{sfe}(\tau),~\text{as}~n\rightarrow \infty.
	\end{align*}
\end{thm} 

Unlike the weighted bootstrap, the covariate-adaptive bootstrap can mimic the cross-sectional dependence, and thus, produces an asymptotically valid standard error for the SFE estimator.  
\section{Proof of Theorem \ref{thm:sfe}}
\label{sec:proofsfe}
Define $\tilde{\beta}_1(\tau) = q(\tau)$, $\tilde{\beta}_0(\tau) = \pi q_1(\tau) + (1-\pi)q_0(\tau)$, $\tilde{\beta}(\tau) = (\tilde{\beta}_0(\tau),\tilde{\beta}_1(\tau))^\prime$, and $\breve{A}_i = (1,A_i - \pi)^\prime$. For arbitrary $b_0$ and $b_1$, let $u_0 = \sqrt{n}(b_0-\tilde{\beta}_0(\tau))$, $u_1 = \sqrt{n}(b_1-\tilde{\beta}_1(\tau))$, $u=(u_0, u_1)' \in \Re^2$, and 
\begin{align*}
L_{sfe,n}(u,\tau)  = \sum_{i=1}^n\left[\rho_\tau(Y_i - \breve{A}'_i\tilde{\beta}(\tau) - (\dot{\tilde{A}}_i'b  - \breve{A}'_i\tilde{\beta}(\tau))) - \rho_\tau(Y_i - \breve{A}'_i\tilde{\beta}(\tau))\right].
\end{align*}  
Then, by the change of variable, we have that 
\begin{align*}
\sqrt{n}(\hat{\beta}_{sfe}(\tau) - \tilde{\beta}(\tau)) = \argmin_{u} L_{sfe,n}(u,\tau).
\end{align*}
Notice that $L_{sfe,n}(u,\tau)$ is convex in $u$ for each $\tau$ and bounded in $\tau$ for each $u$. In the following, we aim to show that there exists 
$$g_{sfe,n}(u,\tau) = - u'W_{sfe,n}(\tau) + \frac{1}{2}u'Q_{sfe}(\tau)u$$ 
such that (1) for each $u$, 
\begin{align*}
\sup_{\tau \in \Upsilon}|L_{sfe,n}(u,\tau) - g_{sfe,n}(u,\tau)-h_{sfe,n}(\tau)| \convP 0,
\end{align*}
where $h_{sfe,n}(\tau)$ does not depend on $u$; 
(2) the maximum eigenvalue of $Q_{sfe}(\tau)$ is bounded from above and the minimum eigenvalue of $Q_{sfe}(\tau)$ is bounded away from $0$ uniformly over $\tau \in \Upsilon$; (3) $W_{sfe,n}(\tau) \convD \tilde{\mathcal{B}}(\tau)$ uniformly over $\tau \in \Upsilon$ for some $\tilde{\mathcal{B}}(\tau)$.\footnote{We abuse the notation and denote the weak limit of $W_{sfe,n}(\tau)$ as $\tilde{\mathcal{B}}(\tau)$. This limit is different from the weak limit of $W_n(\tau)$ in the proof of Theorem \ref{thm:qr}.} Then by \citet[Theorem 2]{K09}, we have 
\begin{align*}
\sqrt{n}(\hat{\beta}_{sfe}(\tau) - \tilde{\beta}(\tau)) = [Q_{sfe}(\tau)]^{-1}W_{sfe,n}(\tau) + r_{sfe,n}(\tau),
\end{align*}
where $\sup_{\tau \in \Upsilon}||r_{sfe,n}(\tau)|| = o_p(1)$. In addition, by (3), we have, uniformly over $\tau \in \Upsilon$, 
\begin{align*}
\sqrt{n}(\hat{\beta}_{sfe}(\tau) - \tilde{\beta}(\tau)) \convD [Q_{sfe}(\tau)]^{-1}\tilde{\mathcal{B}}(\tau) \equiv \mathcal{B}(\tau).
\end{align*}
The second element of $\mathcal{B}(\tau)$ is $\mathcal{B}_{sfe}(\tau)$ stated in Theorem \ref{thm:sfe}. Next, we prove requirements (1)--(3) in three steps. 

\textbf{Step 1.} By Knight's identity (\citep{K98}), we have 
\begin{align*}
& L_{sfe,n}(u,\tau) \\
= & -\sum_{i=1}^{n}(\dot{\tilde{A}}_i'(\tilde{\beta}(\tau) + \frac{u}{\sqrt{n}})  - \breve{A}'_i\tilde{\beta}(\tau))\left(\tau- 1\{Y_i\leq \dot{\tilde{A}}_i'\tilde{\beta}(\tau)\}\right) \\
& + \sum_{i=1}^{n}\int_0^{\dot{\tilde{A}}_i'(\tilde{\beta}(\tau) + \frac{u}{\sqrt{n}})  - \breve{A}'_i\tilde{\beta}(\tau)}\left(1\{Y_i -  \dot{\tilde{A}}_i'\tilde{\beta}(\tau)\leq v\} - 1\{Y_i -  \dot{\tilde{A}}_i'\tilde{\beta}(\tau)\leq 0\} \right)dv\\
\equiv & -L_{1,n}(u,\tau) + L_{2,n}(u,\tau).
\end{align*}

\textbf{Step 1.1.} We first consider $L_{1,n}(u,\tau)$. Note that $\tilde{\beta}_1(\tau) = q(\tau)$ and 
\begin{align}
\label{eq:l1}
& L_{1,n}(u,\tau) \notag \\
= & \sum_{i=1}^{n} \sum_{s \in \mathcal{S}} A_i 1\{S_i = s\}\left(\frac{u_0}{\sqrt{n}} + (1-\hat{\pi}(s))\frac{u_1}{\sqrt{n}} + (\pi - \hat{\pi}(s))q(\tau) \right)\left(\tau- 1\{Y_i(1)\leq q_1(\tau)\}\right) \notag \\
& +  \sum_{i=1}^{n} \sum_{s \in \mathcal{S}} (1-A_i) 1\{S_i = s\}\left(\frac{u_0}{\sqrt{n}}  -\hat{\pi}(s)\frac{u_1}{\sqrt{n}} + (\pi - \hat{\pi}(s))q(\tau) \right)\left(\tau- 1\{Y_i(0)\leq q_0(\tau)\}\right)  \notag \\
\equiv &  L_{1,1,n}(u,\tau) + L_{1,0,n}(u,\tau). 
\end{align}
Let $\iota_1 = (1,1-\pi)'$ and $\iota_0 = (1,-\pi)'$. Note that $\hat{\pi}(s) - \pi = \frac{D_n(s)}{n(s)}$. Then, for $ L_{1,1,n}(u,\tau)$, we have 
\begin{align}
\label{eq:l11}
& L_{1,1,n}(u,\tau) \notag \\
= & \sum_{i=1}^{n} \sum_{s \in \mathcal{S}} A_i 1\{S_i = s\}\left[\frac{u'\iota_1}{\sqrt{n}} + (\pi - \hat{\pi}(s))\left(q(\tau)+\frac{u_1}{\sqrt{n}}\right) \right]\left(\tau- 1\{Y_i(1)\leq q_1(\tau)\}\right) \notag \\
= & \frac{u'\iota_1}{\sqrt{n}}\sum_{i=1}^{n} \sum_{s \in \mathcal{S}} A_i 1\{S_i = s\}\left(\tau- 1\{Y_i(1)\leq q_1(\tau)\}\right) \notag \\
&- \sum_{s \in \mathcal{S}} \frac{D_n(s)}{\sqrt{n}}\frac{u_1}{n(s)}\sum_{i =1}^nA_i1\{S_i = s\}\left(\tau- 1\{Y_i(1)\leq q_1(\tau)\}\right) \notag \\
& + \sum_{s \in \mathcal{S}} (\pi - \hat{\pi}(s))q(\tau)\sum_{i =1}^nA_i1\{S_i = s\}\left(\tau- 1\{Y_i(1)\leq q_1(\tau)\}\right) \notag \\
= & \sum_{s \in \mathcal{S}}\frac{u'\iota_1}{\sqrt{n}}\sum_{i=1}^{n} \biggl[ A_i 1\{S_i = s\}\eta_{i,1}(s,\tau) + (A_i-\pi) 1\{S_i = s\}m_{1}(s,\tau) + \pi 1\{S_i = s\}m_{1}(s,\tau)\biggr] \notag \\
& - 
\sum_{s \in \mathcal{S}}\frac{D_n(s)}{\sqrt{n}}\frac{u_1}{n(s)}\sum_{i =1}^n\biggl[A_i 1\{S_i=s\}\eta_{i,1}(s,\tau)+ (A_i - \pi)1\{S_i=s\}m_1(s,\tau) + \pi1\{S_i=s\}m_1(s,\tau)\biggr] + h_{1,1}(\tau) \notag \\
= & \sum_{s \in \mathcal{S}}\frac{u'\iota_1}{\sqrt{n}}\sum_{i=1}^{n} \biggl[ A_i 1\{S_i = s\}\eta_{i,1}(s,\tau) + (A_i-\pi) 1\{S_i = s\}m_{1}(s,\tau) + \pi 1\{S_i = s\}m_{1}(s,\tau)\biggr] \notag \\
& - \sum_{s \in \mathcal{S}} \frac{u_1D_n(s) \pi m_1(s,\tau)}{\sqrt{n}} + h_{1,1}(\tau) + R_{sfe,1,1}(u,\tau),
\end{align}
where 
\begin{align*}
h_{1,1}(\tau) =  \sum_{s \in \mathcal{S}} (\pi - \hat{\pi}(s))q(\tau)\sum_{i =1}^nA_i1\{S_i = s\}\left(\tau- 1\{Y_i(1)\leq q_1(\tau)\}\right)
\end{align*}
and 
\begin{align*}
R_{sfe,1,1}(u,\tau) = - \sum_{s \in \mathcal{S}}\frac{u_1D_n(s)}{\sqrt{n} n(s)}\sum_{i =1}^n\biggl[A_i 1\{S_i=s\}\eta_{i,1}(s,\tau)+ (A_i - \pi)1\{S_i=s\}m_1(s,\tau)\biggr].
\end{align*}
By the same argument in Lemma \ref{lem:Q} and Assumption \ref{ass:assignment1}.3, we have for every $s\in\mathcal{S}$,
\begin{align}
\label{eq:eta1_sfe}
\sup_{\tau \in \Upsilon}\left|\frac{1}{\sqrt{n}}\sum_{i =1}^n A_i 1\{S_i=s\}\eta_{i,1}(s,\tau)\right| = O_p(1)
\end{align}
and 
\begin{align*}
\sup_{\tau \in \Upsilon}\left|\frac{1}{\sqrt{n}}\sum_{i =1}^n\biggl[(A_i - \pi)1\{S_i=s\}m_1(s,\tau)\biggr]\right| = \sup_{\tau \in \Upsilon}\left|\frac{D_n(s)m_1(s,\tau)}{\sqrt{n}}\right| = O_p(1). 
\end{align*}
In addition, note that $n(s)/n \convP p(s)$. Therefore, 
\begin{align*}
\sup_{\tau \in \Upsilon}|R_{sfe,1,1}(u,\tau)| = O_p(\frac{1}{\sqrt{n}}) = o_p(1).
\end{align*}
Similarly, we have 
\begin{align}
\label{eq:l10}
& L_{1,0,n}(u,\tau) \notag  \\
= & \sum_{s \in \mathcal{S}}\frac{u'\iota_0}{\sqrt{n}}\sum_{i=1}^{n} \biggl[ (1-A_i) 1\{S_i = s\}\eta_{i,0}(s,\tau) - (A_i-\pi) 1\{S_i = s\}m_{0}(s,\tau) + (1-\pi) 1\{S_i = s\}m_{0}(s,\tau)\biggr] \notag \\
& - \sum_{s \in \mathcal{S}} \frac{u_1D_n(s) (1-\pi) m_0(s,\tau)}{\sqrt{n}} + h_{1,0}(\tau) + R_{sfe,1,0}(u,\tau),
\end{align}
where 
\begin{align*}
h_{1,0}(\tau) =  \sum_{s \in \mathcal{S}} (\pi - \hat{\pi}(s))q(\tau)\sum_{i =1}^n(1-A_i)1\{S_i = s\}\left(\tau- 1\{Y_i(0)\leq q_0(\tau)\}\right),
\end{align*}
\begin{align*}
R_{sfe,1,0}(u,\tau) = - \sum_{s \in \mathcal{S}}\frac{u_1D_n(s)}{\sqrt{n} n(s)}\sum_{i =1}^n\biggl[(1-A_i) 1\{S_i=s\}\eta_{i,0}(\tau)- (A_i - \pi)1\{S_i=s\}m_0(s,\tau)\biggr],
\end{align*}
and 
\begin{align*}
\sup_{\tau \in \Upsilon}|R_{sfe,1,0}(\tau)| = O_p(\frac{1}{\sqrt{n}}) = o_p(1).
\end{align*}

Combining \eqref{eq:l1}, \eqref{eq:l11}, \eqref{eq:l10} and letting $\iota_2 = (1,1-2\pi)'$, we have 
\begin{align}
\label{eq:l1end}
L_{1,n}(u,\tau) = & \frac{1}{\sqrt{n}}\sum_{s \in \mathcal{S}}\sum_{i =1}^n\biggl[u'\iota_1 A_i 1\{S_i=s\}\eta_{i,1}(s,\tau) + u'\iota_0(1-A_i)1\{S_i=s\}\eta_{i,0}(s,\tau) \biggr] \notag \\
& + \sum_{s \in \mathcal{S}}u'\iota_2 \frac{D_n(s)}{\sqrt{n}}\left(m_1(s,\tau) - m_0(s,\tau)\right) \notag \\
& + \frac{1}{\sqrt{n}}\sum_{i =1}^n\left(u'\iota_1 \pi m_1(S_i,\tau) + u'\iota_0 (1-\pi) m_0(S_i,\tau)\right) \notag \\
& + R_{sfe,1,1}(u,\tau)+R_{sfe,1,0}(u,\tau) + h_{1,1}(\tau) + h_{1,0}(\tau).
\end{align}

\textbf{Step 1.2.} Next, we consider $L_{2,n}(u,\tau)$. Denote $E_n(s) = \sqrt{n}(\hat{\pi}(s) -\pi)$. Then, 
\begin{align*}
\{E_n(s)\}_{s \in \mathcal{S}} = \left\{\frac{D_n(s)}{\sqrt{n}}\frac{n}{n(s)}\right\}_{s \in \mathcal{S}} \convD \N(0,\Sigma_D') = O_p(1), 
\end{align*}
where $\Sigma_D'=\diag(\gamma(s)/p(s):s\in\mathcal{S})$. In addition,  
\begin{align}
\label{eq:l2}
& L_{2,n}(u,\tau) \notag \\
= & \sum_{s \in \mathcal{S}}\sum_{i=1}^n A_i 1\{S_i = s\}\int_0^{\frac{u'\iota_1}{\sqrt{n}} - \frac{E_n(s)}{\sqrt{n}}\left(q(\tau)+\frac{u_1}{\sqrt{n}}\right)}\left(1\{Y_i(1) \leq q_1(\tau) + v\} - 1\{Y_i(1) \leq q_1(\tau)\}\right)dv \notag \\
& + \sum_{s \in \mathcal{S}}\sum_{i=1}^n (1-A_i) 1\{S_i = s\}\int_0^{\frac{u'\iota_0}{\sqrt{n}} - \frac{E_n(s)}{\sqrt{n}}\left(q(\tau)+\frac{u_1}{\sqrt{n}}\right)}\left(1\{Y_i(0) \leq q_0(\tau) + v\} - 1\{Y_i(0) \leq q_0(\tau)\}\right)dv \notag \\ 
\equiv &  L_{2,1,n}(u,\tau) + L_{2,0,n}(u,\tau). 
\end{align}
By the same argument in \eqref{eq:Qn1}, we have 
\begin{align}
\label{eq:L2sfe}
L_{2,1,n}(u,\tau) \stackrel{d}{=} & \sum_{s \in \mathcal{S}} \sum_{i=N(s) + 1}^{N(s)+n_1(s)}\int_0^{\frac{u'\iota_1}{\sqrt{n}} - \frac{E_n(s)}{\sqrt{n}}\left(q(\tau)+\frac{u_1}{\sqrt{n}}\right)}\left(1\{Y_i^s(1) \leq q_1(\tau) + v\} - 1\{Y_i^s(1) \leq q_1(\tau)\}\right)dv \notag \\
\equiv & \sum_{s \in \mathcal{S}}\left[\Gamma_n^s(N(s)+n_1(s),\tau,E_n(s)) - \Gamma_n^s(N(s),\tau,E_n(s))\right],
\end{align}
where 
\begin{align*}
\Gamma_n^s(k,\tau,e) = \sum_{i=1}^k \int_0^{\frac{u' \iota_1 - e(q(\tau)+\frac{u_1}{\sqrt{n}})}{\sqrt{n}}}\left(1\{Y_i^s(1) \leq q_1(\tau) + v\} - 1\{Y_i^s(1) \leq q_1(\tau)\} \right)dv.
\end{align*}
We want to show, for some any sufficiently large constant $M$, 
\begin{align}
\label{eq:Gamma_sfe}
\sup_{0 < t \leq 1,\tau \in \Upsilon, |e| \leq M}|\Gamma_n^s(\lfloor nt \rfloor,\tau,e) - \mathbb{E}\Gamma_n^s(\lfloor nt \rfloor,\tau,e)| = o_p(1). 
\end{align}
By the same argument in \eqref{eq:Gamma}, it suffices to show that 
\begin{align*}
\sup_{\tau \in \Upsilon, |e| \leq M}n||\mathbb{P}_n - \mathbb{P}||_\mathcal{F} = o_p(1),
\end{align*}
where 
\begin{align*}
\mathcal{F} = \left\{\int_0^{\frac{u' \iota_1 - e(q(\tau)+\frac{u_1}{\sqrt{n}})}{\sqrt{n}}}\left(1\{Y_i^s(1) \leq q_1(\tau) + v\} - 1\{Y_i^s(1) \leq q_1(\tau)\} \right)dv: \tau \in \Upsilon, |e| \leq M \right\}
\end{align*}
with an envelope $F=\frac{|u_0| + |u_1| + M\sup_{ \tau \in \Upsilon}|q(\tau)| + \frac{|u_1|}{\sqrt{n}}}{\sqrt{n}}$. Note that  
\begin{align*}
\sup_{f \in \mathcal{F}}\mathbb{E}f^2 \leq & \sup_{\tau \in \Upsilon}\mathbb{E} \left[\frac{|u_0| + |u_1| + M|q(\tau)| + \frac{|u_1|}{\sqrt{n}}}{\sqrt{n}}1\left\{|Y_i^s(1) - q_1(\tau)|\leq \frac{|u_0| + |u_1| + M|q(\tau)| + \frac{|u_1|}{\sqrt{n}}}{\sqrt{n}}\right\} \right]^2 \\
\lesssim & n^{-3/2}, 
\end{align*}
and $\mathcal{F}$ is a VC-class with a fixed VC index. Then, by \citet[Corollary 5.1]{CCK14}, 
\begin{align}
\label{eq:gammasfe}
\mathbb{E}\sup_{\tau \in \Upsilon, |e| \leq M}|\Gamma^s_{n}(n,\tau,e) - \mathbb{E}\Gamma^s_{n}(n,\tau,e)| = n||\mathbb{P}_n  - \mathbb{P}||_\mathcal{F} \lesssim n\left[\sqrt{\frac{\log(n)}{n^{5/2}}} + \frac{\log(n)}{n^{3/2}}\right] = o(1). 
\end{align} 
In addition, we have 
\begin{align}
\label{eq:Gamma_sfe1}
\mathbb{E}\Gamma_n^s(\lfloor nt \rfloor,\tau,e) = & \lfloor nt \rfloor \int_0^{\frac{u' \iota_1 - e(q(\tau)+\frac{u_1}{\sqrt{n}})}{\sqrt{n}}} [F_1(q_1(\tau)+v|s) - F_1(q_1(\tau)|s)]dv \notag \\
= & t \frac{f_1(q_1(\tau)|s)}{2}(u' \iota_1 - eq(\tau))^2 + o(1),
\end{align}
where $F_j(\cdot|s)$ and $f_j(\cdot|s)$, $j=0,1$ are the conditional CDF and PDF for $Y(j)$ given $S=s$, respectively, and the $o(1)$ term holds uniformly over $\{\tau \in \Upsilon, |e| \leq M\}$. Combining \eqref{eq:Gamma_sfe} and \eqref{eq:Gamma_sfe1} with the fact that $\frac{n_1(s)}{n} \convP \pi p(s)$, we have
\begin{align}
\label{eq:l21}
L_{2,1,n}(u,\tau) = & \sum_{s \in \mathcal{S}}\pi p(s)\frac{f_1(q_1(\tau)|s)}{2}(u' \iota_1 - E_n(s)q(\tau))^2 + R'_{sfe,2,1}(u,\tau) \notag \\
= & \frac{\pi f_1(q_1(\tau))}{2}(u' \iota_1)^2 - \sum_{s \in \mathcal{S}}f_1(q_1(\tau)|s)\frac{\pi D_n(s)u' \iota_{1}}{\sqrt{n}}q(\tau) + h_{2,1}(\tau) + R_{sfe,2,1}(u,\tau), 
\end{align}
where 
\begin{align*}
\sup_{\tau \in \Upsilon}|R'_{sfe,2,1}(u,\tau)| = o_p(1), \quad \sup_{\tau \in \Upsilon}|R_{sfe,2,1}(u,\tau)| = o_p(1),
\end{align*}
and
\begin{align*}
h_{2,1}(\tau) =  \sum_{s \in \mathcal{S}}\frac{\pi f_1(q_1(\tau)|s)}{2}p(s)E^2_n(s)\tilde{\beta}^2_1(\tau).
\end{align*}

Similarly, we have 
\begin{align}
\label{eq:l20}
L_{2,0,n}(u,\tau) = & \frac{(1-\pi)f_0(q_0(\tau))}{2}(u' \iota_0)^2 - \sum_{s \in \mathcal{S}}(1-\pi)f_0(q_0(\tau)|s)\frac{D_n(s)u'\iota_0}{\sqrt{n}}q(\tau) \notag \\
&  + h_{2,0}(\tau) + R_{sfe,2,0}(u,\tau), 
\end{align}
where 
\begin{align*}
\sup_{\tau \in \Upsilon}|R_{sfe,2,0}(u,\tau)| = o_p(1) \quad \text{and} \quad h_{2,0}(\tau) =  \sum_{s \in \mathcal{S}}\frac{(1-\pi)f_0(q_0(\tau)|s)}{2}p(s)E^2_n(s)\tilde{\beta}^2_1(\tau).
\end{align*}

Combining \eqref{eq:l2}, \eqref{eq:l21}, and \eqref{eq:l20}, we have 
\begin{align}
\label{eq:l2end}
L_{2,n}(u,\tau) = & \frac{1}{2}u'Q_{sfe}(\tau)u - \sum_{s \in \mathcal{S}} q(\tau) \left[f_1(q_1(\tau)|s)\pi u'\iota_1 +f_0(q_0(\tau)|s)(1-\pi) u'\iota_0 \right]\frac{D_n(s)}{\sqrt{n}} \notag \\
& + R_{sfe,2,1}(u,\tau)+R_{sfe,2,0}(u,\tau) + h_{2,1}(\tau) + h_{2,0}(\tau).
\end{align}
where 
\begin{align*}
Q_{sfe} = & \pi f_1(q_1(\tau)) \iota_1 \iota_1' + (1-\pi)f_0(q_0(\tau))\iota_0 \iota_0' \\
= & \begin{pmatrix}
\pi f_1(q_1(\tau)) + (1-\pi) f_0(q_0(\tau)) & \pi(1-\pi)(f_1(q_1(\tau))  - f_0(q_0(\tau)) ) \\
\pi(1-\pi)(f_1(q_1(\tau))  - f_0(q_0(\tau)) ) & \pi(1-\pi)((1-\pi)f_1(q_1(\tau)) + \pi f_0(q_0(\tau)))
\end{pmatrix}.
\end{align*}

\textbf{Step 1.3.} Last, by combining \eqref{eq:l1end} and \eqref{eq:l2end}, we have 
\begin{align*}
L_{sfe,n}(u,\tau) = - u'W_{sfe,n}(\tau) + \frac{1}{2}u'Q_{sfe}(\tau)u + R_{sfe}(u,\tau) + h_{sfe,n}(\tau),
\end{align*}
where 
\begin{align}
\label{eq:wsfe}
& W_{sfe,n}(\tau) \notag \\
= &\frac{1}{\sqrt{n}}\sum_{s \in \mathcal{S}}\sum_{i =1}^n\biggl[\iota_1 A_i 1\{S_i=s\}\eta_{i,1}(s,\tau) + \iota_0(1-A_i)1\{S_i=s\}\eta_{i,0}(s,\tau) \biggr] \notag \\
& + \sum_{s \in \mathcal{S}}\biggl\{\iota_2 \left(m_1(s,\tau) - m_0(s,\tau)\right) + q(\tau)\biggl[f_1(q_1(\tau)|s)\pi \iota_1 + f_0(q_0(\tau)|s)(1-\pi) \iota_0\biggr]\biggr\}\frac{D_n(s)}{\sqrt{n}}\notag \\
& + \frac{1}{\sqrt{n}}\sum_{i =1}^n\left(\iota_1 \pi m_1(S_i,\tau) + \iota_0 (1-\pi) m_0(S_i,\tau)\right) \notag \\
\equiv & W_{sfe,n,1}(\tau) + W_{sfe,n,2}(\tau)+W_{sfe,n,3}(\tau),
\end{align}
\begin{align*}
R_{sfe}(u,\tau) = R_{sfe,1,1}(u,\tau) + R_{sfe,1,0}(u,\tau) + R_{sfe,2,1}(u,\tau) + R_{sfe,2,0}(u,\tau)
\end{align*}
such that $\sup_{\tau \in \Upsilon}|R_{sfe}(u,\tau)| = o_p(1)$, and 
\begin{align*}
h_{sfe,n}(\tau) = h_{1,1}(\tau)+h_{1,0}(\tau)+h_{2,1}(\tau)+h_{2,0}(\tau).
\end{align*}
This concludes the proof of Step 1.

\textbf{Step 2}. Note that $\text{det}(Q_{sfe}(\tau)) = \pi(1-\pi) f_0(q_0(\tau))f_1(q_1(\tau))$, which is bounded and bounded away from zero. In addition, it can be shown that the two eigenvalues of $Q_{sfe}(\tau)$ are nonnegative. This leads to the desired result. 

\textbf{Step 3}. 
Lemma \ref{lem:Qsfe} establishes the weak convergence that 
\begin{align*}
(W_{sfe,1,n}(\tau), W_{sfe,2,n}(\tau), W_{sfe,3,n}(\tau)) \convD (\mathcal{B}_{sfe,1}(\tau),\mathcal{B}_{sfe,2}(\tau),\mathcal{B}_{sfe,3}(\tau)), 
\end{align*}
where $(\mathcal{B}_{sfe,1}(\tau),\mathcal{B}_{sfe,2}(\tau),\mathcal{B}_{sfe,3}(\tau))$ are three independent two-dimensional Gaussian processes with covariance kernels $\Sigma_{1}(\tau_1,\tau_2)$, $\Sigma_{2}(\tau_1,\tau_2)$, and $\Sigma_{3}(\tau_1,\tau_2)$, respectively. Therefore, uniformly over $\tau \in \Upsilon$, 
\begin{align*}
W_{sfe,n}(\tau) \convD \tilde{\mathcal{B}}(\tau),
\end{align*}
where $\tilde{\mathcal{B}}(\tau)$ is a two-dimensional Gaussian process with covariance kernel 
$$\tilde{\Sigma}(\tau_1,\tau_2) = \sum_{j=1}^3\Sigma_{j}(\tau_1,\tau_2).$$
Consequently, 
\begin{align*}
\sqrt{n}(\hat{\beta}_{sfe}(\tau) - \tilde{\beta}(\tau)) \convD \mathcal{B}(\tau) \equiv Q_{sfe}^{-1}(\tau)\tilde{\mathcal{B}}(\tau),
\end{align*}
where $\Sigma(\tau_1,\tau_2)$, the covariance kernel of $\mathcal{B}(\tau)$, has the expression that 
\begin{align*}
& \Sigma(\tau_1,\tau_2) \\
= & Q_{sfe}^{-1}(\tau_1)\tilde{\Sigma}(\tau_1,\tau_2)Q_{sfe}^{-1}(\tau_2) \\
= & \biggl\{ \frac{1}{\pi f_1(q_1(\tau_1))f_1(q_1(\tau_2))}\left[\min(\tau_1,\tau_2) - \tau_1\tau_2 - \mathbb{E}m_1(S,\tau_1)m_1(S,\tau_2)\right]\begin{pmatrix}
\pi^2 & \pi \\
\pi & 1
\end{pmatrix} \\
& + \frac{1}{(1-\pi) f_0(q_0(\tau_1))f_0(q_0(\tau_2))}\left[\min(\tau_1,\tau_2) - \tau_1\tau_2 - \mathbb{E}m_0(S,\tau_1)m_0(S,\tau_2)\right]\begin{pmatrix}
(1-\pi)^2 & \pi-1 \\
\pi-1 & 1
\end{pmatrix} \biggr\}\\
& + \biggl\{\mathbb{E}\gamma(S)\biggl[(m_1(S,\tau_1) - m_0(S,\tau_1))\begin{pmatrix}
\frac{\pi}{f_0(q_0(\tau_1))} + \frac{1-\pi}{f_1(q_1(\tau_1))} \\
\frac{1-\pi}{\pi f_1(q_1(\tau_1))} - \frac{\pi}{(1-\pi)f_0(q_0(\tau_1))}
\end{pmatrix} 
+q(\tau_1)\frac{f_1(q_1(\tau_1)|S)}{f_1(q_1(\tau_1))}\begin{pmatrix}
\pi \\
1
\end{pmatrix} \\
& + q(\tau_1)\frac{f_0(q_0(\tau_1)|S)}{f_0(q_0(\tau_1))}\begin{pmatrix}
1-\pi \\
-1
\end{pmatrix}\biggr] \times \biggl[(m_1(S,\tau_2) - m_0(S,\tau_2))\begin{pmatrix}
\frac{\pi}{f_0(q_0(\tau_2))} + \frac{1-\pi}{f_1(q_1(\tau_2))} \\
\frac{1-\pi}{\pi f_1(q_1(\tau_2))} - \frac{\pi}{(1-\pi)f_0(q_0(\tau_2))}
\end{pmatrix}  \\
& +q(\tau_2)\frac{f_1(q_1(\tau_2)|S)}{f_1(q_1(\tau_2))}\begin{pmatrix}
\pi \\
1
\end{pmatrix}
+ q(\tau_2)\frac{f_0(q_0(\tau_2)|S)}{f_0(q_0(\tau_2))}\begin{pmatrix}
1-\pi \\
-1
\end{pmatrix}\biggr] \biggr\}\\
& + \biggl\{\mathbb{E}\biggl[\frac{m_1(S,\tau_1)}{f_1(q_1(\tau_1))}\begin{pmatrix}
\pi \\
1
\end{pmatrix} + \frac{m_0(S,\tau_1)}{f_0(q_0(\tau_1))}\begin{pmatrix}
1-\pi \\
-1
\end{pmatrix} \biggr] \biggl[\frac{m_1(S,\tau_2)}{f_1(q_1(\tau_2))}\begin{pmatrix}
\pi \\
1
\end{pmatrix} + \frac{m_0(S,\tau_2)}{f_0(q_0(\tau_2))}\begin{pmatrix}
1-\pi \\
-1
\end{pmatrix} \biggr]' \biggr\}.
\end{align*}
By checking the $(2,2)$-element of $\Sigma(\tau_1,\tau_2)$, we have 
\begin{align*}
& \Sigma_{sfe}(\tau_1,\tau_2) \\
= &\frac{\min(\tau_1,\tau_2) - \tau_1\tau_2 - \mathbb{E}m_1(S,\tau_1)m_1(S,\tau_2)}{\pi f_1(q_1(\tau_1))f_1(q_1(\tau_2))} + \frac{\min(\tau_1,\tau_2) - \tau_1\tau_2 - \mathbb{E}m_0(S,\tau_1)m_0(S,\tau_2)}{(1-\pi) f_0(q_0(\tau_1))f_0(q_0(\tau_2))} \\
& + \mathbb{E}\gamma(S)\biggl[(m_1(S,\tau_1) - m_0(S,\tau_1))\left(\frac{1-\pi}{\pi f_1(q_1(\tau_1))} - \frac{\pi}{(1-\pi)f_0(q_0(\tau_1))}\right) + q(\tau_1)\left(\frac{f_1(q(\tau_1)|S)}{f_1(q_1(\tau_1))} -\frac{f_0(q(\tau_1)|S)}{f_0(q_0(\tau_1))}\right)\biggr] \\
&\times \biggl[(m_1(S,\tau_2) - m_0(S,\tau_2))\left(\frac{1-\pi}{\pi f_1(q_1(\tau_2))} - \frac{\pi}{(1-\pi)f_0(q_0(\tau_2))}\right) + q(\tau_2)\left(\frac{f_1(q(\tau_2)|S)}{f_1(q_2(\tau_2))} -\frac{f_0(q(\tau_2)|S)}{f_0(q_0(\tau_2))}\right)\biggr] \\
& +\mathbb{E}\biggl[\frac{m_1(S,\tau_1)}{f_1(q_1(\tau_1))} - \frac{m_0(S,\tau_1)}{f_0(q_0(\tau_1))}\biggr] \biggl[\frac{m_1(S,\tau_2)}{f_1(q_1(\tau_2))} - \frac{m_0(S,\tau_2)}{f_0(q_0(\tau_2))}\biggr].
\end{align*}

\section{Proof of Theorem \ref{thm:sfeb}}
\label{sec:proofsfeb}
Note that 
\begin{align*}
\sqrt{n}(\hat{\beta}_{sfe}^w(\tau)-\tilde{\beta}(\tau)) = \argmin_u L_{sfe,n}^w(u,\tau), 
\end{align*} 
where 
\begin{align*}
L_{sfe,n}^w(u,\tau) = \sum_{i=1}^n \xi_{i}\biggl[\rho_\tau(Y_i - \dot{\tilde{A}}_i^{w \prime}(\tilde{\beta}(\tau) + \frac{u}{\sqrt{n}})) - \rho_\tau(Y_i - \breve{A}_i^{\prime}\tilde{\beta}(\tau)) \biggr], 
\end{align*}
$\dot{\tilde{A}}_i^w = (1,\tilde{A}_i^w)'$, $\tilde{A}_i^w = A_i - \hat{\pi}^w(S_i)$, and $$\hat{\pi}^w(s) = \frac{\sum_{i =1}^n \xi_i A_i1\{S_i = s\}}{\sum_{i =1}^n \xi_i 1\{S_i = s\}}.$$

Similar to the proof of Theorem \ref{thm:sfe}, we divide the proof into two steps. In the first step, we show that there exists 
$$g_{sfe,n}^w(u,\tau) = - u'W_{sfe,n}^w(\tau) + \frac{1}{2}u'Q_{sfe}(\tau)u$$
and $h_{sfe,n}^w(\tau)$ independent of $u$
such that for each $u$ 
\begin{align*}
\sup_{\tau \in \Upsilon}|L^w_{sfe,n}(u,\tau) - g_{sfe,n}^w(u,\tau) - h_{sfe,n}^w(\tau)|\convP 0.
\end{align*}
In addition, we will show that $\sup_{\tau \in \Upsilon}||W_{sfe,n}^w(\tau)|| = O_p(1)$. Then, by \citet[Theorem 2]{K09}, we have 
\begin{align*}
\sqrt{n}(\hat{\beta}^w_{sfe}(\tau)-\tilde{\beta}(\tau)) = [Q_{sfe}(\tau)]^{-1}W_{sfe,n}^w(\tau)  + R_{sfe,n}^w(\tau),
\end{align*}
where 
\begin{align*}
\sup_{\tau \in \Upsilon}||R_{sfe,n}^w(\tau)|| = o_p(1). 
\end{align*}
In the second step, we show that, conditionally on data,  
\begin{align*}
\sqrt{n}(\hat{\beta}^w_{sfe,1}(\tau)-\hat{\beta}_{sfe,1}(\tau)) \convD \tilde{\mathcal{B}}_{sfe}(\tau). 
\end{align*}

\textbf{Step 1.} Following Step 1 in the proof of Theorem \ref{thm:sfe}, we have 
\begin{align*}
L^w_{sfe,n}(u,\tau) \equiv - L_{1,n}^w(u,\tau) + L_{2,n}^w(u,\tau),
\end{align*}
where 
\begin{align*}
& L_{1,n}^w(u,\tau) \\
= & \sum_{i =1}^n \sum_{s \in \mathcal{S}}\xi_iA_i 1\{S_i=s\}\left(\frac{u_0}{\sqrt{n}} + (1 - \hat{\pi}^w(s))\frac{u_1}{\sqrt{n}} + (\pi - \hat{\pi}^w(s))q(\tau)\right)(\tau - 1\{Y_i \leq q_1(\tau)\}) \\
& + \sum_{i =1}^n \sum_{s \in \mathcal{S}}\xi_i(1-A_i)1\{S_i=s\}\left(\frac{u_0}{\sqrt{n}}-\hat{\pi}^w(s)\frac{u_1}{\sqrt{n}} + (\pi - \hat{\pi}^w(s))q(\tau)\right)(\tau - 1\{Y_i \leq q_0(\tau)\}) \\
\equiv & L_{1,1,n}^w(u,\tau) + L_{1,0,n}^w(u,\tau),
\end{align*}
\begin{align*}
& L_{2,n}^w(u,\tau) \\
= & \sum_{s \in \mathcal{S}} \sum_{i =1}^n \xi_iA_i 1\{S_i=s\}\int_0^{\frac{u'\iota_1}{\sqrt{n}} - \frac{E_n^w(s)}{\sqrt{n}}\left(q(\tau)+\frac{u_1}{\sqrt{n}}\right) }\left(1\{Y_i \leq q_1(\tau)+v \} - 1\{Y_i \leq q_1(\tau)\}\right)dv \\
& + \sum_{s \in \mathcal{S}} \sum_{i =1}^n \xi_i(1-A_i) 1\{S_i=s\}\int_0^{\frac{u'\iota_0}{\sqrt{n}} - \frac{E_n^w(s)}{\sqrt{n}}\left(q(\tau)+\frac{u_1}{\sqrt{n}}\right) }\left(1\{Y_i \leq q_0(\tau)+v \} - 1\{Y_i \leq q_0(\tau)\}\right)dv \\
\equiv & L_{2,1,n}^w(u,\tau) + L_{2,0,n}^w(u,\tau), 
\end{align*}
and $E_n^w(s) = \sqrt{n}(\hat{\pi}^w(s) - \pi)$.   

\textbf{Step 1.1.} Recall that $\iota_1 = (1,1-\pi)'$ and $\iota_0=(1,-\pi)'$. In addition, denote $\hat{\pi}^w(s) - \pi = \frac{D_n^w(s)}{n^w(s)}$, where 
\begin{align*}
D_n^w(s) = \sum_{i =1}^n \xi_i(A_i-\pi)1\{S_i = s\} \quad \text{and} \quad n^w(s) = \sum_{i =1}^n \xi_i1\{S_i = s\}.
\end{align*}
Then, we have 
\begin{align}
\label{eq:l11nc}
& L_{1,1,n}^w(u,\tau) \notag \\
= & \sum_{s \in \mathcal{S}} \frac{u'\iota_1}{\sqrt{n}} \sum_{i=1}^n \xi_i \left[A_i 1\{S_i =s\}\eta_{i,1}(s,\tau)+ \pi1\{S_i=s\}m_1(s,\tau)\right] + \sum_{s \in \mathcal{S}}\frac{u'\iota_2 D_n^w(s)m_1(s,\tau)}{\sqrt{n}} \notag \\
&+h_{1,1}^w(\tau) + R_{sfe,1,1}^w(u,\tau),
\end{align}
where $\eta_{i,1}(s,\tau) = (\tau - 1\{Y_i(1) \leq q_1(\tau)\}) - m_1(s,\tau)$, 
\begin{align*}
h_{1,1}^w(\tau) = \sum_{s\in \mathcal{S}}(\pi - \hat{\pi}^w(s))q(\tau)\left(\sum_{i =1}^n\xi_i A_i 1\{S_i=s\}(\tau - 1\{Y_i \leq q_1(\tau)\})\right),
\end{align*}
and 
\begin{align}
\label{eq:Rsfe11c}
R_{sfe,1,1}^w(u,\tau)= - \sum_{s \in \mathcal{S}}\frac{u_1D_n^w(s)}{\sqrt{n}n^w(s)}\left\{\sum_{i =1}^n \xi_i\left[A_i 1\{S_i=s\}\eta_{i,1}(s,\tau)+(A_i-\pi)1\{S_i=s\}m_1(s,\tau)\right] \right\}.
\end{align}
By Lemma \ref{lem:Rsfe11c}, we have 
\begin{align*}
\sup_{\tau \in \Upsilon}|R_{sfe,1,1}^w(u,\tau)| = o_p(1). 
\end{align*} 

Similarly, we have 
\begin{align}
\label{eq:l10nc}
& L_{1,0,n}^w(u,\tau) \notag \\
= & \sum_{s \in \mathcal{S}}  \sum_{i=1}^n \xi_i \biggl\{\frac{u'\iota_0}{\sqrt{n}}\left[(1-A_i) 1\{S_i =s\}\eta_{i,0}(s,\tau)+ \pi1\{S_i=s\}m_1(s,\tau)\right]  - \frac{u'\iota_2}{\sqrt{n}}(A_i - \pi)1\{S_i = s\}m_0(s,\tau)\biggr\} \notag \\
&+h_{1,0}^w(\tau) + R_{sfe,1,0}^w(u,\tau),
\end{align}
where 
\begin{align*}
\sup_{\tau \in \Upsilon}|R_{sfe,1,0}^w(u,\tau)| = o_p(1). 
\end{align*}
Combining \eqref{eq:l11nc} and \eqref{eq:l10nc}, we have 
\begin{align*}
& L^w_{1,n}(u,\tau) \\
= & \frac{1}{\sqrt{n}}\sum_{s \in \mathcal{S}}\sum_{i =1}^n \xi_i\biggl[u'\iota_1A_i1\{S_i = s\}\eta_{i,1}(u,\tau) + u'\iota_0(1-A_i)1\{S_i = s\}\eta_{i,0}(u,\tau) \\
& + u'\iota_2 (A_i - \pi)1\{S_i = s\}(m_1(s,\tau) - m_0(s,\tau)) + 1\{S_i = s\}(u'\iota_1 \pi m_1(s,\tau) + u'\iota_0(1-\pi)m_0(s,\tau)) \biggr] \\
& +R_{sfe,1,1}^w(u,\tau)+R_{sfe,1,0}^w(u,\tau) + h_{1,1}^w(\tau)  +h_{1,0}^w(\tau).  
\end{align*}

Furthermore, by Lemma \ref{lem:L2c}, we have 
\begin{align}
\label{eq:l21nc}
L^w_{2,1,n}(u,\tau) = \frac{\pi f_1(q_1(\tau))}{2}(u'\iota_1)^2 - \sum_{s \in \mathcal{S}}f_1(q_1(\tau)|s)\frac{\pi D_n^w(s) u'\iota_1}{\sqrt{n}}q(\tau) + h_{2,1}^w(\tau) + R_{sfe,2,1}^w(u,\tau)
\end{align}
and 
\begin{align}
\label{eq:l20nc}
L^w_{2,0,n}(u,\tau) = \frac{(1-\pi) f_0(q_0(\tau))}{2}(u'\iota_0)^2 - \sum_{s \in \mathcal{S}}f_0(q_0(\tau)|s)\frac{(1-\pi) D_n^w(s) u'\iota_0}{\sqrt{n}}q(\tau) + h_{2,0}^w(\tau) + R_{sfe,2,0}^w(u,\tau),
\end{align}
where 
\begin{align*}
h_{2,1}^w(\tau)  = \sum_{s \in \mathcal{S}} \frac{\pi f_1(q_1(\tau)|s)}{2}p(s)(E_n^w(s))^2q^2(\tau),
\end{align*}
\begin{align*}
h_{2,0}^w(\tau)  = \sum_{s \in \mathcal{S}} \frac{(1-\pi) f_0(q_0(\tau)|s)}{2}p(s)(E_n^w(s))^2q^2(\tau),
\end{align*}
\begin{align*}
\sup_{\tau \in \Upsilon}|R_{sfe,2,1}^w(u,\tau)| = o_p(1),
\end{align*}
and
\begin{align*}
\sup_{\tau \in \Upsilon}|R_{sfe,2,0}^w(u,\tau)| = o_p(1).
\end{align*}
Therefore, 
\begin{align*}
L_{2,n}^w(u,\tau) = & \frac{1}{2}u'Q_{sfe}(\tau)u - \sum_{s \in \mathcal{S}}q(\tau)\left[f_1(q_1(\tau)|s)\pi u'\iota_1 + f_0(q_0(\tau)|s)(1-\pi) u'\iota_0 \right]\frac{D_n^w(s)}{\sqrt{n}} \\
& + R_{sfe,2,1}^w(u,\tau) + R_{sfe,2,0}^w(u,\tau) + h_{2,1}^w(\tau)+h_{2,0}^w(\tau). 
\end{align*}
Combining \eqref{eq:l11nc}, \eqref{eq:l10nc},  \eqref{eq:l21nc}, and \eqref{eq:l20nc}, we have 
\begin{align*}
L_{sfe,n}^w(u,\tau) = - u'\tilde{W}_{sfe,n}^w(\tau) + \frac{1}{2}u'Q_{sfe}u + \tilde{R}_{sfe,n}^w(u,\tau) + h_{sfe,n}^w(\tau), 
\end{align*}
where 
\begin{align*}
& W_{sfe,n}^w(\tau) \\
= & \frac{1}{\sqrt{n}}\sum_{s \in \mathcal{S}} \sum_{i =1}^n \xi_i \biggl[\iota_1A_i 1\{S_i=s\}\eta_{i,1}(s,\tau) + \iota_0(1-A_i)1\{S_i=s\}\eta_{i,0}(s,\tau)\biggr] \\
& + \frac{1}{\sqrt{n}}\sum_{s \in \mathcal{S}}\sum_{i =1}^n \xi_i\biggl\{\iota_2(m_1(s,\tau)- m_0(s,\tau)) + q(\tau)\biggl[f_1(q_1(\tau)|s)\pi \iota_1 + f_0(q_0(\tau)|s)(1-\pi)\iota_0 \biggr]  \biggr\}\\
&\times(A_i - \pi)1\{S_i = s\} + \frac{1}{\sqrt{n}}\sum_{i=1}^n\xi_i(\iota_1 \pi m_1(S_i,\tau) + \iota_0(1-\pi)m_0(S_i,\tau)), 
\end{align*}
\begin{align*}
h_{sfe,n}^w(\tau) = h_{1,1}^w(\tau)+h_{1,0}^w(\tau)+h_{2,1}^w(\tau)+h_{2,0}^w(\tau),
\end{align*}
and
$$\sup_{\tau \in \Upsilon}|\tilde{R}_{sfe,n}^w(u,\tau)| = o_p(1).$$

In addition, by Lemma \ref{lem:wboot}, $\sup_{\tau \in \Upsilon}|W_{sfe,n}^w(\tau)| = O_p(1)$. Then, by \citet[Theorem 2]{K09}, we have 
\begin{align*}
\sqrt{n}(\hat{\beta}^w_{sfe}(\tau)-\tilde{\beta}(\tau)) = [Q_{sfe}(\tau)]^{-1}W_{sfe,n}^w(\tau)  + R_{sfe,n}^w(\tau),
\end{align*}
where 
\begin{align*}
\sup_{\tau \in \Upsilon}||R_{sfe,n}^w(\tau)|| = o_p(1). 
\end{align*}
This concludes Step 1.

\textbf{Step 2.} We now focus on the second element of $\hat{\beta}^w_{sfe}(\tau)$. From Step 1, we know that  
\begin{align*}
& \sqrt{n}(\hat{\beta}^w_{sfe,1}(\tau) - q(\tau)) = \frac{1}{\sqrt{n}}\sum_{s \in \mathcal{S}} \sum_{i =1}^n \xi_i \mathcal{J}_i(s,\tau) + 
R_{sfe,n,1}^w(\tau), 
\end{align*}
where 
\begin{align*}
\mathcal{J}_i(s,\tau) = & \left[\frac{A_i 1\{S_i = s\}\eta_{i,1}(s,\tau)}{\pi f_1(q_1(\tau))} - \frac{(1-A_i)1\{S_i = s\}\eta_{i,0}(s,\tau)}{(1-\pi) f_0(q_0(\tau))}\right] \\
& + \biggl\{\left(\frac{1-\pi}{\pi f_1(q_1(\tau))} - \frac{\pi}{(1-\pi)f_0(q_0(\tau))}\right)(m_1(s,\tau)- m_0(s,\tau)) \\
& + q(\tau)\biggl[\frac{f_1(q_1(\tau)|s)}{f_1(q_1(\tau))} - \frac{f_0(q_0(\tau)|s)}{f_0(q_0(\tau))} \biggr]  \biggr\}(A_i - \pi)1\{S_i = s\} \\
& + \left(\frac{m_1(s,\tau)}{f_1(q_1(\tau))} - \frac{m_0(s,\tau)}{f_0(q_0(\tau))}\right)1\{S_i = s\}
\end{align*}
and
\begin{align*}
\sup_{\tau \in \Upsilon}|R_{sfe,n,1}^w(\tau)| = o_p(1).
\end{align*}
By \eqref{eq:wsfe}, we have  
\begin{align*}
& \sqrt{n}(\hat{\beta}_{sfe,1}(\tau) - q(\tau)) = \frac{1}{\sqrt{n}}\sum_{s \in \mathcal{S}} \sum_{i =1}^n \mathcal{J}_i(s,\tau) + R_{sfe,n,1}(\tau), 
\end{align*}
where 
\begin{align*}
\sup_{\tau \in \Upsilon}|R_{sfe,n,1}(\tau)|= o_p(1).
\end{align*}

Taking the difference of the above two equations, we have 
\begin{align*}
\sqrt{n}(\hat{\beta}^w_{sfe,1}(\tau) - \hat{\beta}_{sfe,1}(\tau)) =  \frac{1}{\sqrt{n}}\sum_{s \in \mathcal{S}} \sum_{i =1}^n (\xi_i-1) \mathcal{J}_i(s,\tau) + R^w(\tau),
\end{align*}
where 
\begin{align*}
\sup_{\tau \in \Upsilon}|R^w(\tau)| = o_p(1).
\end{align*}
Lemma \ref{lem:Qsfec} shows that, conditionally on data,  
\begin{align*}
\frac{1}{\sqrt{n}}\sum_{s \in \mathcal{S}} \sum_{i =1}^n (\xi_i-1) \mathcal{J}_i(s,\tau) \convD \tilde{\mathcal{B}}_{sfe}(\tau),
\end{align*}
where $\tilde{\mathcal{B}}_{sfe}(\tau)$ is a Gaussian process with covariance kernel

\begin{align}
\label{eq:sigma_sfe_tilde}
& \tilde{\Sigma}_{sfe}(\tau_1,\tau_2) \notag\\
= &\frac{\min(\tau_1,\tau_2) - \tau_1\tau_2 - \mathbb{E}m_1(S,\tau_1)m_1(S,\tau_2)}{\pi f_1(q_1(\tau_1))f_1(q_1(\tau_2))} + \frac{\min(\tau_1,\tau_2) - \tau_1\tau_2 - \mathbb{E}m_0(S,\tau_1)m_0(S,\tau_2)}{(1-\pi) f_0(q_0(\tau_1))f_0(q_0(\tau_2))} \notag\\
& + \mathbb{E}\pi(1-\pi)\biggl[(m_1(S,\tau_1) - m_0(S,\tau_1))\left(\frac{1-\pi}{\pi f_1(q_1(\tau_1))} - \frac{\pi}{(1-\pi)f_0(q_0(\tau_1))}\right) \notag\\
& + q(\tau_1)\left(\frac{f_1(q(\tau_1)|S)}{f_1(q_1(\tau_1))} -\frac{f_0(q(\tau_1)|S)}{f_0(q_0(\tau_1))}\right)\biggr] \notag\\
&\times \biggl[(m_1(S,\tau_2) - m_0(S,\tau_2))\left(\frac{1-\pi}{\pi f_1(q_1(\tau_2))} - \frac{\pi}{(1-\pi)f_0(q_0(\tau_2))}\right) + q(\tau_2)\left(\frac{f_1(q(\tau_2)|S)}{f_1(q_2(\tau_2))} -\frac{f_0(q(\tau_2)|S)}{f_0(q_0(\tau_2))}\right)\biggr] \notag\\
& +\mathbb{E}\biggl[\frac{m_1(S,\tau_1)}{f_1(q_1(\tau_1))} - \frac{m_0(S,\tau_1)}{f_0(q_0(\tau_1))}\biggr] \biggl[\frac{m_1(S,\tau_2)}{f_1(q_1(\tau_2))} - \frac{m_0(S,\tau_2)}{f_0(q_0(\tau_2))}\biggr].
\end{align}
This concludes the proof for the SFE estimator.

\section{Proof of Theorem \ref{thm:cabsfe}}
\label{sec:proofcabsfe}
Recall the definition of $\tilde{\beta}(\tau) = (\tilde{\beta}_0(\tau),\tilde{\beta}_1(\tau))'$ in the proof of Theorem \ref{thm:sfe}. Let $u_0 = \sqrt{n}(b_0-\tilde{\beta}_0(\tau))$, $u_1 = \sqrt{n}(b_1-\tilde{\beta}_1(\tau))$ and $u=(u_0, u_1)' \in \Re^2$. Then, 
\begin{align*}
\sqrt{n}(\hat{\beta}^*_{sfe}(\tau)-\tilde{\beta}(\tau)) = \argmin_u L_{sfe,n}^*(u,\tau),
\end{align*}
where 
\begin{align*}
L_{sfe,n}^*(u,\tau) = \sum_{i =1}^n \left[ \rho_\tau(Y_i^* - \dot{\tilde{A}}_i^{* \prime}(\tilde{\beta}(\tau) + \frac{u}{\sqrt{n}})) - \rho_\tau(Y_i^* - \breve{A}_i^{* \prime}\tilde{\beta}(\tau))\right]
\end{align*}
and $\breve{A}_i^* = (1, A_i^* - \pi)'$. Following the proof of Theorem \ref{thm:sfe}, we divide the current proof into two steps. In the first step, we show that there exist 
$$g_{sfe,n}^*(u,\tau) = - u'W_{sfe,n}^*(\tau) + \frac{1}{2}u'Q_{sfe}(\tau)u$$
and $h_{sfe,n}^*(\tau)$ independent of $u$
such that for each $u$ 
\begin{align*}
\sup_{\tau \in \Upsilon}|L^*_{sfe,n}(u,\tau) - g_{sfe,n}^*(u,\tau) - h_{sfe,n}^*(\tau)|\convP 0.
\end{align*}
In addition, we show that $\sup_{\tau \in \Upsilon}||W_{sfe,n}^*(\tau)|| = O_p(1)$. Then, by \citet[Theorem 2]{K09}, we have 
\begin{align*}
\sqrt{n}(\hat{\beta}^*_{sfe}(\tau)-\tilde{\beta}(\tau)) = [Q_{sfe}(\tau)]^{-1}W_{sfe,n}^*(\tau)  + R_{sfe,n}^*(\tau),
\end{align*}
where 
\begin{align*}
\sup_{\tau \in \Upsilon}||R_{sfe,n}^*(\tau)|| = o_p(1). 
\end{align*}
In the second step, we show that, conditionally on data,  
\begin{align*}
\sqrt{n}(\hat{\beta}^*_{sfe,1}(\tau)-\hat{q}(\tau)) \convD \mathcal{B}_{sfe}(\tau). 
\end{align*}

\textbf{Step 1.} Following Step 1 in the proof of Theorem \ref{thm:sfe}, we have 
\begin{align*}
L^*_{sfe,n}(u,\tau) \equiv - L_{1,n}^*(u,\tau) + L_{2,n}^*(u,\tau),
\end{align*}
where 
\begin{align*}
& L_{1,n}^*(u,\tau) \\
= & \sum_{i =1}^n \sum_{s \in \mathcal{S}}A_i^*1\{S_i^*=s\}\left(\frac{u_0}{\sqrt{n}} + (1 - \hat{\pi}^*(s))\frac{u_1}{\sqrt{n}} + (\pi - \hat{\pi}^*(s))q(\tau)\right)(\tau - 1\{Y_i^* \leq q_1(\tau)\}) \\
& + \sum_{i =1}^n \sum_{s \in \mathcal{S}}(1-A_i^*)1\{S_i^*=s\}\left(\frac{u_0}{\sqrt{n}}-\hat{\pi}^*(s)\frac{u_1}{\sqrt{n}} + (\pi - \hat{\pi}^*(s))q(\tau)\right)(\tau - 1\{Y_i^* \leq q_0(\tau)\}) \\
\equiv & L_{1,1,n}^*(u,\tau) + L_{1,0,n}^*(u,\tau),
\end{align*}
\begin{align*}
& L_{2,n}^*(u,\tau) \\
= & \sum_{s \in \mathcal{S}} \sum_{i =1}^n A_i^* 1\{S_i^*=s\}\int_0^{\frac{u'\iota_1}{\sqrt{n}} - \frac{E_n^*(s)}{\sqrt{n}}\left(q(\tau)+\frac{u_1}{\sqrt{n}}\right) }\left(1\{Y_i^* \leq q_1(\tau)+v \} - 1\{Y_i^* \leq q_1(\tau)\}\right)dv \\
& + \sum_{s \in \mathcal{S}} \sum_{i =1}^n (1-A_i^*) 1\{S_i^*=s\}\int_0^{\frac{u'\iota_0}{\sqrt{n}} - \frac{E_n^*(s)}{\sqrt{n}}\left(q(\tau)+\frac{u_1}{\sqrt{n}}\right) }\left(1\{Y_i^* \leq q_0(\tau)+v \} - 1\{Y_i^* \leq q_0(\tau)\}\right)dv \\
\equiv & L_{2,1,n}^*(u,\tau) + L_{2,0,n}^*(u,\tau), 
\end{align*}
and $E_n^*(s) = \sqrt{n}(\hat{\pi}^*(s) - \pi)$.   

\textbf{Step 1.1.} Recall that $\iota_1 = (1,1-\pi)'$ and $\iota_0=(1,-\pi)'$. In addition, $\hat{\pi}^*(s) - \pi = \frac{D_n^*(s)}{n^*(s)}$. Then, 
\begin{align}
\label{eq:l11nstar}
& L_{1,1,n}^*(u,\tau) \notag \\
= & \sum_{s \in \mathcal{S}} \frac{u'\iota_1}{\sqrt{n}} \sum_{i=1}^n \left[A_i^*1\{S_i^*=s\}\eta_{i,1}^*(s,\tau)+ (A_i^* - \pi)1\{S_i^*=s\}m_1(s,\tau)+ \pi1\{S_i^*=s\}m_1(s,\tau)\right]  \notag \\
& - \sum_{s \in \mathcal{S}}\frac{u_1 D^*_n(s)\pi m_1(s,\tau)}{\sqrt{n}}+h_{1,1}^*(\tau) + R_{sfe,1,1}^*(u,\tau),
\end{align}
where $\eta_{i,1}^*(s,\tau) = (\tau - 1\{Y_i^*(1) \leq q_1(\tau)\}) - m_1(s,\tau)$, 
\begin{align*}
h_{1,1}^*(\tau) = \sum_{s\in \mathcal{S}}(\pi - \hat{\pi}^*(s))q(\tau)\left(\sum_{i =1}^n A_i^*1\{S_i^*=s\}(\tau - 1\{Y_i^* \leq q_1(\tau)\})\right),
\end{align*}
and 
\begin{align}
\label{eq:Rsfe11star}
R_{sfe,1,1}^*(u,\tau)= - \sum_{s \in \mathcal{S}}\frac{u_1D_n^*(s)}{\sqrt{n}n^*(s)}\left\{\sum_{i =1}^n A_i^*1\{S_i^*=s\}\eta_{i,1}^*(s,\tau)+(A_i^*-\pi)1\{S_i^*=s\}m_1(s,\tau) \right\}.
\end{align}
Note that 
\begin{align*}
\sup_{s \in \mathcal{S},\tau \in \Upsilon}|\sum_{i =1}^n(A_i^*-\pi)1\{S_i^*=s\}m_1(s,\tau)| = \sup_{s \in \mathcal{S},\tau \in \Upsilon}|D_n^*(s)m_1(s,\tau)| = O_p(\sqrt{n}). 
\end{align*}
In addition, Lemma \ref{lem:Rsfestar} shows   
\begin{align*}
\sup_{s \in \mathcal{S},\tau \in \Upsilon}|\sum_{i =1}^nA_i^*1\{S_i^*=s\}\eta_{i,1}^*(s,\tau)| = O_p(\sqrt{n(s)}).
\end{align*}
Therefore, we have 
\begin{align*}
& \sup_{ \tau \in \Upsilon}|R_{sfe,1,1}^*(u,\tau)| \\
\leq & \sum_{s \in \mathcal{S}}\sup_{s \in \mathcal{S} }\left|\frac{u_1 D_n^*(s)}{\sqrt{n}n^*(s)}\right|\biggl[\sup_{s \in \mathcal{S},\tau \in \Upsilon}|\sum_{i =1}^nA_i^*1\{S_i^*=s\}\eta_{i,1}^*(s,\tau)| +  \sup_{s \in \mathcal{S},\tau \in \Upsilon}|\sum_{i =1}^n(A_i^*-\pi)1\{S_i^*=s\}m_1(s,\tau)|\biggr] \\
= & O_p(1/\sqrt{n}). 
\end{align*}

Similarly, we have 
\begin{align}
\label{eq:l10nstar}
& L_{1,0,n}^*(u,\tau) \notag \\
= & \sum_{s \in \mathcal{S}} \frac{u'\iota_0}{\sqrt{n}} \sum_{i=1}^n \left[(1-A_i^*)1\{S_i^*=s\}\eta_{i,1}^*(s,\tau) - (A_i^* - \pi)1\{S_i^*=s\}m_0(s,\tau)+ (1-\pi)1\{S_i^*=s\}m_0(s,\tau)\right] \notag \\
& - \sum_{s \in \mathcal{S}}\frac{u_1 D^*_n(s)(1-\pi) m_0(s,\tau)}{\sqrt{n}}+h_{1,0}^*(\tau) + R_{sfe,1,0}^*(u,\tau),
\end{align}
where 
\begin{align*}
h_{1,0}^*(\tau) = \sum_{s\in \mathcal{S}}(\pi - \hat{\pi}^*(s))q(\tau)\left(\sum_{i =1}^n (1-A_i^*)1\{S_i^*=s\}(\tau - 1\{Y_i^* \leq q_0(\tau)\})\right),
\end{align*}
and 
\begin{align}
\label{eq:Rsfe10star}
R_{sfe,1,0}^*(u,\tau)= - \sum_{s \in \mathcal{S}}\frac{u_1D_n^*(s)}{\sqrt{n}n^*(s)}\left\{\sum_{i =1}^n (1-A_i^*)1\{S_i^*=s\}\eta_{i,0}^*(s,\tau) -(A_i^*-\pi)1\{S_i^*=s\}m_0(s,\tau) \right\}
\end{align}
such that 
\begin{align*}
\sup_{\tau \in \Upsilon}|R_{sfe,1,0}^*(u,\tau)| = O_p(1/\sqrt{n}). 
\end{align*}
Therefore, 
\begin{align*}
L_{1,n}^*(u,\tau) = & \frac{1}{\sqrt{n}}\sum_{s \in \mathcal{S}}\sum_{i =1}^n\left[u'\iota_1 A_i^*1\{S_i^*=s\}\eta_{i,1}^*(s,\tau) + u'\iota_0(1-A_i^*)1\{S_i^*=s\}\eta_{i,0}^*(s,\tau) \right] \\
& + \sum_{s \in \mathcal{S}} u'\iota_2 \frac{D_n^*(s)}{\sqrt{n}}\left(m_1(s,\tau) - m_0(s,\tau)\right) \\
& + \frac{1}{\sqrt{n}}\sum_{i=1}^n(u'\iota_1\pi m_1(S_i^*,\tau) + u'\iota_0(1-\pi)m_0(S_i^*,\tau)) \\
& + R_{sfe,1,1}^*(u,\tau) + R_{sfe,1,0}^*(u,\tau) + h_{1,1}(\tau) + h_{1,0}(\tau). 
\end{align*}

Furthermore, by Lemma \ref{lem:L2star}, we have 
\begin{align}
\label{eq:l21nstar}
L^*_{2,1,n}(u,\tau) = \frac{\pi f_1(q_1(\tau))}{2}(u'\iota_1)^2 - \sum_{s \in \mathcal{S}}f_1(q_1(\tau)|s)\frac{\pi D_n^*(s) u'\iota_1}{\sqrt{n}}q(\tau) + h_{2,1}^*(\tau) + R_{sfe,2,1}^*(u,\tau)
\end{align}
and 
\begin{align}
\label{eq:l20nstar}
L^*_{2,0,n}(u,\tau) = \frac{(1-\pi) f_0(q_0(\tau))}{2}(u'\iota_0)^2 - \sum_{s \in \mathcal{S}}f_0(q_0(\tau)|s)\frac{(1-\pi) D_n^*(s) u'\iota_0}{\sqrt{n}}q(\tau) + h_{2,0}^*(\tau) + R_{sfe,2,0}^*(u,\tau),
\end{align}
where 
\begin{align*}
h_{2,1}^*(\tau)  = \sum_{s \in \mathcal{S}} \frac{\pi f_1(q_1(\tau)|s)}{2}p(s)(E_n^*(s))^2q^2(\tau),
\end{align*}
\begin{align*}
h_{2,0}^*(\tau)  = \sum_{s \in \mathcal{S}} \frac{(1-\pi) f_0(q_0(\tau)|s)}{2}p(s)(E_n^*(s))^2q^2(\tau),
\end{align*}
\begin{align*}
\sup_{\tau \in \Upsilon}|R_{sfe,2,1}^*(u,\tau)| = o_p(1),
\end{align*}
and
\begin{align*}
\sup_{\tau \in \Upsilon}|R_{sfe,2,0}^*(u,\tau)| = o_p(1).
\end{align*}
Therefore, 
\begin{align*}
L_{2,n}^*(u,\tau) = & \frac{1}{2}u'Q_{sfe}(\tau)u - \sum_{s \in \mathcal{S}}q(\tau)\left[f_1(q_1(\tau)|s)\pi u'\iota_1 + f_0(q_0(\tau)|s)(1-\pi) u'\iota_0 \right]\frac{D_n^*(s)}{\sqrt{n}} \\
& + R_{sfe,2,1}^*(u,\tau) + R_{sfe,2,0}^*(u,\tau) + h_{2,1}^*(\tau)+h_{2,0}^*(\tau). 
\end{align*}
Combining \eqref{eq:l11nstar}, \eqref{eq:l10nstar},  \eqref{eq:l21nstar}, and \eqref{eq:l20nstar}, we have 
\begin{align*}
L_{sfe,n}^*(u,\tau) = - u'W_{sfe,n}^*(\tau) + \frac{1}{2}u'Q_{sfe}u + \tilde{R}_{sfe,n}^*(u,\tau) + h_{sfe,n}^*(\tau), 
\end{align*}
where 
\begin{align*}
& W_{sfe,n}^*(\tau) \\
= & \frac{1}{\sqrt{n}}\sum_{s \in \mathcal{S}} \sum_{i =1}^n \biggl[\iota_1A_i^*1\{S_i^*=s\}\eta_{i,1}^*(s,\tau) + \iota_0(1-A_i^*)1\{S_i^*=s\}\eta_{i,0}^*(s,\tau)\biggr] \\
& + \sum_{s \in \mathcal{S}}\biggl\{\iota_2(m_1(s,\tau)- m_0(s,\tau)) + q(\tau)\biggl[f_1(q_1(\tau)|s)\pi \iota_1 + f_0(q_0(\tau)|s)(1-\pi)\iota_0 \biggr]  \biggr\}\frac{D_n^*(s)}{\sqrt{n}} \\
& + \frac{1}{\sqrt{n}}\sum_{i=1}^n(\iota_1 \pi m_1(S_i^*,\tau) + \iota_0(1-\pi)m_0(S_i^*,\tau)), 
\end{align*}
\begin{align*}
h_{sfe,n}^*(\tau) = h_{1,1}^*(\tau)+h_{1,0}^*(\tau)+h_{2,1}^*(\tau)+h_{2,0}^*(\tau),
\end{align*}
and
$$\sup_{\tau \in \Upsilon}|\tilde{R}_{sfe,n}^*(u,\tau)| = o_p(1).$$

By Lemma \ref{lem:Qsfestar}, $\sup_{\tau \in \Upsilon}|W_{sfe,n}^*(\tau)| = O_p(1)$. Then, by \citet[Theorem 2]{K09}, we have 
\begin{align*}
\sqrt{n}(\hat{\beta}^*_{sfe}(\tau)-\tilde{\beta}(\tau)) = [Q_{sfe}(\tau)]^{-1}W_{sfe,n}^*(\tau)  + R_{sfe,n}^*(\tau),
\end{align*}
where 
\begin{align*}
\sup_{\tau \in \Upsilon}||R_{sfe,n}^*(\tau)|| = o_p(1). 
\end{align*}
This concludes Step 1. 

\textbf{Step 2.} We now focus on the second element of $\hat{\beta}^*_{sfe}(\tau)$. From Step 1, we know that  
\begin{align*}
& \sqrt{n}(\hat{\beta}^*_{sfe,1}(\tau) - q(\tau)) \\
= & \frac{1}{\sqrt{n}}\sum_{s \in \mathcal{S}} \sum_{i =1}^n \left[\frac{A_i^*1\{S_i^* = s\}\eta_{i,1}^*(s,\tau)}{\pi f_1(q_1(\tau))} - \frac{(1-A_i^*)1\{S_i^* = s\}\eta_{i,0}^*(s,\tau)}{(1-\pi) f_0(q_0(\tau))}\right] \\
& + \sum_{s \in \mathcal{S}}\biggl\{\left(\frac{1-\pi}{\pi f_1(q_1(\tau))} - \frac{\pi}{(1-\pi)f_0(q_0(\tau))}\right)(m_1(s,\tau)- m_0(s,\tau)) + q(\tau)\biggl[\frac{f_1(q_1(\tau)|s)}{f_1(q_1(\tau))} - \frac{f_0(q_0(\tau)|s)}{f_0(q_0(\tau))} \biggr]  \biggr\}\frac{D_n^*(s)}{\sqrt{n}} \\
& + \frac{1}{\sqrt{n}}\sum_{i =1}^n\left(\frac{m_1(S_i^*,\tau)}{f_1(q_1(\tau))} - \frac{m_0(S_i^*,\tau)}{f_0(q_0(\tau))}\right) + R_{sfe,n,1}^*(\tau) \\
\equiv & W_{sfe,n,1}^*(\tau) + W_{sfe,n,2}^*(\tau) +W_{sfe,n,3}^*(\tau)  + R_{sfe,n,1}^*(\tau), 
\end{align*}
where 
\begin{align*}
\sup_{\tau \in \Upsilon}|R_{sfe,n,1}^*(\tau)| = o_p(1).
\end{align*}

By \eqref{eq:wipw2}, we have  
\begin{align*}
& \sqrt{n}(\hat{q}(\tau) - q(\tau)) \\
= & \frac{1}{\sqrt{n}}\sum_{s \in \mathcal{S}} \sum_{i =1}^n \left[\frac{A_i 1\{S_i = s\}\eta_{i,1}(s,\tau)}{\pi f_1(q_1(\tau))} - \frac{(1-A_i)1\{S_i = s\}\eta_{i,0}(s,\tau)}{(1-\pi) f_0(q_0(\tau))}\right] \\
& + \frac{1}{\sqrt{n}}\sum_{i =1}^n\left(\frac{m_1(S_i,\tau)}{f_1(q_1(\tau))} - \frac{m_0(S_i,\tau)}{f_0(q_0(\tau))}\right) + R_{ipw,n}(\tau) \\
\equiv  & \mathcal{W}_{n,1}(\tau) + \mathcal{W}_{n,2}(\tau) + R_{ipw,n}(\tau),
\end{align*}
where 
\begin{align*}
\sup_{\tau \in \Upsilon}|R_{ipw,n}(\tau)|= o_p(1).
\end{align*}

Taking the difference of the above two equations, we have 
\begin{align}
\label{eq:betasfestar-qhat}
\sqrt{n}(\hat{\beta}^*_{sfe,1}(\tau) - \hat{q}(\tau)) = (W_{sfe,n,1}^*(\tau)-\mathcal{W}_{n,1}(\tau)) + W_{sfe,n,2}^*(\tau) +(W_{sfe,n,3}^*(\tau)-\mathcal{W}_{n,2}(\tau)) + R^*(\tau),
\end{align}
where 
\begin{align*}
\sup_{\tau \in \Upsilon}|R^*(\tau)| = o_p(1).
\end{align*}
Lemma \ref{lem:Qsfestar} shows that, conditionally on data,  
\begin{align*}
(W_{sfe,n,1}^*(\tau)-\mathcal{W}_{n,1}(\tau)), W_{sfe,n,2}^*(\tau), (W_{sfe,n,3}^*(\tau)-\mathcal{W}_{n,2}(\tau)) \convD (\mathcal{B}_1(\tau),\mathcal{B}_2(\tau),\mathcal{B}_3(\tau)),
\end{align*}
where $(\mathcal{B}_1(\tau),\mathcal{B}_2(\tau),\mathcal{B}_3(\tau))$ are three independent Gaussian processes and $\sum_{j=1}^3\mathcal{B}_j(\tau) \stackrel{d}{=} \mathcal{B}_{sfe}(\tau).$ This concludes the proof.

\section{Technical Lemmas}
\label{sec:lem2}
\begin{lem}
	\label{lem:Qsfe}
	Let $W_{sfe,n,j}(\tau)$, $j = 1,2,3$ be defined as in \eqref{eq:wsfe}. If Assumptions in Theorem \ref{thm:sfe} hold, then uniformly over $\tau \in \Upsilon$, 
	\begin{align*}
	(W_{sfe,n,1}(\tau), W_{sfe,n,2}(\tau), W_{sfe,n,3}(\tau)) \convD (\mathcal{B}_{sfe,1}(\tau),\mathcal{B}_{sfe,2}(\tau),\mathcal{B}_{sfe,3}(\tau)),
	\end{align*}
	where $(\mathcal{B}_{sfe,1}(\tau),\mathcal{B}_{sfe,2}(\tau),\mathcal{B}_{sfe,3}(\tau))$ are three independent two-dimensional Gaussian process with covariance kernels $\Sigma_{sfe,1}(\tau_1,\tau_2)$, $\Sigma_{sfe,2}(\tau_1,\tau_2)$, and $\Sigma_{sfe,3}(\tau_1,\tau_2)$, respectively. The expressions for the three kernels are derived in the proof below. 
\end{lem}
\begin{proof}
	The proofs of weak convergence and the independence among $(\mathcal{B}_{sfe,1}(\tau),\mathcal{B}_{sfe,2}(\tau),\mathcal{B}_{sfe,3}(\tau))$ are similar to that in Lemma \ref{lem:Q}, and thus, are omitted. In the following, we focus on deriving the covariance kernels. 
	
	First, similar to the argument in the proof of Lemma \ref{lem:Q},
	\begin{align*}
	W_{sfe,n,1}(\tau) \stackrel{d}{=} \iota_1 \sum_{s \in \mathcal{S}} \sum_{i=N(s) + 1}^{N(s)+n_1(s)}\frac{1}{\sqrt{n}}\tilde{\eta}_{i,1}(s,\tau) + \iota_0 \sum_{s \in \mathcal{S}} \sum_{i=N(s)+n_1(s) + 1}^{N(s)+n(s)}\frac{1}{\sqrt{n}}\tilde{\eta}_{i,0}(s,\tau).
	\end{align*}
	Therefore, 
	\begin{align*}
	\Sigma_{1}(\tau_1,\tau_2) = & \pi[\min(\tau_1,\tau_2) - \tau_1\tau_2 - \mathbb{E}m_1(S,\tau_1)m_1(S,\tau_2)]\iota_1\iota_1' \\
	& + (1-\pi)[\min(\tau_1,\tau_2) - \tau_1\tau_2 - \mathbb{E}m_0(S,\tau_1)m_0(S,\tau_2)]\iota_0\iota_0'.
	\end{align*}
	
	For $W_{sfe,n,2}(\tau)$, we have 
	\begin{align*}
	\Sigma_{2}(\tau_1,\tau_2) = &\mathbb{E}\gamma(S)\biggl[\iota_2(m_1(S,\tau_1) - m_0(S,\tau_1))+q(\tau_1)\biggl(f_1(q_1(\tau_1)|S)\pi \iota_1 + f_0(q_0(\tau_1)|S)(1-\pi) \iota_0\biggr) \biggr] \\
	& \times \biggl[\iota_2(m_1(S,\tau_2) - m_0(S,\tau_2))+q(\tau_2)\biggl(f_1(q_1(\tau_2)|S)\pi \iota_1 + f_0(q_0(\tau_2)|S)(1-\pi) \iota_0\biggr) \biggr]'.
	\end{align*}
	Next, we have 
	\begin{align*}
	\Sigma_{3}(\tau_1,\tau_2) = \mathbb{E}(\iota_1\pi m_1(S,\tau_1) + \iota_0(1-\pi)m_0(S,\tau_1))(\iota_1\pi m_1(S,\tau_2) + \iota_0(1-\pi)m_0(S,\tau_2))'.
	\end{align*}
	
	In addition, 
	\begin{align*}
	[Q_{sfe}(\tau)]^{-1} = \begin{pmatrix}
	\frac{1-\pi}{f_0(q_0(\tau))} + \frac{\pi}{f_1(q_1(\tau))} &  \frac{1}{f_1(q_1(\tau))}-\frac{1}{f_0(q_0(\tau))}\\ 
	\frac{1}{f_1(q_1(\tau))}-\frac{1}{f_0(q_0(\tau))} & \frac{1}{(1-\pi)f_0(q_0(\tau))} + \frac{1}{\pi f_1(q_1(\tau))} 
	\end{pmatrix}.
	\end{align*}
	
	Therefore, 
	\begin{align*}
	& \Sigma(\tau_1,\tau_2) \\
	= & \biggl\{ \frac{1}{\pi f_1(q_1(\tau_1))f_1(q_1(\tau_2))}\left[\min(\tau_1,\tau_2) - \tau_1\tau_2 - \mathbb{E}m_1(S,\tau_1)m_1(S,\tau_2)\right]\begin{pmatrix}
	\pi^2 & \pi \\
	\pi & 1
	\end{pmatrix} \\
	& + \frac{1}{(1-\pi) f_0(q_0(\tau_1))f_0(q_0(\tau_2))}\left[\min(\tau_1,\tau_2) - \tau_1\tau_2 - \mathbb{E}m_0(S,\tau_1)m_0(S,\tau_2)\right]\begin{pmatrix}
	(1-\pi)^2 & \pi-1 \\
	\pi-1 & 1
	\end{pmatrix} \biggr\}\\
	& + \biggl\{\mathbb{E}\gamma(S)\biggl[(m_1(S,\tau_1) - m_0(S,\tau_1))\begin{pmatrix}
	\frac{\pi}{f_0(q_0(\tau_1))} + \frac{1-\pi}{f_1(q_1(\tau_1))} \\
	\frac{1-\pi}{\pi f_1(q_1(\tau_1))} - \frac{\pi}{(1-\pi)f_0(q_0(\tau_1))}
	\end{pmatrix} 
	+q(\tau_1)\frac{f_1(q_1(\tau_1)|S)}{f_1(q_1(\tau_1))}\begin{pmatrix}
	\pi \\
	1
	\end{pmatrix} \\
	& + q(\tau_1)\frac{f_0(q_0(\tau_1)|S)}{f_0(q_0(\tau_1))}\begin{pmatrix}
	1-\pi \\
	-1
	\end{pmatrix}\biggr] \times \biggl[(m_1(S,\tau_2) - m_0(S,\tau_2))\begin{pmatrix}
	\frac{\pi}{f_0(q_0(\tau_2))} + \frac{1-\pi}{f_1(q_1(\tau_2))} \\
	\frac{1-\pi}{\pi f_1(q_1(\tau_2))} - \frac{\pi}{(1-\pi)f_0(q_0(\tau_2))}
	\end{pmatrix}  \\
	& +q(\tau_2)\frac{f_1(q_1(\tau_2)|S)}{f_1(q_1(\tau_2))}\begin{pmatrix}
	\pi \\
	1
	\end{pmatrix}
	+ q(\tau_2)\frac{f_0(q_0(\tau_2)|S)}{f_0(q_0(\tau_2))}\begin{pmatrix}
	1-\pi \\
	-1
	\end{pmatrix}\biggr] \biggr\}\\
	& + \biggl\{\mathbb{E}\biggl[\frac{m_1(S,\tau_1)}{f_1(q_1(\tau_1))}\begin{pmatrix}
	\pi \\
	1
	\end{pmatrix} + \frac{m_0(S,\tau_1)}{f_0(q_0(\tau_1))}\begin{pmatrix}
	1-\pi \\
	-1
	\end{pmatrix} \biggr] \biggl[\frac{m_1(S,\tau_2)}{f_1(q_1(\tau_2))}\begin{pmatrix}
	\pi \\
	1
	\end{pmatrix} + \frac{m_0(S,\tau_2)}{f_0(q_0(\tau_2))}\begin{pmatrix}
	1-\pi \\
	-1
	\end{pmatrix} \biggr]' \biggr\}.
	\end{align*}
	
\end{proof}

\begin{lem}
	\label{lem:Rsfe11c}
	Recall the definition of $R^w_{sfe,1,1}(u,\tau)$ in \eqref{eq:Rsfe11c}. If Assumptions \ref{ass:assignment1} and \ref{ass:tau} hold, then 
	\begin{align*}
	\sup_{\tau \in \Upsilon}|R^w_{sfe,1,1}(u,\tau)| = o_p(1).
	\end{align*}
\end{lem}
\begin{proof}
	We divide the proof into two steps. In the first step, we show that $\sup_{ s \in \mathcal{S}}|D_n^w(s)| = O_p(\sqrt{n})$. In the second step, we show that 
	\begin{align}
	\label{eq:Mn}
	\sup_{\tau \in \Upsilon,s \in \mathcal{S}}|\sum_{i =1}^n \xi_iA_i 1\{S_i = s\}\eta_{i,1}(s,\tau)| = O_p(\sqrt{n}).
	\end{align}
	Then, 
	\begin{align*}
	& \sup_{\tau \in \Upsilon}|R^w_{sfe,1,1}(u,\tau)| \\
	\leq & \sum_{s \in \mathcal{S}}\frac{|u_1|}{n^w(s)} \sup_{ s \in \mathcal{S}}\left|\frac{D_n^w(s)}{\sqrt{n}}\right|\biggl[\sup_{\tau \in \Upsilon, s \in \mathcal{S}}\biggl|\sum_{i =1}^n \xi_iA_i 1\{S_i = s\}\eta_{i,1}(s,\tau)\biggr| + \sup_{ s \in \mathcal{S}}|D_n^w(s)| \biggr] \\
	= & O_p(1/\sqrt{n}), 
	\end{align*}
	as $n^w(s)/n \convP p(s) > 0$. 
	
	\textbf{Step 1.} Because 
	\begin{align*}
	\sup_{s \in \mathcal{S} }|D_n(s)| = O_p(\sqrt{n}), 
	\end{align*}
	we only need to bound the difference $D_n^w(s) - D_n(s)$. Note that 
	\begin{align}
	\label{eq:Dnc}
	&n(s)^{-1/2}D_n^w(s)-n(s)^{-1/2}D_n(s)= n^{-1/2}\sum_{i=1}^n(\xi_i-1)(A_i-\pi)1\{S_i=s\}.
	\end{align}
	
	We aim to prove that, if $n(s) \rightarrow \infty$ and $D_n(s)/n(s) = o_p(1)$, then conditionally on data, for $s \in \mathcal{S}$, 
	\begin{align}
	\label{eq:clt}
	n(s)^{-1/2}\sum_{i=1}^n(\xi_i-1)(A_i-\pi)1\{S_i =s\} \convD N(0,\pi(1-\pi))
	\end{align}
	and they are independent across $s\in\mathcal{S}$. The independence is straightforward because 
	\begin{align*}
	\frac{1}{n(s)}\sum_{i =1}^n (\xi_i-1)^2(A_i-\pi)^21\{S_i =s\}1\{S_i =s'\} = 0 \quad \text{for} \quad s \neq s'.
	\end{align*}
	
	For the limiting distribution, let $\mathcal{D}_n = \{Y_i,A_i,S_i\}_{i=1}^n$ denote data. According to the Lindeberg-Feller central limit theorem, \eqref{eq:clt} holds because (1)
	\begin{align*}
	n(s)^{-1}\sum_{i=1}^n\mathbb{E}[(\xi_i-1)^2(A_i-\pi)^21\{S_i=s\}|\mathcal{D}_n]=&n(s)^{-1}\sum_{i=1}^n(A_i-2A_i\pi+\pi^2)1\{S_i=s\}\\
	=&n(s)^{-1}\sum_{i=1}^n(A_i-\pi-2(A_i-\pi)\pi+\pi-\pi^2)1\{S_i=s\}\\
	=&\frac{1-2\pi}{n(s)}D_n(s)+\pi(1-\pi)\\
	\convP&\pi(1-\pi),
	\end{align*}
	and (2) for every $\eps>0$, 
	\begin{align*}
	&n(s)^{-1}\sum_{i=1}^n(A_i- \pi)^21\{S_i = s\}\mathbb{E}\left[(\xi_i-1)^21\{|\xi_i-1|(A_i- \pi)^21\{S_i = s\}>\varepsilon\sqrt{n(s)}\}|\mathcal{D}_n\right]\\
	\leq &4 \mathbb{E}(\xi_i-1)^2 1\{2|\xi_i - 1| \geq \varepsilon\sqrt{n(s)}\}\rightarrow 0,
	\end{align*}
	where we use the fact that $|A_i - \pi|1\{S_i = s\} \leq 2$ and $n(s) \rightarrow \infty$. This concludes the proof of Step 1. 
	
	\textbf{Step 2.} By the same rearrangement argument and the fact that $\{\xi_i\}_{i=1}^n \indep \mathcal{D}_n$, we have 
	\begin{align*}
	\sup_{\tau \in \Upsilon, s \in \mathcal{S}}\biggl| \frac{1}{n}\sum_{i=1}^{n} \xi_i A_i1\{S_i = s\}\eta_{i,1}(s,\tau) \biggr| \stackrel{d}{=} \sup_{\tau \in \Upsilon, s \in \mathcal{S}}\biggl| \frac{1}{n}\sum_{i=N(s) + 1}^{N(s)+n_1(s)} \xi_i \tilde{\eta}_{i,1}(s,\tau)\biggr|.
	\end{align*}
	Let $\Gamma_{n,1}(s,t,\tau) = \sum_{i =1}^{\lfloor nt \rfloor}\frac{\xi_i \tilde{\eta}_{i,1}(s,\tau)}{\sqrt{n}}$ and $\mathcal{F} = \{\xi_i \tilde{\eta}_{i,1}(s,\tau): \tau \in \Upsilon, s \in \mathcal{S}\}$ with envelope $F_i = C\xi_i$ and $||F_i||_{P,2}<\infty$. By Lemma \ref{lem:S} and \citet[Theorem 2.14.1]{VW96}, for any $\eps>0$, we can choose $M$ sufficiently large such that
	\begin{align*}
	\mathbb{P}(\sup_{0 < t \leq 1,\tau \in \Upsilon, s\in \mathcal{S}}|\Gamma_{n,1}(s,t,\tau)| \geq M)\leq &\frac{270 \mathbb{E}\sup_{\tau \in \Upsilon, s\in \mathcal{S}}|\Gamma_{n,1}(s,1,\tau)|}{M} \\
	= & \frac{270 \mathbb{E}\sqrt{n}||\mathbb{P}_n - \mathbb{P}||_\mathcal{F}}{M} \lesssim \frac{J(1,\mathcal{F})||F_i||_{P,2}}{M} < \eps.
	\end{align*}
	Therefore,   
	\begin{align*}
	\sup_{0 < t \leq 1,\tau \in \Upsilon, s\in \mathcal{S}}|\Gamma_{n,1}(s,t,\tau)| = O_p(1)
	\end{align*} 
	and 
	\begin{align}
	\label{eq:l5}
	\sup_{\tau \in \Upsilon, s \in \mathcal{S}}\biggl| \frac{1}{n}\sum_{i=1}^{n} \xi_i A_i1\{S_i = s\}\eta_{i,1}(s,\tau) \biggr| \stackrel{d}{=} & \sup_{\tau \in \Upsilon, s \in \mathcal{S}}\frac{1}{\sqrt{n}}\left|\Gamma_{n,1}\left(s,\frac{N(s)+n_1(s)}{n},\tau\right) - \Gamma_{n,1}\left(s,\frac{N(s)}{n},\tau\right)\right| \notag \\
	= & O_p(1/\sqrt{n}).
	\end{align}
	This concludes the proof of Step 2. 
\end{proof}

\begin{lem}
	\label{lem:L2c}
	If Assumptions \ref{ass:assignment1} and \ref{ass:tau} hold, then \ref{eq:l21nc} and \ref{eq:l20nc} hold. 
\end{lem}
\begin{proof}
	We focus on \eqref{eq:l21nc}. Note that 
	\begin{align}
	\label{eq:l21nc'}
	& L_{2,1,n}^w(u,\tau) \notag \\
	= & \sum_{s \in \mathcal{S}} \sum_{i=1}^n \xi_iA_i 1\{S_i = s\} \int_0^{\frac{u'\iota_1}{\sqrt{n}}- \frac{E_n^w(s)}{\sqrt{n}}\left(q(\tau) + \frac{u_1}{\sqrt{n}}\right)}\left(1\{Y_i(1) \leq q_1(\tau)+v\} - 1\{Y_i(1) \leq q_1(\tau)\} \right)dv \notag \\
	= &  \sum_{s \in \mathcal{S}}\sum_{i=1}^{n} \xi_iA_i 1\{S_i = s\}[\phi_i(u,\tau,s,E_n^w(s)) - \mathbb{E}\phi_i(u,\tau,s,E_n^w(s)|S_i=s)] \notag \\
	& + \sum_{s \in \mathcal{S}}\sum_{i=1}^{n}  \xi_iA_i 1\{S_i=s\}\mathbb{E}\phi_i(u,\tau,s,E_n^w(s)|S_i=s), \notag \\
	\end{align} 
	where by Lemma \ref{lem:Rsfe11c}, $E_n^w(s) = \sqrt{n}(\hat{\pi}^w(s) - \pi) = \frac{n}{n^w(s)}\frac{D_n^w(s)}{\sqrt{n}} = O_p(1)$,
	\begin{align*}
	\phi_i(u,\tau,s,e) = \int_0^{\frac{u'\iota_1}{\sqrt{n}}- \frac{e}{\sqrt{n}}\left(q(\tau) + \frac{u_1}{\sqrt{n}}\right)}\left(1\{Y_i(1) \leq q_1(\tau)+v\} - 1\{Y_i(1) \leq q_1(\tau)\} \right)dv,
	\end{align*}
	and $ \mathbb{E}\phi_i(u,\tau,s,E_n^w(s)|S_i=s)$ is interpreted as $ \mathbb{E}(\phi_i(u,\tau,s,e)|S_i=s)$ with $e$ being evaluated at $E_n^w(s)$. 
	
	For the first term on the RHS of \eqref{eq:l21nc'}, by the rearrangement argument in Lemma \ref{lem:Q}, we have
	\begin{align*}
	& \sum_{s \in \mathcal{S}}\sum_{i=1}^{n}\xi_{i}A_i 1\{S_i = s\}[\phi_i(u,\tau,s,E_n^w(s)) - \mathbb{E}\phi_i(u,\tau,s,E_n^w(s)|S_i=s)]  \\
	\stackrel{d}{=} & \sum_{s \in \mathcal{S}}\sum_{i=N(s)+1}^{N(s)+n_1(s)}\xi_{i}[\phi_i^s(u,\tau,s,E_n^w(s)) - \mathbb{E}\phi_i^s(u,\tau,s,E_n^w(s))] ,
	\end{align*}
	where 
	\begin{align*}
	\phi_i^s(u,\tau,s,e) = \int_0^{\frac{u'\iota_1}{\sqrt{n}}- \frac{e}{\sqrt{n}}\left(q(\tau) + \frac{u_1}{\sqrt{n}}\right)}\left(1\{Y^s_i(1) \leq q_1(\tau)+v\} - 1\{Y^s_i(1) \leq q_1(\tau)\} \right)dv.
	\end{align*}
	Similar to \eqref{eq:gammasfe}, we can show that, as $n \rightarrow \infty$,   
	\begin{align}
	\label{eq:term1xi}
	\sup_{\tau \in \Upsilon, s\in \mathcal{S}}\left|\sum_{i=N(s)+1}^{N(s)+n_1(s)}\xi_{i}\left[\phi_i^s(u,\tau,s,E_n^w(s)) - \mathbb{E}\phi_i^s(u,\tau,s,E_n^w(s))\right]\right| = o_p(1). 
	\end{align}
	
	For the second term in \eqref{eq:l21nc'}, we have 
	\begin{align}
	\label{eq:l21nc'''}
	& \sum_{s \in \mathcal{S}}\sum_{i=1}^{n} \xi_i A_i 1\{S_i=s\}\mathbb{E}\phi_i(u,\tau,s,E_n^w(s)|S_i=s) \notag \\
	= & \sum_{s \in \mathcal{S}} \frac{\sum_{i=1}^{n} \xi_i\pi 1\{S_i=s\}}{n} n\mathbb{E}\phi^s_i(u,\tau,s,E_n^w(s)) + \sum_{s \in \mathcal{S}}\frac{D_n^w(s)}{n} n\mathbb{E}\phi^s_i(u,\tau,s,E_n^w(s)) \notag \\
	= & \sum_{s \in \mathcal{S}} \pi p(s) \left[\frac{f_1(q_1(\tau)|s)}{2}(u'\iota_1 - E_n^w(s)q(\tau))^2 + o_p(1)\right] + \sum_{s \in \mathcal{S}}\frac{D_n^w(s)}{n}\left[\frac{f_1(q_1(\tau)|s)}{2}(u'\iota_1 - E_n^w(s)q(\tau))^2 + o_p(1)\right] \notag \\
	= & \frac{\pi f_1(q_1(\tau))}{2}(u'\iota_1)^2 - \sum_{s \in \mathcal{S}}f_1(q_1(\tau)|s)\frac{\pi D_{n}^w(s)u'\iota_1}{\sqrt{n}}q(\tau) + h_{2,1}^w(\tau) + o_p(1),
	\end{align}
	where the $o_p(1)$ term holds uniformly over $(\tau,s) \in \Upsilon \times \mathcal{S}$. The second equality holds by the same calculation in \eqref{eq:Gamma_sfe1} and the fact that $\sum_{i =1}^n \xi_i1\{S_i = s\}/n \convP p(s)$. The last inequality holds because $\frac{D_n^w(s)}{n} = o_p(1)$, $E_n^w(s) = \frac{n}{n^w(s)}\frac{D_n^w(s)}{\sqrt{n}} = O_p(1)$, $ \frac{n}{n^w(s)} \convP 1/p(s)$, and 
	\begin{align*}
	h_{2,1}^w(\tau)  = \sum_{s \in \mathcal{S}} \frac{\pi f_1(q_1(\tau)|s)}{2}p(s)(E_n^w(s))^2q^2(\tau).
	\end{align*}
	
	Combining \eqref{eq:l21nc'}--\eqref{eq:l21nc'''}, we have 
	\begin{align*}
	L_{2,1,n}^w(u,\tau) = \frac{\pi f_1(q_1(\tau))}{2}(u'\iota_1)^2 - \sum_{s \in \mathcal{S}}f_1(q_1(\tau)|s)\frac{\pi D_{n}^w(s)u'\iota_1}{\sqrt{n}}q(\tau) + h_{2,1}^w(\tau) + R_{sfe,2,1}^w(u,\tau),
	\end{align*}
	where 
	\begin{align*}
	h_{2,1}^w(\tau)  = \sum_{s \in \mathcal{S}} \frac{\pi f_1(q_1(\tau)|s)}{2}p(s)(E_n^w(s))^2q^2(\tau)
	\end{align*}
	and
	\begin{align*}
	\sup_{\tau \in \Upsilon}|R_{sfe,2,1}^w(u,\tau)| = o_p(1).
	\end{align*}
	This concludes the proof. 
	
\end{proof}

\begin{lem}
	\label{lem:wboot}
	If Assumptions \ref{ass:assignment1} and \ref{ass:tau} hold, then $\sup_{\tau \in \Upsilon}|| W_{sfe,n}^w(\tau) || = O_p(1)$. 
\end{lem}
\begin{proof}
	It suffices to show that 
	\begin{align}
	\label{eq:wboot_4}
	\sup_{\tau \in \Upsilon, s \in \mathcal{S}}\left|\frac{1}{\sqrt{n}}\sum_{i =1}^n \xi_i A_i 1\{S_i=s\}\eta_{i,1}(s,\tau)\right| = O_p(1) 
	\end{align}
	\begin{align}
	\label{eq:wboot_4'}
	\sup_{\tau \in \Upsilon, s \in \mathcal{S}}\left|\frac{1}{\sqrt{n}}\sum_{i =1}^n \xi_i(1-A_i)1\{S_i=s\}\eta_{i,0}(s,\tau)\right| = O_p(1), 
	\end{align}
	\begin{align}
	\label{eq:wboot_5}
	\sup_{ s \in \mathcal{S}}\left|\frac{1}{\sqrt{n}}\sum_{i =1}^n\xi_i(A_i - \pi)1\{S_i = s\}\right| =O_p(1),
	\end{align}
	and 
	\begin{align}
	\label{eq:wboot_6}
	\sup_{\tau \in \Upsilon}\left\Vert \frac{1}{\sqrt{n}}\sum_{i=1}^n\xi_i(\iota_1 \pi m_1(S_i,\tau) + \iota_0(1-\pi)m_0(S_i,\tau))\right\Vert = O_p(1). 
	\end{align}
	Note that \eqref{eq:wboot_4} holds by the argument in step 2 in the proof of Lemma \ref{lem:Rsfe11c}, \eqref{eq:wboot_4'} holds similarly, \eqref{eq:wboot_5} holds by \eqref{eq:Dnc} and \eqref{eq:clt}, and \eqref{eq:wboot_6} holds by the usual maximal inequality, e.g., \citet[Theorem 2.14.1]{VW96}. This concludes the proof. 	
\end{proof}

\begin{lem}
	\label{lem:Qsfec}
	If Assumptions \ref{ass:assignment1} and \ref{ass:tau} hold, then conditionally on data,  
	\begin{align*}
	\frac{1}{\sqrt{n}}\sum_{s \in \mathcal{S}} \sum_{i =1}^n (\xi_i-1) \mathcal{J}_i(s,\tau) \convD \tilde{\mathcal{B}}_{sfe}(\tau),
	\end{align*}
	where $\tilde{\mathcal{B}}_{sfe}(\tau)$ is a Gaussian process with covariance kernel $\tilde{\Sigma}_{sfe}(\cdot,\cdot)$ defined in \eqref{eq:sigma_sfe_tilde}. 
\end{lem}
\begin{proof}
	In order to show the weak convergence, we only need to show (1) conditional stochastic equicontinuity and (2) conditional convergence in finite dimension. We divide the proof into two steps accordingly. 
	
	\textbf{Step 1.} In order to show the conditional stochastic equicontinuity, it suffices to show that, for any $\eps>0$, as $n \rightarrow \infty$ followed by $\delta \rightarrow 0$, 
	\begin{align*}
	\mathbb{P}_\xi\left( \sup_{\tau_1,\tau_2 \in \Upsilon, \tau_1 < \tau_2 < \tau_1+\delta,s \in \mathcal{S}}\left|\frac{1}{\sqrt{n}}\sum_{i=1}^n (\xi_i - 1)(\mathcal{J}_i(s,\tau_2) - \mathcal{J}_i(s,\tau_1))\right| \geq \eps\right) \convP 0,
	\end{align*}
	where $\mathbb{P}_\xi(\cdot)$ means that the probability operator is with respect to $\xi_1,\cdots,\xi_n$ and conditional on data. Note
	\begin{align*}
	& \mathbb{E}\mathbb{P}_\xi\left( \sup_{\tau_1,\tau_2 \in \Upsilon, \tau_1 < \tau_2 < \tau_1+\delta,s \in \mathcal{S}}\left|\frac{1}{\sqrt{n}}\sum_{i=1}^n (\xi_i - 1)(\mathcal{J}_i(s,\tau_1) - \mathcal{J}_i(s,\tau_1))\right| \geq \eps\right) \\
	= & \mathbb{P}\left( \sup_{\tau_1,\tau_2 \in \Upsilon, \tau_1 < \tau_2 < \tau_1+\delta,s \in \mathcal{S}}\left|\frac{1}{\sqrt{n}}\sum_{i=1}^n (\xi_i - 1)(\mathcal{J}_i(s,\tau_2) - \mathcal{J}_i(s,\tau_1))\right| \geq \eps\right)\\
	\leq &  \mathbb{P}\left( \sup_{\tau_1,\tau_2 \in \Upsilon, \tau_1 < \tau_2 < \tau_1+\delta,s \in \mathcal{S}}\left|\frac{1}{\sqrt{n}}\sum_{i=1}^n (\xi_i - 1)(\mathcal{J}_{i,1}(s,\tau_2) - \mathcal{J}_{i,1}(s,\tau_1))\right| \geq \eps/3\right) \\
	& + \mathbb{P}\left( \sup_{\tau_1,\tau_2 \in \Upsilon, \tau_1 < \tau_2 < \tau_1+\delta,s \in \mathcal{S}}\left|\frac{1}{\sqrt{n}}\sum_{i=1}^n (\xi_i - 1)(\mathcal{J}_{i,2}(s,\tau_2) - \mathcal{J}_{i,2}(s,\tau_1))\right| \geq \eps/3\right)\\
	& + \mathbb{P}\left( \sup_{\tau_1,\tau_2 \in \Upsilon, \tau_1 < \tau_2 < \tau_1+\delta,s \in \mathcal{S}}\left|\frac{1}{\sqrt{n}}\sum_{i=1}^n (\xi_i - 1)(\mathcal{J}_{i,3}(s,\tau_2) - \mathcal{J}_{i,3}(s,\tau_1))\right| \geq \eps/3\right),
	\end{align*}
	where 
	\begin{align*}
	\mathcal{J}_{i,1}(s,\tau) = \frac{A_i 1\{S_i = s\}\eta_{i,1}(s,\tau)}{\pi f_1(q_1(\tau))} - \frac{(1-A_i)1\{S_i = s\}\eta_{i,0}(s,\tau)}{(1-\pi) f_0(q_0(\tau))},
	\end{align*}
	\begin{align*}
	\mathcal{J}_{i,2}(s,\tau) = F_1(s,\tau) (A_i - \pi)1\{S_i = s\},
	\end{align*}
	$$F_1(s,\tau) =  \left(\frac{1-\pi}{\pi f_1(q_1(\tau))} - \frac{\pi}{(1-\pi)f_0(q_0(\tau))}\right)(m_1(s,\tau)- m_0(s,\tau)) + q(\tau)\biggl[\frac{f_1(q_1(\tau)|s)}{f_1(q_1(\tau))} - \frac{f_0(q_0(\tau)|s)}{f_0(q_0(\tau))} \biggr],$$
	\begin{align*}
	\mathcal{J}_{i,3}(s,\tau) =  \left(\frac{m_1(s,\tau)}{f_1(q_1(\tau))} - \frac{m_0(s,\tau)}{f_0(q_0(\tau))}\right)1\{S_i = s\}.
	\end{align*}
	
	Further note that 
	\begin{align*}
	\sum_{i =1}^n (\xi_i-1)\mathcal{J}_{i,1}(s,\tau) \stackrel{d}{=} \sum_{i=N(s) + 1}^{N(s)+n_1(s)}\frac{(\xi_i-1)\tilde{\eta}_{i,1}(s,\tau)}{\pi f_1(q_1(\tau))} - \sum_{i=N(s) +n_1(s)+ 1}^{N(s)+n(s)}\frac{(\xi_i-1)\tilde{\eta}_{i,0}(s,\tau)}{(1-\pi) f_0(q_0(\tau))}
	\end{align*}
	By the same argument in Claim (1) in the proof of Lemma \ref{lem:Q}, we have 
	\begin{align*}
	& \mathbb{P}\left( \sup_{\tau_1,\tau_2 \in \Upsilon, \tau_1 < \tau_2 < \tau_1+\delta,s \in \mathcal{S}}\left|\frac{1}{\sqrt{n}}\sum_{i=1}^n (\xi_i - 1)(\mathcal{J}_{i,1}(s,\tau_2) - \mathcal{J}_i(s,\tau_1))\right| \geq \eps/3\right) \\
	\leq & \frac{3\mathbb{E}\sup_{\tau_1,\tau_2 \in \Upsilon, \tau_1 < \tau_2 < \tau_1+\delta,s \in \mathcal{S}}\left|\frac{1}{\sqrt{n}}\sum_{i=1}^n (\xi_i - 1)(\mathcal{J}_{i,1}(s,\tau_2) - \mathcal{J}_{i,1}(s,\tau_1))\right|}{\eps} \\
	\leq & \frac{3\sqrt{c_2\delta\log(\frac{C}{c_1\delta})} + \frac{3C\log(\frac{C}{c_1\delta})}{\sqrt{n}}}{\eps},
	\end{align*}	
	where $C$, $c_1< c_2$  are some positive constants that are independent of $(n,\eps,\delta)$. By letting $n\rightarrow \infty$ followed by $\delta \rightarrow 0$, the RHS vanishes. 	
	
	For $\mathcal{J}_{i,2}$, we note that $F_1(s,\tau)$ is Lipschitz in $\tau$. Therefore, 
	\begin{align*}
	& \mathbb{P}\left( \sup_{\tau_1,\tau_2 \in \Upsilon, \tau_1 < \tau_2 < \tau_1+\delta,s \in \mathcal{S}}\left|\frac{1}{\sqrt{n}}\sum_{i=1}^n (\xi_i - 1)(\mathcal{J}_{i,2}(s,\tau_2) - \mathcal{J}_{i,2}(s,\tau_1))\right| \geq \eps/3\right) \\
	\leq & \sum_{s \in \mathcal{S}}\mathbb{P}\left(C\delta\left|\frac{1}{\sqrt{n}} \sum_{i =1}^n (\xi_i - 1)(A_i - \pi)1\{S_i = s\}\right| \geq \eps/3 \right) \rightarrow 0 
	\end{align*}
	as $n \rightarrow \infty$ followed by $\delta \rightarrow 0$, in which we use the fact that, by \eqref{eq:clt}, 
	\begin{align*}
	\sup_{s \in \mathcal{S} }\left|\frac{1}{\sqrt{n}} \sum_{i =1}^n (\xi_i - 1)(A_i - \pi)1\{S_i = s\}\right| = O_p(1). 
	\end{align*}
	Last, by the standard maximal inequality (e.g., \citet[Theorem 2.14.1]{VW96}) and the fact that 
	\begin{align*}
	\left(\frac{m_1(s,\tau)}{f_1(q_1(\tau))} - \frac{m_0(s,\tau)}{f_0(q_0(\tau))}\right)
	\end{align*}
	is Lipschitz in $\tau$, we have, as $n \rightarrow \infty$ followed by $\delta \rightarrow 0$,  
	\begin{align*}
	\mathbb{P}\left( \sup_{\tau_1,\tau_2 \in \Upsilon, \tau_1 < \tau_2 < \tau_1+\delta,s \in \mathcal{S}}\left|\frac{1}{\sqrt{n}}\sum_{i=1}^n (\xi_i - 1)(\mathcal{J}_{i,3}(s,\tau_2) - \mathcal{J}_{i,3}(s,\tau_1))\right| \geq \eps/3\right) \rightarrow 0
	\end{align*}
	This concludes the proof of the conditional stochastic equicontinuity.

	\textbf{Step 2.}
	We focus on the one-dimension case and aim to show that, conditionally on data, for fixed $\tau \in \Upsilon$,  
	\begin{align*}
	\frac{1}{\sqrt{n}}\sum_{s \in \mathcal{S}} \sum_{i =1}^n (\xi_i-1) \mathcal{J}_i(s,\tau) \convD \N(0,\tilde{\Sigma}_{sfe}(\tau,\tau)).
	\end{align*}
	The finite-dimensional convergence can be established similarly by the Cram\'{e}r-Wold device. In view of Lindeberg-Feller central limit theorem, we only need to show that (1)
	\begin{align*}
	\frac{1}{n}\sum_{i =1}^n[\sum_{s \in \mathcal{S}}\mathcal{J}_i(s,\tau)]^2 \convP \zeta_Y^2(\pi,\tau) + \tilde{\xi}_A^{\prime 2}(\pi,\tau) + \xi_S^2(\pi,\tau)
	\end{align*}
	and (2)
	\begin{align*}
	\frac{1}{n}\sum_{i =1}^n[\sum_{s \in \mathcal{S}}\mathcal{J}_i(s,\tau)]^2 \mathbb{E}_\xi(\xi-1)^21\{|\sum_{s \in \mathcal{S}}(\xi_i - 1)\mathcal{J}_i(s,\tau)| \geq \eps \sqrt{n}\} \rightarrow 0.
	\end{align*}
	(2) is obvious as $|\mathcal{J}_i(s,\tau)|$ is bounded and $\max_i|\xi_i-1| \lesssim \log(n)$ as $\xi_i$ is sub-exponential. Next, we focus on (1). We have 
	\begin{align*}
	& \frac{1}{n}\sum_{i =1}^n[\sum_{s \in \mathcal{S}}\mathcal{J}_i(s,\tau)]^2 \\
	= & \frac{1}{n}\sum_{i=1}^n\sum_{s \in \mathcal{S}}\biggl\{\biggl[\frac{A_i 1\{S_i = s\}\eta_{i,1}(s,\tau)}{\pi f_1(q_1(\tau))} - \frac{(1-A_i)1\{S_i = s\}\eta_{i,0}(s,\tau)}{(1-\pi) f_0(q_0(\tau))}\biggr] \\
	& + F_1(s,\tau)(A_i -\pi)1\{S_i=s\} + \biggl[ \left(\frac{m_1(s,\tau)}{f_1(q_1(\tau))} - \frac{m_0(s,\tau)}{f_0(q_0(\tau))}\right)1\{S_i = s\}\biggr]\biggr\}^2 \\
	\equiv & \sigma_1^2 + \sigma_2^2 + \sigma_3^2 + 2\sigma_{12} + 2\sigma_{13} + 2 \sigma_{23},
	\end{align*}
	where 
	\begin{align*}
	\sigma_1^2 = \frac{1}{n}\sum_{s \in \mathcal{S}}\sum_{i =1}^n \biggl[\frac{A_i 1\{S_i = s\}\eta_{i,1}(s,\tau)}{\pi f_1(q_1(\tau))} - \frac{(1-A_i)1\{S_i = s\}\eta_{i,0}(s,\tau)}{(1-\pi) f_0(q_0(\tau))}\biggr]^2,
	\end{align*}
	\begin{align*}
	\sigma_2^2 = \frac{1}{n}\sum_{s \in \mathcal{S}}F^2_1(s,\tau)\sum_{i =1}^n (A_i-\pi)^21\{S_i = s\},
	\end{align*}
	\begin{align*}
	\sigma_3^2 = \frac{1}{n}\sum_{i =1}^n \biggl[ \left(\frac{m_1(S_i,\tau)}{f_1(q_1(\tau))} - \frac{m_0(S_i,\tau)}{f_0(q_0(\tau))}\right)\biggr]^2,
	\end{align*}
	\begin{align*}
	\sigma_{12} = \frac{1}{n}\sum_{i=1}^n\sum_{s \in \mathcal{S}}\biggl[\frac{A_i 1\{S_i = s\}\eta_{i,1}(s,\tau)}{\pi f_1(q_1(\tau))} - \frac{(1-A_i)1\{S_i = s\}\eta_{i,0}(s,\tau)}{(1-\pi) f_0(q_0(\tau))}\biggr]F_1(s,\tau)(A_i - \pi)1\{S_i=s\},
	\end{align*}
	
	\begin{align*}
	\sigma_{13} = \frac{1}{n}\sum_{i=1}^n\sum_{s \in \mathcal{S}}\biggl[\frac{A_i 1\{S_i = s\}\eta_{i,1}(s,\tau)}{\pi f_1(q_1(\tau))} - \frac{(1-A_i)1\{S_i = s\}\eta_{i,0}(s,\tau)}{(1-\pi) f_0(q_0(\tau))}\biggr] \biggl[ \left(\frac{m_1(s,\tau)}{f_1(q_1(\tau))} - \frac{m_0(s,\tau)}{f_0(q_0(\tau))}\right)\biggr],
	\end{align*}
	and 
	\begin{align*}
	\sigma_{23} = \sigma_{12} = \frac{1}{n}\sum_{i=1}^n\sum_{s \in \mathcal{S}}F_1(s,\tau)(A_i - \pi)1\{S_i=s\}\biggl[ \left(\frac{m_1(s,\tau)}{f_1(q_1(\tau))} - \frac{m_0(s,\tau)}{f_0(q_0(\tau))}\right)\biggr].
	\end{align*}
	For $\sigma_1^2$, we have 
	\begin{align*}
	\sigma_1^2 = & \frac{1}{n}\sum_{s \in \mathcal{S}}\sum_{i =1}^n \biggl[\frac{A_i 1\{S_i = s\}\eta^2_{i,1}(s,\tau)}{\pi^2 f^2_1(q_1(\tau))} - \frac{(1-A_i)1\{S_i = s\}\eta^2_{i,0}(s,\tau)}{(1-\pi)^2 f^2_0(q_0(\tau))}\biggr] \\
	\stackrel{d}{=} & \frac{1}{n}\sum_{s \in \mathcal{S}} \sum_{i=N(s)+1}^{N(s)+n_1(s)}\frac{\tilde{\eta}^2_{i,1}(s,\tau)}{\pi^2 f^2_1(q_1(\tau))} +  \frac{1}{n}\sum_{s \in \mathcal{S}} \sum_{i=N(s)+n_1(s)+1}^{N(s)+n(s)}\frac{\tilde{\eta}^2_{i,0}(s,\tau)}{(1-\pi)^2 f^2_0(q_0(\tau))} \\
	\convP & \frac{\tau(1-\tau) - \mathbb{E}m_1^s(S,\tau)}{\pi f_1^2(q_1(\tau))} + \frac{\tau(1-\tau) - \mathbb{E}m_0^s(S,\tau)}{(1-\pi) f_0^2(q_0(\tau))} = \zeta_Y^2(\pi,\tau),
	\end{align*}
	where the second equality holds due to the rearrangement argument in Lemma \ref{lem:Q} and the convergence in probability holds due to uniform convergence of the partial sum process. 
	
	For $\sigma_2^2$, by Assumption \ref{ass:assignment1},  
	\begin{align*}
	\sigma_2^2 = \frac{1}{n}\sum_{s \in \mathcal{S}}F_1^2(s,\tau)(D_n(s) - 2\pi D_n(s) + \pi(1-\pi)1\{S_i =s\}) \convP \pi(1-\pi)\mathbb{E}F_1^2(S_i,\tau) = \tilde{\xi}_{A}^{\prime 2}(\pi,\tau).
	\end{align*}
	
	For $\sigma_3^2$, by the law of large number, 
	\begin{align*}
	\sigma_3^2 \convP \mathbb{E}\biggl[ \left(\frac{m_1(S_i,\tau)}{f_1(q_1(\tau))} - \frac{m_0(S_i,\tau)}{f_0(q_0(\tau))}\right)\biggr]^2 = \xi_S^2(\pi,\tau).
	\end{align*}
	
	For $\sigma_{12}$, we have 
	\begin{align*}
	\sigma_{12} = & \frac{1}{n}\sum_{s \in \mathcal{S}}(1-\pi)F_1(s,\tau)\sum_{i=1}^n\frac{A_i 1\{S_i = s\}\eta_{i,1}(s,\tau)}{\pi f_1(q_1(\tau))} - \frac{1}{n}\sum_{s \in \mathcal{S}}\pi F_1(s,\tau)\sum_{i=1}^{n}\frac{(1-A_i)1\{S_i = s\}\eta_{i,0}(s,\tau)}{(1-\pi) f_0(q_0(\tau))} \\
	\stackrel{d}{=} & \frac{1}{n}\sum_{s \in \mathcal{S}}(1-\pi)F_1(s,\tau)\sum_{i=N(s)+1}^{N(s)+n_1(s)}\frac{\tilde{\eta}_{i,1}(s,\tau)}{\pi f_1(q_1(\tau))} - \frac{1}{n}\sum_{s \in \mathcal{S}}\pi F_1(s,\tau)\sum_{i=N(s)+n_1(s)+1}^{N(s)+n(s)}\frac{\tilde{\eta}_{i,0}(s,\tau)}{(1-\pi) f_0(q_0(\tau))} \convP 0,
	\end{align*}
	where the last convergence holds because by Lemma \ref{lem:Q}, 
	\begin{align*}
	\frac{1}{n}\sum_{i=N(s)+1}^{N(s)+n_1(s)}\tilde{\eta}_{i,1}(s,\tau) \convP 0 \quad \text{and} \quad \frac{1}{n}\sum_{i=N(s)+n_1(s)+1}^{N(s)+n(s)}\tilde{\eta}_{i,0}(s,\tau) \convP 0.
	\end{align*}
	By the same argument, we can show that 
	\begin{align*}
	\sigma_{13} \convP 0.
	\end{align*}
	
	Last, for $\sigma_{23}$, by Assumption \ref{ass:assignment1},  
	\begin{align*}
	\sigma_{23} = \sum_{s \in \mathcal{S}}F_1(s,\tau)\biggl[ \left(\frac{m_1(s,\tau)}{f_1(q_1(\tau))} - \frac{m_0(s,\tau)}{f_0(q_0(\tau))}\right)\biggr]\frac{D_n(s)}{n} \convP 0. 
	\end{align*}
	Therefore, we have 
	\begin{align*}
	\frac{1}{n}\sum_{i =1}^n[\sum_{s \in \mathcal{S}}\mathcal{J}_i(s,\tau)]^2 \convP \zeta_Y^2(\pi,\tau) + \tilde{\xi}_A^{\prime 2}(\pi,\tau) + \xi_S^2(\pi,\tau). 
	\end{align*}
\end{proof}

\begin{lem}
	\label{lem:L2star}
	Recall $R_{sfe,2,1}^*(u,\tau)$ and $R_{sfe,2,0}^*(u,\tau)$ defined in \eqref{eq:l21nstar} and \eqref{eq:l20nstar}, respectively. If Assumptions in Theorem \ref{thm:cab} hold, then \eqref{eq:l21nstar} and \eqref{eq:l20nstar} hold and 
	\begin{align*}
	\sup_{\tau \in \Upsilon}|R_{sfe,2,1}^*(u,\tau)| = o_p(1) \quad \text{and} \quad \sup_{\tau \in \Upsilon}|R_{sfe,2,0}^*(u,\tau)| = o_p(1).
	\end{align*}
\end{lem}
\begin{proof}
	We focus on \eqref{eq:l21nstar}. Following the definition of $M_{ni}$ in the proof of Lemma \ref{lem:Rsfestar} and the argument in the Step 1.2 of the proof of Theorem \ref{thm:sfe}, we have 
	\begin{align}
	\label{eq:l21nstar'}
	& L_{2,1,n}^*(u,\tau) \notag \\
	= & \sum_{s \in \mathcal{S}} \sum_{i=N(s) + 1}^{N(s)+n_1(s)}M_{ni}\int_0^{\frac{u'\iota_1}{\sqrt{n}}- \frac{E_n^*(s)}{\sqrt{n}}\left(q(\tau) + \frac{u_1}{\sqrt{n}}\right)}\left(1\{Y_i^s(1) \leq q_1(\tau)+v\} - 1\{Y_i^s(1) \leq q_1(\tau)\} \right)dv \notag \\
	= & \sum_{s \in \mathcal{S}} \sum_{i=N(s) + 1}^{N(s)+n_1(s)}M_{ni}\left[\phi_i(u,\tau,s,E_n^*(s)) - \mathbb{E}\phi_i(u,\tau,E_n^*(s)) \right] + \sum_{s \in \mathcal{S}} \sum_{i=N(s) + 1}^{N(s)+n_1(s)}M_{ni}\mathbb{E}\phi_i(u,\tau,s,E_n^*(s)),
	\end{align} 
	where $E_n^*(s) = \sqrt{n}(\hat{\pi}^*(s) - \pi) = \frac{n}{n^*(s)}\frac{D_n^*(s)}{\sqrt{n}} = O_p(1)$, 
	\begin{align*}
	\phi_i(u,\tau,s,e) = \int_0^{\frac{u'\iota_1}{\sqrt{n}}- \frac{e}{\sqrt{n}}\left(q(\tau) + \frac{u_1}{\sqrt{n}}\right)}\left(1\{Y_i^s(1) \leq q_1(\tau)+v\} - 1\{Y_i^s(1) \leq q_1(\tau)\} \right)dv,
	\end{align*}
	and $ \mathbb{E}\phi_i(u,\tau,s,E_n^*(s))$ is interpreted as $ \mathbb{E}\phi_i(u,\tau,s,e)$ with $e$ being evaluated at $E_n^*(s)$. 
	
	For the first term on the RHS of \eqref{eq:l21nstar'}, similar to \eqref{eq:Mtilde-M}, we have
	\begin{align}
	\label{eq:term1}
	& \sum_{s \in \mathcal{S}} \sum_{i=N(s) + 1}^{N(s)+n_1(s)}M_{ni}\left[\phi_i(u,\tau,s,E_n^*(s)) - \mathbb{E}\phi_i(u,\tau,s,E_n^*(s)) \right] \notag \\
	= & \sum_{s \in \mathcal{S}} \sum_{i=N(s) + 1}^{N(s)+n_1(s)}\xi_i^s\left[\phi_i(u,\tau,s,E_n^*(s)) - \mathbb{E}\phi_i(u,\tau,s,E_n^*(s)) \right] + \sum_{s \in \mathcal{S}}r_n(u,\tau,s,E_n^*(s)),
	\end{align}
	where $\{\xi_i^s\}_{i=1}^n$ is a sequence of i.i.d. Poisson(1) random variables and is independent of everything else, and 
	\begin{align*}
	r_n(u,\tau,s,e) = \text{sign}(N(n_1(s)) - n_1(s))\sum_{j=1}^\infty\frac{\#I_{n}^j(s)}{\sqrt{n}}\frac{1}{\#I_{n}^j(s)}\sum_{i \in I_n^j(s)}\sqrt{n}\left[\phi_i(u,\tau,s,e) - \mathbb{E}\phi_i(u,\tau,s,e) \right].
	\end{align*}
	We aim to show 
	\begin{align}
	\label{eq:rn}
	\sup_{|e| \leq M,\tau \in \Upsilon, s\in \mathcal{S}}|r_n(u,\tau,s,e)| = o_p(1), 
	\end{align}

	Recall that the proof of Lemma \ref{lem:Rsfestar} relies on \eqref{eq:kk} and the fact that 
	$$\mathbb{E}\sup_{n \geq k \geq n_0}\sup_{\tau \in \Upsilon, s \in \mathcal{S}}\left|\frac{1}{k}\sum_{j=1}^k \tilde{\eta}_{j,1}(s,\tau)\right| \rightarrow 0.$$ 
	Using the same argument and replacing $\tilde{\eta}_{j,1}(s,\tau)$ by $\sqrt{n}\left[\phi_i(u,\tau,s,e) - \mathbb{E}\phi_i(u,\tau,s,e) \right]$, in order to show \eqref{eq:rn}, we only need to verify that, as $n \rightarrow \infty$ followed by $n_0 \rightarrow \infty$, 
	\begin{align*}
	\mathbb{E}\sup_{n \geq k \geq n_0}\sup_{|e| \leq M,\tau \in \Upsilon, s\in \mathcal{S}}\left|\frac{1}{k}\sum_{j=1}^k\sqrt{n}\left[\phi_i(u,\tau,s,e) - \mathbb{E}\phi_i(u,\tau,s,e) \right]\right| \rightarrow 0
	\end{align*}
	Because $\sup_{|e| \leq M,\tau \in \Upsilon, s\in \mathcal{S}}\left|\frac{1}{k}\sum_{j=1}^k\sqrt{n}\left[\phi_i(u,\tau,s,e) - \mathbb{E}\phi_i(u,\tau,s,e) \right]\right|$ is bounded as shown below, it suffices to show that, for any $\eps>0$, as $n \rightarrow \infty$ followed by $n_0 \rightarrow \infty$,  
	\begin{align}
	\label{eq:targetp}
	\mathbb{P}\left(\sup_{n \geq k \geq n_0}\sup_{|e| \leq M,\tau \in \Upsilon, s\in \mathcal{S}}\left|\frac{1}{k}\sum_{j=1}^k\sqrt{n}\left[\phi_i(u,\tau,s,e) - \mathbb{E}\phi_i(u,\tau,s,e) \right]\right| \geq \eps \right) \rightarrow 0. 
	\end{align}
	Define the class of functions $\mathcal{F}_n$ as
	\begin{align*}
	\mathcal{F}_n = \{\sqrt{n}\left[\phi_i(u,\tau,s,e) - \mathbb{E}\phi_i(u,\tau,s,e) \right]: |e| \leq M,\tau \in \Upsilon, s\in \mathcal{S} \}.
	\end{align*} 
	Then, $\mathcal{F}_n$ is nested by a VC-class with fixed VC-index. In addition, for fixed $u$, $\mathcal{F}_n$ has a bounded (and independent of $n$) envelope function  
	$$F = |u'\iota_1|+M\left(\max_{\tau \in \Upsilon}|q(\tau)| + \left|u_1\right|\right).$$
	Last, define $\mathcal{I}_l = \{2^l,2^l+1,\cdots,2^{l+1}-1\}$. Then, 
	\begin{align*}
	& \mathbb{P}\left(\sup_{n \geq k \geq n_0}\sup_{|e| \leq M,\tau \in \Upsilon, s\in \mathcal{S}}\left|\frac{1}{k}\sum_{j=1}^k\sqrt{n}\left[\phi_i(u,\tau,s,e) - \mathbb{E}\phi_i(u,\tau,s,e) \right]\right| \geq \eps\right) \\
	\leq & \sum_{l=\lfloor \log_2(n_0) \rfloor}^{\lfloor \log_2(n) \rfloor+1}\mathbb{P}\left(\sup_{k \in \mathcal{I}_l}\sup_{|e| \leq M,\tau \in \Upsilon, s\in \mathcal{S}}\left|\frac{1}{k}\sum_{j=1}^k\sqrt{n}\left[\phi_i(u,\tau,s,e) - \mathbb{E}\phi_i(u,\tau,s,e) \right]\right| \geq \eps \right) \\
	\leq & \sum_{l=\lfloor \log_2(n_0) \rfloor}^{\lfloor \log_2(n) \rfloor+1}\mathbb{P}\left(\sup_{k \leq 2^{l+1}}\sup_{|e| \leq M,\tau \in \Upsilon, s\in \mathcal{S}}\left|\sum_{j=1}^k\sqrt{n}\left[\phi_i(u,\tau,s,e) - \mathbb{E}\phi_i(u,\tau,s,e) \right]\right| \geq \eps 2^l \right) \\ 
	\leq & \sum_{l=\lfloor \log_2(n_0) \rfloor}^{\lfloor \log_2(n) \rfloor+1}9\mathbb{P}\left(\sup_{|e| \leq M,\tau \in \Upsilon, s\in \mathcal{S}}\left|\sum_{j=1}^{2^{l+1}}\sqrt{n}\left[\phi_i(u,\tau,s,e) - \mathbb{E}\phi_i(u,\tau,s,e) \right]\right| \geq \eps 2^l/30 \right) \\ 
	\leq & \sum_{l=\lfloor \log_2(n_0) \rfloor}^{\lfloor \log_2(n) \rfloor+1}\frac{270 \mathbb{E}\sup_{|e| \leq M,\tau \in \Upsilon, s\in \mathcal{S}}\left|\sum_{j=1}^{2^{l+1}}\sqrt{n}\left[\phi_i(u,\tau,s,e) - \mathbb{E}\phi_i(u,\tau,s,e) \right]\right|}{\eps 2^l} \\
	\leq & \sum_{l=\lfloor \log_2(n_0) \rfloor}^{\lfloor \log_2(n) \rfloor+1}\frac{C_1}{\eps 2^{l/2}} \\
	\leq & \frac{2C_1}{\eps \sqrt{n_0}} \rightarrow 0,
	\end{align*}
	where the first inequality holds by the union bound, the second inequality holds because on $\mathcal{I}_l$, $2^{l+1} \geq k \geq 2^l$, the third inequality follows the same argument in the proof of Theorem \ref{thm:qr}, the fourth inequality is due to the Markov inequality, the fifth inequality follows the standard maximal inequality such as \citet[Theorem 2.14.1]{VW96} and the constant $C_1$ is independent of $(l,\eps,n)$, and the last inequality holds by letting $n \rightarrow \infty$. Because $\eps$ is arbitrary, we have established \eqref{eq:targetp}, and thus, \eqref{eq:rn}, which further implies that 
	\begin{align*}
	\sup_{\tau \in \Upsilon, s\in \mathcal{S}}|r_n(u,\tau,s,E_n^*(s))| = o_p(1), 
	\end{align*}
	For the leading term of \eqref{eq:term1}, we have 
	\begin{align*}
	& \sum_{s \in \mathcal{S}} \sum_{i=N(s) + 1}^{N(s)+n_1(s)}\xi_i^s\left[\phi_i(u,\tau,s,E_n^*(s)) - \mathbb{E}\phi_i(u,\tau,s,E_n^*(s)) \right] \\
	= & \sum_{s \in \mathcal{S}}\left[ \Gamma_n^{s*}(N(s),\tau,E_n^*(s)) - \Gamma_n^{s*}(N(s)+n_1(s),\tau,E_n^*(s))\right],
	\end{align*}
	where 
	\begin{align*}
	\Gamma_n^{s*}(k,\tau,e) = & \sum_{i=1}^k \xi_i^s\int_0^{\frac{u'\iota_1 - e(q(\tau)+\frac{u_1}{\sqrt{n}})}{\sqrt{n}}}\left(1\{Y_i^s(1) \leq q_1(\tau)+v\} - 1\{Y_i^s(1) \leq q_1(\tau)\} \right)dv \\
	& - k \mathbb{E}\left[\int_0^{\frac{u'\iota_1 - e(q(\tau)+\frac{u_1}{\sqrt{n}})}{\sqrt{n}}}\left(1\{Y_i^s(1) \leq q_1(\tau)+v\} - 1\{Y_i^s(1) \leq q_1(\tau)\} \right)dv\right].
	\end{align*}
	By the same argument in \eqref{eq:Gamma_sfe}, we can show that 
	\begin{align*}
	\sup_{0 < t \leq 1,\tau \in \Upsilon, |e| \leq M}|\Gamma_n^{s*}(k,\tau,e) | = o_p(1), 
	\end{align*}
	where we need to use the fact that the Poisson(1) random variable has an exponential tail and thus
	\begin{align*}
	\mathbb{E}\sup_{i\in\{1,\cdots,n\},s\in \mathcal{S}} \xi_i^s = O(\log(n)). 
	\end{align*}
	Therefore, 
	\begin{align}
	\label{eq:l21nstar''}
	\sup_{\tau \in \Upsilon}\left|\sum_{s \in \mathcal{S}} \sum_{i=N(s) + 1}^{N(s)+n_1(s)}M_{ni}\left[\phi_i(u,\tau,s,E_n^*(s)) - \mathbb{E}\phi_i(u,\tau,E_n^*(s)) \right]\right| = o_p(1).
	\end{align}
	For the second term on the RHS of \eqref{eq:l21nstar'}, we have 
	\begin{align}
	\label{eq:l21nstar'''}
	\sum_{s \in \mathcal{S}} \sum_{i=N(s) + 1}^{N(s)+n_1(s)}M_{ni}\mathbb{E}\phi_i(u,\tau,s,e) = & \sum_{s \in \mathcal{S}}n_1^*(s)\mathbb{E}\phi_i(u,\tau,s,e) \notag \\
	= & \sum_{s \in \mathcal{S}} \pi p(s) \frac{f_1(q_1(\tau)|s)}{2}(u'\iota_1 - eq(\tau))^2 +o(1),
	\end{align}
	where the $o(1)$ term holds uniformly over $(\tau,e) \in \Upsilon \times [-M,M]$, the first equality holds because $\sum_{i=N(s) + 1}^{N(s)+n_1(s)}M_{ni} = n^*_1(s)$ and the second equality holds by the same calculation in \eqref{eq:Gamma_sfe1} and the facts that $n^*(s)/n \convP p(s)$ and 
	\begin{align*}
	\frac{n^*_1(s)}{n} = \frac{D_n^*(s)+\pi n^*(s)}{n} \convP \pi p(s).
	\end{align*}
	
	Combining \eqref{eq:l21nstar}, \eqref{eq:l21nstar'}, \eqref{eq:l21nstar''}, \eqref{eq:l21nstar'''}, and the facts that $E_n^*(s) = \frac{n}{n^*(s)}\frac{D_n^*(s)}{\sqrt{n}}$ and $ \frac{n}{n^*(s)} \convP 1/p(s)$, we have 
	\begin{align*}
	L_{2,1,n}^*(u,\tau) = \frac{\pi f_1(q_1(\tau))}{2}(u'\iota_1)^2 - \sum_{s \in \mathcal{S}}f_1(q_1(\tau)|s)\frac{\pi D_{n}^*(s)u'\iota_1}{\sqrt{n}}q(\tau) + h_{2,1}^*(\tau) + R_{sfe,2,1}^*(u,\tau),
	\end{align*}
	where 
	\begin{align*}
	h_{2,1}^*(\tau)  = \sum_{s \in \mathcal{S}} \frac{\pi f_1(q_1(\tau)|s)}{2}p(s)(E_n^*(s))^2q^2(\tau)
	\end{align*}
	and
	\begin{align*}
	\sup_{\tau \in \Upsilon}|R_{sfe,2,1}^*(u,\tau)| = o_p(1).
	\end{align*}
	This concludes the proof. 
\end{proof}

\begin{lem}
	\label{lem:Qsfestar}
	Recall the definition of $(W_{sfe,n,1}^*(\tau)-\mathcal{W}_{n,1}(\tau), W_{sfe,n,2}^*(\tau), W_{sfe,n,3}^*(\tau)-\mathcal{W}_{n,2}(\tau))$ in \eqref{eq:betasfestar-qhat}. If Assumptions in Theorem \ref{thm:cab} hold, then conditionally on data, 
	\begin{align*}
	(W_{sfe,n,1}^*(\tau)-\mathcal{W}_{n,1}(\tau), W_{sfe,n,2}^*(\tau), W_{sfe,n,3}^*(\tau)-\mathcal{W}_{n,2}(\tau)) \convD (\mathcal{B}_1(\tau),\mathcal{B}_2(\tau),\mathcal{B}_3(\tau)),
	\end{align*}
	where $(\mathcal{B}_1(\tau),\mathcal{B}_2(\tau),\mathcal{B}_3(\tau))$ are three independent Gaussian processes with covariance kernels
	\begin{align*}
	\Sigma_1(\tau_1,\tau_2) = \frac{\min(\tau_1,\tau_2) - \tau_1\tau_2 - \mathbb{E}m_1(S,\tau_1)m_1(S,\tau_2)}{\pi f_1(q_1(\tau_1))f_1(q_1(\tau_2))} + \frac{\min(\tau_1,\tau_2) - \tau_1\tau_2 - \mathbb{E}m_0(S,\tau_1)m_0(S,\tau_2)}{(1-\pi) f_0(q_0(\tau_1))f_0(q_0(\tau_2))},
	\end{align*}
	\begin{align*}
	& \Sigma_2(\tau_1,\tau_2) \\
	= & \mathbb{E}\gamma(S)\biggl[(m_1(S,\tau_1) - m_0(S,\tau_1))\left(\frac{1-\pi}{\pi f_1(q_1(\tau_1))} - \frac{\pi}{(1-\pi)f_0(q_0(\tau_1))}\right) + q(\tau_1)\left(\frac{f_1(q(\tau_1)|S)}{f_1(q_1(\tau_1))} -\frac{f_0(q(\tau_1)|S)}{f_0(q_0(\tau_1))}\right)\biggr]\\
	&\times \biggl[(m_1(S,\tau_2) - m_0(S,\tau_2))\left(\frac{1-\pi}{\pi f_1(q_1(\tau_2))} - \frac{\pi}{(1-\pi)f_0(q_0(\tau_2))}\right) + q(\tau_2)\left(\frac{f_1(q(\tau_2)|S)}{f_1(q_2(\tau_2))} -\frac{f_0(q(\tau_2)|S)}{f_0(q_0(\tau_2))}\right)\biggr], 
	\end{align*}
	and 
	\begin{align*}
	\Sigma_3(\tau_1,\tau_2) = \mathbb{E}\biggl[\frac{m_1(S,\tau_1)}{f_1(q_1(\tau_1))} - \frac{m_0(S,\tau_1)}{f_0(q_0(\tau_1))}\biggr] \biggl[\frac{m_1(S,\tau_2)}{f_1(q_1(\tau_2))} - \frac{m_0(S,\tau_2)}{f_0(q_0(\tau_2))}\biggr],
	\end{align*}
	respectively. 
\end{lem}

\begin{proof}
	Let $\mathcal{A}_n = \{(A_i^*,S_i^*, A_i, S_i): i =1,\cdots,n\}$. Following the definition of $M_{ni}$ and arguments in the proof of Lemma \ref{lem:Rsfestar}, we have 
	\begin{align*}
	& \{W_{sfe,n,1}^*(\tau)-\mathcal{W}_{n,1}(\tau)|\mathcal{A}_n\} \\
	\stackrel{d}{=} &\left\{\sum_{s \in \mathcal{S}} \frac{1}{\sqrt{n}}\left[\sum_{i=N(s) + 1}^{N(s)+n_1(s)}(M_{ni}-1)\left(\frac{\tilde{\eta}_{i,1}(s,\tau)}{\pi f_1(q_1(\tau))}\right)  - \sum_{i=N(s) +n_1(s)+ 1}^{N(s)+n(s)}(M_{ni}-1)\left(\frac{\tilde{\eta}_{i,0}(s,\tau)}{(1-\pi) f_0(q_0(\tau))}\right)\right] \biggl|\mathcal{A}_n\right\}\\
	= & \left\{\sum_{s \in \mathcal{S}}\frac{1}{\sqrt{n}}\left[\sum_{i=N(s) + 1}^{N(s)+n_1(s)}(\xi_{i}^s-1)\frac{\tilde{\eta}_{i,1}(s,\tau)}{\pi f_1(q_1(\tau))} - \sum_{i=N(s) +n_1(s)+ 1}^{N(s)+n(s)}(\xi_{i}^s-1)\frac{\tilde{\eta}_{i,0}(s,\tau)}{(1-\pi) f_0(q_0(\tau))}\right] + R_1(\tau) \biggl|\mathcal{A}_n\right\},
	\end{align*}
	where $\sup_{\tau \in \Upsilon}|R_1(\tau)| = o_p(1)$ and $\{\xi_{i}^s\}_{i=1}^n$, $s \in \mathcal{S}$ are sequences of i.i.d. Poisson(1) random variables that are independent of $\mathcal{A}_n$ and across $s \in \mathcal{S}$. In addition, by the same argument in the proof of Lemma \ref{lem:Q}, we have 
	\begin{align*}
	& \sum_{s \in \mathcal{S}}\frac{1}{\sqrt{n}}\left[\sum_{i=N(s) + 1}^{N(s)+n_1(s)}(\xi_{i}^s-1)\frac{\tilde{\eta}_{i,1}(s,\tau)}{\pi f_1(q_1(\tau))} - \sum_{i=N(s) +n_1(s)+ 1}^{N(s)+n(s)}(\xi_{i}^s-1)\frac{\tilde{\eta}_{i,0}(s,\tau)}{(1-\pi) f_0(q_0(\tau))} \right] \\
	= & \sum_{s \in \mathcal{S}}\frac{1}{\sqrt{n}}\left[\sum_{i= \lfloor nF(s)\rfloor + 1}^{\lfloor n(F(s)+\pi p(s))\rfloor}(\xi_{i}^s-1)\frac{\tilde{\eta}_{i,1}(s,\tau)}{\pi f_1(q_1(\tau))} - \sum_{i=\lfloor n(F(s)+\pi p(s))\rfloor+1}^{\lfloor n(F(s)+ p(s))\rfloor}(\xi_{i}^s-1) \frac{\tilde{\eta}_{i,0}(s,\tau)}{(1-\pi) f_0(q_0(\tau))} \right] + R_2(\tau) \\
	\equiv & W_{1}^*(\tau) + R_2(\tau),
	\end{align*}
	where  $\sup_{\tau \in \Upsilon}|R_2(\tau)| = o_p(1)$. Because both $ W_{sfe,n,2}^*(\tau)$ and $W_{sfe,n,3}^*(\tau)-\mathcal{W}_{n,2}(\tau)$ are in the $\sigma$-field generated by $\mathcal{A}_n$, we have 
	\begin{align*}
	& (W_{sfe,n,1}^*(\tau)-\mathcal{W}_{n,1}(\tau), W_{sfe,n,2}^*(\tau), W_{sfe,n,3}^*(\tau)-\mathcal{W}_{n,2}(\tau)) \\
	\stackrel{d}{=} & (W_{1}^*(\tau) + R_1(\tau)+R_2(\tau), W_{sfe,n,2}^*(\tau), W_{sfe,n,3}^*(\tau)-\mathcal{W}_{n,2}(\tau)). 
	\end{align*}
	In addition, note that $\{\xi_i^s\}_{i=1}^n$ and $\{\tilde{\eta}_{i,1}(s,\tau),\tilde{\eta}_{i,1}(s,\tau)\}_{i=1}^n$ are independent of $\mathcal{A}_n$, therefore, $W_{1}^*(\tau) \indep (W_{sfe,n,2}^*(\tau),W_{sfe,n,3}^*(\tau)-\mathcal{W}_{n,2}(\tau))$. Applying \citet[Theorem 2.9.6]{VW96} to each segment 
	$$\lfloor nF(s)\rfloor + 1,\cdots,\lfloor n(F(s)+\pi p(s))\rfloor \quad \text{or} \quad \lfloor n(F(s)+\pi p(s))\rfloor+1,\cdots,\lfloor n(F(s)+ p(s))\rfloor$$ for $s \in \mathcal{S}$ and noticing that $\{\tilde{\eta}_{i,1}(s,\tau)\}_{i=1}^n$ and $\{\tilde{\eta}_{i,0}(s,\tau)\}_{i=1}^n$ are two i.i.d. sequences for each $s \in \mathcal{S}$, independent of each other, and independent across $s$, we have, conditionally on $\{\tilde{\eta}_{i,1}(s,\tau),\tilde{\eta}_{i,0}(s,\tau)\}_{i=1}^n$, $s \in \mathcal{S}$,   
	\begin{align*}
	W_1^*(\tau) \convD \mathcal{B}_1(\tau)
	\end{align*}  
	with the covariance kernel $\Sigma_1(\tau_1,\tau_2)$. 
	
	For $W^*_{sfe,n,2}(\tau)$, we note that it depends on data only through $\{S_i^*\}_{i=1}^n$. By Assumption \ref{ass:bassignment}, 
	\begin{align*}
	W^*_{sfe,n,2}(\tau)|\{S_i^*\}_{i=1}^n \convD \mathcal{B}_2(\tau)
	\end{align*}
	with the covariance kernel $\Sigma_2(\tau_1,\tau_2)$. 
	
	Last, for $W^*_{sfe,n,3}(\tau) - \mathcal{W}_{n,2}(\tau)$, note that $\{S_i^*\}$ is sampled by the standard bootstrap procedure. Therefore, directly applying \citet[Theorem 3.6.2]{VW96}, we have 
	\begin{align*}
	W^*_{sfe,n,3}(\tau) - \mathcal{W}_{n,2}(\tau) = \frac{1}{\sqrt{n}}\sum_{i =1}^n (\xi'_i-1) \biggl[\frac{m_1(S_i,\tau)}{f_1(q_1(\tau))} - \frac{m_0(S_i,\tau)}{f_0(q_0(\tau))}\biggr]+R_3(\tau)
	\end{align*}
	where $\sup_{\tau \in \Upsilon}|R_3(\tau)| = o_p(1)$, $\{\xi_i'\}_{i=1}^n$ is a sequence of i.i.d. Poisson(1) random variables that is independent of data and $\{\xi^s_i\}_{i=1}^n$, $s \in \mathcal{S}$. By \citet[Theorem 3.6.2]{VW96}, conditionally on data $\{S_i\}_{i=1}^n$, 
	\begin{align*}
	\frac{1}{\sqrt{n}}\sum_{i =1}^n (\xi'_i-1) \biggl[\frac{m_1(S_i,\tau)}{f_1(q_1(\tau))} - \frac{m_0(S_i,\tau)}{f_0(q_0(\tau))}\biggr] \convD \mathcal{B}_3(\tau),
	\end{align*}
	where $\mathcal{B}_3(\tau)$ has the covariance kernel $\Sigma_3(\tau_1,\tau_2)$. Furthermore, $\mathcal{B}_2(\tau)$ and  $\mathcal{B}_3(\tau)$ are independent as $\Sigma_2(\tau_1,\tau_2)$ is not a function of $\{S_i^*\}_{i=1}^n$. This concludes the proof. 
	
\end{proof}

\section{Additional Simulation Results}
\label{sec:addsim}
\subsection{DGPs}
We consider the following four DGPs with parameters $\gamma =4$, $\sigma =2$, and $\mu$ which will be specified later. DGPs 1 and 3 correspond to DGPs 1 and 2 in Section \ref{sec:sim} in the main paper. 
\begin{enumerate}
	\item Let $Z$ be the standardized $\text{Beta}(2,2)$ distributed, $S_i = \sum_{j=1}^4\{Z_i \leq g_j\}$, and $(g_1,\cdots,g_4) = (-0.25\sqrt{20},0,0.25\sqrt{20},0.5\sqrt{20})$. The outcome equation is 
	$$Y_i = A_i \mu + \gamma Z_i + \eta_i,$$ 
	where $\eta_i = \sigma A_i \eps_{i,1} + (1-A_i)\eps_{i,2}$ and $(\eps_{i,1},\eps_{i,2})$ are jointly standard normal. 
	\item Let $S$ be the same as in DGP1. The outcome equation is 
	\begin{align*}
	Y_i = A_i \mu+\gamma Z_i A_i - \gamma(1-A_i)(\log(Z_i+3)1\{Z_i \leq 0.5\}) + \eta_i.
	\end{align*}
	where $\eta_i = \sigma A_i \eps_{i,1} + (1-A_i)\eps_{i,2}$ and $(\eps_{i,1},\eps_{i,2})$ are jointly standard normal.  
	\item Let $Z$ be uniformly distributed on $[-2,2]$, $S_i = \sum_{j=1}^4\{Z_i \leq g_j\}$, and $(g_1,\cdots,g_4) = (-1,0,1,2)$. The outcome equation is 
	\begin{align*}
	Y_i = A_i \mu+A_im_{i,1}   + (1-A_i)m_{i,0} + \eta_i,
	\end{align*}
	where $m_{i,0} = \gamma Z_i^2 1\{|Z_i|\geq 1\} + \frac{\gamma}{4}(2 - Z_i^2)1\{|Z_i|<1\}$, $\eta_i = \sigma(1+Z_i^2)A_i\eps_{i,1} + (1+Z_i^2)(1-A_i)\eps_{i,2}$, and $(\eps_{i,1},\eps_{i,2})$ are mutually independent $T(3)/3$ distributed. 
	\item Let $Z_i$ be normally distributed with mean $0$ and variance $4$, $S_i = \sum_{j=1}^4\{Z_i \leq g_j\}$, $(g_1,\cdots,g_4) = (2\Phi^{-1}(0.25),2\Phi^{-1}(0.5),2\Phi^{-1}(0.75),\infty)$, and $\Phi(\cdot)$ is the standard normal CDF. The outcome equation is 
	\begin{align*}
	Y_i = A_i \mu+A_i m_{i,1} + (1-A_i)m_{i,0} + \eta_i,
	\end{align*}
	where $m_{i,0} = -\gamma Z_i^2/4$, $m_{i,1} = \gamma Z_i^2/4$, 
	\begin{align*}
	\eta_i = \sigma(1 + 0.5 \exp(-Z_i^2/2))A_i\eps_{i,1} + (1+0.5\exp(-Z_i^2/2))(1-A_i)\eps_{i,2},
	\end{align*}
	and $(\eps_{i,1},\eps_{i,2})$ are jointly standard normal.
\end{enumerate}

When $\pi = \frac{1}{2}$, for each DGP, we consider four randomization schemes:
\begin{enumerate}
	\item SRS: Treatment assignment is generated as in Example \ref{ex:srs}.
	\item WEI: Treatment assignment is generated as in Example \ref{ex:wei} with $\phi(x) = (1-x)/2$.
	\item BCD: Treatment assignment is generated as in Example \ref{ex:bcd} with $\lambda = 0.75$. 
	\item SBR: Treatment assignment is generated as in Example \ref{ex:sbr}. 
\end{enumerate}
When $\pi \neq 0.5$, we focus on SRS and SBR. We conduct the simulations with sample sizes $n=200$ and $400$. The numbers of simulation replications and bootstrap samples are 1000. Under the null, $\mu=0$ and the true parameters of interest are computed by simulations with $10^6$ sample size and $10^4$ replications. Under the alternative, we perturb the true values by $\mu = 1$ and $\mu = 0.75$ for $n = 200$ and $400$, respectively. We consider the following eight t-statistics. 

\begin{enumerate}
	\item ``s/naive": the point estimator is computed by the simple QR and its standard error $\sigma_{naive}$ is computed as 
	\begin{align}
	\label{eq:stddev1}
	\sigma^2_{naive} = & \frac{\tau(1-\tau) - \frac{1}{n}\sum_{i =1}^n\hat{m}^2_1(S_i,\tau) }{\pi\hat{f}_1^2(\hat{q}_1(\tau))}+\frac{\tau(1-\tau) - \frac{1}{n}\sum_{i =1}^n\hat{m}^2_0(S_i,\tau) }{(1-\pi)\hat{f}_0^2(\hat{q}_0(\tau))} \notag \\
	& + \frac{1}{n}\sum_{i =1}^n \pi(1-\pi)\left(\frac{\hat{m}_1(S_i,\tau)}{\pi \hat{f}_1(\hat{q}_1(\tau))} + \frac{\hat{m}_0(S_i,\tau)}{(1-\pi) \hat{f}_0(\hat{q}_0(\tau))}\right)^2 \notag \\
	& + \frac{1}{n}\sum_{i =1}^n \left(\frac{\hat{m}_1(S_i,\tau)}{ \hat{f}_1(\hat{q}_1(\tau))} - \frac{\hat{m}_0(S_i,\tau)}{ \hat{f}_0(\hat{q}_0(\tau))}\right)^2,
	\end{align}
	where $\hat{q}_j(\tau)$ is the $\tau$-the empirical quantile of $Y_i|A_i = j$, $$\hat{m}_{i,1}(s,\tau) = \frac{\sum_{i =1}^nA_i1\{S_i = s\}(\tau - 1\{Y_i \leq \hat{q}_1(\tau)\})}{n_1(s)},$$ 
	$$\hat{m}_{i,0}(s,\tau) = \frac{\sum_{i =1}^n(1-A_i)1\{S_i = s\}(\tau - 1\{Y_i \leq \hat{q}_0(\tau)\})}{n(s)-n_1(s)},$$ and for $j=0,1$, $\hat{f}_j(\cdot)$ is computed by the kernel density estimation using the observations $Y_i$ provided that $A_i = j$, bandwidth $h_j = 1.06\hat{\sigma}_jn_j^{-1/5}$, and the Gaussian kernel function, where $\hat{\sigma}_j$ is the standard deviation of the observations $Y_i$ provided that $A_i = j$, and $n_j = \sum_{i =1}^n \{A_i = j\}$, $j=0,1$. 
	\item ``s/adj": exactly the same as the ``s/naive" method with one difference: replacing $\pi(1-\pi)$ in $\sigma^2_{naive}$ by $\gamma(S_i)$.
	\item ``s/W": the point estimator is computed by the simple QR and its standard error $\sigma_B$ is computed by the weighted bootstrap procedure. The bootstrap weights $\{\xi_i\}_{i=1}^n$ are generated from the standard exponential distribution. Denote $\{\hat{\beta}^w_{1,b}\}_{b=1}^B$ as the collection of $B$ estimates obtained by the simple QR applied to the samples generated by the weighted bootstrap procedure. Then, 
	\begin{align*}
	\sigma_{B} = \frac{\hat{Q}(0.9) - \hat{Q}(0.1)}{\Phi^{-1}(0.9) - \Phi^{-1}(0.1)},
	\end{align*}
	where $\Phi(\cdot)$ is the standard normal CDF and $\hat{Q}(\tau)$ is the $\tau$-th empirical quantile of $\{\hat{\beta}^w_{1,b}\}_{b=1}^B$. 
	\item ``sfe/W": the same as above with one difference: the estimation method for both the original and bootstrap samples is the QR with strata fixed effects.
	\item ``ipw/W": the same as above with one difference: the estimation method for both the original and bootstrap samples is the inverse propensity score weighted QR.  
	\item  ``s/CA": the point estimator is computed by the simple QR and its standard error $\sigma_{CA}$ is computed by the covariate-adaptive bootstrap procedure. Denote $\{\hat{\beta}^*_{1,b}\}_{b=1}^B$ as the collection of $B$ estimates obtained by the simple QR applied to the samples generated by the covariate-adaptive bootstrap procedure. Then, 
	\begin{align*}
	\sigma_{CA} = \frac{\hat{Q}(0.9) - \hat{Q}(0.1)}{\Phi^{-1}(0.9) - \Phi^{-1}(0.1)},
	\end{align*}
	where $\hat{Q}(\tau)$ is the $\tau$-th empirical quantile of $\{\hat{\beta}^*_{1,b}\}_{b=1}^B$. 
	\item ``sfe/CA": the same as above with one difference:  the estimation method for both the original and bootstrap samples is the QR with strata fixed effects.
	\item ``ipw/CA": the same as above with one difference:  the estimation method for both the original and bootstrap samples is the inverse propensity score weighted QR.  
\end{enumerate}

\subsection{QTE, $H_0$, $\pi = 0.5$}
\begin{table}[H]
	\centering
	\caption{$H_0$, $n = 200$, $\tau = 0.25$}
	\begin{tabular}{l|l|cccccccc}
		\multicolumn{1}{c}{M} & \multicolumn{1}{c}{A}   & \multicolumn{1}{l}{s/naive} & \multicolumn{1}{c}{s/adj} & \multicolumn{1}{c}{s/W} & \multicolumn{1}{c}{sfe/W} & \multicolumn{1}{c}{ipw/W} & \multicolumn{1}{c}{s/CA } & \multicolumn{1}{c}{sfe/CA} & \multicolumn{1}{c}{ipw/CA} \\ \hline 
		1     & SRS   & 0.042 & 0.042 & 0.051 & 0.039 & 0.047 & 0.046 & 0.044 & 0.046 \\
		& WEI   & 0.011 & 0.038 & 0.018 & 0.043 & 0.046 & 0.037 & 0.047 & 0.047 \\
		& BCD   & 0.004 & 0.041 & 0.010 & 0.043 & 0.043 & 0.045 & 0.048 & 0.048 \\
		& SBR   & 0.003 & 0.047 & 0.003 & 0.047 & 0.054 & 0.049 & 0.046 & 0.046 \\  \hline 
		2     & SRS   & 0.045 & 0.045 & 0.060 & 0.062 & 0.066 & 0.056 & 0.069 & 0.069 \\
		& WEI   & 0.023 & 0.037 & 0.049 & 0.056 & 0.066 & 0.068 & 0.064 & 0.068 \\
		& BCD   & 0.021 & 0.037 & 0.032 & 0.049 & 0.057 & 0.063 & 0.059 & 0.057 \\
		& SBR   & 0.025 & 0.042 & 0.037 & 0.050 & 0.054 & 0.057 & 0.054 & 0.053 \\  \hline 
		3     & SRS   & 0.042 & 0.042 & 0.045 & 0.045 & 0.054 & 0.055 & 0.044 & 0.058 \\
		& WEI   & 0.042 & 0.043 & 0.037 & 0.044 & 0.045 & 0.045 & 0.043 & 0.045 \\
		& BCD   & 0.052 & 0.056 & 0.044 & 0.050 & 0.057 & 0.057 & 0.057 & 0.055 \\
		& SBR   & 0.046 & 0.053 & 0.041 & 0.043 & 0.048 & 0.052 & 0.048 & 0.047 \\  \hline 
		4     & SRS   & 0.054 & 0.054 & 0.048 & 0.046 & 0.049 & 0.046 & 0.043 & 0.048 \\
		& WEI   & 0.050 & 0.051 & 0.045 & 0.035 & 0.047 & 0.051 & 0.043 & 0.055 \\
		& BCD   & 0.056 & 0.059 & 0.040 & 0.030 & 0.049 & 0.047 & 0.044 & 0.048 \\
		& SBR   & 0.061 & 0.065 & 0.044 & 0.032 & 0.053 & 0.057 & 0.051 & 0.053 \\
	\end{tabular}%
	\label{tab:200_1}%
\end{table}%

\begin{table}[H]
	\centering
	\caption{$H_0$, $n = 200$, $\tau = 0.5$}
	\begin{tabular}{l|l|cccccccc}
		\multicolumn{1}{c}{M} & \multicolumn{1}{c}{A}   & \multicolumn{1}{l}{s/naive} & \multicolumn{1}{c}{s/adj} & \multicolumn{1}{c}{s/W} & \multicolumn{1}{c}{sfe/W} & \multicolumn{1}{c}{ipw/W} & \multicolumn{1}{c}{s/CA } & \multicolumn{1}{c}{sfe/CA} & \multicolumn{1}{c}{ipw/CA} \\ \hline 
		1     & SRS   & 0.045 & 0.045 & 0.047 & 0.043 & 0.044 & 0.044 & 0.039 & 0.039 \\
		& WEI   & 0.012 & 0.040 & 0.014 & 0.044 & 0.043 & 0.037 & 0.041 & 0.035 \\
		& BCD   & 0.002 & 0.057 & 0.003 & 0.040 & 0.041 & 0.044 & 0.039 & 0.039 \\
		& SBR   & 0.001 & 0.057 & 0.001 & 0.045 & 0.046 & 0.045 & 0.045 & 0.044 \\  \hline 
		2     & SRS   & 0.045 & 0.045 & 0.057 & 0.066 & 0.061 & 0.048 & 0.064 & 0.066 \\
		& WEI   & 0.033 & 0.065 & 0.037 & 0.056 & 0.065 & 0.065 & 0.056 & 0.061 \\
		& BCD   & 0.022 & 0.062 & 0.027 & 0.048 & 0.056 & 0.057 & 0.057 & 0.054 \\
		& SBR   & 0.017 & 0.050 & 0.017 & 0.040 & 0.046 & 0.048 & 0.048 & 0.046 \\  \hline 
		3     & SRS   & 0.004 & 0.004 & 0.047 & 0.045 & 0.052 & 0.052 & 0.047 & 0.053 \\
		& WEI   & 0.006 & 0.006 & 0.045 & 0.050 & 0.058 & 0.052 & 0.053 & 0.057 \\
		& BCD   & 0.010 & 0.010 & 0.045 & 0.050 & 0.051 & 0.050 & 0.050 & 0.053 \\
		& SBR   & 0.008 & 0.011 & 0.048 & 0.048 & 0.053 & 0.046 & 0.051 & 0.047 \\  \hline 
		4     & SRS   & 0.013 & 0.013 & 0.050 & 0.036 & 0.051 & 0.055 & 0.035 & 0.043 \\
		& WEI   & 0.011 & 0.011 & 0.043 & 0.033 & 0.051 & 0.049 & 0.043 & 0.052 \\
		& BCD   & 0.013 & 0.013 & 0.049 & 0.041 & 0.053 & 0.055 & 0.047 & 0.052 \\
		& SBR   & 0.013 & 0.013 & 0.040 & 0.033 & 0.047 & 0.046 & 0.044 & 0.045 \\
	\end{tabular}%
	\label{tab:200_2_2}%
\end{table}%

\begin{table}[H]
	\centering
	\caption{$H_0$,  $n = 200$, $\tau = 0.75$}
	\begin{tabular}{l|l|cccccccc}
		\multicolumn{1}{c}{M} & \multicolumn{1}{c}{A}   & \multicolumn{1}{l}{s/naive} & \multicolumn{1}{c}{s/adj} & \multicolumn{1}{c}{s/W} & \multicolumn{1}{c}{sfe/W} & \multicolumn{1}{c}{ipw/W} & \multicolumn{1}{c}{s/CA } & \multicolumn{1}{c}{sfe/CA} & \multicolumn{1}{c}{ipw/CA} \\ \hline 
		1     & SRS   & 0.052 & 0.052 & 0.053 & 0.044 & 0.044 & 0.048 & 0.041 & 0.042 \\
		& WEI   & 0.012 & 0.042 & 0.014 & 0.043 & 0.046 & 0.037 & 0.039 & 0.045 \\
		& BCD   & 0.002 & 0.047 & 0.002 & 0.051 & 0.054 & 0.055 & 0.053 & 0.053 \\
		& SBR   & 0.001 & 0.026 & 0.003 & 0.030 & 0.035 & 0.030 & 0.033 & 0.035 \\  \hline 
		2     & SRS   & 0.052 & 0.052 & 0.066 & 0.057 & 0.058 & 0.053 & 0.048 & 0.058 \\
		& WEI   & 0.021 & 0.045 & 0.027 & 0.047 & 0.052 & 0.057 & 0.051 & 0.054 \\
		& BCD   & 0.013 & 0.046 & 0.025 & 0.051 & 0.060 & 0.067 & 0.061 & 0.060 \\
		& SBR   & 0.008 & 0.036 & 0.012 & 0.037 & 0.046 & 0.046 & 0.046 & 0.050 \\  \hline 
		3     & SRS   & 0.058 & 0.058 & 0.048 & 0.054 & 0.047 & 0.058 & 0.054 & 0.051 \\
		& WEI   & 0.053 & 0.055 & 0.041 & 0.044 & 0.047 & 0.047 & 0.048 & 0.046 \\
		& BCD   & 0.042 & 0.043 & 0.026 & 0.026 & 0.033 & 0.033 & 0.032 & 0.034 \\
		& SBR   & 0.048 & 0.052 & 0.040 & 0.036 & 0.046 & 0.051 & 0.043 & 0.048 \\  \hline 
		4     & SRS   & 0.044 & 0.044 & 0.057 & 0.059 & 0.062 & 0.053 & 0.051 & 0.065 \\
		& WEI   & 0.034 & 0.034 & 0.044 & 0.029 & 0.053 & 0.048 & 0.044 & 0.054 \\
		& BCD   & 0.029 & 0.032 & 0.040 & 0.019 & 0.045 & 0.047 & 0.043 & 0.047 \\
		& SBR   & 0.034 & 0.037 & 0.042 & 0.025 & 0.051 & 0.055 & 0.049 & 0.051 \\
	\end{tabular}%
	\label{tab:200_3}%
\end{table}%

\begin{table}[H]
	\centering
	\caption{$H_0$,  $n = 400$, $\tau = 0.25$}
	\begin{tabular}{l|l|cccccccc}
		\multicolumn{1}{c}{M} & \multicolumn{1}{c}{A}   & \multicolumn{1}{l}{s/naive} & \multicolumn{1}{c}{s/adj} & \multicolumn{1}{c}{s/W} & \multicolumn{1}{c}{sfe/W} & \multicolumn{1}{c}{ipw/W} & \multicolumn{1}{c}{s/CA } & \multicolumn{1}{c}{sfe/CA} & \multicolumn{1}{c}{ipw/CA} \\ \hline 
		1     & SRS   & 0.047 & 0.047 & 0.053 & 0.041 & 0.039 & 0.049 & 0.040 & 0.040 \\
		& WEI   & 0.009 & 0.043 & 0.017 & 0.041 & 0.042 & 0.045 & 0.044 & 0.043 \\
		& BCD   & 0.002 & 0.042 & 0.003 & 0.037 & 0.040 & 0.035 & 0.036 & 0.037 \\
		& SBR   & 0.002 & 0.043 & 0.004 & 0.034 & 0.034 & 0.036 & 0.032 & 0.030 \\ \hline 
		2     & SRS   & 0.046 & 0.046 & 0.056 & 0.059 & 0.059 & 0.055 & 0.057 & 0.059 \\
		& WEI   & 0.035 & 0.046 & 0.046 & 0.056 & 0.062 & 0.065 & 0.061 & 0.060 \\
		& BCD   & 0.030 & 0.044 & 0.037 & 0.055 & 0.065 & 0.060 & 0.060 & 0.057 \\
		& SBR   & 0.026 & 0.049 & 0.042 & 0.058 & 0.067 & 0.063 & 0.062 & 0.066 \\ \hline 
		3     & SRS   & 0.044 & 0.044 & 0.039 & 0.041 & 0.042 & 0.042 & 0.041 & 0.043 \\
		& WEI   & 0.042 & 0.045 & 0.048 & 0.041 & 0.048 & 0.051 & 0.046 & 0.049 \\
		& BCD   & 0.039 & 0.040 & 0.041 & 0.040 & 0.044 & 0.046 & 0.047 & 0.048 \\
		& SBR   & 0.048 & 0.051 & 0.046 & 0.048 & 0.052 & 0.056 & 0.056 & 0.055 \\ \hline 
		4     & SRS   & 0.056 & 0.056 & 0.039 & 0.042 & 0.041 & 0.041 & 0.043 & 0.042 \\
		& WEI   & 0.052 & 0.055 & 0.038 & 0.034 & 0.045 & 0.042 & 0.044 & 0.044 \\
		& BCD   & 0.054 & 0.058 & 0.040 & 0.026 & 0.045 & 0.044 & 0.045 & 0.043 \\
		& SBR   & 0.061 & 0.068 & 0.049 & 0.027 & 0.047 & 0.054 & 0.055 & 0.051 \\
	\end{tabular}%
	\label{tab:400_1}%
\end{table}%

\begin{table}[H]
	\centering
	\caption{$H_0$,  $n = 400$, $\tau = 0.5$}
	\begin{tabular}{l|l|cccccccc}
		\multicolumn{1}{c}{M} & \multicolumn{1}{c}{A}   & \multicolumn{1}{l}{s/naive} & \multicolumn{1}{c}{s/adj} & \multicolumn{1}{c}{s/W} & \multicolumn{1}{c}{sfe/W} & \multicolumn{1}{c}{ipw/W} & \multicolumn{1}{c}{s/CA } & \multicolumn{1}{c}{sfe/CA} & \multicolumn{1}{c}{ipw/CA} \\ \hline 
		1     & SRS   & 0.042 & 0.042 & 0.054 & 0.046 & 0.040 & 0.046 & 0.050 & 0.041 \\
		& WEI   & 0.010 & 0.049 & 0.008 & 0.047 & 0.047 & 0.046 & 0.043 & 0.042 \\
		& BCD   & 0.003 & 0.045 & 0.002 & 0.043 & 0.043 & 0.035 & 0.039 & 0.040 \\
		& SBR   & 0.002 & 0.046 & 0.000 & 0.035 & 0.037 & 0.036 & 0.036 & 0.037 \\ \hline 
		2     & SRS   & 0.050 & 0.050 & 0.055 & 0.049 & 0.047 & 0.051 & 0.052 & 0.050 \\
		& WEI   & 0.018 & 0.048 & 0.025 & 0.041 & 0.046 & 0.045 & 0.048 & 0.045 \\
		& BCD   & 0.011 & 0.042 & 0.011 & 0.041 & 0.046 & 0.045 & 0.046 & 0.043 \\
		& SBR   & 0.017 & 0.051 & 0.014 & 0.042 & 0.050 & 0.053 & 0.047 & 0.050 \\ \hline 
		3     & SRS   & 0.012 & 0.012 & 0.043 & 0.046 & 0.048 & 0.046 & 0.050 & 0.050 \\
		& WEI   & 0.014 & 0.016 & 0.057 & 0.055 & 0.060 & 0.055 & 0.058 & 0.057 \\
		& BCD   & 0.013 & 0.013 & 0.055 & 0.059 & 0.061 & 0.051 & 0.053 & 0.052 \\
		& SBR   & 0.006 & 0.006 & 0.040 & 0.040 & 0.039 & 0.038 & 0.039 & 0.038 \\ \hline 
		4     & SRS   & 0.019 & 0.019 & 0.056 & 0.052 & 0.064 & 0.056 & 0.051 & 0.061 \\
		& WEI   & 0.018 & 0.018 & 0.060 & 0.046 & 0.065 & 0.064 & 0.062 & 0.066 \\
		& BCD   & 0.015 & 0.015 & 0.057 & 0.046 & 0.066 & 0.063 & 0.059 & 0.067 \\
		& SBR   & 0.021 & 0.021 & 0.057 & 0.043 & 0.060 & 0.062 & 0.062 & 0.062 \\
	\end{tabular}%
	\label{tab:400_2_2}%
\end{table}%

\begin{table}[H]
	\centering
	\caption{$H_0$,  $n = 400$, $\tau = 0.75$}
	\begin{tabular}{l|l|cccccccc}
		\multicolumn{1}{c}{M} & \multicolumn{1}{c}{A}   & \multicolumn{1}{l}{s/naive} & \multicolumn{1}{c}{s/adj} & \multicolumn{1}{c}{s/W} & \multicolumn{1}{c}{sfe/W} & \multicolumn{1}{c}{ipw/W} & \multicolumn{1}{c}{s/CA } & \multicolumn{1}{c}{sfe/CA} & \multicolumn{1}{c}{ipw/CA} \\ \hline 
		1     & SRS   & 0.051 & 0.051 & 0.056 & 0.055 & 0.056 & 0.052 & 0.055 & 0.054 \\
		& WEI   & 0.007 & 0.041 & 0.014 & 0.055 & 0.053 & 0.051 & 0.050 & 0.051 \\
		& BCD   & 0.006 & 0.038 & 0.004 & 0.046 & 0.048 & 0.041 & 0.042 & 0.046 \\
		& SBR   & 0.004 & 0.033 & 0.002 & 0.044 & 0.043 & 0.042 & 0.043 & 0.042 \\ \hline 
		2     & SRS   & 0.048 & 0.048 & 0.073 & 0.055 & 0.061 & 0.060 & 0.057 & 0.059 \\
		& WEI   & 0.020 & 0.039 & 0.024 & 0.046 & 0.053 & 0.048 & 0.051 & 0.053 \\
		& BCD   & 0.012 & 0.048 & 0.020 & 0.050 & 0.051 & 0.057 & 0.055 & 0.051 \\
		& SBR   & 0.011 & 0.047 & 0.014 & 0.046 & 0.052 & 0.050 & 0.052 & 0.052 \\ \hline 
		3     & SRS   & 0.054 & 0.054 & 0.050 & 0.045 & 0.052 & 0.049 & 0.044 & 0.052 \\
		& WEI   & 0.053 & 0.055 & 0.049 & 0.047 & 0.053 & 0.050 & 0.049 & 0.054 \\
		& BCD   & 0.059 & 0.063 & 0.038 & 0.041 & 0.045 & 0.044 & 0.043 & 0.043 \\
		& SBR   & 0.049 & 0.051 & 0.042 & 0.044 & 0.043 & 0.049 & 0.049 & 0.049 \\ \hline 
		4     & SRS   & 0.054 & 0.054 & 0.057 & 0.053 & 0.063 & 0.055 & 0.056 & 0.063 \\
		& WEI   & 0.047 & 0.051 & 0.055 & 0.043 & 0.064 & 0.055 & 0.061 & 0.059 \\
		& BCD   & 0.049 & 0.051 & 0.054 & 0.033 & 0.063 & 0.062 & 0.056 & 0.063 \\
		& SBR   & 0.046 & 0.048 & 0.047 & 0.026 & 0.051 & 0.057 & 0.056 & 0.053 \\
	\end{tabular}%
	\label{tab:400_3}%
\end{table}%

\subsection{QTE, $H_1$, $\pi = 0.5$}

\begin{table}[H]
	\centering
	\caption{$H_1$,  $n = 200$, $\tau = 0.25$}
	\begin{tabular}{l|l|cccccccc}
		\multicolumn{1}{c}{M} & \multicolumn{1}{c}{A}   & \multicolumn{1}{l}{s/naive} & \multicolumn{1}{c}{s/adj} & \multicolumn{1}{c}{s/W} & \multicolumn{1}{c}{sfe/W} & \multicolumn{1}{c}{ipw/W} & \multicolumn{1}{c}{s/CA } & \multicolumn{1}{c}{sfe/CA} & \multicolumn{1}{c}{ipw/CA} \\ \hline 
		1     & SRS   & 0.191 & 0.191 & 0.203 & 0.354 & 0.356 & 0.205 & 0.340 & 0.342 \\
		& WEI   & 0.126 & 0.257 & 0.147 & 0.359 & 0.358 & 0.279 & 0.345 & 0.350 \\
		& BCD   & 0.105 & 0.372 & 0.122 & 0.379 & 0.375 & 0.361 & 0.369 & 0.365 \\
		& SBR   & 0.099 & 0.400 & 0.114 & 0.378 & 0.382 & 0.411 & 0.375 & 0.368 \\  \hline 
		2     & SRS   & 0.284 & 0.284 & 0.315 & 0.352 & 0.376 & 0.319 & 0.345 & 0.378 \\
		& WEI   & 0.270 & 0.319 & 0.314 & 0.356 & 0.364 & 0.359 & 0.363 & 0.369 \\
		& BCD   & 0.282 & 0.333 & 0.304 & 0.361 & 0.375 & 0.390 & 0.385 & 0.383 \\
		& SBR   & 0.290 & 0.346 & 0.296 & 0.335 & 0.361 & 0.387 & 0.358 & 0.356 \\  \hline 
		3     & SRS   & 0.712 & 0.712 & 0.694 & 0.688 & 0.698 & 0.704 & 0.677 & 0.686 \\
		& WEI   & 0.701 & 0.707 & 0.678 & 0.685 & 0.680 & 0.699 & 0.687 & 0.674 \\
		& BCD   & 0.712 & 0.720 & 0.673 & 0.686 & 0.695 & 0.699 & 0.698 & 0.698 \\
		& SBR   & 0.672 & 0.684 & 0.659 & 0.639 & 0.647 & 0.673 & 0.647 & 0.638 \\  \hline 
		4     & SRS   & 0.166 & 0.166 & 0.124 & 0.112 & 0.132 & 0.135 & 0.131 & 0.128 \\
		& WEI   & 0.166 & 0.170 & 0.126 & 0.098 & 0.125 & 0.144 & 0.139 & 0.133 \\
		& BCD   & 0.165 & 0.176 & 0.126 & 0.094 & 0.155 & 0.157 & 0.145 & 0.157 \\
		& SBR   & 0.167 & 0.175 & 0.122 & 0.088 & 0.139 & 0.145 & 0.133 & 0.140 \\
	\end{tabular}%
	\label{tab:200_1'}%
\end{table}%

\begin{table}[H]
	\centering
	\caption{$H_1$,  $n = 200$, $\tau = 0.5$}
	\begin{tabular}{l|l|cccccccc}
		\multicolumn{1}{c}{M} & \multicolumn{1}{c}{A}   & \multicolumn{1}{l}{s/naive} & \multicolumn{1}{c}{s/adj} & \multicolumn{1}{c}{s/W} & \multicolumn{1}{c}{sfe/W} & \multicolumn{1}{c}{ipw/W} & \multicolumn{1}{c}{s/CA } & \multicolumn{1}{c}{sfe/CA} & \multicolumn{1}{c}{ipw/CA} \\ \hline 
		1     & SRS   & 0.183 & 0.183 & 0.193 & 0.443 & 0.441 & 0.200 & 0.431 & 0.429 \\
		& WEI   & 0.116 & 0.295 & 0.138 & 0.442 & 0.447 & 0.298 & 0.437 & 0.436 \\
		& BCD   & 0.072 & 0.472 & 0.095 & 0.450 & 0.453 & 0.434 & 0.446 & 0.448 \\
		& SBR   & 0.085 & 0.485 & 0.099 & 0.463 & 0.460 & 0.457 & 0.453 & 0.448 \\  \hline 
		2     & SRS   & 0.267 & 0.267 & 0.256 & 0.359 & 0.366 & 0.265 & 0.358 & 0.371 \\
		& WEI   & 0.248 & 0.346 & 0.247 & 0.358 & 0.394 & 0.346 & 0.378 & 0.389 \\
		& BCD   & 0.229 & 0.402 & 0.233 & 0.358 & 0.396 & 0.388 & 0.395 & 0.392 \\
		& SBR   & 0.232 & 0.404 & 0.234 & 0.365 & 0.392 & 0.399 & 0.401 & 0.391 \\  \hline 
		3     & SRS   & 0.797 & 0.797 & 0.904 & 0.897 & 0.916 & 0.902 & 0.897 & 0.913 \\
		& WEI   & 0.802 & 0.807 & 0.907 & 0.903 & 0.909 & 0.913 & 0.902 & 0.906 \\
		& BCD   & 0.796 & 0.804 & 0.902 & 0.910 & 0.911 & 0.908 & 0.911 & 0.906 \\
		& SBR   & 0.771 & 0.774 & 0.897 & 0.896 & 0.901 & 0.899 & 0.894 & 0.899 \\  \hline 
		4     & SRS   & 0.176 & 0.176 & 0.312 & 0.269 & 0.317 & 0.316 & 0.297 & 0.316 \\
		& WEI   & 0.171 & 0.175 & 0.289 & 0.255 & 0.307 & 0.309 & 0.297 & 0.298 \\
		& BCD   & 0.169 & 0.174 & 0.299 & 0.262 & 0.313 & 0.329 & 0.311 & 0.316 \\
		& SBR   & 0.163 & 0.165 & 0.283 & 0.255 & 0.304 & 0.302 & 0.298 & 0.298 \\
	\end{tabular}%
	\label{tab:200_2_2'}%
\end{table}%

\begin{table}[H]
	\centering
	\caption{$H_1$,  $n = 200$, $\tau = 0.75$}
	\begin{tabular}{l|l|cccccccc}
		\multicolumn{1}{c}{M} & \multicolumn{1}{c}{A}   & \multicolumn{1}{l}{s/naive} & \multicolumn{1}{c}{s/adj} & \multicolumn{1}{c}{s/W} & \multicolumn{1}{c}{sfe/W} & \multicolumn{1}{c}{ipw/W} & \multicolumn{1}{c}{s/CA } & \multicolumn{1}{c}{sfe/CA} & \multicolumn{1}{c}{ipw/CA} \\ \hline 
		1     & SRS   & 0.198 & 0.198 & 0.215 & 0.362 & 0.358 & 0.216 & 0.353 & 0.355 \\
		& WEI   & 0.143 & 0.293 & 0.153 & 0.361 & 0.368 & 0.315 & 0.362 & 0.364 \\
		& BCD   & 0.108 & 0.377 & 0.131 & 0.356 & 0.360 & 0.355 & 0.353 & 0.353 \\
		& SBR   & 0.079 & 0.386 & 0.105 & 0.397 & 0.396 & 0.381 & 0.403 & 0.386 \\  \hline 
		2     & SRS   & 0.268 & 0.268 & 0.315 & 0.386 & 0.439 & 0.322 & 0.391 & 0.434 \\
		& WEI   & 0.238 & 0.339 & 0.285 & 0.396 & 0.430 & 0.390 & 0.417 & 0.428 \\
		& BCD   & 0.209 & 0.407 & 0.263 & 0.398 & 0.428 & 0.425 & 0.428 & 0.418 \\
		& SBR   & 0.206 & 0.427 & 0.267 & 0.439 & 0.455 & 0.450 & 0.465 & 0.456 \\  \hline 
		3     & SRS   & 0.698 & 0.698 & 0.607 & 0.594 & 0.619 & 0.634 & 0.609 & 0.622 \\
		& WEI   & 0.668 & 0.673 & 0.607 & 0.606 & 0.616 & 0.631 & 0.623 & 0.624 \\
		& BCD   & 0.690 & 0.698 & 0.607 & 0.612 & 0.616 & 0.635 & 0.618 & 0.621 \\
		& SBR   & 0.669 & 0.675 & 0.596 & 0.614 & 0.633 & 0.617 & 0.631 & 0.630 \\  \hline 
		4     & SRS   & 0.163 & 0.163 & 0.158 & 0.122 & 0.167 & 0.173 & 0.140 & 0.169 \\
		& WEI   & 0.144 & 0.152 & 0.152 & 0.105 & 0.175 & 0.169 & 0.152 & 0.178 \\
		& BCD   & 0.133 & 0.138 & 0.151 & 0.085 & 0.170 & 0.177 & 0.173 & 0.172 \\
		& SBR   & 0.146 & 0.154 & 0.143 & 0.090 & 0.175 & 0.171 & 0.177 & 0.180 \\
	\end{tabular}%
	\label{tab:200_3'}%
\end{table}%

\begin{table}[H]
	\centering
	\caption{$H_1$,  $n = 400$, $\tau = 0.25$}
	\begin{tabular}{l|l|cccccccc}
		\multicolumn{1}{c}{M} & \multicolumn{1}{c}{A}   & \multicolumn{1}{l}{s/naive} & \multicolumn{1}{c}{s/adj} & \multicolumn{1}{c}{s/W} & \multicolumn{1}{c}{sfe/W} & \multicolumn{1}{c}{ipw/W} & \multicolumn{1}{c}{s/CA } & \multicolumn{1}{c}{sfe/CA} & \multicolumn{1}{c}{ipw/CA} \\ \hline 
		1     & SRS   & 0.206 & 0.206 & 0.229 & 0.403 & 0.417 & 0.231 & 0.401 & 0.405 \\
		& WEI   & 0.163 & 0.332 & 0.173 & 0.408 & 0.413 & 0.337 & 0.408 & 0.413 \\
		& BCD   & 0.121 & 0.430 & 0.143 & 0.420 & 0.422 & 0.421 & 0.419 & 0.413 \\
		& SBR   & 0.128 & 0.451 & 0.144 & 0.428 & 0.429 & 0.458 & 0.426 & 0.423 \\ \hline 
		2     & SRS   & 0.312 & 0.312 & 0.345 & 0.422 & 0.415 & 0.351 & 0.416 & 0.416 \\
		& WEI   & 0.312 & 0.352 & 0.332 & 0.405 & 0.424 & 0.378 & 0.408 & 0.426 \\
		& BCD   & 0.299 & 0.378 & 0.333 & 0.392 & 0.405 & 0.403 & 0.415 & 0.413 \\
		& SBR   & 0.330 & 0.389 & 0.345 & 0.401 & 0.407 & 0.426 & 0.410 & 0.406 \\ \hline 
		3     & SRS   & 0.763 & 0.763 & 0.734 & 0.730 & 0.740 & 0.738 & 0.732 & 0.738 \\
		& WEI   & 0.763 & 0.764 & 0.739 & 0.739 & 0.748 & 0.744 & 0.746 & 0.746 \\
		& BCD   & 0.781 & 0.783 & 0.760 & 0.760 & 0.768 & 0.772 & 0.774 & 0.767 \\
		& SBR   & 0.766 & 0.773 & 0.745 & 0.739 & 0.744 & 0.763 & 0.751 & 0.744 \\ \hline 
		4     & SRS   & 0.177 & 0.177 & 0.129 & 0.108 & 0.136 & 0.127 & 0.121 & 0.133 \\
		& WEI   & 0.170 & 0.176 & 0.129 & 0.096 & 0.139 & 0.139 & 0.131 & 0.143 \\
		& BCD   & 0.178 & 0.185 & 0.132 & 0.089 & 0.141 & 0.141 & 0.139 & 0.138 \\
		& SBR   & 0.180 & 0.186 & 0.129 & 0.102 & 0.134 & 0.147 & 0.135 & 0.133 \\
	\end{tabular}%
	\label{tab:400_1'}%
\end{table}%

\begin{table}[H]
	\centering
	\caption{$H_1$,  $n = 400$, $\tau = 0.5$}
	\begin{tabular}{l|l|cccccccc}
		\multicolumn{1}{c}{M} & \multicolumn{1}{c}{A}   & \multicolumn{1}{l}{s/naive} & \multicolumn{1}{c}{s/adj} & \multicolumn{1}{c}{s/W} & \multicolumn{1}{c}{sfe/W} & \multicolumn{1}{c}{ipw/W} & \multicolumn{1}{c}{s/CA } & \multicolumn{1}{c}{sfe/CA} & \multicolumn{1}{c}{ipw/CA} \\ \hline 
		1     & SRS   & 0.218 & 0.218 & 0.232 & 0.504 & 0.502 & 0.235 & 0.497 & 0.502 \\
		& WEI   & 0.147 & 0.356 & 0.160 & 0.503 & 0.503 & 0.350 & 0.498 & 0.507 \\
		& BCD   & 0.089 & 0.526 & 0.117 & 0.498 & 0.502 & 0.493 & 0.495 & 0.496 \\
		& SBR   & 0.089 & 0.550 & 0.109 & 0.520 & 0.518 & 0.524 & 0.526 & 0.519 \\ \hline 
		2     & SRS   & 0.301 & 0.301 & 0.309 & 0.402 & 0.426 & 0.306 & 0.413 & 0.423 \\
		& WEI   & 0.287 & 0.387 & 0.281 & 0.402 & 0.418 & 0.372 & 0.411 & 0.420 \\
		& BCD   & 0.268 & 0.451 & 0.262 & 0.400 & 0.443 & 0.434 & 0.434 & 0.441 \\
		& SBR   & 0.260 & 0.433 & 0.252 & 0.403 & 0.421 & 0.418 & 0.431 & 0.420 \\ \hline 
		3     & SRS   & 0.897 & 0.897 & 0.956 & 0.957 & 0.956 & 0.957 & 0.956 & 0.957 \\
		& WEI   & 0.892 & 0.892 & 0.954 & 0.944 & 0.948 & 0.951 & 0.942 & 0.948 \\
		& BCD   & 0.887 & 0.889 & 0.952 & 0.949 & 0.954 & 0.957 & 0.954 & 0.956 \\
		& SBR   & 0.900 & 0.902 & 0.954 & 0.954 & 0.954 & 0.958 & 0.962 & 0.957 \\ \hline 
		4     & SRS   & 0.234 & 0.234 & 0.345 & 0.317 & 0.351 & 0.353 & 0.339 & 0.343 \\
		& WEI   & 0.222 & 0.224 & 0.336 & 0.326 & 0.352 & 0.352 & 0.335 & 0.358 \\
		& BCD   & 0.226 & 0.230 & 0.346 & 0.321 & 0.349 & 0.368 & 0.359 & 0.365 \\
		& SBR   & 0.238 & 0.242 & 0.369 & 0.350 & 0.380 & 0.379 & 0.374 & 0.377 \\
	\end{tabular}%
	\label{tab:400_2_2'}%
\end{table}%

\begin{table}[H]
	\centering
	\caption{$H_1$,  $n = 400$, $\tau = 0.75$}
	\begin{tabular}{l|l|cccccccc}
		\multicolumn{1}{c}{M} & \multicolumn{1}{c}{A}   & \multicolumn{1}{l}{s/naive} & \multicolumn{1}{c}{s/adj} & \multicolumn{1}{c}{s/W} & \multicolumn{1}{c}{sfe/W} & \multicolumn{1}{c}{ipw/W} & \multicolumn{1}{c}{s/CA } & \multicolumn{1}{c}{sfe/CA} & \multicolumn{1}{c}{ipw/CA} \\ \hline 
		1     & SRS   & 0.218 & 0.218 & 0.237 & 0.430 & 0.435 & 0.242 & 0.438 & 0.435 \\
		& WEI   & 0.163 & 0.321 & 0.176 & 0.441 & 0.437 & 0.344 & 0.433 & 0.432 \\
		& BCD   & 0.136 & 0.422 & 0.152 & 0.421 & 0.420 & 0.417 & 0.417 & 0.416 \\
		& SBR   & 0.103 & 0.446 & 0.124 & 0.459 & 0.459 & 0.448 & 0.463 & 0.461 \\ \hline 
		2     & SRS   & 0.300 & 0.300 & 0.337 & 0.445 & 0.479 & 0.335 & 0.449 & 0.479 \\
		& WEI   & 0.258 & 0.369 & 0.313 & 0.446 & 0.465 & 0.414 & 0.453 & 0.463 \\
		& BCD   & 0.247 & 0.462 & 0.295 & 0.451 & 0.476 & 0.483 & 0.481 & 0.477 \\
		& SBR   & 0.227 & 0.444 & 0.276 & 0.472 & 0.490 & 0.471 & 0.496 & 0.492 \\ \hline 
		3     & SRS   & 0.763 & 0.763 & 0.710 & 0.702 & 0.707 & 0.712 & 0.701 & 0.715 \\
		& WEI   & 0.773 & 0.776 & 0.696 & 0.701 & 0.700 & 0.720 & 0.709 & 0.706 \\
		& BCD   & 0.753 & 0.755 & 0.705 & 0.716 & 0.720 & 0.720 & 0.717 & 0.726 \\
		& SBR   & 0.746 & 0.750 & 0.684 & 0.699 & 0.705 & 0.692 & 0.709 & 0.708 \\ \hline 
		4     & SRS   & 0.209 & 0.209 & 0.199 & 0.140 & 0.221 & 0.208 & 0.149 & 0.221 \\
		& WEI   & 0.201 & 0.208 & 0.191 & 0.110 & 0.203 & 0.206 & 0.178 & 0.204 \\
		& BCD   & 0.195 & 0.200 & 0.199 & 0.121 & 0.213 & 0.224 & 0.213 & 0.220 \\
		& SBR   & 0.198 & 0.203 & 0.198 & 0.114 & 0.229 & 0.214 & 0.230 & 0.225 \\
	\end{tabular}%
	\label{tab:400_3'}%
\end{table}%

\subsection{QTE, $H_0$, $\pi = 0.7$}
\begin{table}[H]
	\centering
	\caption{$H_0$,  $n = 200$, $\tau = 0.25$}
	\begin{tabular}{l|l|cccccccc}
		\multicolumn{1}{c}{M} & \multicolumn{1}{c}{A}   & \multicolumn{1}{l}{s/naive} & \multicolumn{1}{c}{s/adj} & \multicolumn{1}{c}{s/W} & \multicolumn{1}{c}{sfe/W} & \multicolumn{1}{c}{ipw/W} & \multicolumn{1}{c}{s/CA } & \multicolumn{1}{c}{sfe/CA} & \multicolumn{1}{c}{ipw/CA} \\ \hline 
		1     & SRS   & 0.042 & 0.042 & 0.046 & 0.042 & 0.036 & 0.036 & 0.039 & 0.039 \\
		& SBR   & 0.002 & 0.014 & 0.005 & 0.053 & 0.052 & 0.049 & 0.050 & 0.047 \\ \hline 
		2     & SRS   & 0.037 & 0.037 & 0.051 & 0.059 & 0.057 & 0.061 & 0.057 & 0.064 \\
		& SBR   & 0.032 & 0.036 & 0.042 & 0.046 & 0.048 & 0.055 & 0.055 & 0.055 \\ \hline 
		3     & SRS   & 0.046 & 0.046 & 0.046 & 0.047 & 0.039 & 0.045 & 0.049 & 0.043 \\
		& SBR   & 0.040 & 0.044 & 0.032 & 0.031 & 0.034 & 0.041 & 0.037 & 0.040 \\ \hline 
		4     & SRS   & 0.098 & 0.098 & 0.067 & 0.075 & 0.069 & 0.062 & 0.057 & 0.066 \\
		& SBR   & 0.057 & 0.066 & 0.043 & 0.016 & 0.062 & 0.061 & 0.066 & 0.064 \\
	\end{tabular}%
	\label{tab:200_1_70}%
\end{table}%

\begin{table}[H]
	\centering
	\caption{$H_0$,  $n = 200$, $\tau = 0.5$}
	\begin{tabular}{l|l|cccccccc}
		\multicolumn{1}{c}{M} & \multicolumn{1}{c}{A}   & \multicolumn{1}{l}{s/naive} & \multicolumn{1}{c}{s/adj} & \multicolumn{1}{c}{s/W} & \multicolumn{1}{c}{sfe/W} & \multicolumn{1}{c}{ipw/W} & \multicolumn{1}{c}{s/CA } & \multicolumn{1}{c}{sfe/CA} & \multicolumn{1}{c}{ipw/CA} \\ \hline 
		1     & SRS   & 0.048 & 0.048 & 0.052 & 0.045 & 0.047 & 0.034 & 0.040 & 0.044 \\
		& SBR   & 0.001 & 0.007 & 0.002 & 0.039 & 0.040 & 0.044 & 0.038 & 0.037 \\ \hline 
		2     & SRS   & 0.057 & 0.057 & 0.065 & 0.051 & 0.058 & 0.050 & 0.051 & 0.053 \\
		& SBR   & 0.022 & 0.034 & 0.021 & 0.053 & 0.053 & 0.050 & 0.059 & 0.053 \\ \hline 
		3     & SRS   & 0.016 & 0.016 & 0.052 & 0.046 & 0.054 & 0.051 & 0.048 & 0.053 \\
		& SBR   & 0.004 & 0.005 & 0.039 & 0.038 & 0.048 & 0.045 & 0.046 & 0.048 \\ \hline 
		4     & SRS   & 0.009 & 0.009 & 0.046 & 0.037 & 0.049 & 0.046 & 0.045 & 0.051 \\
		& SBR   & 0.004 & 0.005 & 0.036 & 0.016 & 0.052 & 0.049 & 0.043 & 0.046 \\
	\end{tabular}%
	\label{tab:200_2_2_70}%
\end{table}%

\begin{table}[H]
	\centering
	\caption{$H_0$,  $n = 200$, $\tau = 0.75$}
	\begin{tabular}{l|l|cccccccc}
		\multicolumn{1}{c}{M} & \multicolumn{1}{c}{A}   & \multicolumn{1}{l}{s/naive} & \multicolumn{1}{c}{s/adj} & \multicolumn{1}{c}{s/W} & \multicolumn{1}{c}{sfe/W} & \multicolumn{1}{c}{ipw/W} & \multicolumn{1}{c}{s/CA } & \multicolumn{1}{c}{sfe/CA} & \multicolumn{1}{c}{ipw/CA} \\ \hline 
		1     & SRS   & 0.052 & 0.052 & 0.057 & 0.045 & 0.049 & 0.044 & 0.040 & 0.043 \\
		& SBR   & 0.002 & 0.008 & 0.004 & 0.033 & 0.034 & 0.036 & 0.036 & 0.036 \\ \hline
		2     & SRS   & 0.042 & 0.042 & 0.061 & 0.055 & 0.067 & 0.047 & 0.055 & 0.068 \\
		& SBR   & 0.006 & 0.014 & 0.009 & 0.029 & 0.037 & 0.042 & 0.039 & 0.040 \\ \hline
		3     & SRS   & 0.056 & 0.056 & 0.043 & 0.038 & 0.054 & 0.048 & 0.046 & 0.054 \\
		& SBR   & 0.055 & 0.057 & 0.048 & 0.042 & 0.050 & 0.053 & 0.052 & 0.052 \\ \hline
		4     & SRS   & 0.019 & 0.019 & 0.038 & 0.032 & 0.046 & 0.045 & 0.042 & 0.042 \\
		& SBR   & 0.022 & 0.022 & 0.044 & 0.028 & 0.045 & 0.044 & 0.038 & 0.042 \\
	\end{tabular}%
	\label{tab:200_3_70}%
\end{table}

\begin{table}[H]
	\centering
	\caption{$H_0$,  $n = 400$, $\tau = 0.25$}
	\begin{tabular}{l|l|cccccccc}
		\multicolumn{1}{c}{M} & \multicolumn{1}{c}{A}   & \multicolumn{1}{l}{s/naive} & \multicolumn{1}{c}{s/adj} & \multicolumn{1}{c}{s/W} & \multicolumn{1}{c}{sfe/W} & \multicolumn{1}{c}{ipw/W} & \multicolumn{1}{c}{s/CA } & \multicolumn{1}{c}{sfe/CA} & \multicolumn{1}{c}{ipw/CA} \\ \hline 
		1     & SRS   & 0.044 & 0.044 & 0.054 & 0.039 & 0.041 & 0.038 & 0.040 & 0.042 \\
		& SBR   & 0.003 & 0.015 & 0.003 & 0.051 & 0.052 & 0.043 & 0.046 & 0.046 \\ \hline 
		2     & SRS   & 0.034 & 0.034 & 0.057 & 0.058 & 0.054 & 0.062 & 0.058 & 0.053 \\
		& SBR   & 0.031 & 0.034 & 0.040 & 0.044 & 0.049 & 0.051 & 0.051 & 0.051 \\ \hline 
		3     & SRS   & 0.037 & 0.037 & 0.029 & 0.034 & 0.036 & 0.033 & 0.033 & 0.039 \\
		& SBR   & 0.045 & 0.049 & 0.037 & 0.037 & 0.042 & 0.044 & 0.040 & 0.041 \\ \hline 
		4     & SRS   & 0.073 & 0.073 & 0.044 & 0.054 & 0.046 & 0.045 & 0.048 & 0.041 \\
		& SBR   & 0.065 & 0.076 & 0.036 & 0.014 & 0.060 & 0.058 & 0.062 & 0.060 \\
	\end{tabular}%
	\label{tab:400_1_70}%
\end{table}%

\begin{table}[H]
	\centering
	\caption{$H_0$,  $n = 400$, $\tau = 0.5$}
	\begin{tabular}{l|l|cccccccc}
		\multicolumn{1}{c}{M} & \multicolumn{1}{c}{A}   & \multicolumn{1}{l}{s/naive} & \multicolumn{1}{c}{s/adj} & \multicolumn{1}{c}{s/W} & \multicolumn{1}{c}{sfe/W} & \multicolumn{1}{c}{ipw/W} & \multicolumn{1}{c}{s/CA } & \multicolumn{1}{c}{sfe/CA} & \multicolumn{1}{c}{ipw/CA} \\ \hline 
		1     & SRS   & 0.044 & 0.044 & 0.051 & 0.037 & 0.039 & 0.048 & 0.036 & 0.037 \\
		& SBR   & 0.001 & 0.002 & 0.000 & 0.035 & 0.039 & 0.035 & 0.040 & 0.040 \\ \hline 
		2     & SRS   & 0.062 & 0.062 & 0.062 & 0.049 & 0.049 & 0.059 & 0.041 & 0.048 \\
		& SBR   & 0.015 & 0.029 & 0.015 & 0.034 & 0.040 & 0.040 & 0.042 & 0.037 \\ \hline 
		3     & SRS   & 0.007 & 0.007 & 0.039 & 0.036 & 0.042 & 0.042 & 0.042 & 0.047 \\
		& SBR   & 0.006 & 0.006 & 0.035 & 0.037 & 0.036 & 0.037 & 0.041 & 0.037 \\ \hline 
		4     & SRS   & 0.013 & 0.013 & 0.046 & 0.029 & 0.061 & 0.053 & 0.035 & 0.054 \\
		& SBR   & 0.009 & 0.010 & 0.033 & 0.025 & 0.056 & 0.054 & 0.052 & 0.050 \\
	\end{tabular}%
	\label{tab:400_2_2_70}%
\end{table}%

\begin{table}[H]
	\centering
	\caption{$H_0$,  $n = 400$, $\tau = 0.75$}
	\begin{tabular}{l|l|cccccccc}
		\multicolumn{1}{c}{M} & \multicolumn{1}{c}{A}   & \multicolumn{1}{l}{s/naive} & \multicolumn{1}{c}{s/adj} & \multicolumn{1}{c}{s/W} & \multicolumn{1}{c}{sfe/W} & \multicolumn{1}{c}{ipw/W} & \multicolumn{1}{c}{s/CA } & \multicolumn{1}{c}{sfe/CA} & \multicolumn{1}{c}{ipw/CA} \\ \hline 
		1     & SRS   & 0.049 & 0.049 & 0.053 & 0.046 & 0.050 & 0.043 & 0.048 & 0.050 \\
		& SBR   & 0.001 & 0.006 & 0.002 & 0.038 & 0.041 & 0.037 & 0.036 & 0.036 \\ \hline 
		2     & SRS   & 0.050 & 0.050 & 0.065 & 0.050 & 0.049 & 0.056 & 0.052 & 0.052 \\
		& SBR   & 0.010 & 0.019 & 0.015 & 0.041 & 0.048 & 0.042 & 0.041 & 0.041 \\ \hline 
		3     & SRS   & 0.044 & 0.044 & 0.031 & 0.042 & 0.039 & 0.032 & 0.038 & 0.039 \\
		& SBR   & 0.057 & 0.059 & 0.040 & 0.036 & 0.044 & 0.043 & 0.043 & 0.043 \\ \hline 
		4     & SRS   & 0.034 & 0.034 & 0.051 & 0.046 & 0.049 & 0.051 & 0.046 & 0.051 \\
		& SBR   & 0.028 & 0.028 & 0.044 & 0.040 & 0.045 & 0.045 & 0.045 & 0.046 \\
	\end{tabular}%
	\label{tab:400_3_70}%
\end{table}%

\subsection{QTE, $H_1$, $\pi = 0.7$}

\begin{table}[H]
	\centering
	\caption{$H_1$,  $n = 200$, $\tau = 0.25$}
	\begin{tabular}{l|l|cccccccc}
		\multicolumn{1}{c}{M} & \multicolumn{1}{c}{A}   & \multicolumn{1}{l}{s/naive} & \multicolumn{1}{c}{s/adj} & \multicolumn{1}{c}{s/W} & \multicolumn{1}{c}{sfe/W} & \multicolumn{1}{c}{ipw/W} & \multicolumn{1}{c}{s/CA } & \multicolumn{1}{c}{sfe/CA} & \multicolumn{1}{c}{ipw/CA} \\ \hline
		1     & SRS   & 0.152 & 0.152 & 0.176 & 0.359 & 0.313 & 0.187 & 0.343 & 0.339 \\
		& SBR   & 0.065 & 0.186 & 0.100 & 0.346 & 0.336 & 0.357 & 0.341 & 0.338 \\ \hline
		2     & SRS   & 0.314 & 0.314 & 0.334 & 0.361 & 0.325 & 0.347 & 0.367 & 0.365 \\
		& SBR   & 0.309 & 0.334 & 0.336 & 0.355 & 0.368 & 0.383 & 0.375 & 0.376 \\ \hline
		3     & SRS   & 0.704 & 0.704 & 0.671 & 0.665 & 0.626 & 0.685 & 0.663 & 0.691 \\
		& SBR   & 0.697 & 0.716 & 0.663 & 0.671 & 0.669 & 0.702 & 0.686 & 0.688 \\ \hline
		4     & SRS   & 0.136 & 0.136 & 0.097 & 0.094 & 0.129 & 0.106 & 0.093 & 0.122 \\
		& SBR   & 0.116 & 0.127 & 0.081 & 0.050 & 0.103 & 0.107 & 0.105 & 0.106 \\
	\end{tabular}%
	\label{tab:200_70_1'}%
\end{table}%

\begin{table}[H]
	\centering
	\caption{$H_1$,  $n = 200$, $\tau = 0.5$}
	\begin{tabular}{l|l|cccccccc}
		\multicolumn{1}{c}{M} & \multicolumn{1}{c}{A}   & \multicolumn{1}{l}{s/naive} & \multicolumn{1}{c}{s/adj} & \multicolumn{1}{c}{s/W} & \multicolumn{1}{c}{sfe/W} & \multicolumn{1}{c}{ipw/W} & \multicolumn{1}{c}{s/CA } & \multicolumn{1}{c}{sfe/CA} & \multicolumn{1}{c}{ipw/CA} \\ \hline 
		1     & SRS   & 0.170 & 0.170 & 0.172 & 0.411 & 0.425 & 0.167 & 0.407 & 0.406 \\
		& SBR   & 0.043 & 0.212 & 0.060 & 0.445 & 0.455 & 0.457 & 0.435 & 0.434 \\ \hline 
		2     & SRS   & 0.287 & 0.287 & 0.280 & 0.371 & 0.364 & 0.275 & 0.374 & 0.360 \\
		& SBR   & 0.258 & 0.327 & 0.236 & 0.367 & 0.387 & 0.372 & 0.383 & 0.381 \\ \hline 
		3     & SRS   & 0.771 & 0.771 & 0.891 & 0.882 & 0.903 & 0.895 & 0.883 & 0.894 \\
		& SBR   & 0.760 & 0.769 & 0.892 & 0.896 & 0.911 & 0.901 & 0.904 & 0.900 \\ \hline 
		4     & SRS   & 0.145 & 0.145 & 0.265 & 0.218 & 0.305 & 0.264 & 0.241 & 0.301 \\
		& SBR   & 0.128 & 0.136 & 0.235 & 0.177 & 0.288 & 0.290 & 0.284 & 0.287 \\
	\end{tabular}%
	\label{tab:200_70_2'}%
\end{table}%

\begin{table}[H]
	\centering
	\caption{$H_1$,  $n = 200$, $\tau = 0.75$}
	\begin{tabular}{l|l|cccccccc}
		\multicolumn{1}{c}{M} & \multicolumn{1}{c}{A}   & \multicolumn{1}{l}{s/naive} & \multicolumn{1}{c}{s/adj} & \multicolumn{1}{c}{s/W} & \multicolumn{1}{c}{sfe/W} & \multicolumn{1}{c}{ipw/W} & \multicolumn{1}{c}{s/CA } & \multicolumn{1}{c}{sfe/CA} & \multicolumn{1}{c}{ipw/CA} \\ \hline 
		1     & SRS   & 0.181 & 0.181 & 0.183 & 0.342 & 0.340 & 0.188 & 0.340 & 0.338 \\
		& SBR   & 0.072 & 0.175 & 0.076 & 0.353 & 0.364 & 0.342 & 0.357 & 0.357 \\ \hline 
		2     & SRS   & 0.279 & 0.279 & 0.321 & 0.404 & 0.427 & 0.341 & 0.400 & 0.427 \\
		& SBR   & 0.243 & 0.341 & 0.293 & 0.430 & 0.451 & 0.430 & 0.454 & 0.435 \\ \hline 
		3     & SRS   & 0.662 & 0.662 & 0.586 & 0.559 & 0.599 & 0.605 & 0.569 & 0.592 \\
		& SBR   & 0.631 & 0.639 & 0.572 & 0.564 & 0.597 & 0.594 & 0.601 & 0.598 \\ \hline 
		4     & SRS   & 0.150 & 0.150 & 0.201 & 0.164 & 0.199 & 0.208 & 0.189 & 0.211 \\
		& SBR   & 0.143 & 0.145 & 0.193 & 0.166 & 0.206 & 0.206 & 0.208 & 0.205 \\
	\end{tabular}%
	\label{tab:200_70_3'}%
\end{table}%

\begin{table}[H]
	\centering
	\caption{$H_1$,  $n = 400$, $\tau = 0.25$}
	\begin{tabular}{l|l|cccccccc}
		\multicolumn{1}{c}{M} & \multicolumn{1}{c}{A}   & \multicolumn{1}{l}{s/naive} & \multicolumn{1}{c}{s/adj} & \multicolumn{1}{c}{s/W} & \multicolumn{1}{c}{sfe/W} & \multicolumn{1}{c}{ipw/W} & \multicolumn{1}{c}{s/CA } & \multicolumn{1}{c}{sfe/CA} & \multicolumn{1}{c}{ipw/CA} \\ \hline 
		1     & SRS   & 0.181 & 0.181 & 0.192 & 0.351 & 0.354 & 0.202 & 0.346 & 0.351 \\
		& SBR   & 0.083 & 0.233 & 0.113 & 0.392 & 0.392 & 0.407 & 0.394 & 0.392 \\ \hline 
		2     & SRS   & 0.362 & 0.362 & 0.406 & 0.403 & 0.415 & 0.408 & 0.415 & 0.424 \\
		& SBR   & 0.350 & 0.381 & 0.388 & 0.412 & 0.426 & 0.426 & 0.422 & 0.419 \\ \hline 
		3     & SRS   & 0.781 & 0.781 & 0.743 & 0.751 & 0.758 & 0.746 & 0.750 & 0.759 \\
		& SBR   & 0.791 & 0.797 & 0.752 & 0.765 & 0.777 & 0.781 & 0.778 & 0.779 \\ \hline 
		4     & SRS   & 0.160 & 0.160 & 0.082 & 0.072 & 0.112 & 0.097 & 0.095 & 0.116 \\
		& SBR   & 0.133 & 0.154 & 0.091 & 0.044 & 0.119 & 0.119 & 0.121 & 0.120 \\
	\end{tabular}%
	\label{tab:400_70_1'}%
\end{table}%

\begin{table}[H]
	\centering
	\caption{$H_1$,  $n = 400$, $\tau = 0.5$}
	\begin{tabular}{l|l|cccccccc}
		\multicolumn{1}{c}{M} & \multicolumn{1}{c}{A}   & \multicolumn{1}{l}{s/naive} & \multicolumn{1}{c}{s/adj} & \multicolumn{1}{c}{s/W} & \multicolumn{1}{c}{sfe/W} & \multicolumn{1}{c}{ipw/W} & \multicolumn{1}{c}{s/CA } & \multicolumn{1}{c}{sfe/CA} & \multicolumn{1}{c}{ipw/CA} \\ \hline 
		1     & SRS   & 0.184 & 0.184 & 0.187 & 0.468 & 0.479 & 0.194 & 0.460 & 0.466 \\
		& SBR   & 0.042 & 0.220 & 0.059 & 0.486 & 0.498 & 0.505 & 0.480 & 0.482 \\ \hline 
		2     & SRS   & 0.322 & 0.322 & 0.298 & 0.405 & 0.404 & 0.303 & 0.412 & 0.400 \\
		& SBR   & 0.262 & 0.342 & 0.237 & 0.376 & 0.399 & 0.385 & 0.389 & 0.389 \\ \hline 
		3     & SRS   & 0.867 & 0.867 & 0.939 & 0.930 & 0.933 & 0.941 & 0.932 & 0.936 \\
		& SBR   & 0.883 & 0.888 & 0.948 & 0.952 & 0.952 & 0.955 & 0.952 & 0.952 \\ \hline 
		4     & SRS   & 0.209 & 0.209 & 0.327 & 0.275 & 0.354 & 0.341 & 0.308 & 0.351 \\
		& SBR   & 0.194 & 0.217 & 0.310 & 0.256 & 0.365 & 0.364 & 0.359 & 0.356 \\
	\end{tabular}%
	\label{tab:400_70_2'}%
\end{table}%

\begin{table}[H]
	\centering
	\caption{$H_1$,  $n = 400$, $\tau = 0.75$}
	\begin{tabular}{l|l|cccccccc}
		\multicolumn{1}{c}{M} & \multicolumn{1}{c}{A}   & \multicolumn{1}{l}{s/naive} & \multicolumn{1}{c}{s/adj} & \multicolumn{1}{c}{s/W} & \multicolumn{1}{c}{sfe/W} & \multicolumn{1}{c}{ipw/W} & \multicolumn{1}{c}{s/CA } & \multicolumn{1}{c}{sfe/CA} & \multicolumn{1}{c}{ipw/CA} \\ \hline 
		1     & SRS   & 0.217 & 0.217 & 0.224 & 0.411 & 0.409 & 0.219 & 0.411 & 0.408 \\
		& SBR   & 0.103 & 0.246 & 0.107 & 0.419 & 0.418 & 0.400 & 0.421 & 0.420 \\ \hline 
		2     & SRS   & 0.335 & 0.335 & 0.378 & 0.485 & 0.505 & 0.384 & 0.468 & 0.501 \\
		& SBR   & 0.278 & 0.384 & 0.329 & 0.479 & 0.500 & 0.487 & 0.504 & 0.493 \\ \hline 
		3     & SRS   & 0.708 & 0.708 & 0.661 & 0.628 & 0.665 & 0.665 & 0.629 & 0.672 \\
		& SBR   & 0.705 & 0.706 & 0.652 & 0.631 & 0.665 & 0.673 & 0.672 & 0.673 \\ \hline 
		4     & SRS   & 0.205 & 0.205 & 0.226 & 0.221 & 0.245 & 0.234 & 0.234 & 0.240 \\
		& SBR   & 0.205 & 0.205 & 0.249 & 0.209 & 0.248 & 0.258 & 0.256 & 0.258 \\
	\end{tabular}%
	\label{tab:400_70_3'}%
\end{table}%

\subsection{ATE, $\pi = 0.5$}
\begin{table}[H]
	\centering
	\caption{$H_0$, $n=200$, $\pi = 0.5$}
	\begin{tabular}{c|c|ccccccccc}
		\multicolumn{1}{c}{M} & \multicolumn{1}{c}{A}      & \multicolumn{1}{c}{s/naive} & \multicolumn{1}{c}{s/adj} & \multicolumn{1}{c}{sfe/adj} & \multicolumn{1}{c}{s/W} & \multicolumn{1}{c}{sfe/W} & \multicolumn{1}{c}{ipw/W} & \multicolumn{1}{c}{s/CA } & \multicolumn{1}{c}{sfe/CA} & \multicolumn{1}{c}{ipw/CA} \\ \hline
		1     & SRS   & 0.059 & 0.057 & 0.051 & 0.061 & 0.055 & 0.057 & 0.053 & 0.048 & 0.049 \\
		& WEI   & 0.006 & 0.048 & 0.062 & 0.004 & 0.068 & 0.068 & 0.051 & 0.065 & 0.065 \\
		& BCD   & 0.001 & 0.089 & 0.056 & 0.000 & 0.058 & 0.058 & 0.071 & 0.056 & 0.056 \\
		& SBR   & 0.000 & 0.067 & 0.061 & 0.000 & 0.064 & 0.064 & 0.059 & 0.061 & 0.061 \\ \hline
		2     & SRS   & 0.062 & 0.061 & 0.061 & 0.061 & 0.059 & 0.062 & 0.060 & 0.057 & 0.059 \\
		& WEI   & 0.027 & 0.060 & 0.050 & 0.029 & 0.046 & 0.054 & 0.057 & 0.052 & 0.053 \\
		& BCD   & 0.014 & 0.058 & 0.053 & 0.016 & 0.053 & 0.052 & 0.052 & 0.052 & 0.049 \\
		& SBR   & 0.006 & 0.045 & 0.044 & 0.006 & 0.045 & 0.045 & 0.045 & 0.045 & 0.045 \\ \hline
		3     & SRS   & 0.057 & 0.056 & 0.068 & 0.055 & 0.061 & 0.061 & 0.056 & 0.064 & 0.065 \\
		& WEI   & 0.049 & 0.050 & 0.057 & 0.052 & 0.057 & 0.056 & 0.048 & 0.053 & 0.053 \\
		& BCD   & 0.057 & 0.058 & 0.057 & 0.057 & 0.063 & 0.063 & 0.057 & 0.056 & 0.057 \\
		& SBR   & 0.055 & 0.058 & 0.056 & 0.057 & 0.060 & 0.061 & 0.055 & 0.055 & 0.055 \\ \hline
		4     & SRS   & 0.066 & 0.067 & 0.077 & 0.068 & 0.069 & 0.063 & 0.063 & 0.070 & 0.063 \\
		& WEI   & 0.065 & 0.067 & 0.070 & 0.066 & 0.067 & 0.068 & 0.069 & 0.067 & 0.070 \\
		& BCD   & 0.068 & 0.068 & 0.067 & 0.065 & 0.061 & 0.068 & 0.065 & 0.065 & 0.065 \\
		& SBR   & 0.055 & 0.055 & 0.055 & 0.057 & 0.057 & 0.058 & 0.057 & 0.057 & 0.057 \\
	\end{tabular}%
	\label{tab:ate_200_50}%
\end{table}%

\begin{table}[H]
	\centering
	\caption{$H_1$, $n=200$, $\pi = 0.5$}
	\begin{tabular}{c|c|ccccccccc}
		\multicolumn{1}{c}{M} & \multicolumn{1}{c}{A}      & \multicolumn{1}{c}{s/naive} & \multicolumn{1}{c}{s/adj} & \multicolumn{1}{c}{sfe/adj} & \multicolumn{1}{c}{s/W} & \multicolumn{1}{c}{sfe/W} & \multicolumn{1}{c}{ipw/W} & \multicolumn{1}{c}{s/CA } & \multicolumn{1}{c}{sfe/CA} & \multicolumn{1}{c}{ipw/CA} \\ \hline
		1     & SRS   & 0.062 & 0.058 & 0.047 & 0.064 & 0.054 & 0.054 & 0.062 & 0.048 & 0.047 \\
		& WEI   & 0.006 & 0.056 & 0.049 & 0.007 & 0.052 & 0.052 & 0.062 & 0.050 & 0.051 \\
		& BCD   & 0.000 & 0.063 & 0.054 & 0.000 & 0.055 & 0.055 & 0.044 & 0.054 & 0.054 \\
		& SBR   & 0.000 & 0.053 & 0.054 & 0.000 & 0.056 & 0.057 & 0.054 & 0.055 & 0.055 \\
		2     & SRS   & 0.049 & 0.048 & 0.039 & 0.048 & 0.043 & 0.046 & 0.049 & 0.042 & 0.045 \\
		& WEI   & 0.018 & 0.050 & 0.049 & 0.018 & 0.047 & 0.048 & 0.051 & 0.044 & 0.043 \\
		& BCD   & 0.012 & 0.050 & 0.043 & 0.012 & 0.043 & 0.045 & 0.046 & 0.042 & 0.042 \\
		& SBR   & 0.008 & 0.049 & 0.049 & 0.009 & 0.047 & 0.047 & 0.050 & 0.049 & 0.049 \\
		3     & SRS   & 0.048 & 0.050 & 0.051 & 0.049 & 0.050 & 0.049 & 0.054 & 0.050 & 0.053 \\
		& WEI   & 0.049 & 0.049 & 0.050 & 0.047 & 0.047 & 0.047 & 0.051 & 0.051 & 0.051 \\
		& BCD   & 0.045 & 0.049 & 0.048 & 0.048 & 0.050 & 0.046 & 0.049 & 0.049 & 0.049 \\
		& SBR   & 0.051 & 0.053 & 0.052 & 0.050 & 0.055 & 0.054 & 0.057 & 0.057 & 0.057 \\
		4     & SRS   & 0.057 & 0.056 & 0.058 & 0.053 & 0.056 & 0.056 & 0.055 & 0.055 & 0.055 \\
		& WEI   & 0.059 & 0.059 & 0.061 & 0.056 & 0.057 & 0.061 & 0.058 & 0.063 & 0.060 \\
		& BCD   & 0.050 & 0.050 & 0.051 & 0.054 & 0.053 & 0.053 & 0.052 & 0.052 & 0.052 \\
		& SBR   & 0.058 & 0.058 & 0.059 & 0.058 & 0.055 & 0.059 & 0.058 & 0.057 & 0.057 \\
	\end{tabular}%
	\label{tab:ate_200_50'}%
\end{table}%

\begin{table}[H]
	\centering
	\caption{$H_0$, $n=400$, $\pi = 0.5$}
	\begin{tabular}{c|c|ccccccccc}
		\multicolumn{1}{c}{M} & \multicolumn{1}{c}{A}      & \multicolumn{1}{c}{s/naive} & \multicolumn{1}{c}{s/adj} & \multicolumn{1}{c}{sfe/adj} & \multicolumn{1}{c}{s/W} & \multicolumn{1}{c}{sfe/W} & \multicolumn{1}{c}{ipw/W} & \multicolumn{1}{c}{s/CA } & \multicolumn{1}{c}{sfe/CA} & \multicolumn{1}{c}{ipw/CA} \\ \hline
		1     & SRS   & 0.063 & 0.061 & 0.042 & 0.063 & 0.043 & 0.045 & 0.055 & 0.042 & 0.042 \\
		& WEI   & 0.005 & 0.050 & 0.050 & 0.006 & 0.052 & 0.052 & 0.052 & 0.050 & 0.050 \\
		& BCD   & 0.000 & 0.067 & 0.052 & 0.000 & 0.059 & 0.059 & 0.051 & 0.059 & 0.059 \\
		& SBR   & 0.000 & 0.059 & 0.058 & 0.000 & 0.057 & 0.057 & 0.063 & 0.060 & 0.060 \\  \hline
		2     & SRS   & 0.061 & 0.057 & 0.055 & 0.058 & 0.055 & 0.054 & 0.061 & 0.054 & 0.051 \\
		& WEI   & 0.018 & 0.051 & 0.064 & 0.019 & 0.063 & 0.064 & 0.052 & 0.064 & 0.064 \\
		& BCD   & 0.009 & 0.045 & 0.046 & 0.006 & 0.046 & 0.047 & 0.043 & 0.049 & 0.049 \\
		& SBR   & 0.014 & 0.062 & 0.060 & 0.016 & 0.065 & 0.065 & 0.063 & 0.063 & 0.063 \\  \hline
		3     & SRS   & 0.050 & 0.049 & 0.050 & 0.050 & 0.049 & 0.051 & 0.052 & 0.048 & 0.048 \\
		& WEI   & 0.046 & 0.047 & 0.049 & 0.047 & 0.046 & 0.047 & 0.048 & 0.047 & 0.046 \\
		& BCD   & 0.049 & 0.049 & 0.049 & 0.049 & 0.050 & 0.050 & 0.050 & 0.050 & 0.050 \\
		& SBR   & 0.055 & 0.056 & 0.056 & 0.059 & 0.058 & 0.059 & 0.055 & 0.056 & 0.056 \\  \hline
		4     & SRS   & 0.057 & 0.057 & 0.055 & 0.056 & 0.056 & 0.059 & 0.054 & 0.051 & 0.056 \\
		& WEI   & 0.051 & 0.051 & 0.053 & 0.052 & 0.054 & 0.054 & 0.051 & 0.051 & 0.052 \\
		& BCD   & 0.056 & 0.056 & 0.056 & 0.054 & 0.056 & 0.056 & 0.054 & 0.053 & 0.053 \\
		& SBR   & 0.056 & 0.058 & 0.058 & 0.055 & 0.056 & 0.057 & 0.057 & 0.057 & 0.057 \\
	\end{tabular}%
	\label{tab:ate_400_50}%
\end{table}%

\begin{table}[H]
	\centering
	\caption{$H_1$, $n=400$, $\pi = 0.5$}
	\begin{tabular}{c|c|ccccccccc}
		\multicolumn{1}{c}{M} & \multicolumn{1}{c}{A}      & \multicolumn{1}{c}{s/naive} & \multicolumn{1}{c}{s/adj} & \multicolumn{1}{c}{sfe/adj} & \multicolumn{1}{c}{s/W} & \multicolumn{1}{c}{sfe/W} & \multicolumn{1}{c}{ipw/W} & \multicolumn{1}{c}{s/CA } & \multicolumn{1}{c}{sfe/CA} & \multicolumn{1}{c}{ipw/CA} \\ \hline
		1     & SRS   & 0.422 & 0.422 & 0.964 & 0.416 & 0.968 & 0.966 & 0.415 & 0.964 & 0.962 \\
		& WEI   & 0.387 & 0.732 & 0.969 & 0.393 & 0.969 & 0.969 & 0.732 & 0.967 & 0.968 \\
		& BCD   & 0.341 & 0.962 & 0.971 & 0.350 & 0.969 & 0.968 & 0.955 & 0.968 & 0.968 \\
		& SBR   & 0.357 & 0.967 & 0.967 & 0.368 & 0.966 & 0.966 & 0.967 & 0.965 & 0.965 \\ \hline
		2     & SRS   & 0.572 & 0.568 & 0.806 & 0.579 & 0.795 & 0.805 & 0.568 & 0.796 & 0.805 \\
		& WEI   & 0.577 & 0.723 & 0.813 & 0.575 & 0.814 & 0.810 & 0.728 & 0.811 & 0.808 \\
		& BCD   & 0.606 & 0.809 & 0.813 & 0.618 & 0.817 & 0.821 & 0.802 & 0.810 & 0.810 \\
		& SBR   & 0.601 & 0.828 & 0.829 & 0.603 & 0.832 & 0.836 & 0.830 & 0.834 & 0.834 \\ \hline
		3     & SRS   & 0.804 & 0.801 & 0.803 & 0.798 & 0.798 & 0.799 & 0.804 & 0.803 & 0.803 \\
		& WEI   & 0.804 & 0.804 & 0.806 & 0.802 & 0.800 & 0.803 & 0.803 & 0.803 & 0.803 \\
		& BCD   & 0.816 & 0.818 & 0.820 & 0.822 & 0.825 & 0.825 & 0.819 & 0.819 & 0.819 \\
		& SBR   & 0.821 & 0.823 & 0.823 & 0.816 & 0.820 & 0.819 & 0.822 & 0.822 & 0.822 \\ \hline
		4     & SRS   & 0.228 & 0.230 & 0.229 & 0.225 & 0.227 & 0.228 & 0.234 & 0.226 & 0.226 \\
		& WEI   & 0.229 & 0.230 & 0.230 & 0.225 & 0.223 & 0.228 & 0.233 & 0.235 & 0.234 \\
		& BCD   & 0.221 & 0.224 & 0.225 & 0.227 & 0.225 & 0.231 & 0.231 & 0.231 & 0.233 \\
		& SBR   & 0.224 & 0.226 & 0.225 & 0.224 & 0.225 & 0.230 & 0.235 & 0.235 & 0.235 \\
	\end{tabular}%
	\label{tab:ate_400_50'}%
\end{table}%

\subsection{ATE, $\pi = 0.7$}
\begin{table}[H]
	\centering
	\caption{$H_0$, $n=200$, $\pi = 0.7$}
	\begin{tabular}{c|c|ccccccccc}
		\multicolumn{1}{c}{M} & \multicolumn{1}{c}{A}      & \multicolumn{1}{c}{s/naive} & \multicolumn{1}{c}{s/adj} & \multicolumn{1}{c}{sfe/adj} & \multicolumn{1}{c}{s/W} & \multicolumn{1}{c}{sfe/W} & \multicolumn{1}{c}{ipw/W} & \multicolumn{1}{c}{s/CA } & \multicolumn{1}{c}{sfe/CA} & \multicolumn{1}{c}{ipw/CA} \\ \hline
		1     & SRS   & 0.050 & 0.045 & 0.056 & 0.051 & 0.056 & 0.062 & 0.046 & 0.054 & 0.055 \\
		& SBR   & 0.000 & 0.004 & 0.051 & 0.000 & 0.061 & 0.064 & 0.064 & 0.060 & 0.059 \\ \hline
		2     & SRS   & 0.048 & 0.055 & 0.074 & 0.055 & 0.049 & 0.056 & 0.045 & 0.049 & 0.057 \\
		& SBR   & 0.013 & 0.030 & 0.041 & 0.013 & 0.024 & 0.051 & 0.056 & 0.049 & 0.051 \\ \hline
		3     & SRS   & 0.059 & 0.060 & 0.066 & 0.060 & 0.060 & 0.064 & 0.058 & 0.055 & 0.064 \\
		& SBR   & 0.051 & 0.053 & 0.052 & 0.053 & 0.045 & 0.057 & 0.056 & 0.056 & 0.055 \\ \hline
		4     & SRS   & 0.057 & 0.057 & 0.056 & 0.058 & 0.056 & 0.068 & 0.054 & 0.057 & 0.058 \\
		& SBR   & 0.047 & 0.050 & 0.044 & 0.051 & 0.037 & 0.054 & 0.054 & 0.055 & 0.055 \\
	\end{tabular}%
	\label{tab:ate_200_70}%
\end{table}%

\begin{table}[H]
	\centering
	\caption{$H_1$, $n=200$, $\pi = 0.7$}
	\begin{tabular}{c|c|ccccccccc}
		\multicolumn{1}{c}{M} & \multicolumn{1}{c}{A}      & \multicolumn{1}{c}{s/naive} & \multicolumn{1}{c}{s/adj} & \multicolumn{1}{c}{sfe/adj} & \multicolumn{1}{c}{s/W} & \multicolumn{1}{c}{sfe/W} & \multicolumn{1}{c}{ipw/W} & \multicolumn{1}{c}{s/CA } & \multicolumn{1}{c}{sfe/CA} & \multicolumn{1}{c}{ipw/CA} \\ \hline
		1     & SRS   & 0.329 & 0.328 & 0.934 & 0.336 & 0.943 & 0.946 & 0.326 & 0.941 & 0.941 \\
		& SBR   & 0.220 & 0.631 & 0.938 & 0.233 & 0.946 & 0.949 & 0.932 & 0.943 & 0.943 \\ \hline
		2     & SRS   & 0.581 & 0.578 & 0.687 & 0.582 & 0.619 & 0.756 & 0.571 & 0.601 & 0.758 \\
		& SBR   & 0.598 & 0.699 & 0.747 & 0.599 & 0.686 & 0.768 & 0.752 & 0.766 & 0.764 \\ \hline
		3     & SRS   & 0.773 & 0.779 & 0.758 & 0.769 & 0.741 & 0.784 & 0.773 & 0.729 & 0.782 \\
		& SBR   & 0.771 & 0.773 & 0.772 & 0.777 & 0.763 & 0.782 & 0.782 & 0.780 & 0.781 \\ \hline
		4     & SRS   & 0.149 & 0.154 & 0.121 & 0.153 & 0.140 & 0.168 & 0.154 & 0.141 & 0.165 \\
		& SBR   & 0.144 & 0.151 & 0.129 & 0.153 & 0.118 & 0.175 & 0.172 & 0.170 & 0.169 \\
	\end{tabular}%
	\label{tab:ate_200_70'}%
\end{table}

\begin{table}[H]
	\centering
	\caption{$H_0$, $n=400$, $\pi = 0.7$}
	\begin{tabular}{c|c|ccccccccc}
		\multicolumn{1}{c}{M} & \multicolumn{1}{c}{A}      & \multicolumn{1}{c}{s/naive} & \multicolumn{1}{c}{s/adj} & \multicolumn{1}{c}{sfe/adj} & \multicolumn{1}{c}{s/W} & \multicolumn{1}{c}{sfe/W} & \multicolumn{1}{c}{ipw/W} & \multicolumn{1}{c}{s/CA } & \multicolumn{1}{c}{sfe/CA} & \multicolumn{1}{c}{ipw/CA} \\ \hline
		1     & SRS   & 0.062 & 0.059 & 0.065 & 0.061 & 0.056 & 0.056 & 0.062 & 0.060 & 0.061 \\
		& SBR   & 0.000 & 0.000 & 0.034 & 0.000 & 0.039 & 0.040 & 0.045 & 0.045 & 0.044 \\ \hline
		2     & SRS   & 0.052 & 0.050 & 0.087 & 0.054 & 0.055 & 0.052 & 0.050 & 0.057 & 0.051 \\
		& SBR   & 0.013 & 0.029 & 0.040 & 0.012 & 0.027 & 0.044 & 0.042 & 0.044 & 0.042 \\ \hline
		3     & SRS   & 0.042 & 0.041 & 0.049 & 0.045 & 0.043 & 0.052 & 0.040 & 0.040 & 0.046 \\
		& SBR   & 0.028 & 0.028 & 0.031 & 0.029 & 0.025 & 0.032 & 0.035 & 0.036 & 0.034 \\ \hline
		4     & SRS   & 0.053 & 0.055 & 0.043 & 0.058 & 0.053 & 0.058 & 0.055 & 0.050 & 0.056 \\
		& SBR   & 0.050 & 0.051 & 0.043 & 0.051 & 0.035 & 0.054 & 0.055 & 0.055 & 0.053 \\
	\end{tabular}%
	\label{tab:ate_400_70}%
\end{table}%

\begin{table}[H]
	\centering
	\caption{$H_1$, $n=400$, $\pi = 0.7$}
	\begin{tabular}{c|c|ccccccccc}
		\multicolumn{1}{c}{M} & \multicolumn{1}{c}{A}      & \multicolumn{1}{c}{s/naive} & \multicolumn{1}{c}{s/adj} & \multicolumn{1}{c}{sfe/adj} & \multicolumn{1}{c}{s/W} & \multicolumn{1}{c}{sfe/W} & \multicolumn{1}{c}{ipw/W} & \multicolumn{1}{c}{s/CA } & \multicolumn{1}{c}{sfe/CA} & \multicolumn{1}{c}{ipw/CA} \\ \hline
		1     & SRS   & 0.384 & 0.380 & 0.972 & 0.381 & 0.971 & 0.976 & 0.382 & 0.970 & 0.973 \\
		& SBR   & 0.250 & 0.736 & 0.970 & 0.254 & 0.972 & 0.972 & 0.967 & 0.973 & 0.974 \\ \hline
		2     & SRS   & 0.616 & 0.628 & 0.753 & 0.622 & 0.693 & 0.796 & 0.617 & 0.690 & 0.795 \\
		& SBR   & 0.659 & 0.759 & 0.806 & 0.665 & 0.740 & 0.827 & 0.817 & 0.827 & 0.827 \\ \hline
		3     & SRS   & 0.818 & 0.817 & 0.805 & 0.812 & 0.793 & 0.821 & 0.816 & 0.793 & 0.829 \\
		& SBR   & 0.833 & 0.838 & 0.836 & 0.831 & 0.824 & 0.840 & 0.838 & 0.839 & 0.837 \\ \hline
		4     & SRS   & 0.177 & 0.172 & 0.145 & 0.180 & 0.162 & 0.195 & 0.181 & 0.171 & 0.186 \\
		& SBR   & 0.181 & 0.190 & 0.164 & 0.184 & 0.142 & 0.202 & 0.202 & 0.202 & 0.200 \\
	\end{tabular}%
	\label{tab:ate_400_70'}%
\end{table}%
	
	\bibliographystyle{chicago}
	\bibliography{BCAR}
	
\end{document}